\documentclass[a4paper,9.5pt]{article}
\usepackage{etoolbox}
\newtoggle{longversion}
\toggletrue{longversion}

\usepackage[dvipsnames]{xcolor}

\usepackage{url}
\usepackage{microtype}

\usepackage{hyperref}
\hypersetup{
  colorlinks   = true, 
  linkcolor    = black!50!brown, 
  citecolor   = black!50!brown, 
}
\usepackage[normalem]{ulem}

\usepackage{amsmath}
\usepackage{amssymb}
\usepackage{amsthm}
\usepackage{verbatim}
\usepackage{bm}
\usepackage{color,graphicx,xcolor}
\usepackage{mdwtab}
\usepackage{mathtools,tikz}
\usepackage{hhline}
\usepackage{multirow}
\usepackage{pdfpages}
\usepackage{enumitem}  
\usepackage{caption}
\usepackage{subcaption}
\DeclareCaptionFormat{myformat}{\fontsize{8}{9}\selectfont#1#2#3}
\usepackage{appendix}
\usepackage{authblk}
 \usepackage{circledsteps}
\iftoggle{longversion}{
\linespread{0.995}
}
{
\setlength{\abovedisplayskip}{0.3em}
\setlength{\belowdisplayskip}{0.3pt}
\setlength{\abovedisplayshortskip}{0pt}
\setlength{\belowdisplayshortskip}{0pt}
\twocolumn
\newcommand{\longonly}[1]{}
\newcommand{\shortonly}[1]{#1}

}

\iftoggle{longversion}{
\onecolumn
\newcommand{\longonly}[1]{#1} 
\newcommand{\shortonly}[1]{} 
}
{
\setlength{\abovedisplayskip}{0.3em}
\setlength{\belowdisplayskip}{0.3pt}
\setlength{\abovedisplayshortskip}{0pt}
\setlength{\belowdisplayshortskip}{0pt}
\twocolumn
\newcommand{\longonly}[1]{}
\newcommand{\shortonly}[1]{#1}

}

\usepackage{cite}
\usepackage{amsmath,amssymb,amsfonts,amsthm,xspace}
\usepackage{algorithmic}
\usepackage{graphicx}
\usepackage{textcomp}

\usetikzlibrary{arrows,shapes.geometric}
\allowdisplaybreaks

\theoremstyle{definition}
\newtheorem{example}{Example}
\newtheorem{remark}{Remark}
\newtheorem{thm}{Theorem}
\newtheorem{defn}{Definition}

\newtheorem*{defn*}{Definition}
\newtheorem*{lemma*}{Lemma}

\newtheorem{lemma}[thm]{Lemma}

\newtheoremstyle{thmnum}{\topsep}{\topsep}{\itshape}{0pt}{\bfseries}{.}{ }{\thmname{#1}\thmnote{ \bfseries #3}}
\theoremstyle{thmnum}
\newtheorem{duplicatelemma}{Lemma}
\newtheorem{duplicatedefn}{Definition}
\usepackage{pgffor}
\foreach \x in {a,...,z}{%
\expandafter\xdef\csname vec\x \endcsname{\noexpand\ensuremath{\noexpand\bm{\x}}}
}

\foreach \x in {A,...,Z}{%
\expandafter\xdef\csname vec\x \endcsname{\noexpand\ensuremath{\noexpand\bm{\x}}}
}

\foreach \x in {A,...,Z}{%
\expandafter\xdef\csname c\x \endcsname{\noexpand\ensuremath{\noexpand\mathcal{\x}}}
}

\foreach \x in {A,...,Z}{%
\expandafter\xdef\csname bb\x \endcsname{\noexpand\ensuremath{\noexpand\mathbb{\x}}}
}

\newcommand{\Prob}{\ensuremath{{\mathbb P}}}

\newcommand{\defineqq}{\ensuremath{\stackrel{\textup{\tiny def}}{=}}}

\def\tlx{\ensuremath{\tilde{x}}}
\def\tla{\ensuremath{\tilde{a}}}
\def\tlb{\ensuremath{\tilde{b}}}
\def\tlc{\ensuremath{\tilde{c}}}

\usetikzlibrary{shapes,arrows,positioning,decorations.pathreplacing,calc}

\def\msg{\ensuremath{m}} 
\def\Msg{\ensuremath{M}} 
\def\Msgh{\ensuremath{\hat{\Msg}}} 
\def\nummsg{\mbox{$N$}} 

\def\deterministic{\ensuremath{\mathrm{deterministic}}}
\def\random{\ensuremath{\mathrm{random}}}
\def\strong{\ensuremath{\mathrm{strong}}}
\def\weak{\ensuremath{\mathrm{weak}}}

\newcommand{\bmac}{\ensuremath{\text{byzantine-MAC}}\xspace}
\newcommand{\bmacs}{\ensuremath{\text{byzantine-MACs}}\xspace}

\newcommand{\mach}{\ensuremath{W_{Y|X_1X_2X_3}}\xspace}
\newcommand{\machijk}{\ensuremath{W_{Y|X_iX_jX_k}}\xspace}

\newcommand{\indep}{\raisebox{0.05em}{\rotatebox[origin=c]{90}{$\models$}}\xspace}

\newcommand{\inb}[1]{\left\{#1\right\}}
\newcommand{\inp}[1]{\left(#1\right)}
\newcommand{\insq}[1]{\left[#1\right]}
\newcommand{\dm}[2]{\ensuremath{\cD^{(#1)}_{#2}}\xspace}
\newcommand{\set}[1]{\left\{#1\right\}}

\newcommand{\blue}{\textcolor{black}}
\newcommand{\pink}{\textcolor{black}}

\definecolor{bluishgreen}{RGB}{0,158,115}
\definecolor{vermillion}{RGB}{213,94,0}
\definecolor{myblue}{RGB}{0,114,200}
\definecolor{myorange}{RGB}{230,159,0}
\definecolor{reddishPurple}{RGB}{204,121,167}

\def\BibTeX{{\rm B\kern-.05em{\sc i\kern-.025em b}\kern-.08em
    T\kern-.1667em\lower.7ex\hbox{E}\kern-.125emX}}
\begin{document}

\title{Byzantine Multiple Access Channels --- Part I: Reliable Communication\thanks{This work was presented in part at the 2019 IEEE Information Theory Workshop \cite{8989065}.\\
N. Sangwan and V. M. Prabhakaran were supported by DAE under project no. RTI4001. N. Sangwan was additionally supported by the TCS Foundation through the TCS Research Scholar Program.  The work of M. Bakshi was supported in part by the Research Grants Council of the Hong Kong Special Administrative Region, China, under Grant GRF 14300617, and in part by the National Science Foundation under Grant No. CCF-2107526. The work
of B. K. Dey was supported by Bharti Centre for
Communication in IIT Bombay. V. M. Prabhakaran was additionally supported by SERB through project MTR/2020/000308.}}

\author[1]{Neha Sangwan}
\author[2]{Mayank Bakshi}
\author[3]{Bikash Kumar Dey}
\author[1]{Vinod M. Prabhakaran}
\affil[1]{Tata Institute of Fundamental Research, Mumbai, India}
\affil[2]{Arizona State University, Tempe, AZ, USA}
\affil[3]{Indian Institute of Technology Bombay, Mumbai, India}

\makeatletter
\patchcmd{\@maketitle}
  {\addvspace{0.5\baselineskip}\egroup}
  {\addvspace{-1.5\baselineskip}\egroup}
  {}
  {}
\makeatother

\maketitle

\begin{abstract}
We study communication over a Multiple Access Channel (MAC) where users can possibly be adversarial. The receiver is unaware of the identity of the adversarial users (if any). When all users are non-adversarial, we want their messages to be decoded reliably. When a user behaves adversarially, we require that the honest users' messages be decoded reliably. An adversarial user can mount an attack by sending any input into the channel rather than following the protocol. It turns out that the $2$-user MAC capacity region follows from the point-to-point Arbitrarily Varying Channel (AVC) capacity. For the $3$-user MAC in which at most one user may be malicious, we characterize the capacity region for  deterministic codes and randomized codes (where each user shares an independent random secret key with the receiver). These results are then generalized for the $k$-user MAC where the adversary may control all users in one out of a collection of given subsets. 
\end{abstract}

\section{Introduction} 
\subsection{Motivation and setup}\label{sec:motivation_and_setup}

Communication systems such as the wireless Internet-of-Things (IoTs), which consist of devices of varying security levels connected over a wireless network, pose new security challenges \cite{8232533,ammar2018internet}. Since, the devices share the same communication medium, a malicious\footnote{We use `malicious' and `adversarial' interchangeably.} device may attempt to cause decoding errors for the honest device(s). This motivates the present problem. We study the uplink of a communication network in which users may behave maliciously.

Consider a Multiple Access Channel (MAC) with users who are potentially adversarial. An adversarial user may not follow the protocol and may choose its channel input maliciously to disrupt the communication of the other users. The receiver is unaware of whether any of the users is adversarial and the identity of the adversarial user(s) (if any). We call such a channel a ``\bmac''. 
If all users are non-adversarial ({\em i.e.}, honest), we require that their messages be reliably decoded. 
However, if some of the users are adversarial, the decoder must correctly recover the messages of all the other (honest) users. \pink{Adversarial users have no side information about the messages of the honest users.}
We call this the problem of {\em reliable communication in a \bmac}.

\subsection{Related works}\label{sec:related_works}
The present model is different from other well-studied models involving \pink{non-byzantine users and adversaries}, both {\em passive} and {\em active}. 
\pink{In all such models, the adversary is different from all the legitimate communicating parties and its identity is known to all parties.}

For example, a wiretap channel \cite{Wyner75} has a passive eavesdropper who gets a noisy version of the communication between the sender and the receiver. The goal is to ensure reliable and private (from the eavesdropper) communication from the sender to the receiver. On the other hand, in Arbitrarily Varying Channels (AVC) \cite{BlackwellBTAMS60,survey}
the  adversary is active and controls the channel. The adversary can change the channel law for each channel use with the goal of jamming the communication between the sender and the receiver. 
Arbitrarily Varying Multiple Access Channels (AV-MAC) \cite{BeemerCNS20,Jahn81,Gubner90,AhlswedeC99,PeregS19,WieseB13}, which consider a Multiple Access Channel (MAC) where the channel law is controlled by an adversary, have also been studied. 
Jahn~\cite{Jahn81} obtained the randomized coding capacity region where each user has independent randomness shared with the receiver. He also showed that this region is also the deterministic coding capacity region under average probability of error whenever the
latter has a non-empty interior, a result along the lines of Ahlswede’s dichotomy for the AVC~\cite{Ahlswede78}. Gubner~\cite{Gubner90} proved necessary conditions (non-symmetrizability conditions) for the deterministic coding capacity region to be non-empty. Ahlswede and Cai~\cite{AhlswedeC99} showed that Gubner’s necessary conditions are also sufficient for the deterministic coding capacity region to have a non-empty interior. More recently, Pereg and Steinberg~\cite{PeregS19} obtained the capacity region for the AV-MAC with state constraints. Wiese and Boche~\cite{WieseB13} considered the two-user AV-MAC with conferencing encoders. 
In a recent work, Beemer, Graves, Kliewer, Kosut, and Yu~\cite{BeemerCNS20} considered an authentication like model in a two-user AV-MAC, where all states, except one, are treated as adversarial states. Under adversarial states, the decoder's output can be a declaration of the presence of an  adversary while also decoding at least one user’s message.

\pink{In contrast to these works, the current model has {\em byzantine users}, {\em i.e.}, one of the legitimate users is potentially adversarial.} 
There are other works on \pink{models with {\em byzantine users} in the information theory literature}, mostly in the setting of network coding. Byzantine attacks on the nodes and edges of networks have been studied under omniscient and weaker adversarial models in~\cite{KTong} and~\cite{Jaggi7,WangSK10}, respectively. He and Yener~\cite{Yener} considered a Gaussian two-hop network with an eavesdropping and byzantine adversarial relay where the requirement is decoding the message with secrecy and byzantine attack detection.  La and Anantharam~\cite{LaADIMACS04} studied the MAC with strategic users modeled as a cooperative game with the objective of fairly allocating communication rate among the users.


\pink{For the \bmac, in a previous work \cite{NehaBDP19}, we looked at a weaker decoding guarantee than the present model, called {\em authenticated communication}. 
Under this weaker guarantee, the decoder must still  reliably recover the messages when all the users are honest. However, if any user behaves adversarially, the decoder may either output the correct messages for the honest users or declare an error, {\em i.e.}, an adversary should not be able to cause an undetected erroneous output for the honest users.
In a similar model of {\em communication with adversary identification} \cite{NehaBDP21} in a \bmac with two users, a slightly stronger decoding guarantee was considered.
Again, reliable decoding was required when all users are honest. In the presence of a malicious user, the decoder may either output a pair of messages out of which the message of the honest user is correct, or declare an error together with the identity of the malicious user. Both these models are different from the present model, where we always require reliable decoding of the honest users' messages and the decoder may never declare an error\footnote{Journal versions of~\cite{NehaBDP19} and~\cite{NehaBDP21} are in preparation. \pink{Together with the present paper these constitute our multi-part study of Byzantine MACs encompassing various decoding requirements.}\label{footnote:2}}.}

\subsection{Two-user \bmac}
For the $2$-user \bmac, consider the problem of reliable communication when any one of the users might be adversarial (though the decoder does not know {\em a priori} which, if any, of the users is adversarial). Clearly, each user can at least achieve the capacity of the AVC where the other user's input is treated as the channel state. Further, it is also easy to see that no higher rate is possible as, for the honest user's perspective,  the other user, when adversarial, can behave exactly like an adversary in the AVC setup (see Figure~\ref{fig:2userAVMAC}(a) and (b)). 
Thus, the capacity region is the rectangular region defined by the AVC capacities of the two users' channels (Figure~\ref{fig:2userAVMAC}(c)), \emph{i.e.}, there is no trade-off between the rates\footnote{This observation holds true under deterministic coding, stochastic encoding (where the encoders have private randomness), and randomized coding settings under both maximum and average probabilities of error. 
A similar observation can be made for a $k$-user \bmac where up to $k-1$ users may adversarially collude.}. Thus, the simplest non-trivial case is that of a 3-user \bmac with at most one adversarial user.

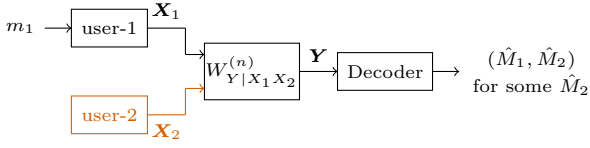
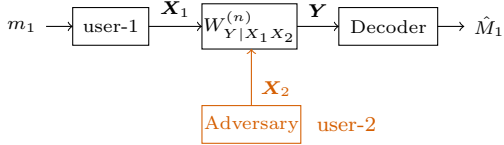
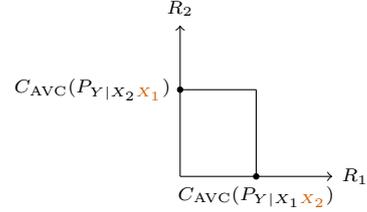
\begin{figure}
 \begin{subfigure}[b]{0.5\textwidth}
 	 \begin{subfigure}[b]{0.9\textwidth}
 	 \centering
	\begin{tikzpicture}[scale=0.5]
	\draw (1.7,6) rectangle ++(2,1) node[pos=.5]{\scriptsize user-1};
	\draw (4.2+1,4.6) rectangle ++(1.3+1.2,1.5) node[pos=.5]{\scriptsize $W^{(n)}_{Y|X_1X_2}$};
	\draw[->][vermillion] (3.7+1,4.9) -- ++ (0.5,0);
	\draw[-][vermillion] (3.7+1, 4.9) -- ++(0,-0.7);
	\draw[-][vermillion] (3.7+1, 4.2) -- node[below]{\scriptsize $\vecX_2$} ++(-1, 0);
	\draw[vermillion] (1.7,3.7) rectangle ++(2,1) node[pos=.5]{\scriptsize user-2};
	\draw[->] (1,6.5) node[anchor=east]{\scriptsize $m_1$} -- ++ (0.7,0) ;
	\draw[->] (3.7+1,5.80) -- ++(0.5,0);
	\draw[-] (3.7+1,5.8) -- ++ (0,0.7);
	\draw[-] (2.7+1, 6.5) -- node[above]{\scriptsize $\vecX_1$} ++(1,0);
	\draw[->] (5.7+1+1,5.35) -- node[above] {\scriptsize $\vecY$} ++ (1,0);
	\draw (6.7+1+1, 4.85) rectangle ++(2.5,1) node[pos=0.5]{\scriptsize Decoder};
	\draw[->] (7.7+1+1+1+0.5, 5.35) -- ++ (0.7,0) ;
	\node at (13.8, 5) {\scriptsize for some $\hat{M}_2$};
	\node at (13.8, 5.7) {\scriptsize $(\hat{M}_1, \hat{M}_2)$};
	\end{tikzpicture}
	\caption{A two-user \bmac where user-2 is malicious.}
	\end{subfigure}
 	\begin{subfigure}[b]{0.9\textwidth}
 	\centering
	\begin{tikzpicture}[scale=0.5]
	\draw[white] (3.2,7) rectangle ++(2.5,1.2) ;
	\draw (-0.2,6) rectangle ++(2,1) node[pos=.5]{ \scriptsize user-1};
	\draw (3.2,5.9) rectangle ++(2.5,1.2) node[pos=.5]{\scriptsize $W^{(n)}_{Y|X_1X_2}$};
	\draw[->,color=vermillion] (4.5, 4.4) node[anchor = south west]{\scriptsize $\vecX_2$} -- ++ (0, 1.5) ;
	\draw[color=vermillion] (3.2, 3.4) rectangle ++(2.6,1) node[pos=0.5]{\scriptsize Adversary};
	\draw[->] (-0.7-0.2,6.5) node[anchor=east]{\scriptsize $m_1$} -- ++ (0.7,0) ;
	\draw[->] (2-0.2, 6.5) -- node[above]{\scriptsize $\vecX_1$} ++(1.4,0);
	\draw[->] (5.7,6.5) -- node[above] {\scriptsize $\vecY$} ++ (1.1,0);
	\draw (6.8, 6) rectangle ++(2.6,1) node[pos=0.5]{\scriptsize Decoder};
	\draw[->] (9.4, 6.5) -- ++ (0.7,0)  node[right] {\scriptsize $\hat{M}_1$};
	\node[align=center, color = vermillion] at (7,3.9) {{\footnotesize user-2}};
\end{tikzpicture}
\caption{A malicious user-2 can simulate an AVC from user-1 to the receiver where the input of user-2 is treated as the adversarial state. Thus, user-1 cannot communicate reliably an any rate above the capacity of this AVC. On the other hand, user-1 can achieve all the rates below the capacity of this AVC by using an appropriate AVC code. 
}
\end{subfigure}
\end{subfigure}
\begin{subfigure}[b]{0.5\textwidth}
\centering
\begin{tikzpicture}[scale=0.5]
	\draw[->] (0,0)  --  ++(0,4) node[anchor= south]{\scriptsize $R_2$} ;
	\draw[->] (0,0) --  ++ (4,0)node[anchor= west] {\scriptsize $R_1$};
	\draw[-] (2,0) node[anchor= north]{\scriptsize ${C}_{\text{AVC}}(P_{Y|X_1{\color{vermillion} X_2}})$} --  ++ (0,2.3) ;
	\draw[-] (0,2.3) node[anchor= east] {\scriptsize ${C}_{\text{AVC}}(P_{Y|X_2{\color{vermillion} X_1}})$} --  ++ (2,0);
	\filldraw [black] (0,2.3) circle (2pt);
	\filldraw [black] (2,0) circle (2pt);
\end{tikzpicture}
	\vspace{0.4cm}
\caption{Capacity region of a two-user \bmac. \pink{${C}_{\text{AVC}}(P_{Y|X{\color{vermillion} S}})$  is the capacity of AVC $P_{Y|X{\color{vermillion} S}}$ with input $X$,  state $S$ and output $Y$.}}\label{fig:2userAVMAC_1c}
\end{subfigure}
\caption{Capacity region of a two-user \bmac is given by the rectangular capacity region obtained by treating the channel from each user to the receiver as an AVC with the other user's input as the AVC state sequence.}\label{fig:2userAVMAC}
\end{figure}

\subsection{Three-user \bmac \pink{with at most one adversary}}\label{sec:three_user_1.4}
It turns out that all the key ideas can be presented in the relatively simpler setting of a 3-user \bmac (Figure~\ref{fig:advMAC_3}) with at most one adversarial user. 
The general results then build on this. 
For this model, we characterize the capacity region under randomized coding where each user shares independent secret keys with the decoder, and deterministic coding with an average probability of error criterion.

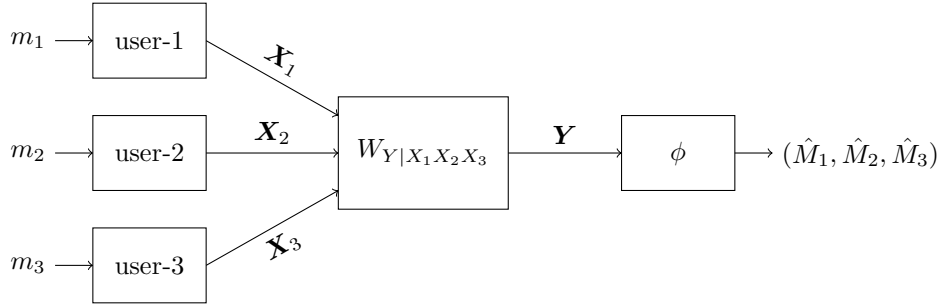
\begin{figure}[h]\centering
\resizebox{0.7\columnwidth}{!}{\begin{tikzpicture}[scale=0.5]
	\draw (2,0) rectangle ++(3,2) node[pos=.5]{user-3};
	\draw (2,3) rectangle ++(3,2) node[pos=.5]{user-2};
	\draw (2,6) rectangle ++(3,2) node[pos=.5]{user-1};
	\draw (8.5,2.5) rectangle ++(4.5,3) node[pos=.5]{$\mach$};
	\draw (16,3) rectangle ++(3,2) node[pos=.5]{$\phi$};
	\draw[->] (1,1) node[anchor=east]{$\msg_3$} -- ++ (1,0) ;
	\draw[->] (5,1) -- node[midway,below,sloped] {$\vecX_3$} ++ (3.5,2);
	
	\draw[->] (1,4) node[anchor=east]{$\msg_2$} -- ++ (1,0) ;
	\draw[->] (5,4) -- node[above] {$\vecX_2$} ++ (3.5,0);

	\draw[->] (1,7) node[anchor=east]{$\msg_1$} -- ++ (1,0) ;
	\draw[->] (5,7) -- node[midway,above,sloped] {$\vecX_1$} ++ (3.5,-2);

	\draw[->] (13,4) -- node[above] {$\vecY$} ++ (3,0);
	\draw[->] (19,4) -- ++ (1,0) node[anchor=west]{$(\hat{M}_1,\hat{M}_2,\hat{M}_3)$};
\end{tikzpicture}}
\caption{Byzantine-MAC: At most one user may be adversarial. Reliable decoding of the messages of all honest users is required. Clearly, no decoding guarantees are given for an adversarial user.}\label{fig:advMAC_3}
\end{figure}
\subsubsection{Randomized coding}\label{intro:random_coding}
Consider a three-user \bmac in which each user shares independent randomness with the decoder which is unknown to the other users. Notice that similar to the two-user \bmac where a malicious user could induce an AVC from the honest user to the receiver,  in a three-user \bmac, a malicious user-$i$, $i\in \{1, 2, 3\}$ can induce a two-user AV-MAC $W^{(i)}$ from the honest users $\{1, 2, 3\}\setminus\{i\}$ to the receiver, where the input of the malicious user is treated as the adversarially chosen state sequence. For instance, if a rate triple $(R_1,R_2,R_3)$ is achievable for the \bmac, then the rate pair $(R_1,R_2)$ is also be achievable over the two-user AV-MAC $W^{(3)}$.\footnote{In fact, a stronger necessary condition follows by noting that the encoder of each user must not depend on the knowledge of which user, if any, is the adversary. Thus, as in compound channels, the same code should work for $W^{(i)}, \, i\in \{1, 2, 3\}$. We use this observation in our converse arguments. } We use this intuition to show the converse of the randomized coding capacity region (Theorem~\ref{thm:random}). 
We show the achievability by using a randomized code (from \cite{Jahn81}\footnote{Note that similar to the current model, in the AV-MAC model of \cite{Jahn81}, users share independent randomness with the decoder.}) for the two-user AV-MAC $W^{(i)}, \, i\in \{1, 2, 3\}$ and using random hashes for each message, generated using the independent shared randomness. See Section~\ref{sec:random_proof_sketch} for a sketch of achievability and Section~\ref{app:A} for a detailed proof of achievability and converse.

\subsubsection{Deterministic coding} For deterministic coding, before discussing the capacity region, let us consider the following question: in which channels can all users {\em communicate  reliably}?

In the AVC literature, the channels over which reliable communication is infeasible are called {\em symmetrizable channels} \cite{Ericson,CsiszarN88}.	
In a symmetrizable AVC, the adversary can mount an attack which introduces a spurious message that can be confused with the actual message, resulting in a non-vanishing probability of error.


To answer the question, we first recall that a malicious user-$i$, $i\in \{1, 2, 3\}$, in a three-user \bmac, can induce a two-user AV-MAC $W^{(i)}$ formed by treating user-$i$'s input as an adversarially chosen state and the inputs of other two users as the inputs of legitimate users in the two-user AV-MAC. 
Thus, we inherit the 
symmetrizability conditions~\cite[Definition~3.1-3.3]{Gubner90} from the three AV-MAC $W^{(1)}, W^{(2)}$ and $W^{(3)}$. We show that, in addition to the symmetrizability conditions inherited from the AV-MAC, fully characterizing the feasibility of reliable communication of a 3-user \bmac requires three additional symmetrizability conditions (Eq.~\eqref{eq:symm3}). Roughly speaking, each of these conditions reflect whether or not an adversarial user at a node $k$ can attack in a manner that is  also consistent with an adversarial user at a node $j\neq k$ while resulting in a decoding ambiguity about the remaining user's message (see Figure~\ref{fig:symm3}). Example~\ref{ex:symmetrizable} (page~\pageref{ex:symmetrizable}) shows that the new symmetrizability conditions are not redundant given the symmetrizability conditions inherited from the two-user AV-MAC.

We characterize the deterministic coding capacity region under the average error criterion for most channels.\footnote{Our characterization for deterministic codes is incomplete for channels in which some, but not all users are symmetrizable (for an appropriate notion of symmetrizability for a 3-user \bmac). See remark~\ref{rem:gap}. \pink{We only study average probability of error under deterministic coding since the capacity under maximum probability of error remains open for multiple access channels (even with non-byzantine users)\cite{dueck1978maximal}.}}\label{footnote:12}
There are two different  approaches towards showing the achievability for the AVC using  deterministic codes. We show achievability for the $3$-user \bmac using both approaches and show a more general result for $k$-user \bmac using one of them. 

\paragraph*{First approach.}\label{intro:first_approach} The first approach uses a {\em randomness reduction} argument of Ahlswede~\cite{Ahlswede78} (and its extension for AV-MAC by Jahn~\cite{Jahn81}). He showed that given a randomized code of achievable rate $R$ and block-length $n$, there exists 
another randomized code of achievable rate $R$ which requires only $O(\log{n})$ bits of randomness. This small amount of shared randomness can be established using deterministic codes. This shows the surprising fact that when the deterministic capacity is positive (which is the case for non-symmetrizable channels), it is in fact equal to the randomized coding capacity.
Thus, to show achievability under deterministic codes, it suffices to show that all non-symmetrizable channels admit positive rates. 
Ahlswede and Cai in \cite{AhlswedeC99} took this route for the achievability proof of the two-user AV-MAC. 
For \bmacs, we may follow a similar recipe (in fact, we do this for the general $k$-user \bmac). We show a {\em randomness reduction} argument along the lines of Jahn~\cite{Jahn81} and Ahlswede~\cite{Ahlswede78} (see Appendix~\ref{app:randomness_reduction}). 
With this and the randomized coding scheme discussed above (Section~\ref{intro:random_coding}), all that remains is to show that in a non-symmetrizable \bmac, all users can transmit at positive rates using deterministic coding. 
The main difference from \cite{AhlswedeC99} in showing this, is that the code should be able to handle any user behaving maliciously. Please see Section~\ref{sec:k_user} for details.

\paragraph*{Second approach.}\label{intro:second_approach} The second approach is a direct argument based on the method of types which establishes a deterministic code. The technique does not rely on the achievability of the randomized coding capacity. For the (point-to-point) AVC, Csisz\'ar and Narayan \cite{CsiszarN88} established the deterministic coding capacity using such an approach. 
Their achievability proof is based on a  concentration result \cite[Lemma A1]{CsiszarN88}. A similar approach for multi-user channels (e.g. AV-MAC, \bmac etc.) requires extending this concentration result.
We specialize the concentration result in \cite[Theorem 2.1]{SJ} to obtain just such an extension (Lemma~\ref{lemma:codebook} on page~\pageref{lemma:codebook}). 
This allows us to directly achieve all rate triples in the capacity region of a non-symmetrizable three-user \bmac (see Section~\ref{sec:det_proofs}). Our technique can also be used to obtain the deterministic coding capacity region of a two-user AV-MAC directly. 
We believe that this technique may have applications in other multi-user deterministic coding settings for adversarial channels and may be of independent interest.


\subsection{$k$-user \bmac}
In Section~\ref{sec:k_user}, we consider a general $k$-user \bmac in which an adversary may control all users in any one of a set of subsets of users, called an {\em adversary structure}\footnote{The term `adversary structure' is borrowed from cryptography. An adversary structure is a collection of subsets of users. The adversary may control all the users in any one of these subsets and use them to mount an attack (see, e.g., \cite{fitzi1999general,hirt2000player,wu2022two}).} (see Fig.~\ref{fig:advMAC}). The receiver is unaware which of these subsets the adversary controls.
We provide a general symmetrizability condition for the $k$-user \bmacs. 
On the achievability side, we take the first achievability approach described above (see Section~\ref{intro:first_approach}) and show a {\em randomness reduction} argument along the lines of Jahn~\cite{Jahn81} and Ahlswede~\cite{Ahlswede78}.
We then show that as long as the given \bmac is non-symmetrizable, {\em i.e.}, none of the symmetrizability conditions hold, the deterministic coding capacity region has a non-empty interior, in other words, all users can communicate at positive rates. 
Finally, we characterize the randomized coding capacity region using the same ideas as that for the three-user case. For the $k$-user \bmac, we do not pursue a direct proof using the second achievability approach described above (in Section~\ref{intro:second_approach}) as it appears to be cumbersome.

\begin{figure}\centering
\resizebox{0.9\columnwidth}{!}{\begin{tikzpicture}[scale=0.5]
	\draw (2,0) rectangle ++(3,1.2) node[pos=.5]{user-$k$};
	\draw[vermillion] (2,5.7) rectangle ++(3,1.2) node[pos=.5]{user-2};
	\draw[vermillion] (2,7.4) rectangle ++(3,1.2) node[pos=.5]{user-1};
	\draw (2,4) rectangle ++(3,1.2) node[pos=.5]{user-3};
	\draw (8.5,2.5) rectangle ++(4.5,3) node[pos=.5]{$W^{(n)}_{Y|X_1X_2\ldots X_k}$};
	\draw (16,3) rectangle ++(3,2) node[pos=.5]{Decoder};
	\draw[->] (1,0.6) node[anchor=east]{$\msg_k$} -- ++ (1,0) ;
	\draw[->] (5,1) -- node[midway,below,sloped] {$\scriptstyle \vecX_k$} ++ (3.5,2);
	
	\draw[->] (1,4.6) node[anchor=east]{$\msg_3$} -- ++ (1,0) ;
	\draw[vermillion,->] (5,6.3) -- node[above,,sloped] {$\scriptstyle \vecX_2$} ++ (3.5,-1.6);
	\draw[->] (5,4.6) -- node[above,,sloped] {$\scriptstyle \vecX_3$} ++ (3.5,-0.2);

	\draw[vermillion,->] (5,8) -- node[midway,above,sloped] {$\scriptstyle \vecX_1$} ++ (3.5,-3);

	\node at (3.5,2.8) {$\vdots$};
	\draw[->] (13,4) -- node[above] {$\scriptstyle \vecY$} ++ (3,0);
	\draw[->] (19,4) -- ++ (1,0) node[anchor=west]{$\stackrel{\inp{\hat{M}_1,\hat{M}_2,\hat{M}_3,\ldots,\hat{M}_k}}{\scriptstyle \text{where }(\hat{M}_3,\ldots,\hat{M}_k) = \inp{m_3,\ldots,m_k}\text{ w.h.p.}}$};
\end{tikzpicture}}
\caption{Consider a $k$-user byzantine-MAC where the set containing users~1
and~2 belongs to the adversary structure. The figure depicts the case
when users~1 and~2 deviate from the protocol under the control of an
adversary. Then we require reliable decoding of the messages of all
the honest users, {\em i.e.}, users 3 to $k$. 
}\label{fig:advMAC}
\end{figure}
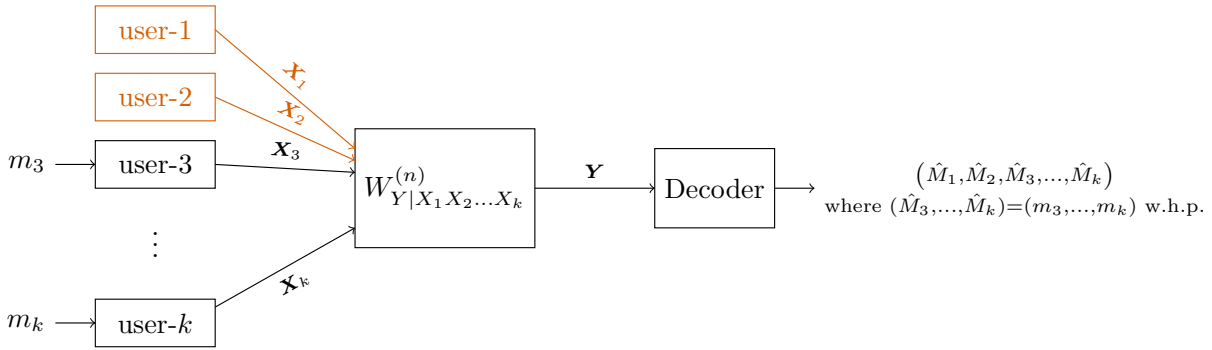

\subsection{Summary of contributions}
\begin{itemize}
\item  We introduce the model of reliable communication in a \bmac, where malicious users may deviate from the prescribed protocol. The model requires that decoded messages should be correct for the honest users with high probability. 
\item  We completely characterize the capacity region under both deterministic codes (with an average probability of error criterion) and randomized codes for any $k$-user \bmac. 
\item We also provide an alternate direct achievability for the 3-user \bmac, in the spirit of \cite{CsiszarN88}, where the achievability is based on a recent concentration result. This technique can be used to obtain a similar direct achievability for the 2-user AV-MAC (see Section~\ref{intro:second_approach}) and may be of independent interest. 
\end{itemize}

\subsection{Outline}
The system model and main results for the 3-user \bmac are given in Section~\ref{section:MAC2_model_main} (Page \pageref{section:MAC2_model_main}). This section also contains the proof sketches. The main proofs of the results in Section~\ref{section:MAC2_model_main} are given in Section~\ref{sec:proof_3_user} (page \pageref{sec:proof_3_user}). Others are deferred to the appendices. Section~\ref{sec:k_user} presents the $k$-user \bmac model and gives main results. All the proofs of theorems in this section are given in the appendices.
\section{Notation}
Random variables are denoted by capital letters (possibly indexed) like $X_1, X_2, X_3, Y,$ etc. 
The corresponding alphabets are denoted by calligraphic letters in the same format, for example, the random variable $X_1$ has alphabet $\cX_1$. 
Its $n$-product set is denoted by $\cX_1^n$. We use bold faced letters to denote $n$-length vectors, for example, $\vecx$ denotes a vector in $\cX^n$ and $\vecX$ denotes a random vector taking values in $\cX^n$. 
For a random variable $X$, we denote its distribution by $P_X$ and use the notation $X\sim P_{X}$ to indicate this. 
For an alphabet $\cX$, let $\mathcal{P}^n_{\cX}$ denote the set of all empirical distributions of $n$ length strings from $\cX^n$. 
For a random variable $X\sim P_{X}$ such that  $P_X \in \mathcal{P}^n_{\cX}$, let $T^n_X$ be the set of all $n$-length strings with empirical distribution $P_X$. 
For $\vecx\in \cX^n$, the statement $\vecx\in T^n_{X}$ defines $P_{X}$ as the empirical distribution of $\vecx$ and a random variable $X\sim P_{X}$. 
For a set $\cS$, $2^{\cS}$ denotes it power set, $\cS^c$ denotes its compliment and $\textsf{int}(\cS)$ denotes its interior. A uniform distribution on any set $\cS$ is denoted by $\textsf{Unif}(\cS)$.
For any $n$, the set $\inb{1, 2, \ldots, n}$ will be denoted by $[1:n]$. We will use the acronyms `w.h.p.' to mean `with high probability'. \pink{For any real number $A$, we use $|A|^{+}$ to mean $A$ if $A\geq 0$. Otherwise, $|A|^{+}=0$.}

The following notation will be used in Section~\ref{sec:k_user}. For any sets $\cS_i$, $i\in [1:k]$ and for $\cB\subseteq[1:k]$, $\cS_{\cB}$ denotes the product set $\times_{i\in \cB}\cS_i$.  The tuple $(s_i\in \cS_i: i\in [1:k])$ will be denoted by $s_{[1:k]}\in \cS_{[1:k]}$ and when restricted to $\cB$, we write $s_{\cB}\in \cS_{\cB}$. The notation $g_{\cB}(s_{\cB})$ denotes $(g_i(s_i): i\in \cB)$ for function $g_i$ defined on $\cS_i$, $i\in [1:k]$.

\section{The three user \bmac \pink{with at most one adversary}}\label{section:MAC2_model_main}
\subsection{System model}\label{section:MAC2_model}

Consider the 3-user \bmac setup shown in Fig.~\ref{fig:advMAC_3}. The memoryless channel \mach has input alphabets $\mathcal{X}_1,\mathcal{X}_2,\mathcal{X}_3$, and output alphabet ${\mathcal Y}$. 
\shortonly{An $(\nummsg_1,\nummsg_2,\nummsg_3,n)$ deterministic {\em code} for the multiple access channel \mach consists of: \vspace{-0.25em}
\begin{enumerate}[label=(\roman*)]
\item three message sets, $\mathcal{M}_i = \{1,\ldots,\nummsg_i\}$, $i=1,2,3$,
\item three encoders, $f_{i}:\mathcal{M}_i\rightarrow \mathcal{X}_i^n$, $i=1,2,3$, and
\item a decoder, $\phi:\mathcal{Y}^n\rightarrow\mathcal{M}_1\times\mathcal{M}_2\times\mathcal{M}_3.$ 
\end{enumerate}}

\longonly{\begin{defn}[Deterministic code]\label{code:det_3user}
An $(\nummsg_1,\nummsg_2,\nummsg_3,n)$ deterministic {\em code} for the \bmac \mach consists of the following: 
\begin{enumerate}[label=(\roman*)]
\item three message sets, $\mathcal{M}_i = \{1,\ldots,\nummsg_i\}$, $i=1,2,3$,
\item three encoders, $f_{i}:\mathcal{M}_i\rightarrow \mathcal{X}_i^n$, $i=1,2,3$, and
\item a decoder, $\phi:\mathcal{Y}^n\rightarrow\mathcal{M}_1\times\mathcal{M}_2\times\mathcal{M}_3.$ 
\end{enumerate}
\end{defn}}
\noindent We define the average probability of error $P_e$ as the maximum of average error probabilities under four different scenarios, one corresponding to all users being honest and three corresponding to exactly one user being adversarial. Let $(\Msgh_1,\Msgh_2,\Msgh_3)=\phi(Y^n)$. 
\[ P_e(f_1,f_2,f_3,\phi)\defineqq\max \{ P_{e,0}, P_{e,1}, P_{e,2}, P_{e,3} \},\]
where the terms on the right-hand side are defined below. Note that our notation suppresses their dependence on the code. $P_{e,0}$ is the average probability of error when none of the users are adversarial,
\begin{align*}
&P_{e,0}\defineqq \frac{1}{\nummsg_1\nummsg_2\nummsg_3}
\sum_{(m_1,m_2,m_3)\in \cM_1\times\cM_2\times\cM_3} e_0(m_1,m_2,m_3),\text{ where}\\
&e_0(m_1,m_2,m_3)=
\Prob\Big( (\Msgh_1,\Msgh_2,\Msgh_3) \neq (m_1,m_2,m_3) \Big|\vecX_1=f_1(m_1),  \vecX_2=f_2(m_2), \vecX_3=f_3(m_3)\Big).
\end{align*}
$P_{e,i}$, $i=1,2,3$ is the average error {probability under worst case deterministic attacks} when user-$i$ is adversarial. $P_{e,1}$ is as below.  $P_{e,2},P_{e,3}$ are defined similarly.
\begin{align}
&P_{e,1}\defineqq \max_{\vecx_1\in \cX^n_1} \frac{1}{\nummsg_2\nummsg_3} \sum_{(m_2,m_3)\in \cM_2\times\cM_3} e_1(\vecx_1,m_2,m_3),\text{ where }\nonumber\\
&e_1(\vecx_1,m_2,m_3)= \Prob\Big( (\Msgh_2,\Msgh_3) \neq (m_2,m_3) \Big| \vecX_1=\vecx_1, \vecX_2=f_2(m_2), \vecX_3=f_3(m_3) \Big).\label{eq:error_adv}
\end{align}
{We emphasize that 
\begin{enumerate}[label=\alph*)]
\item the decoder is unaware of whether any of the users is adversarial and the identity of the adversarial user (if any).
\item the adversary knows the encoders and the decoder, but is unaware of the messages transmitted by the other (non-adversarial) users\footnote{Recall that at most one user is adversarial.}.
\end{enumerate}}
Note that it is sufficient to define $P_{e,i}$ under deterministic attacks by the adversarial user. {To see this, consider the setting where user-1 is adversarial. Then, under any randomized attack $\tilde{\vecX}_1\sim Q$ for any distribution $Q$ on $\cX^n_1$, }
\begin{align}
&\bbE_{Q}\insq{\frac{1}{\nummsg_2\nummsg_3} \sum_{m_2,m_3} \Prob\Big( (\Msgh_2,\Msgh_3) \neq (m_2,m_3) \Big| \vecX_1=\tilde{\vecX_1}, \vecX_2=f_2(m_2), \vecX_3=f_3(m_3) \Big)}\nonumber\\
&=\sum_{\vecx_1\in\cX^n_1}Q(\vecx_1)\frac{1}{\nummsg_2\nummsg_3} \sum_{m_2,m_3} \Prob\Big( (\Msgh_2,\Msgh_3) \neq (m_2,m_3) \Big| \vecX_1={\vecx_1}, \vecX_2=f_2(m_2), \vecX_3=f_3(m_3) \Big)\nonumber\\
&\leq \sum_{\vecx_1\in\cX^n_1}Q(\vecx_1)P_{e,1}\nonumber\\
&=P_{e,1}. \label{eq:random_attacks}
\end{align}
In other words, the probability of error is maximized when the adversarial user selects a deterministic attack vector (that depends only on the channel and the deterministic code used). 
We also note that 
\begin{align}\label{eq:honest_upper_bound}
P_{e,0}\leq P_{e,1}+P_{e,2}+P_{e,3}.
\end{align} This is because
\begin{align*}
P_{e, o}=&\frac{1}{N_1N_2N_3}\sum_{m_1,m_2,m_3}\bbP\Big( (\Msgh_1,\Msgh_2,\Msgh_3) \neq (m_1,m_2,m_3) \Big|\vecX_1=f_1(m_1),  \vecX_2=f_2(m_2), \vecX_3=f_3(m_3)\Big)\\
=&\frac{1}{N_1N_2N_3}\sum_{m_1,m_2,m_3}\bbP\Big( \{(\Msgh_1,\Msgh_2) \neq (m_1,m_2)\}\cup\{(\Msgh_2,\Msgh_3) \neq (m_2,m_3)\}\cup\{(\Msgh_1,\Msgh_3) \neq (m_1,m_3)\} \\
&\qquad\qquad\qquad\qquad\qquad\qquad\Big|\vecX_1=f_1(m_1),  \vecX_2=f_2(m_2), \vecX_3=f_3(m_3)\Big)\\
\leq &\frac{1}{N_1N_2N_3}\sum_{m_1,m_2,m_3}\Bigg\{\bbP\Big( \{(\Msgh_1,\Msgh_2) \neq (m_1,m_2)\} \Big|\vecX_1=f_1(m_1),  \vecX_2=f_2(m_2), \vecX_3=f_3(m_3)\Big)\\
&\qquad\qquad\qquad\qquad+\bbP\Big( \{(\Msgh_2,\Msgh_3) \neq (m_2,m_3)\} \Big|\vecX_1=f_1(m_1),  \vecX_2=f_2(m_2), \vecX_3=f_3(m_3)\Big)\\
&\qquad\qquad\qquad\qquad+\bbP\Big( \{(\Msgh_1,\Msgh_3) \neq (m_1,m_3)\} \Big|\vecX_1=f_1(m_1),  \vecX_2=f_2(m_2), \vecX_3=f_3(m_3)\Big)\Bigg\}\\
&\leq P_{e,1}+P_{e,2}+P_{e,3}.
\end{align*}

\begin{defn}[{Achievable rate triple and the deterministic coding capacity region}]
We say a rate triple $(R_1,R_2,R_3)$ is {\em achievable} if there is a sequence of $(\lfloor2^{nR_1}\rfloor,\lfloor2^{nR_2}\rfloor,\lfloor2^{nR_3}\rfloor,n)$ codes  {$(f_1^{(n)},f_2^{(n)},f_3^{(n)},\allowbreak\phi^{(n)})$ for $n=1, 2, \ldots$} such that\\ $\lim_{n\rightarrow\infty}P_{e}(f_1^{(n)},f_2^{(n)},f_3^{(n)},\phi^{(n)})\rightarrow0$. The {\em deterministic coding capacity region} $\mathcal{R}_{\deterministic}$ is the closure of the set of all achievable rate triples. 
\end{defn}

\longonly{
\begin{defn}[Randomized code]\label{defn:rand_code}
An $(\nummsg_1,\nummsg_2,\nummsg_3,n)$ randomized {\em code} for the \bmac \mach consists of the following: 
\begin{enumerate}[label=(\roman*)]
\item three message sets, $\mathcal{M}_i = \{1,\ldots,\nummsg_i\}$, $i=1,2,3$,
\item three {\em independent} randomized encoders, $F_{i}:\mathcal{M}_i\rightarrow \mathcal{X}_i^n$ where $F_i\sim P_{F_i}$ takes values in {$\mathcal{F}_i\subseteq\{g:\cM_i\rightarrow \cX_i^n\},\, i = 1,2,3$} and
\item a decoder, $\phi:\mathcal{Y}^n\times\mathcal{F}_1\times\mathcal{F}_2\times\mathcal{F}_3\rightarrow\mathcal{M}_1\times\mathcal{M}_2\times\mathcal{M}_3$ where \\$\phi(\vecy,F_1,F_2,F_3) = (\phi_1(\vecy,F_1,F_2,F_3),\phi_2(\vecy,F_1,F_2,F_3),\phi_3(\vecy,F_1,F_2,F_3))$ for \pink{some functions}  $\phi_i:\mathcal{Y}^n\times\mathcal{F}_1\times\mathcal{F}_2\times\mathcal{F}_3\rightarrow\mathcal{M}_i, \, i = 1,2,3$. 
\end{enumerate}
\end{defn}}
\pink{In other words, a randomized code consists of {\em independent} random encoding maps
$F_1,F_2,F_3$ and a decoder $\phi$ (which takes $F_1,F_2,F_3$ also as inputs), {\em i.e.},
the encoders randomize independently of each other and their randomization is
available to the decoder. This is similar to the randomized code of
Jahn~\cite{Jahn81} for 2-user AV-MACs. Notice that the decoder is a randomized decoder since the decoding function $\phi$ takes the random encoding maps $F_1,F_2,F_3$ as inputs\footnote{\pink{Any additional private randomness at the decoder can be subsumed as part of the randomness shared with each encoder in a slightly more general definition of randomized code ({\em i.e.}, a slight generalization of Definition~\ref{defn:rand_code}) for which our converse in Section~\ref{app:A} continues to hold. In this generalization, the users first sample $(F_i, \vecB_{i}); i =1, 2, 3$ where $\vecB_i$ are uniform bit strings, independent of $F_i$. Now, any additional private randomness at the decoder may be thought of as a bit string $\vecD$ which is XOR of $\vecB_1$, $\vecB_2$ and $\vecB_3$. Even, when one of the users, say user $i$, maliciously chooses $\vecB_i$, note that $\vecD$ remains uniform and unknown to user $i$.\label{footnote:remark1}}}.
We emphasize that each byzantine user is unaware of the encoding maps of the other users. We also assume that the (byzantine) user-$i$ samples its encoder $F_i$ which is then made available to the decoder. Notice that the decoder $\phi$ is a function which maps the channel output as well as the random encoding maps to the decoded messages. This allows the adversarial user to adversarially choose its encoding map (in addition to its channel input) as part of its attack and thus attempt to influence the decoding.  This means that an adversarial user $i$ may choose $\vecx_i\in \cX^n_i$ as input to the channel and any $f_i\in \cF_i$ as the encoding map. This is shown in Fig.~\ref{fig:weak_adv_strong}. We denote the {\em
randomized coding capacity region} by $\mathcal{R}_{\random}$. 
We also consider another adversarial model, called the {\em weak} adversary. An adversary is a {\em weak} adversary if it
does not have access to its own random encoding map when choosing its input vector, that is, the random encoding map $F_i$ is sampled according to $P_{F_i}$ and the adversarial input to the channel $\vecx_i$ is chosen independent of $F_i$ (see Fig.~\ref{fig:weak_adv_weak})\footnote{\pink{An intermediate model is the one where the adversary knows the random encoding map but does not have control over it. That is, for a malicious user $i$, $F_i\sim P_{F_i}$ and the input to the channel $\vecx_i$ can be chosen as a function of $F_i$. In the proof of Theorem~\ref{thm:random}, the achievability is proved for the default adversary (who is stronger) while the converse is proved for the weak adversary. Hence, the capacity region for this intermediate model is the same as in Theorem~\ref{thm:random}.}}.
We denote the corresponding randomized coding capacity region by $\mathcal{R}_{\random}^{\weak}$. We show converse for the weak adversary.
Clearly,
$\mathcal{R}_{\random} \subseteq \mathcal{R}_{\random}^{\weak}$. Thus, a converse bound on $\mathcal{R}_{\random}^{\weak}$ is also a converse bound on $\mathcal{R}_{\random}$.}

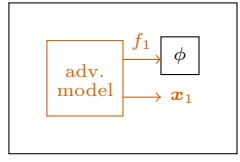
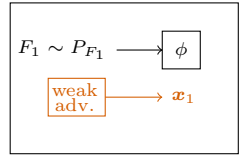
\begin{figure}[h]
\centering

\begin{subfigure}[t]{.25\columnwidth}
\centering
\begin{tikzpicture}[scale=0.5]
	\draw (0,0) rectangle ++(6,4);
	\draw (4,2.1) rectangle ++(1,1) node[pos=.5]{\scriptsize $\phi$};
	\draw[vermillion] (1,1)  rectangle ++(2,2) ;
	
	\node[vermillion] at (2,2.2) {\scriptsize adv.};
	\node[vermillion] at (2,1.7) {\scriptsize model};
	\draw[vermillion,->] (3,1.5)  --++ (1,0)  node[right]{\scriptsize $\vecx_1$};

	\draw[vermillion,->] (3,2.5)  -- node[above]{\scriptsize $f_1$} ++ (1,0) ;

\end{tikzpicture}
\caption{An adversarial user-1 chooses an encoder map $f_1\in \cF_1$ and an input vector $\vecx_1\in \cX_1$ together.}
\label{fig:weak_adv_strong}
\end{subfigure}
\hspace{1cm}
\begin{subfigure}[t]{.25\columnwidth}
\centering
\begin{tikzpicture}[scale=0.5]
	\draw (0,0) rectangle ++(6,4);
	\draw (4,2.25) rectangle ++(1,1) node[pos=.5]{\scriptsize $\phi$};
	\draw[vermillion] (1,1)  rectangle ++(1.5,1) ;
	
	\node[vermillion] at (1.7,1.7) {\scriptsize weak};
	\node[vermillion] at (1.7,1.3) {\scriptsize adv.};
	\draw[vermillion,->] (2.5,1.5)  --  ++ (1.5,0) node[right]{\scriptsize $\vecx_1$};

	\draw[->] (2.5+0.3,2.75) node[anchor=east]{\scriptsize $F_1\sim P_{F_1}$} -- ++ (1.2,0) ;

\end{tikzpicture}
\caption{User-1 as a weak adversary sends a malicious input $\vecx_1$ independent and unaware of the choice of random encoder $F_1\sim P_{F_1}$.}
\label{fig:weak_adv_weak}
\end{subfigure}

\caption{A figure depicting various adversary models for randomized coding when user-1 is malicious.}
\label{fig:weak_adv}
\end{figure}

\pink{Analogous to the deterministic case, the average probability of error $P^{\text{rand}}_{e}$ is defined as 
\[ P^{\text{rand}}_{e}(P_{F_1},P_{F_2},P_{F_3},\phi)\defineqq\max \{ P^{\text{rand}}_{e,0}, P^{\text{rand}}_{e,1}, P^{\text{rand}}_{e,2}, P^{\text{rand}}_{e,3} \},\]
where 
\begin{align*}
&P^{\text{rand}}_{e,0}\defineqq \frac{1}{\nummsg_1\nummsg_2\nummsg_3}
\sum_{m_1,m_2,m_3} e^{\text{rand}}_0(m_1,m_2,m_3),\text{ where}\\
&e^{\text{rand}}_0(m_1,m_2,m_3)=
\Prob\Big( \phi(Y^n,F_1,F_2,F_3) \neq (m_1,m_2,m_3) \Big|\vecX_1=F_1(m_1),  \vecX_2=F_2(m_2), \vecX_3=F_3(m_3)\Big).
\end{align*}
The probability is over independent $F_i\sim P_{F_i}, \, i=1, 2, 3$ and the randomness in the channel.
 $P^{\text{rand}}_{e,1}$ is as below, $P^{\text{rand}}_{e,2}, P^{\text{rand}}_{e,3}$ are  defined similarly.
\begin{align}
&P^{\text{rand}}_{e,1}\defineqq \max_{\vecx_1\in\mathcal{X}^n, f_1\in\mathcal{F}_1} \frac{1}{\nummsg_2\nummsg_3} \sum_{m_2,m_3} e_{f_1}(\vecx_1,m_2,m_3),\text{ where}\label{eq:p_rand_e_1}\\
&e_{f_1}(\vecx_1,m_2,m_3)= \Prob\Big( (\phi_2(\vecY,f_1,F_2,F_3),\phi_3(\vecY,f_1,F_2,F_3)) \neq (m_2,m_3) \Big|\vecX_1=\vecx_1, \vecX_2=F_2(m_2), \vecX_3=F_3(m_3) \Big).\nonumber
\end{align} The probability is over independent $F_i\sim P_{F_i}, \, i=2, 3$ and the channel. Restricting the attacks to deterministic attacks is without loss of generality along the lines of \eqref{eq:random_attacks}.
\noindent We define achievable rate triples and and capacity region
for randomized codes in a similar manner as the deterministic case\footnote{\pink{Along the lines of \cite[Problem 12.6 (b)]{CK11}, one can show that }for randomized
codes, the capacity region will remain unchanged for maximum and average
probabilities of error criteria. Hence we only consider the average error
criterion here.\label{footnote:max_avg_Error}}.
\begin{defn}[{Achievable rate triple and randomized coding capacity regions}]
We say a rate triple $(R_1,R_2,R_3)$ is {\em achievable}, if there is a sequence of $(\lfloor2^{nR_1}\rfloor,\lfloor2^{nR_2}\rfloor,\lfloor2^{nR_3}\rfloor,n)$ codes  $\{F_1^{(n)},F_2^{(n)},F_3^{(n)},\phi^{(n)}\}_{n=1}^\infty$ such that \\$\lim_{n\rightarrow\infty}P^{\text{rand}}_{e}(P_{F_1^{(n)}},P_{F_2^{(n)}},P_{F_3^{(n)}},\phi^{(n)})\rightarrow0$. The {\em randomized coding capacity region} $\mathcal{R_{\text{random}}}$ is the closure of the set of all achievable rate triples. 
\end{defn}
The probability of error $P^{\text{weak}}_{e}$ and the capacity region $\mathcal{R}_{\random}^{\weak}$ for randomized codes with a weak adversary are defined by replacing  $P^{\text{rand}}_{e,i}$ with $P^{\text{weak}}_{e,i},\, i = 1,2,3$ in the definition of $P^{\text{rand}}_e(P_{F_1}, P_{F_2}, P_{F_3}, \phi)$,  where
\begin{align}
&P^{\text{weak}}_{e,1}\defineqq \max_{\vecx_1\in \mathcal{X}^n} \frac{1}{\nummsg_2\nummsg_3} \sum_{m_2,m_3} e^{\text{weak}}_{1}(\vecx_1,m_2,m_3),\label{eq:pe1_weak}\\
&\text{ where}\quad e^{\text{weak}}_{1}(\vecx_1,m_2,m_3)= \Prob\Big( (\phi_2(\vecY,F_1,F_2,F_3),\phi_3(\vecY,F_1,F_2,F_3)) \neq (m_2,m_3) \Big|\nonumber\\
&\qquad\qquad\qquad\qquad \vecX_1=\vecx_1, \vecX_2=F_2(m_2), \vecX_3=F_3(m_3) \Big).\nonumber
\end{align}
The probability is over independent $F_i\sim P_{F_i}, \, i=1, 2, 3$ and the channel. $P^{\text{weak}}_{e,2}$ and $P^{\text{weak}}_{e,3}$ are  defined similarly.
}

\subsection{Main results}
\subsubsection{Deterministic coding capacity region}\label{sec:deterministic_capacity_region}
We first present our results for \pink{the three user \bmac with at most one adversary under  deterministic coding.} Analogous to the notion of symmetrizability {\cite{CsiszarN88,AhlswedeC99}} in the AVC and AV-MAC literature, we give conditions under which at least one user cannot communicate with positive rate.

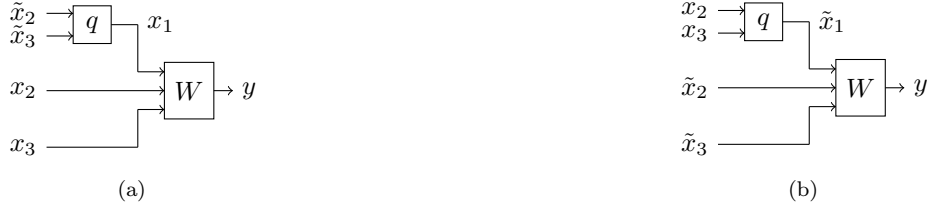
\begin{figure}[h]
\centering
\begin{subfigure}[t]{.49\columnwidth}
\centering
\begin{tikzpicture}[scale=0.5]
	\draw (1.7,6) rectangle ++(1,1) node[pos=.5]{$q$};
	\draw (4.1,4) rectangle ++(1.3,1.5) node[pos=.5]{$W$};
	\draw[->] (1,4.75) node[anchor=east]{$x_2$} -- ++ (3.1,0);

	\draw[->] (1,6.8) node[anchor=east]{$\tilde{x}_2$} -- ++ (0.7,0) ;
	\draw[->] (1,6.2) node[anchor=east]{$\tilde{x}_3$} -- ++ (0.7,0) ;
	\draw[->] (3.4,4.25)  -- ++(0.7,0);
	\draw[->] (3.4,5.25) -- ++(0.7,0);
	\draw[-] (3.4,5.25) -- ++ (0,1.25);
	\draw[-] (3.4,4.25) -- ++ (0,-1);
	\draw[-] (2.7, 6.5) -- ++(0.7,0) node[anchor=west]{{${x}_1$}};
	\draw[-] (1,3.25) node[anchor=east]{$x_3$} -- ++ (2.4,0);

	\draw[->] (5.4,4.75) --  ++ (0.5,0)node[anchor= west] {$y$};
\end{tikzpicture}
\caption{}
\label{fig:symm1a}
\end{subfigure}
\begin{subfigure}[t]{.49\columnwidth}
\centering
\begin{tikzpicture}[scale=0.5]
	\draw (1.7,6) rectangle ++(1,1) node[pos=.5]{$q$};
	\draw (4.1,4) rectangle ++(1.3,1.5) node[pos=.5]{$W$};
	\draw[->] (1,4.75) node[anchor=east]{$\tilde{x}_2$} -- ++ (3.1,0);

	\draw[->] (1,6.8) node[anchor=east]{$x_2$} -- ++ (0.7,0) ;
	\draw[->] (1,6.2) node[anchor=east]{$x_3$} -- ++ (0.7,0) ;
	\draw[->] (3.4,4.25) -- ++(0.7,0);
	\draw[->] (3.4,5.25) -- ++(0.7,0);
	\draw[-] (3.4,5.25) -- ++ (0,1.25);
	\draw[-] (3.4,4.25) -- ++ (0,-1);
	\draw[-] (2.7, 6.5) -- ++(0.7,0)node[anchor=west]{ 	{$\tilde{x}_1$}};
	\draw[-] (1,3.25) node[anchor=east]{$\tilde{x}_3$} -- ++ (2.4,0);

	\draw[->] (5.4,4.75) --  ++ (0.5,0)node[anchor= west] {$y$};
\end{tikzpicture}
\caption{}
\label{fig:symm1bb}
\end{subfigure}
\caption{We say $\mach$ is $\cX_2\times\cX_3$-{symmetrizable by}~$\cX_1$ if, for each $(x_2,\tlx_2,x_3,\tlx_3)$, the conditional output distributions in the two cases above are the same. Thus, the receiver is unable to tell whether users~2 and~3 are sending $(x_2,x_3)$ or $(\tlx_2,\tlx_3)$.}
\label{fig:symm1}
\end{figure}

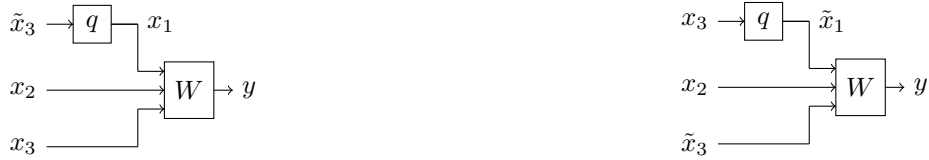
\begin{figure}[h]
\centering
\begin{subfigure}[b]{.49\columnwidth}
\centering
\begin{tikzpicture}[scale=0.5]
	\draw (1.7,6) rectangle ++(1,1) node[pos=.5]{$q$};
	\draw (4.1,4) rectangle ++(1.3,1.5) node[pos=.5]{$W$};
	\draw[->] (1,4.75) node[anchor=east]{$x_2$} -- ++ (3.1,0);
	
	\draw[->] (1,6.5) node[anchor=east]{$\tilde{x}_3$} -- ++ (0.7,0) ;
	\draw[->] (3.4,4.25) -- ++(0.7,0);
	\draw[->] (3.4,5.25) -- ++(0.7,0);
	\draw[-] (3.4,5.25) -- ++ (0,1.25);
	\draw[-] (3.4,4.25) -- ++ (0,-1);
	\draw[-] (2.7, 6.5) -- ++(0.7,0);
	\draw[-] (2.7, 6.5) -- ++(0.7,0) node[anchor=west]{{${x}_1$}};
	\draw[-] (1,3.25) node[anchor=east]{$x_3$} -- ++ (2.4,0);

	\draw[->] (5.4,4.75) --  ++ (0.5,0)node[anchor= west] {$y$};
\end{tikzpicture}
\end{subfigure}
\begin{subfigure}[b]{.49\columnwidth}
\centering
\begin{tikzpicture}[scale=0.5]
	\draw (1.7,6) rectangle ++(1,1) node[pos=.5]{$q$};
	\draw (4.1,4) rectangle ++(1.3,1.5) node[pos=.5]{$W$};
	\draw[->] (1,4.75) node[anchor=east]{$x_2$} -- ++ (3.1,0);

	\draw[->] (1,6.5) node[anchor=east]{$x_3$} -- ++ (0.7,0) ;
	\draw[->] (3.4,4.25) -- ++(0.7,0);
	\draw[->] (3.4,5.25) -- ++(0.7,0);
	\draw[-] (3.4,5.25) -- ++ (0,1.25);
	\draw[-] (3.4,4.25) -- ++ (0,-1);
	\draw[-] (2.7, 6.5) -- ++(0.7,0);
	\draw[-] (2.7, 6.5) -- ++(0.7,0) node[anchor=west]{{$\tilde{x}_1$}};
	\draw[-] (1,3.25) node[anchor=east]{$\tilde{x}_3$} -- ++ (2.4,0);

	\draw[->] (5.4,4.75) --  ++ (0.5,0)node[anchor= west] {$y$};
\end{tikzpicture}
\end{subfigure}
\caption{We say $\mach$ is $\cX_3|\cX_2$-{symmetrizable by}~$\cX_1$ if, for each $(x_2,x_3,\tlx_3)$, the conditional output distributions in the two cases above are the same. The receiver is unable to tell whether user-3 is sending $x_3$ or $\tlx_3$.}
\label{fig:symm2}
\end{figure}
\begin{figure}[h]
\centering
\begin{subfigure}[t]{.49\columnwidth}
\centering
\begin{tikzpicture}[scale=0.5]
	\draw (1.7,6) rectangle ++(1,1) node[pos=.5]{$q$};
	\draw (4.1,4) rectangle ++(1.3,1.5) node[pos=.5]{$W$};
	\draw[->] (1,4.75) node[anchor=east]{$x_2$} -- ++ (3.1,0);
	
	\draw[->] (1,6.8) node[anchor=east]{$\tlx_1$} -- ++ (0.7,0) ;
	\draw[->] (1,6.2) node[anchor=east]{$\tilde{x}_3$} -- ++ (0.7,0) ;
	\draw[->] (3.4,4.25) -- ++(0.7,0);
	\draw[->] (3.4,5.25) -- ++(0.7,0);
	\draw[-] (3.4,5.25) -- ++ (0,1.25);
	\draw[-] (3.4,4.25) -- ++ (0,-1);
	\draw[-] (2.7, 6.5) -- ++(0.7,0);

	\draw[-] (2.7, 6.5) -- ++(0.7,0) node[anchor=west]{{$x_1$}};
	\draw[-] (1,3.25) node[anchor=east]{$x_3$} -- ++ (2.4,0);

	\draw[->] (5.4,4.75) --  ++ (0.5,0)node[anchor= west] {$y$};
\end{tikzpicture}
\end{subfigure}
\begin{subfigure}[t]{.49\columnwidth}
\centering
\begin{tikzpicture}[scale=0.5]
	\draw (1.7,4.25) rectangle ++(1,1) node[pos=.5]{$q'$};
	\draw (4.1,4) rectangle ++(1.3,1.5) node[pos=.5]{$W$};
	\draw[->] (1,4.45) node[anchor=east]{$x_3$} -- ++ (0.7,0);
	\draw[->] (1,5.05) node[anchor=east]{$x_2$} -- ++ (0.7,0);
	\draw[->] (2.7,4.75) -- ++ (1.4,0);

	\node at (3.2, 5.1) {{${\tlx}_2$}};
	\draw[-] (1,6.25) node[anchor=east]{$\tlx_1$} -- ++ (2.4,0) ;
	\draw[->] (3.4,4.25) -- ++(0.7,0);
	\draw[->] (3.4,5.25) -- ++(0.7,0);
	\draw[-] (3.4,5.25) -- ++ (0,1);
	\draw[-] (3.4,4.25) -- ++ (0,-1);
	\draw[-] (1,3.25) node[anchor=east]{$\tilde{x}_3$} -- ++ (2.4,0);

	\draw[->] (5.4,4.75) --  ++ (0.5,0)node[anchor= west] {$y$};
\end{tikzpicture}
\end{subfigure}
\caption{We say $\mach$ is $\cX_3$-{symmetrizable by}~$\cX_1/\cX_2$ if, for each $(\tilde{x}_1,x_2,x_3,\tlx_3)$, the conditional output distributions in the two cases above are the same. The receiver is unable to tell whether user-3 is sending $x_3$ (and  user-1 being malicious) or user-3 is sending $\tlx_3$ (and user-2 being malicious).}
\label{fig:symm3}
\end{figure}
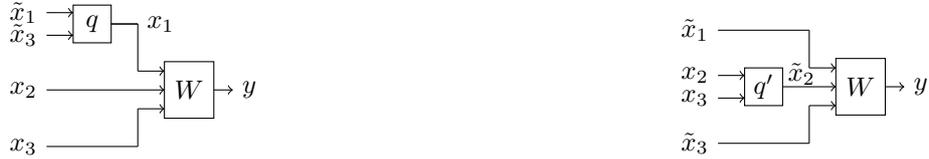

\paragraph{Symmetrizability conditions}\label{sec:sym_conditions_defn}
\begin{defn}\label{defn:symm}
Let $(i,j,k)$ be some permutation of $(1,2,3)$. We define three symmetrizability conditions for $\mach$ (See Fig.~\ref{fig:symm1}-\ref{fig:symm3}).
\begin{enumerate}
\item \label{item:symm1} We say that $\mach$ is $\cX_j\times\cX_k$-{\em symmetrizable by}~$\cX_i$ if for some distribution $q(x_i|x_j,x_k)$
\begin{align}
&\sum_{{x}_i} q({x}_i|\tlx_j,\tlx_k)\machijk(y|x_i,x_j,x_k)\notag\\
&\qquad = \sum_{\tilde{x}_i} q(\tilde{x}_i|x_j,x_k)\machijk(y|\tilde{x}_i,\tlx_j,\tlx_k),\notag\\
&\qquad\qquad\quad \forall\; x_j,\tlx_j\in\mathcal{X}_j,\; x_k,\tlx_k\in\mathcal{X}_k,\; y\in\mathcal{Y}. \label{eq:symm1}
\end{align}
\item \label{item:symm2} We say that $\mach$ is $\cX_k|\cX_j$-{\em symmetrizable by}~$\cX_i$ if for some distribution $q(x_i|x_k)$
\begin{align}
&\sum_{{x}_i} q(x_i|\tlx_k)\machijk(y|x_i,x_j,x_k)\notag\\
&\qquad = \sum_{\tilde{x}_i} q(\tlx_i|x_k)\machijk(y|\tlx_i,x_j,\tlx_k),\notag\\
&\qquad\qquad\quad \forall\; x_j\in\mathcal{X}_j,\; x_k,\tlx_k\in\mathcal{X}_k,\; y\in\mathcal{Y}. \label{eq:symm2}
\end{align}
\item \label{item:symm3} We say that $\mach$ is $\cX_k$-{\em symmetrizable by}~$\cX_i/\cX_j$ if for some pair of distributions $q(x_i|\tlx_i,\tlx_k)$ and $q'(\tlx_j|x_j,x_k)$
\begin{align}
&\sum_{x_i} q(x_i|\tlx_i,\tlx_k)\machijk(y|x_i,x_j,x_k)\notag\\
&\qquad = \sum_{\tlx_j} q'(\tlx_j|x_j,x_k)\machijk(y|\tlx_i,\tlx_j,\tlx_k),\notag\\
&\qquad\quad\;\forall\; \tlx_i\in\mathcal{X}_i,\;x_j\in\mathcal{X}_j,\; x_k,\tlx_k\in\mathcal{X}_k,\; y\in\mathcal{Y}. \label{eq:symm3}
\end{align}
\end{enumerate}
\end{defn}
We say that {\em user-$k$ is symmetrizable} if any of the above three symmetrizability conditions \eqref{eq:symm1}-\eqref{eq:symm3} holds for some distinct $i,j \in \{1,2,3\}\setminus\{k\}$.{We say that the channel is {\em not symmetrizable} if  user-$k$ is not symmetrizable for every $k\in\{1,2,3\}$.} \pink{In Section~\ref{sec:k_user_det}, we generalize the symmetrizability conditions to more than three users and provide a unified way of looking at them.}

\pink{The first two symmetrizability conditions arise from the possibility that the decoder cannot
tell apart different messages of honest user(s) when a particular user behaves adversarially.
These symmetrizability conditions are thus inherited from those for the AV-MAC model. Specifically, symmetrizability conditions for the two-user AV-MAC with ${X}_i$ as the state and ${X}_j,{X}_k$ as the inputs are also symmetrizability conditions for our problem. Thus, the first two conditions \eqref{eq:symm1}-\eqref{eq:symm2} (Figures~\ref{fig:symm1} and~\ref{fig:symm2}) follow from two-user AV-MAC symmetrizability conditions given by Gubner~\cite{Gubner90}. Notice that \eqref{eq:symm1} involves a distribution $q(x_i|x_j, x_k)$ whereas \eqref{eq:symm2} involves $q(x_i|x_k)$. The third condition \eqref{eq:symm3}(Figure~\ref{fig:symm3}) is new (see Section~\ref{para:new_cond}) and arises from the byzantine nature of the users in this problem. In a \bmac, the decoder may not be able to tell apart two messages since while one message is explained by the possibility of another user (say $j$) behaving adversarially, the other message may be explained by the
possibility of a third user (say $k$) behaving adversarially.  We discuss the implications of the third condition in Section~\ref{sec:3.2.1.2} where we argue that a symmetrizable user cannot communicate reliably using deterministic codes. }

\paragraph{Symmetrizability implies non-feasibility of communication.}\label{sec:3.2.1.2}
Suppose that \eqref{eq:symm3} holds for $(i, j, k) = (1, 2, 3)$. 
Thus, user-3 is symmetrizable. 
{In the following, we show that user-3 cannot communicate reliably at positive rates.} 
For fixed vectors ($\tilde{\vecx}_1,  \vecx_2, \vecx_3, \tilde{\vecx}_3$), {Eq.~}\eqref{eq:symm3} with $(i, j, k) = (1, 2, 3)$ implies that the output is same under the following two cases: 
{\begin{enumerate}[label=(\roman*)]
\item User 1 sends $\vecX_1\sim q^n(.|\tilde{\vecx_1}, \tilde{\vecx}_3)$, {\em i.e.}, a vector distributed as the output of the memoryless channel $q$ on input $(\tilde{\vecx}_1, \tilde{\vecx}_3)$ (see Figure~\ref{fig:symm3} which depicts single use of the channel), user-2 and user-3 send $\vecx_2$ and $\vecx_3$ respectively, and 
\item User 2 sends $\vecX_2\sim (q')^n(.|\vecx_2, {\vecx_3})$, and user-1 and user-3 send $\tilde{\vecx}_1$ and $\tilde{\vecx}_3$ respectively. 
\end{enumerate}}
Hence, for a given $(\cM_1, \cM_2, \cM_3, n)$ code $(f_1, f_2, f_3, \phi)$ and independent $\tilde{M}_1\sim \textsf{Unif}(\cM_1)$, $M_2\sim \textsf{Unif}(\cM_2)$, $M_3\sim \textsf{Unif}(\cM_3)$ and $\tilde{M}_3\sim \textsf{Unif}(\cM_3)$, the output {distributions are identical }in the following two cases:
{\begin{enumerate}[label=(\roman*)]
\item  User 1 sends $\vecX_1\sim q^n(.|f_1(\tilde{M}_1), f_3(\tilde{M}_3))$, user-2 and user-3 send $f_2(M_2)$ and $f_3(M_3)$ respectively, and 
\item User 2 sends $\vecX_2\sim (q')^n(.|f_2(M_2), f_3(M_3))$, and user-1 and user-3 send $f_1(\tilde{M}_1)$ and $f_3(\tilde{M}_3)$ respectively. 
\end{enumerate}}
Thus, the receiver is unable to tell apart the two possibilities, {\em i.e.}, whether user-1 is malicious with user-3 sending $M_3$ or user-2 is malicious with user-3 sending $\tilde{M}_3$. We can argue along the similar lines to show that the symmetrized user(s) in \eqref{eq:symm1} or \eqref{eq:symm2} cannot communicate reliably. On the other hand, we can show that when no user is symmetrizable, we can work at positive rates. This brings us to our main result.

\paragraph{Deterministic capacity region.}
Let $\mathcal{R}$ be the set of all rate triples $(R_1,R_2,R_3)$ such that for some $p(u)p(x_1|u)p(x_2|u)p(x_3|u)$, the following conditions hold for all permutations $(i,j,k)$ of $(1,2,3)$:
\begin{align}
R_i &\leq \min_{q(x_k|u)} I(X_i;Y|X_j,U),\quad\text{and} \label{eq:rateconstraint1}\\
R_i+R_j &\leq \min_{q(x_k|u)} I(X_i,X_j;Y|U), \label{eq:rateconstraint2}
\end{align}
where the mutual information terms above are evaluated using the joint distribution
$p(u) p(x_i|u)p(x_j|u)q(x_k|u) \allowbreak \mach(y|x_1,x_2,x_3)$. Here, \pink{$U$  is an auxiliary random variable distributed over some alphabet $\cU$ with $|\mathcal{U}|\leq 6$. The bound on the cardinality of $\cU$ can be shown using the convex cover method~\cite[Appendix~C]{YHKEG}.}
\begin{thm}\label{thm:symmetrizability}
$\mathcal{R}_{\deterministic} = \mathcal{R}$ if $\mach$ is not symmetrizable. Furthermore, $\textup{int}(\mathcal{R}_{\deterministic})\neq\varnothing$ only if $\mach$ is not symmetrizable.
\end{thm} 
\begin{remark}\label{rem:gap} As argued, we prove the converse part of Theorem~\ref{thm:symmetrizability} by showing that if user-$k$ is symmetrizable, then any achievable rate triple $(R_1,R_2,R_3)$ must be such that $R_k=0$. Our capacity region characterization does not cover the case where some (but not all) users are symmetrizable. In this case, by Theorem~\ref{thm:random} (which shows that $\cR_{\random} = \cR$), $\mathcal{R}$ restricted to rates of non-symmetrizable users is clearly an outer bound on $\mathcal{R}_{\deterministic}$. It is tempting to conjecture that these regions are equal. A similar result for the two-user AV-MAC was recently studied by Pereg and Steinberg~\cite{PeregS19} for the case where the users can privately randomize. \end{remark}

\paragraph{Overview of the proof of Theorem~\ref{thm:symmetrizability}.}\label{para:det_Achiev} 
The {detailed} proof of Theorem~\ref{thm:symmetrizability} is given in Section~\ref{sec:det_proofs}. 
Here, we describe the main ideas behind the achievability. The codebooks used in the achievability are obtained by a random coding argument (see Lemma~\ref{lemma:codebook}). We will
briefly describe the decoder here and also point out  its connection to the non-symmetrizability of the channel. \pink{A high-level proof-idea is also given in the flowchart in Figure~\ref{fig:flowchart1}.}

\begin{figure}[h]
\centering
\begin{subfigure}[t]{.49\columnwidth}
\centering
\begin{tikzpicture}[scale=0.5]
	\draw (1.7,6) rectangle ++(1,1) node[pos=.5]{$q$};
	\draw (4.1,4) rectangle ++(1.3,1.5) node[pos=.5]{$W$};
	\draw[->] (1,4.75) node[anchor=east]{{\color{myblue}{$\scriptstyle{f_2(m_2)}$}} $\scriptstyle{X_2}$} -- ++ (3.1,0);

	\draw[->] (1,6.8) node[anchor=east]{{\color{myblue}{$ \scriptstyle{f_2(m'_2)}$}} $\scriptstyle{X'_2}$} -- ++ (0.7,0) ;
	\draw[->] (1,6.2) node[anchor=east]{{\color{myblue}{$ \scriptstyle{f_3(m'_3)}$}} $\scriptstyle{X'_3}$} -- ++ (0.7,0) ;
	\draw[->] (3.4,4.25)  -- ++(0.7,0);
	\draw[->] (3.4,5.25) -- ++(0.7,0);
	\draw[-] (3.4,5.25) -- ++ (0,1.25);
	\draw[-] (3.4,4.25) -- ++ (0,-1);
	\draw[-] (2.7, 6.5) -- ++(0.7,0) node[anchor=west]{{$\scriptstyle{{X}_1}$} {\color{myblue}$\scriptstyle{\vecx_1}$}};
	\draw[-] (1,3.25) node[anchor=east]{{\color{myblue}$\scriptstyle{f_3(m_3)}$} $\scriptstyle{X_3}$} -- ++ (2.4,0);

	\draw[->] (5.4,4.75) --  ++ (0.5,0)node[anchor= west] {$\scriptstyle{Y}$ {\color{myblue}$\scriptstyle{\vecy}$}};
\end{tikzpicture}
\caption{$X'_2 X'_3-X_1-X_2X_3Y$ Markov chain holds approximately.}
\label{fig:symm1aa}
\end{subfigure}
\begin{subfigure}[t]{.49\columnwidth}
\centering
\begin{tikzpicture}[scale=0.5]
	\draw (1.7,6) rectangle ++(1,1) node[pos=.5]{$q$};
	\draw (4.1,4) rectangle ++(1.3,1.5) node[pos=.5]{$W$};
	\draw[->] (1,4.75) node[anchor=east]{{\color{myblue}{$\scriptstyle{f_2(m'_2)}$}} $\scriptstyle{X'_2}$} -- ++ (3.1,0);

	\draw[->] (1,6.8) node[anchor=east]{{\color{myblue}{$ \scriptstyle{f_2({m}_2)}$}} $\scriptstyle{{X}_2}$} -- ++ (0.7,0) ;
	\draw[->] (1,6.2) node[anchor=east]{{\color{myblue}{$ \scriptstyle{f_3({m}_3)}$}} $\scriptstyle{{X}_3}$} -- ++ (0.7,0) ;
	\draw[->] (3.4,4.25)  -- ++(0.7,0);
	\draw[->] (3.4,5.25) -- ++(0.7,0);
	\draw[-] (3.4,5.25) -- ++ (0,1.25);
	\draw[-] (3.4,4.25) -- ++ (0,-1);
	\draw[-] (2.7, 6.5) -- ++(0.7,0) node[anchor=west]{{$\scriptstyle{X'_1}$} {\color{myblue}$\scriptstyle{\vecx'_1}$}};
	\draw[-] (1,3.25) node[anchor=east]{{\color{myblue}$\scriptstyle{f_3(m'_3)}$} $\scriptstyle{X'_3}$} -- ++ (2.4,0);

	\draw[->] (5.4,4.75) --  ++ (0.5,0)node[anchor= west] {$\scriptstyle{Y}$ {\color{myblue}$\scriptstyle{\vecy}$}};
\end{tikzpicture}
\caption{${X}_2 {X}_3-X'_1-X'_2X'_3Y$ Markov chain holds approximately.}
\label{fig:symm1b}
\end{subfigure}
\caption{The subfigure $(a)$  above describes the decoding condition \ref{cond:1}. The quantities in blue describe it operationally while the random variables describe the single-letter joint distribution. Non-symmetrizability implies that, for $(\vecx_1,\vecx'_1, f_2(m_2),f_2(m'_2), f_3(m_3),   f_3(m'_3), \vecy)\in T^n_{X_1X'_1X_2X'_2X_3 X'_3Y}$, both $(a)$ and $(b)$ cannot hold simultaneously (see Fig~\ref{fig:symm1}). }
\label{fig:symm1_dec}
\end{figure}
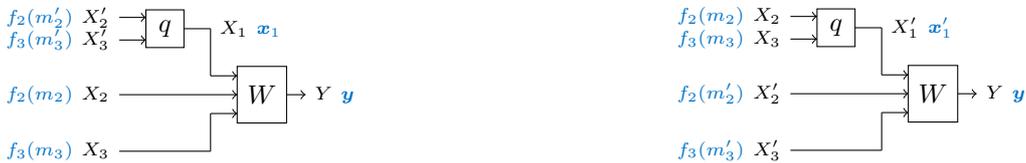

\begin{figure}[h]
\centering
\begin{subfigure}[b]{.49\columnwidth}
\centering
\begin{tikzpicture}[scale=0.5]
	\draw (1.7,6) rectangle ++(1,1) node[pos=.5]{$q$};
	\draw (4.1,4) rectangle ++(1.3,1.5) node[pos=.5]{$W$};
	\draw[->] (1,4.75) node[anchor=east]{{\color{myblue}{$ \scriptstyle{f_2({m}_2)}$}} $\scriptstyle{{X}_2}$} -- ++ (3.1,0);
	
	\draw[->] (1,6.5) node[anchor=east]{{\color{myblue}$\scriptstyle{f_3(m'_3)}$} $\scriptstyle{X'_3}$} -- ++ (0.7,0) ;
	\draw[->] (3.4,4.25) -- ++(0.7,0);
	\draw[->] (3.4,5.25) -- ++(0.7,0);
	\draw[-] (3.4,5.25) -- ++ (0,1.25);
	\draw[-] (3.4,4.25) -- ++ (0,-1);
	\draw[-] (2.7, 6.5) -- ++(0.7,0);
	\draw[-] (2.7, 6.5) -- ++(0.7,0) node[anchor=west]{{$\scriptstyle{{X}_1}$} {\color{myblue}$\scriptstyle{\vecx_1}$}};
	\draw[-] (1,3.25) node[anchor=east]{{\color{myblue}{$ \scriptstyle{f_3({m}_3)}$}} $\scriptstyle{{X}_3}$} -- ++ (2.4,0);

	\draw[->] (5.4,4.75) --  ++ (0.5,0)node[anchor= west] {$\scriptstyle{Y}$ {\color{myblue}$\scriptstyle{\vecy}$}};
\end{tikzpicture}
\caption{$X'_3-X_1-X_2X_3Y$ Markov chain holds approximately.}
\label{fig:symm2aa}
\end{subfigure}
\begin{subfigure}[b]{.49\columnwidth}
\centering
\begin{tikzpicture}[scale=0.5]
	\draw (1.7,6) rectangle ++(1,1) node[pos=.5]{$q$};
	\draw (4.1,4) rectangle ++(1.3,1.5) node[pos=.5]{$W$};
	\draw[->] (1,4.75) node[anchor=east]{{\color{myblue}{$ \scriptstyle{f_2({m}_2)}$}} $\scriptstyle{{X}_2}$} -- ++ (3.1,0);

	\draw[->] (1,6.5) node[anchor=east]{{\color{myblue}{$ \scriptstyle{f_3({m}_3)}$}} $\scriptstyle{{X}_3}$} -- ++ (0.7,0) ;
	\draw[->] (3.4,4.25) -- ++(0.7,0);
	\draw[->] (3.4,5.25) -- ++(0.7,0);
	\draw[-] (3.4,5.25) -- ++ (0,1.25);
	\draw[-] (3.4,4.25) -- ++ (0,-1);
	\draw[-] (2.7, 6.5) -- ++(0.7,0);
	\draw[-] (2.7, 6.5) -- ++(0.7,0) node[anchor=west]{{$\scriptstyle{X'_1}$} {\color{myblue}$\scriptstyle{\vecx'_1}$}};
	\draw[-] (1,3.25) node[anchor=east]{{\color{myblue}$\scriptstyle{f_3(m'_3)}$} $\scriptstyle{X'_3}$} -- ++ (2.4,0);

	\draw[->] (5.4,4.75) --  ++ (0.5,0)node[anchor= west] {$\scriptstyle{Y}$ {\color{myblue}$\scriptstyle{\vecy}$}};
\end{tikzpicture}
\caption{${X}_3-X'_1-X_2X'_3Y$ Markov chain holds approximately.}
\end{subfigure}
\caption{The subfigure $(a)$  above describes the decoding condition \ref{cond:2}. The quantities in blue describe it operationally while the random variables describe the single-letter joint distribution. Non-symmetrizability implies that, for $(\vecx_1,\vecx'_1,  f_2(m_2), f_3(m_3),  f_3(m'_3), \vecy)\in T^n_{X_1X'_1X_2X_3 X'_3Y}$, both $(a)$ and $(b)$ cannot hold simultaneously (see Fig~\ref{fig:symm2}).}
\label{fig:symm2_dec}
\end{figure}
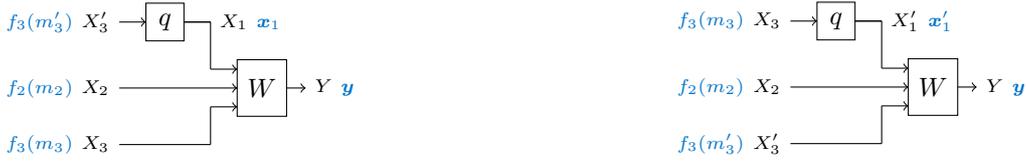
\begin{figure}[h]
\centering
\begin{subfigure}[t]{.49\columnwidth}
\centering
\begin{tikzpicture}[scale=0.5]
	\draw (1.7,6) rectangle ++(1,1) node[pos=.5]{$q$};
	\draw (4.1,4) rectangle ++(1.3,1.5) node[pos=.5]{$W$};
	\draw[->] (1,4.75) node[anchor=east]{{\color{myblue}{$ \scriptstyle{f_2({m}_2)}$}} $\scriptstyle{{X}_2}$} -- ++ (3.1,0);
	
	\draw[->] (1,6.8) node[anchor=east]{{\color{myblue}$\scriptstyle{f_1(m'_1)}$} $\scriptstyle{X'_1}$} -- ++ (0.7,0) ;
	\draw[->] (1,6.2) node[anchor=east]{{\color{myblue}$\scriptstyle{f_3(m'_3)}$} $\scriptstyle{X'_3}$} -- ++ (0.7,0) ;
	\draw[->] (3.4,4.25) -- ++(0.7,0);
	\draw[->] (3.4,5.25) -- ++(0.7,0);
	\draw[-] (3.4,5.25) -- ++ (0,1.25);
	\draw[-] (3.4,4.25) -- ++ (0,-1);
	\draw[-] (2.7, 6.5) -- ++(0.7,0);

	\draw[-] (2.7, 6.5) -- ++(0.7,0) node[anchor=west] {{$\scriptstyle{{X}_1}$} {\color{myblue}$\scriptstyle{{\vecx}_1}$}};
	\draw[-] (1,3.25) node[anchor=east]{{\color{myblue}{$ \scriptstyle{f_3({m}_3)}$}} $\scriptstyle{{X}_3}$} -- ++ (2.4,0);

	\draw[->] (5.4,4.75) --  ++ (0.5,0)node[anchor= west]{$\scriptstyle{Y}$ {\color{myblue}$\scriptstyle{\vecy}$}};
\end{tikzpicture}
\caption{$X'_1X'_3-X_1-X_2X_3Y$ Markov chain holds approximately.}
\label{fig:symm3aa}
\end{subfigure}
\begin{subfigure}[t]{.49\columnwidth}
\centering
\begin{tikzpicture}[scale=0.5]
	\draw (1.4,4.25) rectangle ++(1,1) node[pos=.5]{$q'$};
	\draw (4.1,4) rectangle ++(1.3,1.5) node[pos=.5]{$W$};
	\draw[->] (1,4.45) node[anchor=east]{{\color{myblue}{$ \scriptstyle{f_3({m}_3)}$}} $\scriptstyle{{X}_3}$} -- ++ (0.4,0);
	\draw[->] (1,5.05) node[anchor=east]{{\color{myblue}{$ \scriptstyle{f_2({m}_2)}$}} $\scriptstyle{{X}_2}$} -- ++ (0.4,0);
	\draw[->] (2.4,4.75) -- ++ (1.7,0);

	\node at (2.9, 5.2) {\color{myblue}{$ \scriptstyle{\vecx'_2}$}};
	\node at (2.9, 4.4) {{$ \scriptstyle{X'_2}$}};
	\draw[-] (1,6.25) node[anchor=east]{{\color{myblue}$\scriptstyle{f_1(m'_1)}$} $\scriptstyle{X'_1}$} -- ++ (2.4,0) ;
	\draw[->] (3.4,4.25) -- ++(0.7,0);
	\draw[->] (3.4,5.25) -- ++(0.7,0);
	\draw[-] (3.4,5.25) -- ++ (0,1);
	\draw[-] (3.4,4.25) -- ++ (0,-1);
	\draw[-] (1,3.25) node[anchor=east]{{\color{myblue}$\scriptstyle{f_3(m'_3)}$} $\scriptstyle{X'_3}$} -- ++ (2.4,0);

	\draw[->] (5.4,4.75) --  ++ (0.5,0)node[anchor= west] {$\scriptstyle{Y}$ {\color{myblue}$\scriptstyle{\vecy}$}};
\end{tikzpicture}
\caption{$X_2X_3-X'_2-X'_1X'_3Y$ Markov chain holds approximately.}
\end{subfigure}
\caption{The subfigure $(a)$ above describes the decoding condition \ref{cond:3}. The quantities in blue describe it operationally while the random variables describe the single-letter joint distribution. Non-symmetrizability implies that, for $(\vecx_1, f_1(m'_1),f_2(m_2),\vecx'_2, f_3(m_3),   f_3(m'_3), \vecy)\in T^n_{X_1X'_1X_2X'_2X_3 X'_3Y}$, both $(a)$ and $(b)$ cannot hold simultaneously (see Fig~\ref{fig:symm3}).}
\label{fig:symm3_dec}
\end{figure}
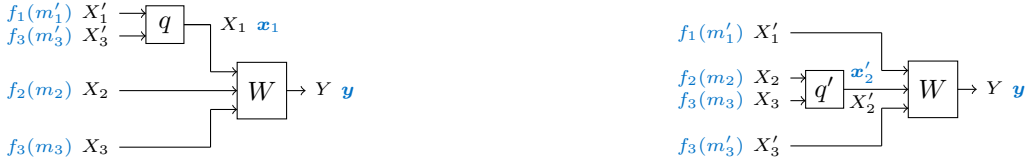

Upon receiving the channel output, the decoder works by separately collecting potential candidates for each user's input message and subjecting them to further checks. Finally, we will show that there will be at most one potential candidate for each user (see Lemma~\ref{lemma:dec}) which the decoder outputs. {In the following, we describe these steps for decoding user-3's message. Similar procedures are also employed for user-1 and user-2's decoding.

Let $\cC_i, \, i=1, 2, 3$ denote the codebook of user-$i$. {Given a received vector $\vecy\in \cY^n$, we say that the message $m_3\in \cM_3$ of user-3 is a ``candidate'' with an ``explanation'' $(\vecx_1,\vecx_2)\in \left(\mathcal{X}^n_1\times \cC_2\right)\cup \left(\cC_1\times\mathcal{X}^n_2\right)$ if the tuple $(\vecx_1,\vecx_2,f_3(m_3),\vecy)$ is jointly typical with respect to a joint distribution that corresponds to independent channel inputs and the channel output following the channel conditional distribution given the inputs. The choice of the set of explanations is motivated by the fact that at most one user can be malicious. Note that, in general, a candidate message may have multiple explanations.}
\begin{enumerate}
\item {The decoder first forms a list of all candidate messages of user-3 along with their explanations.}
\item {The list of such candidate messages is then pruned by only keeping those messages that ``account'' for every other candidate message in the sense described below. Suppose that the candidate message $m_3$ has an explanation of the form $(\vecx_1,f_2(m_2))$ for some $\vecx_1\in\cX_1^n$ and $m_2\in\cM_2$. Similar procedures are followed if the explanation for $m_3$ is of the form $(f_1(m_1),\vecx_2)$ by interchanging the roles of user-1 and user-2 below.
 Let  $m'_3$ be another candidate message. We say that $m_3$ accounts for $m'_3$ if one of the following three conditions is satisfied.}
\begin{enumerate}
\item \label{cond:1} $m'_3$ has an explanation $(\vecx'_1,f_2(m'_2))$  for some $m'_2\neq m_2$, such that the collection \\$(\vecx_1,  f_2(m_2), f_2(m'_2),f_3(m_3), f_3(m'_3),  \vecy)$ may be interpreted as typical instances drawn from a distribution $P_{X_1X_2X'_2X_3 X'_3Y}$ that specifies that $X'_2 X'_3$ and $X_2X_3Y$ are (roughly) conditionally independent given $X_1$. 

The condition \ref{cond:1} may be interpreted as follows: $\vecx_1$, $f_2(m_2)$, and $f_3(m_3)$  as inputs to the channel are a more plausible explanation of the channel output than the alternative input $(f_2(m_2'), f_3(m_3'))$, which is  part of (adversarial) user-$1$'s attack strategy (see Fig~\ref{fig:symm1a} and Fig.~\ref{fig:symm1aa}). It can be shown that for a non-symmetrizable channel, an analogue of Fig.~\ref{fig:symm1bb} (see Fig.~\ref{fig:symm1b}), which (roughly)  corresponds to the Markov chain ${X}_2 {X}_3-X_1-X'_2X'_3Y$,  cannot simultaneously hold (also see proof of Lemma~\ref{lemma:dec}).

\item \label{cond:2} $m'_3$ has an explanation $(\vecx'_1,f_2(m_2))$ such that the collection $(\vecx_1,  f_2(m_2), f_3(m_3),f_3(m'_3), \vecy)$ may be interpreted as typical instances drawn from a distribution $P_{X_1X_2X_3 X'_3Y}$ that specifies that $X'_3$ and $X_2X_3Y$ are (roughly) conditionally independent given $X_1$ (see Fig.~\ref{fig:symm2_dec}).  
\item \label{cond:3} $m'_3$ has an explanation $(f_1(m'_1),\vecx'_2)$ such that the collection $(\vecx_1, f_1(m'_1), f_2(m_2), f_3(m_3), f_3(m'_3),\vecy)$ may be interpreted as typical instances drawn from a distribution $P_{X_1X'_1X_2X_3 X'_3Y}$ that specifies that $X'_1X'_3$ and $X_2X_3Y$ are (roughly) conditionally independent given $X_1$ (see Fig.~\ref{fig:symm3_dec}).
\end{enumerate}
\end{enumerate}


 See the decoder definition below for a complete description, which accounts for all candidates. 
Items {\bf(a)} and {\bf(b)} in the decoder definition are similar to the decoding conditions in~\cite{AhlswedeC99} where user-$i$ is the adversary and $\vecx_i$ is the state. Item~{\bf(c)} is associated with our new non-symmetrizability criterion (see Fig.~\ref{fig:symm3}) and handles the situation in which an adversarial user tries to make another user appear adversarial while pretending to act honestly.

\paragraph*{Decoder:}{
Let $\eta>0$. For a received vector $\vecy\in \cY^n$, the decoder outputs $\phi(\vecy)=(m_1,m_2,m_3)\in \mathcal{M}_1\times\mathcal{M}_2\times\mathcal{M}_3$ if  $\vecy\in \dm{1}{m_1}\cap \dm{2}{m_2} \cap \dm{3}{m_3}$ where $\dm{i}{m_i}, \, i=1,2,3$ is defined as below.\\
$\vecy \in {\dm{3}{m_3}}$ if there exists  some permutation $(i, j)$ of $(1,2)$, $m_j \in \mathcal{M}_j,\,\vecx_i\in \mathcal{X}^n_i$, and random variables $X_i, X_j, X_3$ with $(\vecx_{i},f_j(m_j), f_3(m_3),  \vecy) \in T^{n}_{X_iX_jX_3Y}$ and $D(P_{X_iX_jX_3 Y}||P_{X_i}\times P_{X_j}\times P_{X_3}\times W)< \eta$ such that the following hold:
	\begin{description}			
			 \item[(a)] \underline{Disambiguating ($m_j,m_3$) from ($m_j',m'_3$):} For every $(m'_j, m'_3) \in \mathcal{M}_j\times\mathcal{M}_3$, $m'_j \neq m_j$, $m'_3\neq m_3$, $\vecx'_i\in \mathcal{X}^n_i$, and random variables $X_i',X_j', X_3'$ such that $(\vecx_i, \vecx'_i, f_j(m_j), f_j(m'_j),f_3(m_3), f_3(m'_3),  \vecy) \in T^{n}_{X_i X_i'X_j X_j' X_3X'_3 Y}$ and  $D(P_{X_i'X'_jX'_k Y}||P_{X'_1}\times P_{X'_j}\times P_{X'_k}\times W)< \eta$, we require that $I(X_jX_3Y;X'_jX'_3|X_i) < \eta$.

			\item[(b)] \underline{Disambiguating $m_3$ from $m_3'$:} For every $m'_3 \in \mathcal{M}_3$, $m'_3 \neq m_3$, $\vecx'_i\in \mathcal{X}^n_i$, and random variables $X_i',X_3'$ such that $(\vecx_i, \vecx'_i, f_j(m_j),f_3(m_3),f_3(m'_3), \allowbreak \vecy) \in T^{n}_{X_i X_i'X_j X_3 X_3' Y}$ and  $D(P_{X_i'X_jX_3' Y}||P_{X_i'}\times P_{X_j}\times P_{X_3'}\times W)< \eta$, we require that $I(X_jX_3Y;X_3'|X_i) < \eta$.

			\item[(c)] \underline{Disambiguating ($m_j,m_3$) from ($m_i,m_3'$):} If there exist $(m_i,m'_3) \in \mathcal{M}_i\times \mathcal{M}_3$, $m'_3 \neq m_3$, $\vecx_j\in \mathcal{X}^n_j$,  and random variables $X_i',X_j', X_3'$ such that $(\vecx_i,f_i(m_i),f_j(m_j), \vecx_j, f_3(m_3), f_3(m'_3), \vecy) \in T^{n}_{X_i X_i' X_j X_j' X_3 X_3'Y}$ and  \\$D(P_{X_i'X_j'X_3' Y}||P_{X_i'}\times P_{X_j}\times P_{X_3'}\times W)< \eta$, we require that $I(X_jX_3Y;X_i'X_3'|X_i) < \eta$.
	\end{description}
The decoding sets $\dm{1}{m_1}$ and $\dm{2}{m_2}$ are defined similarly. 
If $\vecy\notin \dm{1}{m_1}\cap \dm{2}{m_2} \cap \dm{3}{m_3}$ for {all} $(m_1, m_2, m_3)\in \cM_1\times\cM_2\times\cM_3$, the decoder outputs $(1,1,1)$.

} Note that the decoder is not well defined if  $\vecy\in \dm{1}{m_1}\cap \dm{2}{m_2} \cap \dm{3}{m_3}$ and $\vecy\in \dm{1}{m'_1}\cap \dm{2}{m'_2} \cap \dm{3}{m'_3}$ for $(m_1, m_2, m_3)\neq (m'_1, m'_2, m'_2)$. This is ruled out by
the following lemma {(proved in Appendix~\ref{app:disamb})} which guarantees that, for sufficiently small $\eta>0$, there is at most one triple $(m_1, m_2, m_3)$ such that $\vecy\in \dm{1}{m_1}\cap \dm{2}{m_2} \cap \dm{3}{m_3}$. This is analogous to~\cite[Lemma~4]{CsiszarN88}.
\begin{lemma}[Disambiguity of decoding]\label{lemma:dec}
Suppose the channel \mach is not symmetrizable. Let $P_{X_1}\in \mathcal{P}^n_{\cX_1}, P_{X_2} \in \mathcal{P}^n_{\cX_2}$ and $P_{X_3}\in \mathcal{P}^n_{\cX_3}$ be distributions such that for some $\alpha >0,$ $\min_{x_1}P_{X_1}(x_1),\min_{x_2}P_{X_2}(x_2),\allowbreak \min_{x_3}P_{X_3}(x_3)\geq \alpha$. Let $f_1:\mathcal{M}_1\rightarrow T^{n}_{X_1}, \, f_2:\mathcal{M}_2\rightarrow T^{n}_{X_2}$ and $f_3:\mathcal{M}_3\rightarrow T^{n}_{X_3}$ be any encoding maps. There exists a choice of $\eta>0$  such that if $(\tilde{m}_1,\tilde{m}_2,\tilde{m}_3)\neq(m_1,m_2,m_3),$ then  $(\dm{1}{\tilde{m}_1}\cap\dm{2}{\tilde{m}_2}\cap\dm{3}{\tilde{m}_3})\cap (\dm{1}{m_1}\cap\dm{2}{m_2}\cap\dm{3}{m_3}) = \emptyset$.
\end{lemma}

Notice that the decoder definition does not require consistency of the input message for the same user. For example, when $\vecy\in \dm{1}{m_1}\cap \dm{2}{m_2} \cap \dm{3}{m_3}$, in which case the decoder outputs $(m_1, m_2, m_3)$, the message $m_2$ plays no special role in  $\dm{1}{m_1}$ or $\dm{3}{m_3}$. That is, an ``explanation'' for the candidate $m_1$ may be $(f_2(\tilde{m}_2), \vecx_3)\in \cX^n_2\times\cX^n_3$ which passes checks $\mathbf{(a)},\, \mathbf{(b)}$ and $\mathbf{(c)}$ in the definition of $\dm{1}{m_1}$ where $\tilde{m}_2$ need not be same as $m_2$ or even be unique (for instance, there might be another simultaneous ``explanation'' $(f_2({m}'_2), \vecx'_3)$). At the same time, an ``explanation'' for the candidate $m_3$ which passes checks $\mathbf{(a)},\, \mathbf{(b)}$ and $\mathbf{(c)}$ of $\dm{3}{m_3}$ may be $(\vecx_1, f_2(\hat{m}_2))\in \cX^n_1\times\cX^n_2$  where $\tilde{m}_2$ need not be same as $\hat{m}_2$ or $m_2$.

\paragraph{``$\cX_k$-symmetrizable by~$\cX_i/\cX_j$'' is new.}\label{para:new_cond} The following example shows that the third symmetrizability condition \eqref{eq:symm3} does not imply  the others. The channel below is neither $\cX_j\times\cX_k$-{symmetrizable by}~$\cX_i$ nor $\cX_k|\cX_j$-{symmetrizable by}~$\cX_i$ for any permutation $(i,j,k)$ of $(1,2,3)$. However, it is $\cX_3$-symmetrizable by~$\cX_1/\cX_2$.
\begin{example}\label{ex:symmetrizable}
Let $\cX_1=\cX_2=\cY=\{0,1\}^3$ and $\cX_3=\{0,1\}$. 
Consider the channel $\mach$ (where the output is $Y=(Y_1, Y_2, Y_3)$) defined by
\begin{align*}
(Y_1,Y_2)&=(C_1,C_2),\\
Y_3&=\left\{\begin{array}{ll} 
B_1\oplus(A_1\odot X_3) &\text{w.p. } 1/2\\
B_2\oplus(A_2 \odot X_3) &\text{w.p. } 1/2
\end{array}\right.
\end{align*}
where $\odot$ denotes multiplication and $\oplus$ denotes addition modulo~2.

To see that this channel is $\cX_3$-symmetrizable by~$\cX_1/\cX_2$, consider the ``deterministic'' $q((a_1,b_1,c_1)|(\tla_1,\tlb_1,\tlc_1),\tlx_3)$ and $q'((\tla_2,\tlb_2,\tlc_2)|(a_2,b_2,c_2),x_3)$, for $(a_1,b_1,c_1),(\tla_1,\tlb_1,\tlc_1)\in \cX_1, (\tla_2,\tlb_2,\tlc_2),(a_2,b_2,c_2)\in \cX_2$ and $\tlx_3, x_3\in \cX_3$, defined as follows: let $g,g':\{0,1\}^4\rightarrow\{0,1\}^2$ be defined as
\begin{align*}
g((\tla_1,\tlb_1,\tlc_1),\tlx_3) &= (0,\tlb_1 \oplus (\tla_1\odot \tlx_3),\tlc_1),\\
g'((a_2,b_2,c_2),x_3) &= (0,b_2 \oplus (a_2\odot x_3),c_2).
\end{align*}
Then 
\begin{align*}
q((a_1,b_1,c_1)|(\tla_1,\tlb_1,\tlc_1),\tlx_3)
&= 1_{\{(a_1,b_1,c_1)=g((\tla_1,\tlb_1,\tlc_1),\tlx_3)\}},\\
q'((\tla_2,\tlb_2,\tlc_2)|(a_2,b_2,c_2),x_3)
&= 1_{\{(\tla_2,\tlb_2,\tlc_2)=g'((a_2,b_2,c_2),x_3)\}}.
\end{align*}
Consider the two cases shown in Figure~\ref{fig:symm3} with $\tlx_1=(\tla_1,\tlb_1,\tlc_1)$, $x_2=(a_2,b_2,c_2)$, and $q$ and $q'$ defined as above.
It follows that, in both the cases, the channel output $Y$ has the same conditional distribution given each input. In particular,
\begin{align*}
(Y_1,Y_2)&=(c_1,c_2),\\
Y_3&=\left\{\begin{array}{ll} 
\tlb_1\oplus(\tla_1\odot \tlx_3) &\text{w.p. } 1/2\\
b_2\oplus(a_2 \odot x_3) &\text{w.p. } 1/2.
\end{array}\right.
\end{align*}
\noindent This shows that the symmetrizability condition~\eqref{eq:symm3} holds for $(i,j,k)=(1,2,3)$. 

Since $(Y_1,Y_2)=(C_1,C_2)$, it is clear that neither user-1 nor user-2 is symmetrizable. It only remains {to be shown} that the channel is neither $\cX_3|\cX_2$-{symmetrizable by}~$\cX_1$ nor $\cX_3|\cX_1$-{symmetrizable by}~$\cX_2$. Suppose the channel is $\cX_3|\cX_2$-{symmetrizable by}~$\cX_1$. Then, to satisfy \eqref{eq:symm2} for $x_2=(0,0,c_2)$ and $(x_3,\tlx_3)=(0,1)$, it {must hold that}
\begin{align*}
&q(0,0,0|1)+q(0,0,1|1)+q(1,0,0|1)+q(1,0,1|1)\\
&\qquad= q(1,1,0|0)+q(1,1,1|0)+q(0,0,0|0)+q(0,0,1|0).
\end{align*}
However, to satisfy \eqref{eq:symm2} for $x_2=(1,0,c_2)$ and $(x_3,\tlx_3)=(0,1)$, {we must also satisfy
\begin{align*}
&1+q(0,0,0|1)+q(0,0,1|1)+q(1,0,0|1)+q(1,0,1|1)\\
&\qquad= q(1,1,0|0)+q(1,1,1|0)+q(0,0,0|0)+q(0,0,1|0),
\end{align*}which is a contradiction.}
Hence, the channel is not $\cX_3|\cX_2$-{symmetrizable by}~$\cX_1$. By symmetry, it is {also} not  $\cX_3|\cX_1$-{symmetrizable by}~$\cX_2$.
\end{example}
\pink{Next, the following examples also show that none of the three types of symmetrizability conditions given in Definition~\ref{defn:symm} are redundant given the others. Example~\ref{ex:1symmetrizable_main_text1} gives a channel that is $\cX_k|\cX_j$-{\em symmetrizable by}~$\cX_i$ for every permutation $(i,j,k)$ of $(1,2,3)$ but does not satisfy other any other symmetrizability condition from Definition~\ref{defn:symm}. Example~\ref{ex:1symmetrizable_main_text2} gives a channel that is $\cX_1\times\cX_2$-{\em symmetrizable by}~$\cX_3$ but does not satisfy other forms of symmetrizability conditions ({\em i.e.}, conditions of the form 2 and 3 in Definition~\ref{defn:symm}). We skip the detailed proofs here as these properties can be verified following similar arguments as the AVMAC examples from~\cite{Gubner90} and~\cite{AhlswedeC99}.} 
\pink{\begin{example}[{\cite[Example~on~pg.~264]{Gubner90}}]\label{ex:1symmetrizable_main_text1} Let $\cX_1=\cX_2=\cX_3=\{0,1\}$ and $\cY=\{0,1,2,3\}$. Consider the channel $\mach$ defined by \[Y=X_1+X_2+X_3,\] where $+$ denotes addition over integers. 
This channel is $\cX_k|\cX_j$-{\em symmetrizable by}~$\cX_i$ but is neither $\cX_j\times\cX_k$-{\em symmetrizable by}~$\cX_i$ nor $\cX_k$-{\em symmetrizable by}~$\cX_i/\cX_j$ for any permutation $(i,j,k)$ of $(1,2,3)$. 
\end{example}}
\pink{\begin{example}[{\cite[Example~1]{AhlswedeC99}}]~\label{ex:1symmetrizable_main_text2} Let $\cX_1=\cX_2=\cX_3=\cY=\{0,1\}$. Consider the channel $\mach$ defined by 
\begin{align*}
	Y=\begin{cases}
		X_3 &  X_1\oplus X_2\oplus X_3=0,\\
		Z\sim\mbox{Ber}(1/2) &  X_1\oplus X_2\oplus X_3=1.
	\end{cases}
\end{align*}
The channel  is $\cX_1\times\cX_2$-{\em symmetrizable by}~$\cX_3$ but is neither   $\cX_k|\cX_j$-{\em symmetrizable by}~$\cX_i$ nor $\cX_k$-{\em symmetrizable by}~$\cX_i/\cX_j$ for any permutation $(i,j,k)$ of $(1,2,3)$.  \end{example}}

\subsubsection{Randomized coding capacity region}\label{sec:random_proof_sketch}
\begin{thm}\label{thm:random}
\pink{The randomized coding capacity region of the 3-user \bmac with at most one adversarial user is given by} \pink{\[ \mathcal{R}_{\random} = \mathcal{R}_{\random}^{\weak} = \mathcal{R}.\]}
\end{thm} 
\begin{remark}\label{remark:2}
The statement $\mathcal{R}_{\deterministic} = \mathcal{R},\text{ if }\textup{int}(\mathcal{R}_{\deterministic})\neq\varnothing$
 can also be shown directly using  the extension, provided in \cite{Jahn81},  of the elimination 
technique~\cite{Ahlswede78} to first show that
$n^2$-valued randomness at each encoder is sufficient to achieve
any rate-triple in $\mathcal{R}_{\random}^{\weak}$ (see Lemma~\ref{lemma:randomness_reduction}). A deterministic
code of small rate can be used to send $2\log_2 n$ bits out of each
message. These message bits are then used as the encoder randomness
in the next phase to communicate the rest of the message bits using
a randomized code.

\end{remark}

\pink{Below, we sketch the proof of achievability. \longonly{A detailed achievability proof and a converse proof for the weak adversary case are available in Section~\ref{app:A}}. \shortonly{Full details and a converse proof for the weak adversary case are available in the extended draft~\cite{ExtendedDraft}.} }
\vspace{-0.25em}
\begin{proof}[Proof sketch (achievability of Theorem~\ref{thm:random})] 
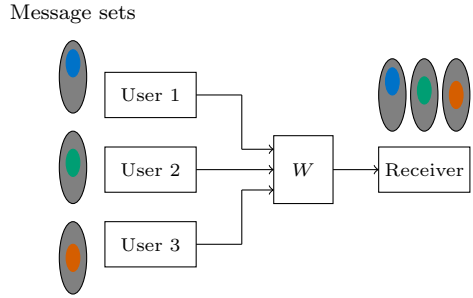
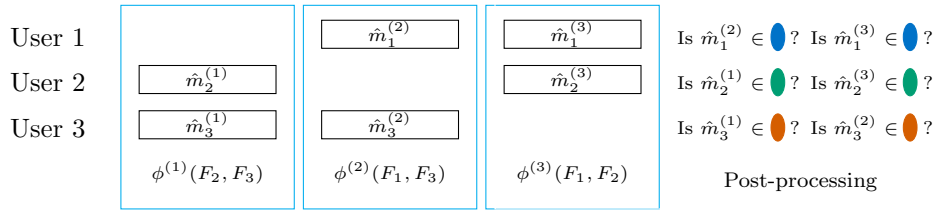
\begin{figure}
\centering
\begin{subfigure}[t]{.8\columnwidth}
\centering
    \begin{tikzpicture}[scale=0.6]

    \node at (0,8.8) {\footnotesize Message sets};
    \draw [rotate around={90:(0,7.4)},line width=0.4pt,color=black,fill=gray,fill opacity=0.2] (0,7.4) ellipse (0.8cm and 0.3cm);
\draw [rotate around={90:(0,7.7)},line width=0.4pt,color=myblue,fill=myblue,fill opacity=0.46] (0,7.7) ellipse (0.3cm and 0.15cm);

\draw [rotate around={90:(0,5.4)},line width=0.4pt,color=black,fill=gray,fill opacity=0.2] (0,5.4) ellipse (0.8cm and 0.3cm);
\draw [rotate around={90:(0,5.5)},line width=0.4pt,color=bluishgreen,fill=bluishgreen,fill opacity=0.46] (0,5.5) ellipse (0.3cm and 0.15cm);

\draw [rotate around={90:(0,3.4)},line width=0.4pt,color=black,fill=gray,fill opacity=0.2] (0,3.4) ellipse (0.8cm and 0.3cm);
\draw [rotate around={90:(0,3.4)},line width=0.4pt,color=vermillion,fill=vermillion,fill opacity=0.46] (0,3.4) ellipse (0.3cm and 0.15cm);

    \draw [rotate around={90:(7, 7)},line width=0.4pt,color=black,fill=gray,fill opacity=0.2] (7, 7) ellipse (0.8cm and 0.3cm);
\draw [rotate around={90:(7, 7.3)},line width=0.4pt,color=myblue,fill=myblue,fill opacity=0.46] (7, 7.3) ellipse (0.3cm and 0.15cm);

\draw [rotate around={90:(8.4, 7)},line width=0.4pt,color=black,fill=gray,fill opacity=0.2] (8.4, 7) ellipse (0.8cm and 0.3cm);
\draw [rotate around={90:(8.4, 7)},line width=0.4pt,color=vermillion,fill=vermillion,fill opacity=0.46] (8.4, 7) ellipse (0.3cm and 0.15cm);

\draw [rotate around={90:(7.7, 7)},line width=0.4pt,color=black,fill=gray,fill opacity=0.2] (7.7, 7) ellipse (0.8cm and 0.3cm);
\draw [rotate around={90:(7.7, 7.1)},line width=0.4pt,color=bluishgreen,fill=bluishgreen,fill opacity=0.46] (7.7, 7.1) ellipse (0.3cm and 0.15cm);



    \draw (0.7,6.5) rectangle ++(2,1) node[pos=.5]{\scriptsize User 1};
    \draw (4.4,4.6) rectangle ++(1.3,1.5) node[pos=.5]{\scriptsize $W$};
    \draw (1.7-1,4.85) rectangle ++(2,1) node[pos=.5]{\scriptsize User 2};
    \node[align=center] at (3.2,5.7) {};

    \draw[->] (2.7,5.35)  -- ++ (1.7,0) ;
    \draw[->] (3.7,4.9) -- ++ (0.7,0);
    \draw[-] (3.7, 4.9) -- ++(0,-0.7-0.5);
    \draw[-] (3.7, 4.2-0.5) -- node[below]{} ++(-1, 0);
    \draw (1.7-1,3.7-0.5) rectangle ++(2,1) node[pos=.5]{\scriptsize User 3};
    \draw[->] (3.7,5.80) -- ++(0.7,0);
    \draw[-] (3.7,5.8) -- ++ (0,0.7+0.5);
    \draw[-] (2.7, 6.5+0.5) -- node[above]{} ++(1,0);

    \draw[->] (5.7,5.35) -- node[above] {} ++ (1,0);
    \draw (6.7, 4.85) rectangle ++(2,1) node[pos=0.5]{\scriptsize Receiver};
\end{tikzpicture}
\caption{The figure shows the message sets for each user. For each message set, a subset (shown in color) is picked randomly using the independent randomness shared with the decoder. Only the messages in the colored subset are valid messages and used for communication in the \bmac.}
\end{subfigure}

\begin{subfigure}[t]{.8\columnwidth}
\centering
    \begin{tikzpicture}[scale=0.6]
    \node at (0,0) {User 3};
    \node at (0, 1) {User 2};
    \node at (0, 2) {User 1};

    \draw [rotate around={90:(15.5+0.5,0)},line width=0.4pt,color=vermillion,fill=vermillion,fill opacity=0.46] (15.5+0.5,0) ellipse (0.3cm and 0.15cm) ;
    \draw [rotate around={90:(18.1+0.8,0)},line width=0.4pt,color=vermillion,fill=vermillion,fill opacity=0.46] (18.1+0.8,0) ellipse (0.3cm and 0.15cm) ;
    \node at (16.1+0.5,0) {\scriptsize Is $\hat{m}_3^{(1)}\in\phantom{m}$?\,\,\,Is $\hat{m}_3^{(2)}\in\phantom{m}$?};

    \draw [rotate around={90:(15.5+0.5,1)},line width=0.4pt,color=bluishgreen,fill=bluishgreen,fill opacity=0.46] (15.5+0.5,1) ellipse (0.3cm and 0.15cm) ;
    \draw [rotate around={90:(18.1+0.8,1)},line width=0.4pt,color=bluishgreen,fill=bluishgreen,fill opacity=0.46] (18.1+0.8,1) ellipse (0.3cm and 0.15cm) ;
    \node at (16.1+0.5,1) {\scriptsize Is $\hat{m}_2^{(1)}\in\phantom{m}$?\,\,\,Is $\hat{m}_2^{(3)}\in\phantom{m}$?};

    \draw [rotate around={90:(15.5+0.5,2)},line width=0.4pt,color=myblue,fill=myblue,fill opacity=0.46] (15.5+0.5,2) ellipse (0.3cm and 0.15cm) ;
    \draw [rotate around={90:(18.1+0.8,2)},line width=0.4pt,color=myblue,fill=myblue,fill opacity=0.46] (18.1+0.8,2) ellipse (0.3cm and 0.15cm) ;
    \node at (16.1+0.5,2) {\scriptsize Is $\hat{m}_1^{(2)}\in\phantom{m}$?\,\,\,Is $\hat{m}_1^{(3)}\in\phantom{m}$?};

    \draw (2,-0.25) rectangle ++(3,0.62) node[pos=.5]{\scriptsize$\hat{m}_3^{(1)}$};
    \draw (6,-0.25) rectangle ++(3,0.62) node[pos=.5]{\scriptsize$\hat{m}_3^{(2)}$};
    \draw (2,0.75) rectangle ++(3,0.62) node[pos=.5]{\scriptsize$\hat{m}_2^{(1)}$};
    \draw (6,1.75) rectangle ++(3,0.62) node[pos=.5]{\scriptsize$\hat{m}_1^{(2)}$};
    \draw (10,1.75) rectangle ++(3,0.62) node[pos=.5]{\scriptsize$\hat{m}_1^{(3)}$}; 
    \draw (10,0.75) rectangle ++(3,0.62) node[pos=.5]{\scriptsize$\hat{m}_2^{(3)}$};
    \draw[cyan] (1.6,-1.8) rectangle ++(3.8,4.5);
    \draw[cyan] (5.6,-1.8) rectangle ++(3.8,4.5);
    \draw[cyan] (9.6,-1.8) rectangle ++(3.8,4.5);

    \draw[white,opacity=0.0] (9.8,-1.8) rectangle ++(3.8,5.5);

    \node at (3.5, -1) {\scriptsize $\phi^{(1)}(F_2, F_3)$};
    \node at (7.5, -1) {\scriptsize $\phi^{(2)}(F_1, F_3)$};
    \node at (11.5, -1) {\scriptsize $\phi^{(3)}(F_1, F_2)$};

    \node at (16.5, -1.2) {\footnotesize Post-processing };
    \end{tikzpicture}
        \caption{The decoder works in two steps. Suppose the encoders are $F_1$, $F_2$ and $F_3$. To decode user-$1$'s message, the decoder uses the decoder $\phi^{(2)}(F_1, F_3)$ of AV-MAC $W^{(2)}$ and $\phi^{(3)}(F_1, F_2)$ of AV-MAC $W^{(3)}$, to get candidates $\hat{m}_1^{(2)}$ and $\hat{m}_1^{(3)}$ respectively. These candidates pass through a further post-processing step where a candidate which does not belong to the set of valid messages is rejected. If user-$1$ is honest, at least one of the the decoders among $\phi^{(2)}(F_1, F_3)$ and $\phi^{(3)}(F_1, F_2)$ will output correctly with high probability. Since a small set of valid messages was chosen using independent shared randomness, a malicious user cannot correlate their attack with messages of honest user(s) with high probability. This ensures that erroneous messages are rejected in the post-processing step. The decoder for other users proceeds similarly.}
    \end{subfigure}
\caption{A figure depicting the decoder for the randomized code of a 3-user \bmac}\label{proofsketch:random}
    \end{figure}

The scheme is depicted in Figure~\ref{proofsketch:random}. The achievability uses the two-user AV-MAC randomized code used in the proof of \cite[Theorem~1]{Jahn81}.
Let $(R_1,R_2,R_3)$ be a rate triple such that, for some $\blue{p(u)p(x_1|u)p(x_2|u)p(x_3|u)}$, the following conditions hold for all permutations $(i,j,k)$ of $(1,2,3)$: \begin{align*}
R_i &< \min_{q(x_k\blue{|u})} I(X_i;Y|\blue{U,}X_j),\quad\text{and}\\
R_i+R_j &< \min_{q(x_k\blue{|u})} I(X_i,X_j;Y\blue{|U}),
\end{align*} 
with the mutual information terms evaluated using the joint distribution
\blue{$p(u)p(x_i|u)p(x_j|u)q(x_k|u)W(y|x_1,x_2,x_3)$.} \pink{We show the achievability of these rate triplets.}
\pink{Note that for the AV-MAC $W^{(k)}$, the rate pair $(R_i,R_j)$ is achievable by the first part of the direct result of ~\cite[Theorem~1]{Jahn81} (see \cite[Section~III-C]{Jahn81})}\footnote{\pink{Note that Jahn's proof does not involve the auxiliary random variable $U$. However, it can be easily incorporated along the lines of \cite{370123}.}\label{footnote:gubner}}. 
Here, $W^{(k)}$ is the two-user AV-MAC formed by the channel inputs from user-$k$ as the state and the remaining channel inputs as the inputs of the legitimate users of the AV-MAC. 

In order to design a code for the \bmac, for each user-$i\in \{1, 2, 3\}$, we consider the randomized encoder $F_i$ which maps each message independently to a codeword generated i.i.d. according to $p_i$. The realization of each user's codebook is shared with the decoder as independent shared randomness, {\em i.e.}, $F_1, F_2$ and $F_3$ are independent. We note that the encoders $F_i$ and $F_j$ are identical to the ones in the proof of \cite[Theorem~1, Section~III-C]{Jahn81} for the two user AV-MAC $W^{(k)}$. The decoder for the \bmac will be implemented using the decoders $\phi^{(1)}(F_2, F_3)$, $\phi^{(2)}(F_1, F_3)$ and $\phi^{(3)}(F_1, F_2)$ (from the proof of \cite[Theorem~1, Section~III-C]{Jahn81}) of AV-MACs $W^{(1)}$, $W^{(2)}$ and $W^{(3)}$ corresponding to the encoder pairs $(F_2, F_3)$, $(F_1, F_3)$ and $(F_1, F_2)$ respectively. 
It is clear that if, say, user-$3$ is malicious, and honest users-1 and 2 send $m_1$ and $m_2$ respectively, then the output $\inp{\hat{m}_1^{(3)}, \hat{m}_2^{(3)}}$ of the decoder $\phi^{(3)}(F_1, F_2)$ will match $(m_1, m_2)$ with high probability. 
However, since the decoder does not know the identity of the malicious user, there is an additional decoded message $\hat{m}_1^{(2)}$ from decoder $\phi^{(2)}(F_1, F_3)$ for user-1 (and similarly there is an additional decoded message for user-2 from $\phi^{(1)}(F_2, F_3)$). The message $\hat{m}_1^{(2)}$ can potentially be different from $\hat{m}_1^{(3)}$. This is because the decoder $\phi^{(2)}(F_1, F_3)$ assumes that user-$3$ is honest and no decoding guarantees are available for its output when user-$3$ is in fact malicious. When $\hat{m}_1^{(3)}\neq \hat{m}_1^{(2)}$, it is not clear what the receiver should output as the decoded message for user-1. This is where we can leverage the independent shared randomness shared by each user with the receiver.

We use a form of random hashing in order to add a further post-processing step which filters the outputs of the decoders of the AV-MACs as follows. Using the randomness they share with the receiver, each user randomly selects a subset of the original message set which is of nearly the same rate but is only a small fraction of original message set in size. These randomly selected subsets will be the valid message sets for communication in the \bmac. If the decoders of AV-MACs decode to messages which are not in these randomly selected subsets, they will be rejected in the post-processing step. Since these subsets are chosen using the independent shared randomness between each user and the receiver, their identity is hidden from the malicious user. For a malicious user-$3$, the output of $\phi^{(3)}(F_1, F_2)$ will be correct with high probability and will be accepted in the post-processing step as honest users-$1$ and $2$ will only send valid messages. On the other hand, the outputs of $\phi^{(1)}(F_2, F_3)$ and $\phi^{(2)}(F_1, F_3)$ will be rejected with high probability if they are different from the output of $\phi^{(3)}(F_1, F_2)$. This is because the size of valid message set is only a very small fraction of the original message set, so an arbitrary decoded message will fall outside the set of valid messages with high probability. This crucially uses the fact that these sets are constructed using independent shared randomness which protects the identity of the set of valid messages (and thus the set of valid codewords) and prevents the malicious user from correlating the attack with those messages. These ideas are formalized in Section~\ref{app:A}.
\end{proof}

\section{Proofs}\label{sec:proof_3_user}
 \pink{In this section, we present proofs of the results presented in the previous section for the three user \bmac with at most one adversary.}
\subsection{Deterministic coding capacity region (Theorem~\ref{thm:symmetrizability})}\label{sec:det_proofs}
\begin{proof}[Proof (Converse of Theorem~\ref{thm:symmetrizability})]
The outer bound on the rate region, when non-empty, follows from Theorem~\ref{thm:random}. Below, we show that a symmetrizable user cannot communicate.
\shortonly{
Clearly, symmetrizability conditions for the two-user AV-MAC with $\mathcal{X}_i$ as the state alphabet and $\mathcal{X}_j,\mathcal{X}_k$ as the input alphabets are also symmetrizability conditions for our problem. Conditions~1 and 2 follow from Gubner~\cite{Gubner90}.

To show condition~3, consider $(i,j,k)=(1,2,3)$, the other cases follow similarly. Suppose $q(\tlx_1|x_1,x_3)$ and $q'(\tlx_j|x_j,x_k)$ satisfy \eqref{eq:symm3}. {\em i.e.},
\begin{align}
&\sum_{\tlx_1} q(\tlx_1|x_1,\tlx_3)\mach(y|\tlx_1,x_2,x_3)\notag\\
&\qquad = \sum_{\tlx_2} q'(\tlx_2|x_2,x_3)\mach(y|x_1,\tlx_2,\tlx_3),\notag\\
&\qquad\quad\;\forall\; x_1\in\mathcal{X}_1,\;x_2\in\mathcal{X}_2,\; x_3,\tlx_3\in\mathcal{X}_3,\; y\in\mathcal{Y}. \label{eq:conversesymm}
\end{align}
Let $m_3,\tilde{m}_3\in\mathcal{M}_3$ be distinct, and let ${\vecx}_3=f_3(m_3)$ and $\tilde{\vecx}_3=f_3(\tilde{m}_3)$. We consider two different settings in which user-3 sends ${\vecx}_3$ and $\tilde{\vecx}_3$ respectively:
\begin{enumerate}
\item[(i)] In the first setting, user-1 is adversarial. It chooses an $M_1\sim\textup{Unif}(\mathcal{M}_1)$. Let ${\vecX}_1=f_1(M_1)$. To produce its input $\tilde{\vecX}_1$ to the channel, it passes $({\vecX}_1,\tilde{\vecx}_3)$ through $q^n$, the $n$-fold product of the channel $q(\tlx_1|x_1,x_3)$. user-2, being non-adversarial, sends as its input to the channel $\vecX_2=f_2(M_2)$, where $M_2\sim\textup{Unif}(\mathcal{M}_2)$. user-3 sends ${\vecx}_3$ corresponding to message ${m}_3$. The distribution of the received vector in this case is
\begin{align*}
\frac{1}{\nummsg_1\nummsg_2}\sum_{m_1,m_2}
\prod_{t=1}^n &\sum_{\tlx_{1,t}}q(\tlx_{1,t}|f_{1,t}(m_1),\tlx_{3,t})\\&  \mach(y_t|\tlx_{1,t},f_{2,t}(m_2),x_{3,t}).
\end{align*}
\item[(ii)] In the second setting, user-2 is adversarial. It chooses an $M_2\sim\textup{Unif}(\mathcal{M}_2)$. Let ${\vecX}_2=f_2(M_2)$. To produce its input $\tilde{\vecX}_2$ to the channel, it passes $({\vecX}_2,{\vecx}_3)$ through $q'^n$, the $n$-fold product of the channel $q'(\tlx_2|x_2,x_3)$. user-1, being non-adversarial now, sends as its input to the channel ${\vecX}_1=f_1(M_1)$, where $M_1\sim\textup{Unif}(\mathcal{M}_1)$. user-3 here sends $\tilde{\vecx}_3$ corresponding to message $\tilde{m}_3$. Here, the distribution of the received vector is
\begin{align*}
\frac{1}{\nummsg_1\nummsg_2}\sum_{m_1,m_2}
\prod_{t=1}^n &\sum_{\tlx_{2,t}}q'(\tlx_{2,t}|f_{2,t}(m_2),x_{3,t})\\&  \mach(y_t|f_{1,t}(m_1),\tlx_{2,t},\tlx_{3,t}).
\end{align*}
\end{enumerate}
By \eqref{eq:conversesymm}, the above two distributions are identical. Hence, for any decoder, the sum of probabilities of decoding error for messages $m_3$ and $\tilde{m}_3$ must be at least 1, {\em i.e.},
if we define $e^{3}_1(m_3,\tilde{\vecx}_1)\defineqq \frac{1}{\nummsg_2}\sum_{m_2'} e_1(\tilde{\vecx}_1,m_2',m_3)$ and similarly $e^{3}_2(\tilde{m}_3,\tilde{\vecx}_2)\defineqq\frac{1}{\nummsg_1}\sum_{m_1'} e_2(m_1',\tilde{\vecx}_2,\tilde{m}_3)$, then
\[ \pink{\bbE}_{\tilde{\vecX}_1}[e^{3}_1(m_3,\tilde{\vecX}_1)] + \pink{\bbE}_{\tilde{\vecX}_2}[e^{3}_2(\tilde{m}_3,\tilde{\vecX}_2)] \geq 1.\]
	Note that the distribution of $\tilde{\vecX}_1$ (resp. $\tilde{\vecX}_2$) does not depend on $m_3$ (resp. $\tilde{m}_3$). Arguing along the lines of~\cite[(3.29) in page~187]{CsiszarN88}, we can show that the average probability of error is at least 1/8. \shortonly{See the extended version~\cite{ExtendedDraft} for more details.}
\end{proof}

}

\longonly{ Clearly, symmetrizability conditions for the two-user AV-MAC with $\mathcal{X}_i$ as the state alphabet and $\mathcal{X}_j,\mathcal{X}_k$ as the input alphabets are also symmetrizability conditions for our problem. Conditions~\ref{item:symm1} and~\ref{item:symm2} follow from Gubner~\cite{Gubner90}.

To analyze the rate region when condition~\ref{item:symm3} holds, consider $(i,j,k)=(1,2,3)$, the other cases follow similarly. Suppose $q(\tlx_1|x_1,x_3)$ and $q'(\tlx_j|x_j,x_k)$ satisfy \eqref{eq:symm3}, {\em i.e.},
\begin{align}
&\sum_{\tlx_1} q(\tlx_1|x_1,\tlx_3)\mach(y|\tlx_1,x_2,x_3)\notag\\
&\qquad = \sum_{\tlx_2} q'(\tlx_2|x_2,x_3)\mach(y|x_1,\tlx_2,\tlx_3),\notag\\
&\qquad\quad\;\forall\; x_1\in\mathcal{X}_1,\;x_2\in\mathcal{X}_2,\; x_3,\tlx_3\in\mathcal{X}_3,\; y\in\mathcal{Y}. \label{eq:conversesymm}
\end{align}
Let $m_3,\tilde{m}_3\in\mathcal{M}_3$ be distinct, and let ${\vecx}_3=f_3(m_3)$ and $\tilde{\vecx}_3=f_3(\tilde{m}_3)$. We consider two different settings in which user-3 sends ${\vecx}_3$ and $\tilde{\vecx}_3$ respectively:
\begin{enumerate}
\item[(i)] In the first setting, user-1 is adversarial. It chooses an $M_1\sim\textup{Unif}(\mathcal{M}_1)$. Let ${\vecX}_1=f_1(M_1)$. To produce its input $\tilde{\vecX}_{1,\tilde{m}_3}$ to the channel, it passes $({\vecX}_1,\tilde{\vecx}_3)$ through $q^n$, the $n$-fold product of the channel $q(\tlx_1|x_1,x_3)$. user-2, being non-adversarial, sends as its input to the channel ${\vecX}_2=f_2(M_2)$, where $M_2\sim\textup{Unif}(\mathcal{M}_2)$. user-3 sends ${\vecx}_3$ corresponding to message ${m}_3$. The distribution of the received vector in this case is
\begin{align*}
\frac{1}{\nummsg_1\nummsg_2}\sum_{m_1,m_2}
\prod_{t=1}^n &\sum_{\tlx_{1,\tilde{m}_3,t}}q(\tlx_{1,\tilde{m}_3,t}|f_{1,t}(m_1),\tlx_{3,t}) \mach(y_t|\tlx_{1,\tilde{m}_3,t},f_{2,t}(m_2),x_{3,t}).
\end{align*}
\item[(ii)] In the second setting, user-2 is adversarial. It chooses an $M_2\sim\textup{Unif}(\mathcal{M}_2)$. Let ${\vecX}_2=f_2(M_2)$. To produce its input $\tilde{\vecX}_{2,m_3}$ to the channel, it passes $({\vecX}_2,{\vecx}_3)$ through $q'^n$, the $n$-fold product of the channel $q'(\tlx_2|x_2,x_3)$. user-1, being non-adversarial now, sends as its input to the channel ${\vecX}_1=f_1(M_1)$, where $M_1\sim\textup{Unif}(\mathcal{M}_1)$. user-3 here sends $\tilde{\vecx}_3$ corresponding to message $\tilde{m}_3$. Here, the distribution of the received vector is
\begin{align*}
\frac{1}{\nummsg_1\nummsg_2}\sum_{m_1,m_2}
\prod_{t=1}^n &\sum_{\tlx_{2,m_3,t}}q'(\tlx_{2,m_3,t}|f_{2,t}(m_2),x_{3,t})\mach(y_t|f_{1,t}(m_1),\tlx_{2,m_3,t},\tlx_{3,t}).
\end{align*}
\end{enumerate}
By \eqref{eq:conversesymm}, the above two distributions are identical. Hence, for any decoder, the sum of probabilities of decoding error for messages $m_3$ and $\tilde{m}_3$ must be at least 1, {\em i.e.},
if we define $e^{3}_1(m_3,\tilde{\vecx}_1)\defineqq \frac{1}{\nummsg_2}\sum_{m_2'} e_1(\tilde{\vecx}_1,m_2',m_3)$ and similarly $e^{3}_2(\tilde{m}_3,\tilde{\vecx}_2)\defineqq\frac{1}{\nummsg_1}\sum_{m_1'} e_2(m_1',\tilde{\vecx}_2,\tilde{m}_3)$, then
\shortonly{
\[ \pink{\bbE}_{\tilde{\vecX}_1}[e^{3}_1(m_3,\tilde{\vecX}_1)] + \pink{\bbE}_{\tilde{\vecX}_2}[e^{3}_2(\tilde{m}_3,\tilde{\vecX}_2)] \geq 1.\]}
\longonly{
\begin{align*}
\pink{\bbE}_{\tilde{\vecX}_{1,\tilde{m}_3}}[e^{3}_1(m_3,\tilde{\vecX}_{1,\tilde{m}_3})]& + \pink{\bbE}_{\tilde{\vecX}_{2,m_3}}[e^{3}_2(\tilde{m}_3,\tilde{\vecX}_{2,m_3})] = \\
&\quad \sum_{{\vecy}:\phi({\vecy})\neq m_3}\left(\frac{1}{\nummsg_1\nummsg_2}\sum_{m_1,m_2}\prod_{t=1}^n \sum_{\tlx_{1,t}}q(\tlx_{1,t}|f_{1,t}(m_1),\tlx_{3,t}) \mach(y_t|\tlx_{1,t},f_{2,t}(m_2),x_{3,t})\right)\\
&\quad+\sum_{{\vecy}:\phi({\vecy})\neq \tilde{m}_3}\left(\frac{1}{\nummsg_1\nummsg_2}\sum_{m_1,m_2}
\prod_{t=1}^n \sum_{\tlx_{2,t}}q'(\tlx_{2,t}|f_{2,t}(m_2),x_{3,t})\mach(y_t|f_{1,t}(m_1),\tlx_{2,t},\tlx_{3,t})\right)\\
&\quad\stackrel{\text{(a)}}{\geq} 1,
\end{align*}
where (a) follows from \eqref{eq:conversesymm}.

}
	Note that the distribution of $\tilde{\vecX}_1$ (resp. $\tilde{\vecX}_2$) does not depend on $m_3$ (resp. $\tilde{m}_3$). Arguing along the lines of~\cite[(3.29) in page~187]{CsiszarN88},\shortonly{ we can show that the average probability of error is at least 1/8. See the extended version~\cite{ExtendedDraft} for more details.}

\longonly{
\begin{align*}
2P_{e}(f_1,f_2,f_3,\phi)&\geq P_{e,1}+P_{e,2}\\
&\geq \frac{1}{N_3}\sum_{m_3}\pink{\bbE}_{\tilde{\vecX}_1}[e^{3}_1(m_3,\tilde{\vecX}_1)] + \frac{1}{N_3}\sum_{m_3}\pink{\bbE}_{{\vecX}_2}[e^{3}_1(m_3,{\vecX}_2)]
\end{align*}
for any attack vectors $\tilde{\vecX}_1$ and $\tilde{\vecX}_2$. In particular, for the attack vectors $\frac{1}{N_3}\sum_{\tilde{m}_3}\tilde{X}_{1,\tilde{m_3}}$ and $\frac{1}{N_3}\sum_{m_3}\tilde{X}_{2,m_3}$,
\begin{align*}
2P_{e}(f_1,f_2,f_3,\phi)\geq \frac{1}{N^2_3}\sum_{\tilde{m}_3}\sum_{m_3}\left(\pink{\bbE}_{\tilde{\vecX}_{1,\tilde{m}_3}}[e^{3}_1(m_3,\tilde{\vecX}_{1,\tilde{m}_3})] + \pink{\bbE}_{\tilde{\vecX}_{2,m_3}}[e^{3}_2(\tilde{m}_3,\tilde{\vecX}_{2,m_3})]\right).
\end{align*}
For $m_3\neq \tilde{m}_3$, the term in brackets on the right is upper bounded by 1, otherwise it is upper bounded by zero. Thus, 
\begin{align*}
P_{e}(f_1,f_2,f_3,\phi)&\geq \frac{N_3(N_3-1)/2}{2N_3^2}\\
&\geq \frac{1}{8}.
\end{align*}}}
\end{proof}
Next, we turn to achievability of Theorem~\ref{thm:symmetrizability}. It uses \cite[Theorem 2.1]{SJ} which provides a concentration result for dependent random variables. We use it to obtain the codebook given below. \pink{This codebook is a generalization of the codebook for the point-to-point AVC (Lemma 3) studied in \cite{CsiszarN88}. In particular, \eqref{lemma_eq7b} is similar to \cite[Lemma 3, (3.1)]{CsiszarN88}. \eqref{lemma_eq7a} and \eqref{lemma_eq7c} are generalizations of \cite[Lemma 3, (3.1)]{CsiszarN88} to a pair of messages. Similarly, \eqref{code_eq2} is a generalization of \cite[Lemma 3, (3.2)]{CsiszarN88}, and \eqref{code_eq3}, \eqref{code_eq5} and \eqref{code_eq4} are generalizations of \cite[Lemma 3, (3.3)]{CsiszarN88}.}
\pink{As we mentioned in Section~\ref{intro:second_approach}, proving these generalizations  requires  establishing an analogue of the concentration result \cite[Lemma A1]{CsiszarN88} for multi-user channels.
We specialize the concentration result in \cite[Theorem 2.1]{SJ} to obtain such an extension. We illustrate the proof idea by proving \eqref{lemma_eq7a} immediately following the lemma statement.} 
For the complete proof, please refer to Appendix~\ref{appendix:codebook}.

\begin{lemma}[\pink{Codebook Lemma}]\label{lemma:codebook}
For any  $\epsilon>0,\,  n\geq n_0(\epsilon), \, N_1, N_2, N_3\geq\exp(n\epsilon)$ and types $P_1\in \cP_{\cX_1}^n$, $ P_2\in \cP_{\cX_2}^n$,  $P_3\in \cP_{\cX_3}^n$, there exists codebooks $\vecx_{11}, \ldots,  \vecx_{1N_1}\in \mathcal{X}_1^n, \vecx_{21}, \ldots,  \vecx_{2N_2}\in \mathcal{X}_2^n$, $\vecx_{31}, \ldots,  \vecx_{3N_3}\in \mathcal{X}_3^n$ whose codewords are of type $P_1, P_2$, $P_3$ respectively such that for every permutation $(i,j,k)$ of $(1,2,3)$; for every $(\vecx_i,\vecx_j, \vecx_k)  \in \mathcal{X}^n_i \times \mathcal{X}^n_j \times \mathcal{X}^n_k$; for every joint type $P_{X_iX'_iX_jX'_jX_kX'_k}\in \cP^n_{\cX_{i}\times \cX_{i}\times\cX_{j}\times\cX_j\times\cX_k\times\cX_k}$; and for $R_i \defineqq (1/ n )\log_2{N_i},R_j \defineqq (1/ n )\log_2{N_j}$, and $R_k \defineqq (1/ n )\log_2{N_k}$; the following holds:
\begin{align}
&|\{u\in [1:N_i]:(\vecx_{iu}, \vecx_i, \vecx_j, \vecx_k)\in T^{n}_{X'_i  X_i X_j X_k}\}|\leq\exp{\left(n\left(\left|R_1-I(X_1';X_1X_2X_3)\right|^{+}+\epsilon/2\right)\right)};\label{lemma_eq7b}\\
&|\{(u, v)\in [1:N_i]\times[1:N_j]:(\vecx_{iu},\vecx_{jv}, \vecx_i, \vecx_j, \vecx_k)\in T^{n}_{X'_i X'_j X_i X_j X_k}\}|\nonumber\\
&\qquad\qquad\leq\exp{\left(n\left(\left||R_i-I(X_i';X_iX_jX_k)|^{+}+|R_j-I(X_j';X_iX_jX_k)|^{+}-I(X_i';X_j'|X_iX_jX_k)\right|^{+}+\epsilon/2\right)\right)};\label{lemma_eq7a}\\
&|\{(u, w)\in [1:N_i]\times[1:N_k]:(\vecx_{iu},\vecx_{kw}, \vecx_i, \vecx_j, \vecx_k)\in T^{n}_{X'_iX'_jX_iX_j X_k}\}|\nonumber\\
&\qquad\qquad\leq\exp{\left(n\left(\left||R_i-I(X_i';X_iX_jX_k)|^{+}+|R_k-I(X_k';X_iX_jX_k)|^{+}-I(X_i';X_k'|X_iX_jX_k)\right|^{+} + \epsilon/2\right)\right)};\label{lemma_eq7c}\\
&\frac{1}{N_iN_j}|\{(r,s)\in [1:N_i]\times[1:N_j]: (\vecx_{ir},\vecx_{js}, \vecx_k)\in T^{n}_{X_i X_j X_k} \}| < \exp\left(-\frac{n\epsilon}{2}\right), \text{ if } I(X_i;X_k)+I(X_j;X_iX_k)\geq\epsilon; \label{code_eq2}\\
&\frac{1}{N_iN_j}|\{(r, s)\in [1:N_i]\times[1:N_j]: \exists \,(u, v)\in  [1:N_i]\times[1:N_j], \, u\neq r, v\neq s, \, (\vecx_{ir}, \vecx_{js}, \vecx_{iu}, \vecx_{jv},\vecx_k)\in T^{n}_{X_iX_j X_i^{'} X'_j X_k} \}|\nonumber\\
&< \exp\left(-\frac{n\epsilon}{2}\right),\nonumber\\
&\qquad\text{if }I(X_i;X_jX'_iX'_jX_k)+I(X_j;X_i'X_j'X_k)\geq \left|\left|R_i-I(X'_i;X_k)\right|^+ +\left|R_j-I(X'_j;X_k)\right|^+-I(X'_i;X'_j|X_k)\right|^+ + \epsilon; \label{code_eq3}\\
&\frac{1}{N_iN_j}|\{(r, s)\in [1:N_i]\times[1:N_j]: \exists\, (u, w)\in  [1:N_i]\times[1:N_k], \, u\neq r, \, (\vecx_{ir}, \vecx_{js}, \vecx_{iu}, \vecx_{kw}, \vecx_k)\in T^{n}_{X_iX_jX_i^{'} X'_k X_k} \}|\nonumber\\
&\qquad\qquad<\exp\left(-\frac{n\epsilon}{2}\right),\nonumber\\
&\qquad\text{if }I(X_i;X_jX'_iX'_kX_k)+I(X_j;X_i'X_k'X_k)\geq \left|\left|R_i-I(X'_i;X_k)\right|^+ +\left|R_k-I(X'_k;X_k)\right|^+-I(X'_i;X'_k|X_k)\right|^+ + \epsilon;\label{code_eq5}\\
&\frac{1}{N_iN_j}|\{(r, s)\in [1:N_i]\times[1:N_j]: \exists\, u \in [1:N_i], \, u\neq r,\, (\vecx_{ir}, \vecx_{js}, \vecx_{iu},\vecx_k)\in T^{n}_{X_iX_j X_i^{'} X_k} \}|< \exp\left(-\frac{n\epsilon}{2}\right),\nonumber\\
&\qquad\text{if }I(X_i;X_jX'_iX_k)+I(X_j;X_i'X_k)\geq \left|R_i-I(X'_i;X_k)\right|^+  + \epsilon. \label{code_eq4}
\end{align}
\end{lemma}
\begin{proof}[Proof idea]
The existence of a codebook satisfying properties \eqref{lemma_eq7b}-\eqref{code_eq4} is shown by a random coding argument. 
For fixed $(\vecx_i,\vecx_j, \vecx_k)  \in \mathcal{X}^n_i \times \mathcal{X}^n_j \times \mathcal{X}^n_k$and joint type $P_{X_iX'_iX_jX'_jX_kX'_k}\in \cP^n_{\cX_{i}\times \cX_{i}\times\cX_{j}\times\cX_j\times\cX_k\times\cX_k}$, we will show that the probability that each of the  statements~\eqref{lemma_eq7b}-\eqref{code_eq4} does not hold, falls doubly exponentially in $n$. Since $|\cX_i^n|$, $|\cX_j^n|$, $|\cX_k^n|$ and $|\cP^n_{\cX_{i}\times \cX_{i}\times\cX_{j}\times\cX_j\times\cX_k\times\cX_k}|$  grow  only exponentially in $n$, a union bound will imply  the existence of a codebook satisfying~~\eqref{lemma_eq7b}-\eqref{code_eq4}. 
\pink{
We first restate~\cite[Theorem 2.1]{SJ} for ready reference.}
\begin{lemma}\cite[Theorem 2.1]{SJ}\label{theorem:deletion_lemma_main}
Suppose that $V_{\alpha}, \alpha \in \cJ$, is a finite family of non-negative
random variables and that $\sim$ is a symmetric relation on the index set $\cJ$
such that each $V_{\alpha}$ is independent of $\{V_{\beta} : \beta \nsim \alpha\}$; in other words, the pairs
$(\alpha, \beta)$ with $\alpha\sim\beta$ define the edge set of a (weak) dependency graph for the
variables $V_{\alpha}$. Let $U:=\sum_{\alpha}V_{\alpha}$ and $\mu:=\bbE U =\sum_{\alpha} \bbE V_{\alpha}$. Let further, for $\alpha\in \cJ$, $\tilde{U}_{\alpha}:=\sum_{\beta\sim\alpha}V_{\beta}$.
If $t\geq \mu>0$, then for every real $r>0$,
\begin{align}\label{eq:deletion_concentration}
\bbP(U>\mu+t)\leq e^{-r/3}+\sum_{\alpha\in \cJ}\bbP\inp{\tilde{U}_{\alpha}>\frac{t}{2r}}.
\end{align}
\end{lemma}
\pink{We will now show the analysis of \eqref{lemma_eq7a} using Lemma~\ref{theorem:deletion_lemma_main}.}
Let $T^n_{l}, l\in \{1, 2,3\}$ denote the type class of $P_{l}$. We generate independent random codebooks for each user. The codebook for user $l\in \{1,2,3\}$, denoted by $\inp{\vecX_{l1}, \vecX_{l2}, \ldots, \vecX_{lN_{l}}}$, consists of independent random vectors each distributed uniformly on $T^n_{l}$.  
Fix $(\vecx_i,\vecx_j, \vecx_k)  \in \mathcal{X}^n_i \times \mathcal{X}^n_j \times \mathcal{X}^n_k$and a joint type $P_{X_iX'_iX_jX'_jX_kX'_k}\in \cP^n_{\cX_{i}\times \cX_{i}\times\cX_{j}\times\cX_j\times\cX_k\times\cX_k}$ such that for $l\in \{1,2,3\}, P_{X_l} =P_{X'_l}= P_{l}$ and $(\vecx_i,\vecx_j, \vecx_k)\in T^n_{X_iX_jX_k}$.

In order to apply~Lemma~\ref{theorem:deletion_lemma_main}, let $\cJ = \inb{\inp{ir, js}:(r, s)\in [1:N_i]\times[1:N_j]}$. For every $\inp{ir, js}\in \cJ$, we define binary random variable $V_{\inp{ir, js}}$ as
\begin{align*}
V_{\inp{ir, js}}=\begin{cases}&1,\text{ if }(\vecX_{ir}, \vecX_{js})\in T^n_{X'_iX'_j|X_iX_jX_k}(\vecx_k),\\
& 0,\text{ otherwise}
\end{cases}
\end{align*}
and $U = \sum_{\inp{ir, js}\in \cJ}V_{\inp{ir, js}} = \left|\left\{(r, s)\in[1:N_i]\times[1:N_j]: (\vecX_{ir},\vecX_{js},\vecx_i, \vecx_j, \vecx_{k})\in T^{n}_{X'_iX'_jX_i X_j X_k} \right\}\right|$. Note that ${\inp{ir, js}}\sim {\inp{iu, jv}}$ if and only ${\inp{ir, js}}\cap{\inp{iu, jv}}\neq \emptyset$.   
Thus, for $\inp{ir, js}\in \cJ$,  $\tilde{U}_{\inp{ir, js}} =\sum_{\inp{iu, jv}\in \cJ:\inp{iu, jv}\cap\inp{ir, js}\neq\emptyset}V_{\inp{iu, jv}}$.

Next, we will compute $\mu (= \bbE[U])$.
Note that
\begin{align*}
\bbP\inp{V_{\inp{ir, js}}=1}  &= \frac{|T^n_{X'_iX'_j|X_i X_j X_k}(\vecx_i, \vecx_j, \vecx_{k})|}{|T^n_{X'_i}||T^n_{X'_j}|}\\
&\leq \frac{\exp\inb{nH(X'_iX'_j|X_iX_jX_k)}}{(n+1)^{|\cX_i|+|\cX_j|}\exp\inb{n(H(X'_i)+H(X'_j)}}\\
&= (n+1)^{-\inp{|\cX_i|+|\cX_j|}}\exp\inb{-n\inp{H(X'_iX'_j)-H(X'_iX'_j|X_iX_jX_k) - H(X'_iX'_j) + H(X'_i)+H(X'_j)}}\\
&= (n+1)^{-\inp{|\cX_i|+|\cX_j|}}\exp\inb{-n\inp{I(X'_iX'_j;X_iX_jX_k) + I(X'_i;X'_j)}}\\
&{\leq} \exp\inb{-n\inp{I(X'_iX'_j;X_iX_jX_k) + I(X'_i;X'_j)}}\\
&{=} \exp\inb{-n\inp{I(X'_j;X_iX_jX_k) + I(X'_i;X_iX_jX_k|X'_j) + I(X'_i;X'_j)}}\\
&{=} \exp\inb{-n\inp{I(X'_j;X_iX_jX_k) + I(X'_i;X'_jX_iX_kX_k)}}
\end{align*}
Thus,  
\begin{align*}
\mu = \bbE[U] = &\sum_{(r, s)\in [1:N_i]\times[1:N_j]}\bbE\insq{V_{\inp{ir, js}}} = \sum_{(r, s)\in [1:N_i]\times[1:N_j]}\bbP\inp{V_{\inp{ir, js}}=1}\\
& \leq \exp\left\{n\left(R_i+R_j-I(X'_i;X'_jX_iX_jX_k) - I(X'_j;X_iX_jX_k)\right)\right\}\\
&\leq \exp\left\{n\left|\left|R_i-I(X'_i;X_iX_jX_k)\right|^++\left|R_j-I(X'_j;X_iX_jX_k)\right|^+-I(X'_i;X'_j|X_iX_jX_k)\right|^+\right\}:=E.
\end{align*}
Let $\nu = \exp{\inp{n\epsilon/2}}$. We are interested in $\bbP(U\geq \nu E)$.

\begin{align*}
\bbP(U\geq \nu E) &= \bbP(U-\bbE[U]\geq \nu E-\bbE[U])\\
&\leq \bbP(U-\bbE[U]\geq \nu E-E)\\
&=\bbP(U\geq \bbE[U] + (\nu-1)E)\\
&=\bbP(U\geq \mu + (\nu-1)E)
\end{align*}
Let $t = (\nu-1)E$ and $r = \exp(n\epsilon/8)$. We will use \eqref{eq:deletion_concentration} now.

\begin{align}
\bbP(U>\mu+(\nu-1)E)\leq e^{-\frac{1}{3}\exp(n\epsilon/8)}+\sum_{\inp{ir, js}\in \cJ}\bbP\inp{\tilde{U}_{\inp{ir, js}}>\frac{(\nu-1)E}{2\exp(n\epsilon/8)}}.\label{eqref:required_prob}
\end{align}
We need to analyze $\bbP\inp{\tilde{U}_{\inp{ir, js}}>\frac{(\nu-1)E}{2\exp(n\epsilon/8)}}$.
\begin{align*}
&\bbP\inp{\tilde{U}_{\inp{ir, js}}>\frac{(\nu-1)E}{2\exp(n\epsilon/8)}}\\
&= \bbP\inp{\sum_{\inp{iu, jv}\in \cJ:\inp{iu, jv}\cap\inp{ir, js}\neq\emptyset}V_{\inp{iu, jv}}>\frac{(\nu-1)E}{2\exp(n\epsilon/8)}}\\
& = \bbP\inp{V_{\inp{ir, js}} + \sum_{v\neq s}V_{\inp{ir, jv}} + \sum_{u\neq r}V_{\inp{iu, js}}>\frac{(\nu-1)E}{2\exp(n\epsilon/8)}}\\
& = \bbP\inp{\sum_{v\neq s}V_{\inp{ir, jv}} + \sum_{u\neq r}V_{\inp{iu, js}}>\frac{(\nu-1)E}{2\exp(n\epsilon/8)}-V_{\inp{ir, js}}}\\
& \leq  \bbP\inp{\sum_{v\neq s}V_{\inp{ir, jv}} + \sum_{u\neq r}V_{\inp{iu, js}}>\frac{(\nu-1)E}{2\exp(n\epsilon/8)}-1}\\
& \leq  \bbP\inp{\sum_{v\neq s}V_{\inp{ir, jv}} >\frac{1}{2}\inp{\frac{(\nu-1)E}{2\exp(n\epsilon/8)}-1}}+\bbP\inp{\sum_{u\neq r}V_{\inp{iu, jvs}} >\frac{1}{2}\inp{\frac{(\nu-1)E}{2\exp(n\epsilon/8)}-1}}\\
\end{align*} The last inequality uses a union bound. Note that 
\begin{align*}
\frac{1}{2}\inp{\frac{(\nu-1)E}{2\exp(n\epsilon/8)}-1} &=  \frac{1}{2}\inp{\frac{(\exp{\inp{n\epsilon/2}} -1)E}{2\exp(n\epsilon/8)}-1}\\
&\geq \frac{1}{2}\inp{\frac{(\exp{\inp{n\epsilon/2}} -1)E}{2\exp(n\epsilon/8)}-E}\\
&= \frac{1}{2}\inp{\inp{\frac{(\exp{\inp{n\epsilon/2}} -1)}{2\exp(n\epsilon/8)}-1}E}\\
&\geq  \inp{\inp{\frac{\exp{\inp{3n\epsilon/8}}}{\exp(n\epsilon/8)}}E}\text{ for large }n\\
& = \exp{\inp{n\epsilon/4}}.
\end{align*}
Thus,
\begin{align*}
&\bbP\inp{\sum_{v\neq s}V_{\inp{ir, jv}} >\frac{1}{2}\inp{\frac{(\nu-1)E}{2\exp(n\epsilon/4)}-1}}+\bbP\inp{\sum_{u\neq r}V_{\inp{iu, jvs}} >\frac{1}{2}\inp{\frac{(\nu-1)E}{2\exp(n\epsilon/4)}-1}}\\
&\leq \bbP\inp{\sum_{v\neq s}V_{\inp{ir, jv}} >\exp{\inp{n\epsilon/4}}E}+\bbP\inp{\sum_{u\neq r}V_{\inp{iu, js}} >\exp{\inp{n\epsilon/4}}E}.
\end{align*}
Let us first analyze $\bbP\inp{\sum_{v\neq s}V_{\inp{ir, jv}} >\exp{\inp{n\epsilon/4}}E}$. 
\begin{align*}
\bbP\inp{\sum_{v\neq s}V_{\inp{ir, jv}} >\exp{\inp{n\epsilon/4}}E}
&= \sum_{\vecx_{js}\in T^n_{X'_j|X_i X_jX_k}(\vecx_i, \vecx_j, \vecx_k)}\bbP(\vecX_{js} = \vecx_{js})\bbP\inp{\sum_{v\neq s}V_{\inp{ir, jv}} >\exp{\inp{n\epsilon/4}}E\Big|\vecX_{js} = \vecx_{js}}
\end{align*}
We will apply Lemma~\ref{theorem:deletion_lemma_main} on $\bbP\inp{\sum_{v\neq s}V_{\inp{ir, jv}} >\exp{\inp{n\epsilon/8}}E\Big|\vecX_{js} = \vecx_{js}}$ for  $\cJ' = \inb{\inp{ir, jv}:v\in [1:N_j]\setminus\{s\}}$. For every $\inp{ir, jv}\in \cJ'$, we define binary random variable $V'_{\inp{ir, jv}}$ as
\begin{align*}
V'_{\inp{ir, jv}}=\begin{cases}&1,\text{ if }(\vecX_{jv})\in T^n_{X'_j|X'_iX_i X_j X_k}(\vecx_{ir}, \vecx_i, \vecx_j, \vecx_k),\\
& 0,\text{ otherwise}
\end{cases}
\end{align*}
and $U' = \sum_{\inp{ir, jv}\in \cJ'}V'_{\inp{ir, js}}$. Note that ${\inp{ir, jv}}\sim {\inp{iu, jv'}}$ if and only if $v = v'$.   
Next, we will compute $\bbE[U']$.

\begin{align*}
&\bbE[U'] = \bbE\insq{\sum_{\inp{ir, jv}\in \cJ'}V'_{\inp{ir, jv}}}\nonumber\\
&\leq \sum_{v\neq s}\bbP\inp{V'_{(ir, jv)} = 1}\nonumber\\
&= \sum_{v\neq s}\bbP\inp{\vecX_{jv}\in T^n_{X'_j|X'_i X_iX_jX_k}(\vecx_{ir},\vecx_i, \vecx_j, \vecx_k)}\nonumber\\
&= \sum_{v\neq s}\frac{|T^n_{X'_j|X'_iX_iX_jX_k}(\vecx_{ir},\vecx_i, \vecx_j, \vecx_k)|}{|T^n_{X'_j}|}\nonumber\\
&\leq \exp\inb{nR_j}\frac{\exp\inb{nH(X'_j|X'_iX_iX_jX_k)}}{(n+1)^{|\cX_j|}\exp\inb{nH(X'_j)}}\nonumber\\
&\leq \exp\inb{n\inp{|R_j-I(X'_j;X'_iX_iX_jX_k)|^+}}\nonumber\\
&\leq E.\nonumber
\end{align*} 

Now, 
\begin{align*}
&\bbP\inp{\sum_{v\neq s}V_{\inp{ir, jv}} >\exp{\inp{n\epsilon/4}}E\Big|\vecX_{js} = \vecx_{js}} = \bbP\inp{\sum_{\inp{ir, jv}\in \cJ'}V'_{\inp{ir, jv}} >\exp{\inp{n\epsilon/4}}E}\\
& = \bbP\inp{U'  > \bbE(U')+ \exp{\inp{n\epsilon/4}}E-\bbE(U')}\\
&\stackrel{(a)}{\leq} \bbP\inp{U'  > \bbE(U') + (\exp{\inp{n\epsilon/4}}-1)E}\\
&\stackrel{(b)}{\leq} e^{-\frac{1}{3}\exp(n\epsilon/8)}+\sum_{\inp{ir, jv}\in \cJ'}\bbP\inp{V'_{\inp{ir, jv}}>\frac{(\exp(n\epsilon/4)-1)E}{2\exp(n\epsilon/8)}} \\
&\stackrel{(c)}{\leq} e^{-\frac{1}{3}\exp(n\epsilon/8)}+\sum_{\inp{ir, jv}\in \cJ'}\bbP\inp{V'_{\inp{ir, jv}}>E} \text{ for large }n\\
&\stackrel{(d)}{=} e^{-\frac{1}{3}\exp(n\epsilon/8)}
\end{align*}
where $(a)$ holds because $\bbE[U']\leq E$, $(b)$ uses \eqref{eq:deletion_concentration} for $r = \exp(n\epsilon/8)$, $(c)$ is true for large $n$ and $(d)$ holds because $V'_{\inp{ir, jv}}\in \{0, 1\}$  while $E\geq 1$.
Thus, $$\bbP\inp{\sum_{v\neq s}V_{\inp{ir, jv}} >\exp{\inp{n\epsilon/4}}E}\leq e^{-\frac{1}{3}\exp(n\epsilon/8)}.$$ Similarly, $$\bbP\inp{\sum_{u\neq r}V_{\inp{iu, js}} >\exp{\inp{n\epsilon/4}}E}\leq e^{-\frac{1}{3}\exp(n\epsilon/8)}.$$ 
This implies that
\begin{align*}
\bbP\inp{\tilde{U}_{\inp{ir, js}}>\frac{(\nu-1)E}{2\exp(n\epsilon/8)}}\leq 2e^{-\frac{1}{3}\exp(n\epsilon/8)}.
\end{align*}
Thus, from \eqref{eqref:required_prob},
\begin{align*}
\bbP(U>\mu+(\nu-1)E)&\leq e^{-\frac{1}{3}\exp(n\epsilon/8)}+|\cJ|2e^{-\frac{1}{3}\exp(n\epsilon/8)}\\
&=e^{-\frac{1}{3}\exp(n\epsilon/8)}+|N_i||N_j|2e^{-\frac{1}{3}\exp(n\epsilon/8)}
\end{align*}
which falls doubly exponentially.	
\end{proof}

\begin{proof}[{Proof (Achievability of Theorem~\ref{thm:symmetrizability})}]
\pink{For an input distribution $p(x_1)p(x_2)p(x_3)$, we first show the achievability of the set of rate triples $(R_1,R_2,R_3)$ which, for all permutations $(i,j,k)$ of $(1,2,3)$, satisfy the following conditions:  
\begin{align}
R_i &< \min_{q(x_k)} I(X_i;Y|X_j),\quad\text{and}\label{eq:rateconstraint1_wo_time_sharing}\\
R_i+R_j &< \min_{q(x_k)} I(X_i,X_j;Y),\label{eq:rateconstraint2_wo_time_sharing}
\end{align} where the mutual information terms are evaluated using the joint distribution $p(x_i)p(x_j)q(x_k)W(y|x_1,x_2,x_3)$.} 
Fix distributions $P_{1}\in\mathcal{P}^n_{\cX_1},\,P_{2}\in\mathcal{P}^n_{\cX_2}$ and $P_3\in\mathcal{P}^n_{\cX_3}$ (which approach $p(x_1), p(x_2), p(x_3)$ as $n\rightarrow \infty$). 
For these distributions, consider the codebook given by Lemma~\ref{lemma:codebook} and the decoder as given in Definition~\ref{def:decoder} for \pink{some $\eta>0$ satisfying the condition in Lemma \ref{lemma:dec}. Choose $\epsilon>0$ such that $\eta>6\epsilon$.}
Below, we repeat the decoder from section~\ref{para:det_Achiev} for the sake of completeness. 
\begin{defn}[Decoder]\label{def:decoder}
For a received vector $\vecy\in \cY^n$, some $\eta>0$, $\phi(\vecy)=(m_1,m_2,m_3)\in \mathcal{M}_1\times\mathcal{M}_2\times\mathcal{M}_3$, if  $\vecy\in \cD^{(1)}_{m_1}\cap \cD^{(2)}_{m_2} \cap \cD^{(3)}_{m_3}$ where $\cD^{(i)}_{m_i}, \, i=1,2,3$ is defined as below.\\
$\vecy \in \cD^{(1)}_{m_1}$ if there exists  some permutation $(j, k)$ of $(2,3)$, $m_j \in \mathcal{M}_j,\,\vecx_k\in \mathcal{X}^n_k$, and random variables $X_1, X_j, X_k,  Y$ with $(f_1(m_1), f_j(m_j), \vecx_{k}, \vecy) \in T^{n}_{X_1X_jX_kY}$ and $D(P_{X_1X_jX_kY}||P_{X_1}\times P_{X_j}\times P_{X_k}\times W)< \eta$ such that the following hold:
	\begin{description}			
			 \item[(a)] \underline{Disambiguating ($m_1,m_j$) from ($m_1',m'_j$):} For every $(m'_1, m'_j) \in \mathcal{M}_1\times\mathcal{M}_j$ $m'_1 \neq m_1$, $m'_j\neq m_j$, $\vecx'_k\in \mathcal{X}^n_k$, and random variables $X_1',X_j', X_k'$ such that $(f_1(m_1), f_1(m'_1),f_j(m_j), f_j(m'_j), \vecx_k, \vecx'_k, \vecy) \in T^{n}_{X_1 X_1' X_jX'_j X_k X_k'Y}$ and  $D(P_{X_1'X'_jX'_k Y}||P_{X'_1}\times P_{X'_j}\times P_{X'_k}\times W)< \eta$, we require that $I(X_1X_jY;X'_1X'_j|X_k) < \eta$.
			\item[(b)] \underline{Disambiguating $m_1$ from $m_1'$:} For every $m'_1 \in \mathcal{M}_1$, $m'_1 \neq m_1$, $\vecx'_k\in \mathcal{X}^n_k$, and random variables $X_1',X_k'$ such that $(f_1(m_1),f_1(m'_1),f_j(m_j), \allowbreak \vecx_k, \vecx'_k, \vecy) \in T^{n}_{X_1 X_1'X_jX_k X_k' Y}$ and  $D(P_{X_1'X_jX_k' Y}||P_{X_1'}\times P_{X_j}\times P_{X_k'}\times W)< \eta$, we require that $I(X_1X_jY;X_1'|X_k) < \eta$.
			\item[(c)] \underline{Disambiguating ($m_1,m_j$) from ($m_1',m_k$):} If there exist $(m'_1, m_k) \in \mathcal{M}_1\times \mathcal{M}_k$, $m'_1 \neq m_1$, $\vecx_j\in \mathcal{X}^n_j$,  and random variables $X_1',X_j', X_k'$ such that $(f_1(m_1),f_1(m'_1),f_j(m_j), \vecx_j, \vecx_k, f_k(m_k), \vecy) \in T^{n}_{X_1 X_1'X_j X_j' X_kX_k'Y}$ and  $D(P_{X_1'X_j'X_k' Y}||P_{X_1'}\times P_{X_j'}\times P_{X_k'}\times W)< \eta$, we require that $I(X_1X_jY;X_1'X_k'|X_k) < \eta$.
	\end{description}
The decoding sets $\cD^{(2)}_{m_2}$ and $\cD^{(3)}_{m_3}$ are defined similarly. 
If $\vecy\notin \cD^{(1)}_{m_1}\cap \cD^{(2)}_{m_2} \cap \cD^{(3)}_{m_3}$ for any $(m_1, m_2, m_3)\in \cM_1\times\cM_2\times\cM_3$, the decoder outputs $(1,1,1)$.
\end{defn}

Next, we give some standard properties of joint types as given in~\cite{CsiszarN88} (as {\em Fact 1}, {\em Fact 2} 
 and {\em Fact 3}). As mentioned in~\cite{CsiszarN88}, these bounds  can be found in \cite{CK11}. For finite alphabets $\cX$, $\cY$, the type class $\cP^n_{\cX}$, any channel $V$ from $\cX$ to $\cY$ and random variables $X$ and $Y$ on $\cX$ and $\cY$ respectively with joint distribution $P_{XY}$, the following holds:
\begin{align}
|\cP^n_{\cX}|&\leq (n+1)^{|\cX|}\label{eq:poly_size}\\
(n+1)^{-|\cX|}\exp{\inp{nH(X)}}&\leq |T^n_{X}(\vecx)|\leq \exp{\inp{{nH(X)}}}\qquad \text{if }|T^n_{X}(\vecx)|\neq 0\label{eq:type_class}\\
(n+1)^{-|\cX||\cY|}\exp{\inp{nH(Y|X)}}&\leq |T^n_{Y|X}(\vecx)|\leq \exp{\inp{nH(Y|X)}}\qquad \text{if }|T^n_{Y|X}(\vecx)|\neq 0\label{eq:conditional_type_class}\\
\sum_{\vecy\in T^n_{Y|X}(\vecx)}V^n(\vecy|\vecx)&\leq \exp{\inp{-nD(P_{XY}||P_X\times V)}}\label{eq:sanov}
\end{align}
We first analyze the case when user 3 is adversarial. The probability of error when user 3 is adversarial (see \eqref{eq:error_adv}) is given by 
\pink{\begin{align*}
P_{e,3}&\defineqq \max_{\vecx_3}P_{e,3}(\vecx_3),
\end{align*}
where $P_{e,3}(\vecx_3)$ is the average probability of error for users 1 and 2 when a malicious user 3 sends $\vecx_3$ as input. That is, 
\begin{align}\label{eq:main_result_pex3}
P_{e,3}(\vecx_3):=\frac{1}{N_1N_2} \sum_{r\in \mathcal{M}_1,s\in\mathcal{M}_2} \Prob\Big(\inb{\vecy:\phi(\vecy)\neq (r, s, t)\text{ for some } t\in \cM_3} \, \Big| \vecX_1=\vecx_{1r},\, \vecX_2=\vecx_{2s},\, \vecX_3=\vecx_3 \Big).
\end{align}
We will argue that for every $\vecx_3\in \cX^n_3$, $P_{e,3}(\vecx_3)\rightarrow 0$ as $n\rightarrow \infty$.}
\pink{The analysis of $P_{e,3}(\vecx_3)$ follows the flowchart shown in Figure~\ref{fig:flowchart1}.}
\begin{figure}[!h]
\begin{tikzpicture}[scale=1.2]
 
  \node (A1) at (0,2) [rectangle, draw] {$P_{e,3}(\vecx_3)$};
  \node (A2) at (0,0) [rounded corners=3pt, draw] {Union bound} node[xshift = 1.1cm, right] {eq.~\eqref{eq:second_term}};
  \node (C) at (-1.5, -2) [rounded corners=3pt, draw] {Union bound};
  \node[xshift = 0.8cm, right of = C] {eq.~\eqref{eq:three_terms}};
  \node (D) at (1.5, -2) {};

  \node (E) at (-3, -4)  {\small $\substack{\text{\small  small}\\\text{ (atypical event)}}$};
\node (F) at (-1.5, -4.5)  {$\substack{ \text{\small small}\\\text{ (atypical event)}}$};
 \node (G) at (0, -4) [rounded corners=3pt, draw] {Union bound};
  \node[xshift = 0.8cm, right of = G] {eq.~\eqref{eq_error}};
  \phantom{\node (G1) at (-1.5, -6) [circle, draw] {}; 
 \node (G2) at (0, -6) [circle, draw] {};
  \node (G3) at (1.5, -6) [circle, draw] {};}

  \draw[gray, thick, -] (-2, -5.7) -- (2, -5.7) node[yshift = -0.3cm, midway, below]{\small \textcolor{black}{continued in Figure~\ref{fig:flowchart2}.}};
 	\draw[->] (G) -- (G1) node[midway, right]{$\scriptstyle P_{\textsf{a}}$};
 	\draw[->] (G) -- (G2) node[midway, right]{$\scriptstyle P_{\textsf{b}}$};
 	\draw[->] (G) -- (G3) node[midway, right]{$\scriptstyle P_{\textsf{c}}$};

 	\draw[->] (C) -- (E) node[midway, right]{$\scriptstyle P_{\cA_{\epsilon}}$};
 	\draw[->] (C) -- (F) node[yshift = -0.2cm, midway, right]{$\scriptstyle P_{\cB_{\eta, \epsilon}}$};
 	\draw[->] (C) -- (G) node[midway, right]{$\scriptstyle \scriptstyle P_{\cD_{\eta}}$};
  \draw[->] (A1) -- (A2) node[midway, right]{};
  \draw[->] (A2) -- (C) node[midway, right] {$\scriptstyle P_1(\vecx_3)$};
  \draw[->] (A2) -- (D) node[midway, right] {$\scriptstyle P_2(\vecx_3)$} node[below] {$\vdots$};
\node at (7, -2)  {
\begin{footnotesize}
\begin{tabular}{p{1cm}|p{6cm}} 
$P_{e,3}(\vecx_3)$& the average probability of error when malicious user 3 sends $\vecx_3$\\
\hline
$P_1(\vecx_3)$ & the average probability of error for user 1\\
\hline	
$P_2(\vecx_3)$ & the average probability of error for user 2 \\
\hline	
$P _{\cA_{\epsilon}}$& the probability that channel inputs are atypical\\
\hline	
$P _{\cB_{\eta, \epsilon}}$& the probability that the channel output is atypical\\
\hline	
$P _{\cD_{\eta}}$ & $\cD_{\eta}$ is such that $\cA_{\epsilon}^c\cap\cB_{\eta,\epsilon}^c \subseteq \cD_{\eta}$ \\
\hline	
$P_{\textsf{a}}$ & condition $\bm{\text{(a)}}$ in Definition~\ref{def:decoder} does not hold\\
\hline	
$P_{\textsf{b}}$ & condition $\bm{\text{(b)}}$ in Definition~\ref{def:decoder} does not hold\\
\hline	
$P_{\textsf{c}}$ & condition $\bm{\text{(c)}}$ in Definition~\ref{def:decoder} does not hold\\
\end{tabular}
\end{footnotesize}};
\end{tikzpicture}
\caption{Flowchart depicting the flow of analysis of $P_{e,3}(\vecx_3)$, the average probability of error when malicious user 3 sends $\vecx_3$. In the flowchart, only $P_1(\vecx_3)$ is further broken down and shown. The flowchart is continued in Figure~\ref{fig:flowchart2}.} \label{fig:flowchart1}
\end{figure}

From the decoder definition, we know that for $(r, s)\in \cM_1\times \cM_2$, if $\phi(\vecy)\neq (r, s, t)$ for some $t\in \cM_3$, then $\vecy\notin \cD^{(1)}_{r}\cap \cD^{(2)}_{s}$. In other words, $\vecy\in \left(\cD^{(1)}_{r}\right)^c\cup \left(\cD^{(2)}_{s}\right)^c$.
Thus, 
\pink{\begin{align*}
P_{e,3}\inp{\vecx_3}&= \frac{1}{N_1N_2} \sum_{r\in \mathcal{M}_1,s\in\mathcal{M}_2} \Prob\Big(\inb{\vecy:\vecy\in \left(\cD^{(1)}_{r}\right)^c\cup \left(\cD^{(2)}_{s}\right)^c} \Big| \vecX_1=\vecx_{1r},\, \vecX_2=\vecx_{2s},\, \vecX_3=\vecx_3 \Big)\\
&=  \frac{1}{N_1N_2} \sum_{r\in \mathcal{M}_1,s\in\mathcal{M}_2} \Prob\Big(\inb{\vecy:\vecy\in \left(\cD^{(1)}_{r}\right)^c}\cup\inb{\vecy:\vecy\in \left(\cD^{(2)}_{s}\right)^c} \Big| \vecX_1=\vecx_{1r},\, \vecX_2=\vecx_{2s},\, \vecX_3=\vecx_3 \Big)\\
&\stackrel{(a)}{\leq}  \frac{1}{N_1N_2} \sum_{r\in \mathcal{M}_1,s\in\mathcal{M}_2} \Prob\Big(\inb{\vecy:\vecy\in \left(\cD^{(1)}_{r}\right)^c} \Bigg| \vecX_1=\vecx_{1r},\, \vecX_2=\vecx_{2s},\, \vecX_3=\vecx_3 \Big)\\
&\quad+ \frac{1}{N_1N_2} \sum_{r\in \mathcal{M}_1,s\in\mathcal{M}_2} \Prob\Big(\inb{\vecy:\vecy\in \left(\cD^{(2)}_{s}\right)^c} \Big| \vecX_1=\vecx_{1r},\, \vecX_2=\vecx_{2s},\, \vecX_3=\vecx_3 \Big)
\end{align*}
where $(a)$ uses the union bound. Thus, for
\begin{align*}
P_1(\vecx_3):=\frac{1}{N_1N_2} \sum_{r\in \mathcal{M}_1,s\in\mathcal{M}_2} \Prob\Big( \inb{\vecy:\vecy\notin \cD^{(1)}_{r}} \Big| \pink{\vecX_1=\vecx_{1r},\, \vecX_2=\vecx_{2s},\, \vecX_3=\vecx_3} \Big), 
\end{align*} and 
\begin{align*}
P_2(\vecx_3):=\frac{1}{N_1N_2} \sum_{r\in \mathcal{M}_1,s\in\mathcal{M}_2} \Prob\Big( \inb{\vecy:\vecy\notin \cD^{(2)}_{s}}\Big| \pink{\vecX_1=\vecx_{1r},\, \vecX_2=\vecx_{2s},\, \vecX_3=\vecx_3} \Big),
\end{align*}
we have the following upper bound on $P_{e,3}\inp{\vecx_3}$.}
\pink{\begin{align}
P_{e,3}\inp{\vecx_3}&\leq P_1(\vecx_3) + P_2(\vecx_3)\label{eq:second_term}
\end{align}}
\pink{We will first analyze $P_1(\vecx_3)$. Let 
\begin{align}
\cA_{\epsilon} &\defineqq \{P_{X_1X_2 X_3 Y}\in \cP^{n}_{\cX_1\times\cX_2\times\cX_3\times \cY}:D(P_{X_1 X_2 X_3}||P_{X_1}P_{X_2}P_{X_3})\geq\epsilon\},\label{eq:main_proof_A_epsilon}\\
\cB_{\eta,\epsilon} &\defineqq \{P_{X_1X_2 X_3 Y}\in \cP^{n}_{\cX_1\times\cX_2\times\cX_3\times \cY}: D(P_{X_1X_2 X_3 Y}||P_{X_1X_2X_3}W)\geq \eta-\epsilon\},\label{eq:main_proof_B_epsilon}\\
\text{ and }\cD_{\eta} &\defineqq \{P_{X_1X_2 X_3 Y}\in \cP^{n}_{\cX_1\times\cX_2\times\cX_3\times \cY}: D(P_{X_1X_2 X_3 Y}||P_{X_1}P_{X_2}P_{X_3}W)< \eta\}.\label{eq:main_proof_D_eta}
\end{align} 
In defining $\cB_{\eta, \epsilon}$, recall that $\eta>6\epsilon$. We will use $\cA^c_{\epsilon}$, $\cB^c_{\eta,\epsilon}$ and $\cD^c_{\eta}$ to denote $\cP^{n}_{\cX_1\times\cX_2\times\cX_3\times \cY}\setminus \cA_{\epsilon}$, $\cP^{n}_{\cX_1\times\cX_2\times\cX_3\times \cY}\setminus \cB_{\eta,\epsilon}$ and $\cP^{n}_{\cX_1\times\cX_2\times\cX_3\times \cY}\setminus \cD_{\eta}$ respectively. }

\pink{We first note that $\cA_{\epsilon}^c\cap\cB_{\eta,\epsilon}^c \subseteq \cD_{\eta}$. This is because $D(P_{X_1X_2 X_3 Y}||P_{X_1}P_{X_2}P_{X_3}W) = D(P_{X_1X_2 X_3 Y}||P_{X_1X_2X_3}W) + D(P_{X_1 X_2 X_3}||P_{X_1}P_{X_2}P_{X_3})$ and for $P_{X_1X_2 X_3 Y}\in \cA_{\epsilon}^c\cap\cB_{\eta,\epsilon}^c$, $D(P_{X_1X_2 X_3 Y}||P_{X_1X_2X_3}W) + D(P_{X_1 X_2 X_3}||P_{X_1}P_{X_2}P_{X_3})<\epsilon + \eta-\epsilon = \eta$. Thus, $\cP^{n}_{\cX_1\times\cX_2\times\cX_3\times \cY} = \cA_{\epsilon}\cup\cB_{\eta,\epsilon}\cup\cD_{\eta}$. We  focus on the first term on the RHS of  \eqref{eq:second_term} and split the set of joint types $\cP^{n}_{\cX_1\times\cX_2\times\cX_3\times \cY}$ into $\cA_{\epsilon}$, $\cB_{\eta,\epsilon}$ and $\cD_{\eta}$. \pink{Further, we use loose upper bounds on each of these terms, for example, in the first terms in \eqref{eq:three_terms} below, we upper bound the summand by 1.}
\begin{align}
P_1(\vecx_3)&=\frac{1}{N_1N_2} \sum_{(r,s)} \Prob\Big(  \inb{\vecy:\vecy\notin \cD^{(1)}_{r}} \big| X_1^n=\vecx_{1r},\, X_2^n=\vecx_{2s},\, X_3^n=\vecx_3\Big)\nonumber\\
&=\frac{1}{N_1N_2} \sum_{P_{X_1X_2X_3Y}\in \cP^{n}_{\cX_1\times\cX_2\times\cX_3\times \cY}}\sum_{\substack{(r,s):\\(\vecx_{1r},\vecx_{2s},\vecx_3)\in T^n_{X_1X_2X_3}}}\sum_{\substack{\vecy\in T^n_{Y|X_1X_2X_3}(\vecx_{1r},\vecx_{2s},\vecx_3),\\\vecy\notin \cD^{(1)}_{r}}}W^n(\vecy|\vecx_{1r},\vecx_{2s},\vecx_3)\nonumber\\
&\leq \frac{1}{N_1N_2}\sum_{P_{X_1X_2X_3Y}\in \cA_{\epsilon}}\sum_{\substack{(r, s):\\(\vecx_{1r}, \vecx_{2s},\vecx_3)\in T^n_{X_1X_2X_3}}}1\nonumber\\
& \quad+\frac{1}{N_1N_2}\sum_{P_{X_1X_2X_3Y}\in \cB_{\eta,\epsilon}}\sum_{\substack{(r,s):\\(\vecx_{1r},\vecx_{2s}\vecx_3)\in T^n_{X_1X_2X_3}}}\sum_{\vecy\in T^n_{Y|X_1X_2X_3}(\vecx_{1r},\vecx_{2s},\vecx_3)}W^n(\vecy|\vecx_{1r}, \vecx_{2s}, \vecx_3)\nonumber\\
&\quad+\frac{1}{N_1N_2}\sum_{P_{X_1X_2X_3Y}\in \cD_{\eta}}\sum_{\substack{(r,s):\\(\vecx_{1r},\vecx_{2s},\vecx_3)\in T^n_{X_1X_2X_3}}}\sum_{\substack{\vecy\in T^n_{Y|X_1X_2X_3}(\vecx_{1r},\vecx_{2s},\vecx_3),\\\vecy\notin \cD^{(1)}_{r}}}W^n(\vecy|\vecx_{1r}, \vecx_{2s}, \vecx_3)\nonumber\\
&=:\pink{P_{\cA_{\epsilon}} + P_{\cB_{\eta,\epsilon}} + P_{\cD_{\eta}}},\label{eq:three_terms}
\end{align}
\pink{where we define the $P_{\cA_{\epsilon}}, P_{\cB_{\eta,\epsilon}}$ and $P_{\cD_{\eta}}$ as the three summation terms.}
We will analyze each term on the RHS of \eqref{eq:three_terms} separately. We start with the first term.
\begin{align*}
\pink{P_{\cA_{\epsilon}}}&=\frac{1}{N_1N_2}\sum_{P_{X_1X_2X_3Y}\in \cA_{\epsilon}}\sum_{\substack{(r,s):\\(\vecx_{1r}, \vecx_{2s,\vecx_3})\in T^n_{X_1X_2X_3}}}1 \\
&= \frac{1}{N_1N_2}\sum_{P_{X_1X_2X_3Y}\in \cA_{\epsilon}}|\{(r,s): (\vecx_{1r},\vecx_{2s},\vecx_{3})\in T^{n}_{X_1X_2 X_3}\}|\\
& = \sum_{P_{X_1X_2X_3Y}\in \cA_{\epsilon}}\frac{|\{(r,s): (\vecx_{1r},\vecx_{2s},\vecx_{3})\in T^{n}_{X_1X_2 X_3}\}|}{N_1N_2} \\
&\stackrel{\text{(a)}}{\leq} \sum_{P_{X_1X_2X_3Y}\in \cA_{\epsilon}}\exp\inp{-n\epsilon/2}\\
&\leq |\mathcal{P}^n_{\cX_1\times\cX_2\times\cX_3\times\cY}|\exp\inp{-n\epsilon/2}\\
&\stackrel{\text{(b)}}{\leq} (n+1)^{|\cX_1\times\cX_2\times\cX_3\times\cY|}\exp\inp{-n\epsilon/2}\\
& \rightarrow 0 \text{ as }n\rightarrow\infty. 
\end{align*}
Here, (a) follows from \eqref{code_eq2} (as $I(X_1;X_3)+I(X_2;X_1X_3) = D(P_{X_1X_2X_3}||P_{X_1}P_{X_2}P_{X_3})>\epsilon$ for every $P_{X_1X_2X_3Y}\in \cA_{\epsilon}$ \pink{as defined in \eqref{eq:main_proof_A_epsilon}}). The inequality (b) uses \eqref{eq:poly_size}.
We now analyze the second term. For fixed $r\in [1:N_1]$ and $s\in [1:N_2]$
\begin{align*}
\pink{P_{\cB_{\eta,\epsilon}}}&=\sum_{\substack{ P_{X_1X_2X_3Y}\in \cB_{\eta,\epsilon}:\\(\vecx_{1r},\vecx_{2s},\vecx_3)\in T^n_{X_1X_2X_3}}}\sum_{\vecy\in T^n_{Y|X_1X_2X_3}(\vecx_{1r},\vecx_{2s},\vecx_3)}W^n(\vecy|\vecx_{1r}, \vecx_{2s}, \vecx_3)\\
&\stackrel{\text{(a)}}{\leq} \sum_{\substack{ P_{X_1X_2X_3Y}\in \cB_{\eta,\epsilon}:\\(\vecx_{1r},\vecx_{2s},\vecx_3)\in T^n_{X_1X_2X_3}}}\exp{(-nD(P_{X_1X_2X_3Y}||P_{X_1X_2X_3}W))}\\
&\stackrel{\text{(b)}}{\leq} \sum_{\substack{ P_{X_1X_2X_3Y}\in \cB_{\eta,\epsilon}:\\(\vecx_{1r},\vecx_{2s},\vecx_3)\in T^n_{X_1X_2X_3}}}\exp{(-n(\eta-\epsilon))}\\
&\leq |\mathcal{P}^n_{\cX_1\times\cX_2\times\cX_3\times\cY}|\exp{(-n(\eta-\epsilon))}\\
& \stackrel{\text{(c)}}{\leq} (n+1)^{|\cX_1\times\cX_2\times\cX_3\times\cY|}\exp{(-n(\eta-\epsilon))}\\
& \rightarrow 0 \text{ as } n\rightarrow\infty \text{ as }\pink{\eta>6\epsilon}.
\end{align*}
Here, the inequality (a) uses \eqref{eq:sanov}, (b) follows by noting that $P_{X_1X_2X_3Y}\in \cB_{\eta,\epsilon}$ \pink{(see \eqref{eq:main_proof_B_epsilon})} and thus, $D(P_{X_1X_2X_3Y}||P_{X_1X_2X_3}W) >\eta-\epsilon$. The inequality (c) follows because $\mathcal{P}^n_{\cX_1\times\cX_2\times\cX_3\times\cY}{\leq} (n+1)^{|\cX_1\times\cX_2\times\cX_3\times\cY|}$ by using \eqref{eq:poly_size}.  This shows that the second term on the RHS of \eqref{eq:three_terms} also goes to zero as $n$ goes to infinity.
}

It remains to analyze the third term of \eqref{eq:three_terms}\pink{, that is, $P_{\cD_{\eta}}$}. \pink{This only involves joint distributions $P_{X_1X_2X_3Y}$ which satisfy $D(P_{X_1X_2X_3 Y}||P_{X_1}\times P_{X_2}\times P_{X_3}\times W)\leq \eta$, {\em i.e.}, $P_{X_1X_2X_3Y}\in \cD_{\eta}$ (see \eqref{eq:main_proof_D_eta}).}
\pink{When $P_{X_1X_2X_3Y}\in \cD_{\eta}$, we notice from Definition \ref{def:decoder} that $\vecy\notin \cD_r^{(1)}$ only if for each of $(j,k)=(2,3)$ and $(j,k)=(3,2)$, at least one of the conditions among (a), (b) and (c) in Definition~\ref{def:decoder} fails. Thus, to upper bound $P_{\cD_{\eta}}$, it is sufficient to analyze the probability that at least one of (a)-(c) in Definition~\ref{def:decoder} fails under $(j,k) = (2,3)$. 
}
This implies that at least one of the following holds:
	\begin{description}			
			 \item[(a)] There exists $u \in \mathcal{M}_1$, $u \neq r$, $v \in \mathcal{M}_2$, $v\neq s$, $\vecx'_3\in \mathcal{X}^n_3$, and random variables $X_1',X_2', X_3'$ such that $(f_1(r), f_2(s),  \vecx_3,f_1(u), f_2(v), \vecx'_3, \vecy) \in T^{n}_{X_1X_2X_3X_1'X'_2X_3'Y}$,  $D(P_{X_1'X'_2X'_3 Y}||P_{X'_1} P_{X'_2} P_{X'_3} W)< \eta$ and \\$I(X_1X_2Y;X'_1X'_2|X_3) \geq \eta$.

			\item[(b)] There exists $u \in \mathcal{M}_1$, $u \neq r$, $\vecx'_3\in \mathcal{X}^n_3$, and random variables $X_1',X_3'$ such that $(f_1(r),f_2(s),\vecx_3,\allowbreak  f_1(u), \vecx'_3, \vecy) \in T^{n}_{X_1 X_2 X_3 X_1'X_3' Y}$,  $D(P_{X_1'X_2X_3' Y}||P_{X_1'} P_{X_2} P_{X_3'}W)< \eta$ and $I(X_1X_2Y;X_1'|X_3) \geq \eta$.

			\item[(c)] There exists $u \in \mathcal{M}_1$, $u \neq r$, $\vecx_2\in \mathcal{X}^n_2$, $w\in \mathcal{M}_3$, and random variables $X_1',X_2', X_3'$ such that $(f_1(r),f_2(s),  \vecx_3,\allowbreak f_1(u),\vecx_2,f_3(w), \vecy) \in T^{n}_{X_1  X_2 X_3 X_1'X_2' X_3'Y}$, $D(P_{X_1'X_2'X_3' Y}||P_{X_1'} P_{X_2'} P_{X_3'} W)< \eta$ and $I(X_1X_2Y;X_1'X_3'|X_3) \geq \eta$.
	\end{description}

To analyze these, we define the following sets of distributions:
\begin{align}
&\cQ_1 = \{P_{X_1X_2X_3X'_1X'_2Y}\in \mathcal{P}^n_{\cX_1\times\cX_2\times\cX_3\times\cX_1\times\cX_2\times\cY}:\,P_{X_1X_2X_3Y}\in \cD_{\eta}\cap\cA^c_{\epsilon}, P_{X_1'X'_2X'_3 Y}\in \cD_{\eta}\nonumber\\
&\qquad \qquad \text{ for some }X'_3, \, P_{X_1}=P_{X'_1} = P_1,\, P_{X_2}=P_{X'_2}=P_2\text{ and } I(X_1X_2Y;X'_1X'_2|X_3) \geq \eta\}\label{eq:main_proof_q_1}\\
&\cQ_2 = \{P_{X_1X_2X_3X'_1Y}\in \mathcal{P}^n_{\cX_1\times\cX_2\times\cX_3\times\cX_1\times\cX_3\times\cY}:\,P_{X_1X_2X_3Y}\in \cD_{\eta}\cap\cA^c_{\epsilon}, P_{X_1'X_2X'_3 Y}\in \cD_{\eta} \nonumber\\
 &\qquad\qquad \text{ for some }X'_3, \, P_{X_1}=P_{X'_1} = P_1, \, P_{X_2} = P_2\text{ and } I(X_1X_2Y;X'_1|X_3) \geq \eta\}\label{eq:main_proof_q_2}\\
&\cQ_3 = \{P_{X_1X_2X_3X'_1X'_3Y}\in \mathcal{P}^n_{\cX_1\times\cX_2\times\cX_3\times\cX_1\times\cX_3\times\cY}:\,P_{X_1X_2X_3Y}\in \cD_{\eta}\cap\cA^c_{\epsilon}, P_{X_1'X'_2X'_3 Y}\in \cD_{\eta}\nonumber\\
&\qquad\qquad \text{ for some }X'_2, \, P_{X_1}=P_{X'_1} = P_1, \, P_{X_2} =P_2,\, P_{X'_3}=P_3\text{ and }  I(X_1X_2Y;X'_1X'_3|X_3) \geq \eta\}\label{eq:main_proof_q_3}
\end{align}
For $r,s \in \cM_1\times\cM_2$, $P_{X_1X_2X_3X'_1X'_2Y}\in \cQ_1$, $P_{X_1X_2X_3X'_1Y}\in \cQ_2$ and $P_{X_1X_2X_3X'_1X'_3Y}\in \cQ_3$, define the following sets:
\begin{align*}
&\cE_{r,s,1}(P_{X_1X_2X_3X'_1X'_2Y})= \{\vecy: \exists (u, v) \in \mathcal{M}_1\times \cM_2, u \neq r,v\neq s,\, (\vecx_{1r}, \vecx_{2s}, \vecx_3, \vecx_{1u},\vecx_{2v},  \vecy) \in T^{n}_{X_1  X_2 X_3X_1'X'_2Y}\}\\
&\cE_{r,s,2}(P_{X_1X_2X_3X'_1Y})=  \{\vecy: \exists u \in \mathcal{M}_1, u \neq r,\, (\vecx_{1r},\vecx_{2s}, \vecx_3,\vecx_{1u}, \vecy) \in T^{n}_{X_1 X_2 X_3 X_1' Y}\}\\
&\cE_{r,s,3}(P_{X_1X_2X_3X'_1X'_3Y}) =\{\vecy: \exists (u, t) \in \mathcal{M}_1\times\cM_3, u \neq r, \, (\vecx_{1r},\vecx_{2s}, \vecx_3, \vecx_{1u},\vecx_{3t}, \vecy) \in T^{n}_{X_1 X_2 X_3 X_1' X_3'Y}\}
\end{align*}
Thus, 
\begin{align}
\pink{P_{\cD_{\eta}}}&=\frac{1}{N_1N_2}\sum_{r,s}\sum_{\substack{ P_{X_1X_2X_3Y}\in \cD_{\eta}:\\(\vecx_{1r},\vecx_{2s},\vecx_3)\in T^n_{X_1X_2X_3}}}\sum_{\substack{\vecy\in T^n_{Y|X_1X_2X_3}(\vecx_{1r},\vecx_{2s},\vecx_3)\\\vecy\notin \cD^{(1)}_{r}}}W^n(\vecy|\vecx_{1r}, \vecx_{2s}, \vecx_3)\nonumber\\
& \leq \frac{1}{N_1N_2}\sum_{r,s}\Bigg\{\sum_{P_{X_1X_2X_3X'_1X'_2Y}\in \cQ_1}W^n(\cE_{r,s,1}(P_{X_1X_2X_3X'_1X'_2Y})|\vecx_{1r}, \vecx_{2s}, \vecx_3)\nonumber\\
&\qquad\qquad\qquad+\sum_{P_{X_1X_2X_3X'_1Y}\in \cQ_2}W^n(\cE_{r,s,2}(P_{X_1X_2X_3X'_1Y})|\vecx_{1r}, \vecx_{2s}, \vecx_3)\nonumber\\
&\qquad\qquad\qquad+\sum_{P_{X_1X_2X_3X'_1X'_3Y}\in \cQ_3}W^n(\cE_{r,s,3}(P_{X_1X_2X_3X'_1X'_3Y})|\vecx_{1r}, \vecx_{2s}, \vecx_3)\Bigg\}\nonumber\\
& \leq \frac{1}{N_1N_2}\sum_{r,s}\sum_{P_{X_1X_2X_3X'_1X'_2Y}\in \cQ_1}W^n(\cE_{r,s,1}(P_{X_1X_2X_3X'_1X'_2Y})|\vecx_{1r}, \vecx_{2s}, \vecx_3)\nonumber\\
&\qquad\qquad\qquad+\frac{1}{N_1N_2}\sum_{r,s}\sum_{P_{X_1X_2X_3X'_1Y}\in \cQ_2}W^n(\cE_{r,s,2}(P_{X_1X_2X_3X'_1Y})|\vecx_{1r}, \vecx_{2s}, \vecx_3)\nonumber\\
&\qquad\qquad\qquad+\frac{1}{N_1N_2}\sum_{r,s}\sum_{P_{X_1X_2X_3X'_1X'_3Y}\in \cQ_3}W^n(\cE_{r,s,3}(P_{X_1X_2X_3X'_1X'_3Y})|\vecx_{1r}, \vecx_{2s}, \vecx_3)\nonumber\\
&:=\pink{P_{\textsf{a}}+P_{\textsf{b}}+P_{\textsf{c}}}\label{eq_error}
\end{align}
\pink{where 
\begin{align}
P_{\textsf{a}}&:=\frac{1}{N_1N_2}\sum_{r,s}\sum_{P_{X_1X_2X_3X'_1X'_2Y}\in \cQ_1}W^n(\cE_{r,s,1}(P_{X_1X_2X_3X'_1X'_2Y})|\vecx_{1r}, \vecx_{2s}, \vecx_3),\label{eq:main_proof_error_a}\\
P_{\textsf{b}}&:=\frac{1}{N_1N_2}\sum_{r,s}\sum_{P_{X_1X_2X_3X'_1Y}\in \cQ_2}W^n(\cE_{r,s,2}(P_{X_1X_2X_3X'_1Y})|\vecx_{1r}, \vecx_{2s}, \vecx_3),\label{eq:main_proof_error_b}\\
P_{\textsf{c}} &:= \frac{1}{N_1N_2}\sum_{r,s}\sum_{P_{X_1X_2X_3X'_1X'_3Y}\in \cQ_3}W^n(\cE_{r,s,3}(P_{X_1X_2X_3X'_1X'_3Y})|\vecx_{1r}, \vecx_{2s}, \vecx_3).\label{eq:main_proof_error_c}
\end{align}}
We have three terms in the summation on the RHS of \eqref{eq_error}. We will analyze them one after the other. We will start with the first term. \pink{The analysis follows the flowchart given in Figure~\ref{fig:flowchart2}.}\\

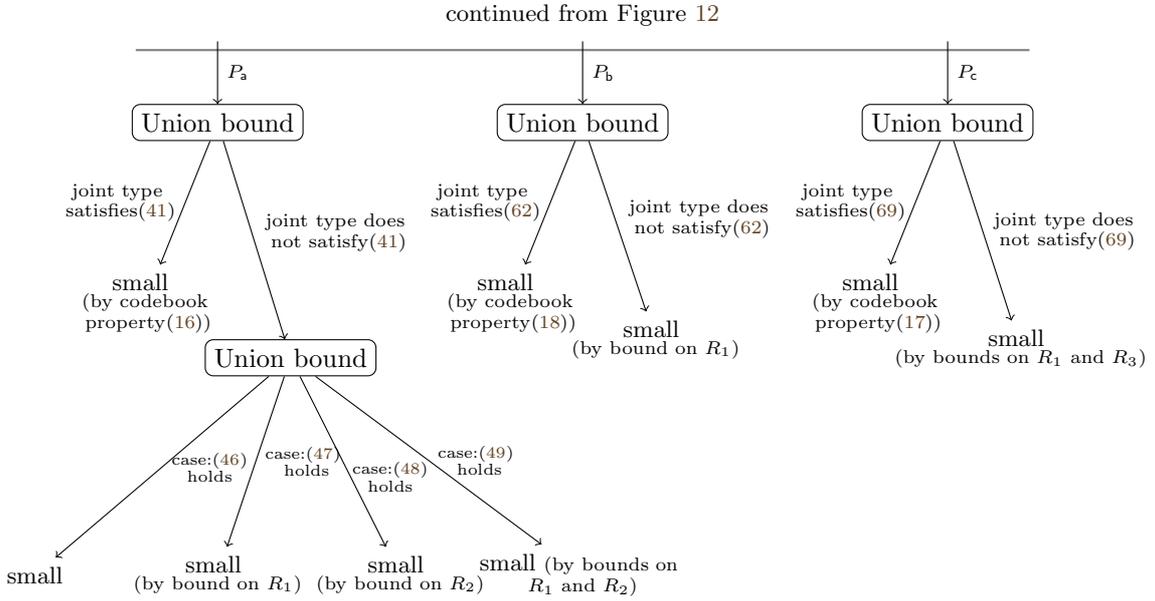
\begin{figure}[!h]
\centering
\begin{tikzpicture}[scale=1.2]
 \node (A) at (-5, 0.8)  {};
 \node (B) at (5, 0.8)  {};
 \node (G10) at (-4, 1)  {};
 \node (G20) at (0, 1) {};
 \node (G30) at (4, 1)  {};

 \node (G1) at (-4, 0) [rounded corners=3pt, draw] {Union bound};
 \node (G11) at (-4.8, -2)  {$\substack{\text{\small  small } \\ \text{(by codebook}\\\text{ property\eqref{code_eq3})}}$};
  \node (G12) at (-3.2, -2.6) [rounded corners=3pt, draw] {Union bound};

  \node (G122) at (-4, -5) {$\substack{\text{\small small }\\ \text{(by bound on $R_1$)}}$};
  \node (G123) at (-2, -5)  {$\substack{\text{\small small }\\ \text{(by bound on $R_2$)}}$};
  \node (G121) at (-6, -5)  {\small small};
  \node (G124) at (0, -5)  {$\substack{\text{{\small small} (by bounds on }\\ \text{$R_1$ and $R_2$)}}$};

 \node (G2) at (0, 0) [rounded corners=3pt, draw] {Union bound};
 \node (G21) at (-0.8, -2)  {$\substack{\text{\small small }\\ \text{(by codebook}\\ \text{ property\eqref{code_eq4})}}$};
  \node (G22) at (0.8, -2.4)  {$\substack{\text{\small small }\\ \text{(by bound on $R_1$)}}$};
 \node (G3) at (4, 0) [rounded corners=3pt, draw] {Union bound};
 \node (G31) at (3.2, -2)  {$\substack{\text{\small small }\\\text{(by codebook}\\ \text{ property\eqref{code_eq5})}}$};
  \node (G32) at (4.8, -2.5) {$\substack{\text{\small small }\\ \text{(by bounds on  $R_1$ and $R_3$)}}$};

  \draw[gray, thick, -] (A) -- (B) node[yshift = 0.2cm, midway, above] {\small\textcolor{black}{continued from Figure~\ref{fig:flowchart1}}};

 	\draw[->] (G10) -- (G1) node[midway, right]{$\scriptstyle P_{\textsf{a}}$};
 	\draw[->] (G20) -- (G2) node[midway, right]{$\scriptstyle P_{\textsf{b}}$};
 	\draw[->] (G30) -- (G3) node[midway, right]{$\scriptstyle P_{\textsf{c}}$};

 	\draw[->] (G12) -- (G121) node[midway, right]{\small$\substack{\text{case:\eqref{code_eq3_codn2_a}}\\ \text{holds}}$};
 	\draw[->] (G12) -- (G122) node[midway, right]{\small$\substack{\text{case:\eqref{code_eq3_codn2_b}}\\\text{ holds}}$};
 	\draw[->] (G12) -- (G123) node[midway, right, yshift = -0.2cm]{\small$\substack{\text{case:\eqref{code_eq3_codn2_c}}\\ \text{holds}}$};
 	\draw[->] (G12) -- (G124) node[midway, right]{\small$\substack{\text{case:\eqref{code_eq3_codn2_d}}\\\text{ holds}}$};

 	\draw[->] (G1) -- (G11) node[midway, left]{$\substack{\text{joint type }\\\text{satisfies} \eqref{eq:p_a_atypical}}$};
 	\draw[->] (G1) -- (G12) node[midway, right, yshift= 0.1cm]{$\substack{\text{joint type does}\\\text{ not satisfy} \eqref{eq:p_a_atypical}}$};
 	\draw[->] (G2) -- (G21) node[midway, left]{$\substack{\text{joint type }\\\text{satisfies} \eqref{eq:p_b_atypical}}$};
 	\draw[->] (G2) -- (G22) node[midway, right, yshift = 0.1cm]{$\substack{\text{joint type does}\\\text{ not satisfy} \eqref{eq:p_b_atypical}}$};
 	\draw[->] (G3) -- (G31) node[midway, left]{$\substack{\text{joint type }\\\text{satisfies} \eqref{eq:p_c_atypical}}$};
 	\draw[->] (G3) -- (G32) node[midway, right]{$\substack{\text{joint type does}\\\text{ not satisfy} \eqref{eq:p_c_atypical}}$};
\end{tikzpicture}
\caption{Flowchart, continued from Figure~\ref{fig:flowchart1}, depicting the flow of analysis of $P_{\textsf{a}}$, $P_{\textsf{b}}$ and $P_{\textsf{c}}$.}\label{fig:flowchart2}
\end{figure}
\noindent\pink{{\em \underline{\textbf{Analysis of $P_{\textsf{a}}$}}}}\\
We will follow the flowchart given in Figure~\ref{fig:flowchart2}. 
\pink{From \eqref{eq:main_proof_error_a}, }
\begin{align}\label{eq:error1}
\pink{P_{\textsf{a}} = }\sum_{P_{X_1X_2X_3X'_1X'_2Y}\in \cQ_1}\frac{1}{N_1N_2}\sum_{r,s}W^n(\cE_{r,s,1}(P_{X_1X_2X_3X'_1X'_2Y})|\vecx_{1r}, \vecx_{2s}, \vecx_3).
\end{align}
Let 
\begin{align*}
P^1_{r,s}(P_{X_1X_2X_3X'_1X'_2Y}) := W^n(\cE_{r,s,1}(P_{X_1X_2X_3X'_1X'_2Y})|\vecx_{1r}, \vecx_{2s}, \vecx_3).
\end{align*}

Note that $P^1_{r,s}(P_{X_1X_2X_3X'_1X'_2Y})$ is upper bounded by the probability of error when $r$ and $s$ are sent by user 1 and user 2 respectively. So, $P^1_{r,s}(P_{X_1X_2X_3X'_1X'_2Y})\leq 1$. Thus, from \eqref{eq:error1}, we see that it is sufficient to show that $P^1_{r,s}(P_{X_1X_2X_3X'_1X'_2Y})$ falls exponentially. 
Let $P_{\text{a, atypical}}$ be the set of joint types satisfying 
\begin{align}\label{eq:p_a_atypical}
I(X_1;X_2X'_1X'_2X_3)+I(X_2;X_1'X_2'X_3)\geq \left|\left|R_1-I(X'_1;X_3)\right|^+ +\left|R_2-I(X'_2;X_3)\right|^+-I(X'_1;X'_2|X_3)\right|^+ + \epsilon, 
\end{align}
From \eqref{code_eq3}, note that when $P_{X_1X_2X_3X'_1X'_2Y}\in \cQ_1$ satisfies \eqref{eq:p_a_atypical}
\begin{align}
&\frac{1}{N_1N_2}\sum_{r,s}P^1_{r,s}(P_{X_1X_2X_3X'_1X'_2Y})\\
&= \frac{1}{N_1N_2}\sum_{\substack{(r,s):\exists (u, v) \text{ satisfying }\\(\vecx_{1r},\vecx_{2s}, \vecx_3, \vecx_{1u},\vecx_{2v}) \in T^{n}_{X_1  X_2 X_3X_1'X'_2} }}W^n\inp{\inb{\vecy:\vecy\in T^{n}_{Y|X_1X_2X_3X_1'X'_2}(\vecx_{1r},\vecx_{2s}, \vecx_3, \vecx_{1u},\vecx_{2v})}\Big|\vecx_{1r}, \vecx_{2s}, \vecx_3}\\
&\leq \frac{1}{N_1N_2}|\{(r, s)\in [1:N_1]\times[1:N_2]: \exists u, v\in  [1:N_1]\times[1:N_2] \, u\neq r, v\neq r, \, (\vecx_{1r}, \vecx_{2s}, \vecx_{1u}, \vecx_{2v},\vecx_3)\in T^{n}_{X_1X_2X_1^{'}X'_2X_3} \}|\nonumber\\
& < \exp\left(-\frac{n\epsilon}{2}\right).\label{eq:small_fraction}
\end{align}
Otherwise, when
\begin{align}
I(X_1;X_2X'_1X'_2X_3)+I(X_2;X_1'X_2'X_3)< \left|\left|R_1-I(X'_1;X_3)\right|^+ +\left|R_2-I(X'_2;X_3)\right|^+-I(X'_1;X'_2|X_3)\right|^+ + \epsilon, \label{code_eq3_codn2}
\end{align}
\pink{depending on the evaluation of $\left|\left|R_1-I(X'_1;X_3)\right|^+ +\left|R_2-I(X'_2;X_3)\right|^+-I(X'_1;X'_2|X_3)\right|^+$, we consider four cases:
\begin{align}
&\left|\left|R_1-I(X'_1;X_3)\right|^+ +\left|R_2-I(X'_2;X_3)\right|^+-I(X'_1;X'_2|X_3)\right|^+ = 0, \label{code_eq3_codn2_a}\\
&\left|\left|R_1-I(X'_1;X_3)\right|^+ +\left|R_2-I(X'_2;X_3)\right|^+-I(X'_1;X'_2|X_3)\right|^+ =  R_1-I(X'_1;X_3)-I(X'_1;X'_2|X_3) , \label{code_eq3_codn2_b}\\
&\left|\left|R_1-I(X'_1;X_3)\right|^+ +\left|R_2-I(X'_2;X_3)\right|^+-I(X'_1;X'_2|X_3)\right|^+= R_2-I(X'_2;X_3)-I(X'_1;X'_2|X_3), \label{code_eq3_codn2_c}\\
&\left|\left|R_1-I(X'_1;X_3)\right|^+ +\left|R_2-I(X'_2;X_3)\right|^+-I(X'_1;X'_2|X_3)\right|^+ = R_1-I(X'_1;X_3) + R_2-I(X'_2;X_3)-I(X'_1;X'_2|X_3), \label{code_eq3_codn2_d}
\end{align}
Before proceeding further, we first argue that \eqref{code_eq3_codn2_a}-\eqref{code_eq3_codn2_d} are the only possible evaluations of \\$\left|\left|R_1-I(X'_1;X_3)\right|^+ +\left|R_2-I(X'_2;X_3)\right|^+-I(X'_1;X'_2|X_3)\right|^+$. To see this, first suppose $R_1\leq I(X'_1;X_3)$. If $R_2\leq I(X'_2;X_3)$, we get \eqref{code_eq3_codn2_a} as mutual information is always non-negative. When $R_2 > I(X'_2;X_3)$, if $R_2 > I(X'_2;X_3) + I(X'_1;X'_2|X_3)$, we get \eqref{code_eq3_codn2_c}. Otherwise, we get \eqref{code_eq3_codn2_a}. Next, suppose $R_1> I(X'_1;X_3)$. In this case, if $R_2\leq I(X'_2;X_3)$, depending on whether $R_1 > I(X'_1;X_3) + I(X'_1;X'_2|X_3)$ or not, we get \eqref{code_eq3_codn2_b} or \eqref{code_eq3_codn2_a} respectively. When $R_2> I(X'_2;X_3)$, we get \eqref{code_eq3_codn2_d} if $R_1 + R_2>I(X'_1;X_3) + I(X'_2;X_3)+I(X'_1;X'_2|X_3)$. Otherwise, we get \eqref{code_eq3_codn2_a}.
}

{Analysing each of the cases \eqref{code_eq3_codn2_a}-\eqref{code_eq3_codn2_d} separately, }we will show that $P^1_{r,s}(P_{X_1X_2X_3X'_1X'_2Y})\rightarrow 0$ exponentially for each $P_{X_1X_2X_3X'_1X'_2Y}\in \cQ_1$. We will show this by using the following upper bound.
\begin{align}
&P^1_{r,s}(P_{X_1X_2X_3X'_1X'_2Y}) = W^n(\cE_{r,s,1}(P_{X_1X_2X_3X'_1X'_2Y})|\vecx_{1r}, \vecx_{2s}, \vecx_3)\nonumber\\
&\quad\leq \sum_{\substack{(u, v):(\vecx_{1r},\vecx_{2s}, \vecx_{3}, \vecx_{1u}, \vecx_{2v})\\\in T^{n}_{X_1 X_2 X_3 X_1' X_2'}}}\quad\sum_{\vecy\in T^n_{Y|X_1X_2X_3X'_1X'_2}(\vecx_{1r},\vecx_{2s}, \vecx_{3}, \vecx_{1u}, \vecx_{2v})}W^n(\vecy|\vecx_{1r},\vecx_{2s}, \vecx_{3})\nonumber\\
&\quad\leq \sum_{\substack{(u, v):(\vecx_{1r},\vecx_{2s}, \vecx_{3}, \vecx_{1u}, \vecx_{2v})\\\in T^{n}_{X_1 X_2 X_3 X_1' X_2'}}}\exp{\left(-n(I(Y;X_1'X_2'|X_1X_2X_3)-\epsilon)\right)}\nonumber\\
&\stackrel{(a)}{\leq}\exp{\left(n\left(\left||R_1-I(X_1';X_1X_2X_3)|^{+}+|R_2-I(X_2';X_1X_2X_3)|^{+}-I(X_1';X_2'|X_1X_2X_3)\right|^{+}-I(Y;X_1'X_2'|X_1X_2X_3)+3\epsilon/2\right)\right)}\label{eq:upper_bound1}
\end{align}
where $(a)$ follows from \eqref{lemma_eq7a}.\\
\pink{{\em \underline{\textbf{Case 1: \eqref{code_eq3_codn2_a} holds.}}}}\\
\pink{We first note that when 
\begin{align}\label{eq:rate_cases1aa}
\left|\left|R_1-I(X'_1;X_3)\right|^+ +\left|R_2-I(X'_2;X_3)\right|^+-I(X'_1;X'_2|X_3)\right|^+ = 0
\end{align} holds, }\eqref{code_eq3_codn2} implies that $I(X_1;X_2X'_1X'_2X_3)+I(X_2;X_1'X_2'X_3)<\epsilon$. This further implies the following:
\begin{align}
\epsilon &> I(X_1;X_2X'_1X'_2X_3)+I(X_2;X_1'X_2'X_3)\nonumber\\
&\stackrel{\pink{(a)}}{\geq} I(X_1;X'_1X'_2|X_2X_3)+I(X_2;X_1'X_2'|X_3)\nonumber\\
&= I(X_1X_2;X'_1X'_2|X_3)\label{eq:rate_cases1bb}
\end{align}
\pink{where $(a)$ holds because $I(X_1;X_2X'_1X'_2X_3) = I(X_1;X'_1X'_2|X_2X_3) + I(X_1;X_2X_3)$ and $I(X_1;X_2X_3)\geq 0$ as mutual information is always non-negative. Next, we will argue that when \eqref{eq:rate_cases1aa} holds, the condition $$\left||R_1-I(X_1';X_1X_2X_3)|^{+}+|R_2-I(X_2';X_1X_2X_3)|^{+}-I(X_1';X_2'|X_1X_2X_3)\right|^{+} = 0$$ also holds and thus \eqref{eq:upper_bound1} evaluates to $\exp{\left(n\left(0-I(Y;X_1'X_2'|X_1X_2X_3)+3\epsilon/2\right)\right)}$.}
\pink{We show this by contradiction. Suppose $\left||R_1-I(X_1';X_1X_2X_3)|^{+}+|R_2-I(X_2';X_1X_2X_3)|^{+}-I(X_1';X_2'|X_1X_2X_3)\right|^{+} > 0$. This implies that at least one of the following three conditions hold.
\begin{align}
&R_1 > I(X_1';X_1X_2X_3) + I(X_1';X_2'|X_1X_2X_3),\label{eq:rate_cases1a}\\
&R_2 > I(X_2';X_1X_2X_3) + I(X_1';X_2'|X_1X_2X_3),\label{eq:rate_cases2b}\\
\text{ \em or }&R_1> I(X_1';X_1X_2X_3), R_2> I(X_2';X_1X_2X_3)\nonumber\\&\text{ and }R_1+R_2 > I(X_1';X_1X_2X_3)+I(X_2';X_1X_2X_3)+I(X_1';X_2'|X_1X_2X_3).\label{eq:rate_cases3c}
\end{align}
If \eqref{eq:rate_cases1a} holds, then 
\begin{align*}
R_1 &> I(X_1';X_1X_2X_3) + I(X_1';X_2'|X_1X_2X_3)\\
&= I(X_1';X_2'X_1X_2X_3)\\
& = I(X_1';X_3) + I(X_1';X_2'|X_3) + I(X_1';X_1X_2|X_2'X_3)\\
& \stackrel{(a)}{\geq}I(X_1';X_3) + I(X_1';X_2'|X_3)
\end{align*} where $(a)$ follows from non-negativity of mutual information. Note that the inequality $(a)$ contradicts \eqref{eq:rate_cases1aa}. The condition \eqref{eq:rate_cases2b} is symmetric and hence leads to a contradiction again. If \eqref{eq:rate_cases3c} holds, then $R_1>I(X_1';X_1X_2X_3) \geq I(X_1';X_3)$ and $R_2> I(X_2';X_1X_2X_3)\geq I(X_2';X_3)$. Furthermore, we have
\begin{align*}
R_1+R_2 &> I(X_1';X_1X_2X_3)+I(X_2';X_1X_2X_3)+I(X_1';X_2'|X_1X_2X_3)\\
& = I(X_1';X_2'X_1X_2X_3) + I(X_2';X_1X_2X_3)\\
&\geq I(X_1';X_3) + I(X_1';X_2'|X_3) + I(X_2';X_1X_2X_3)\\
&\geq I(X_1';X_3) + I(X_1';X_2'|X_3) + I(X_2';X_3).
\end{align*}
These conditions together contradict \eqref{eq:rate_cases1aa}.} 
Thus,  
\begin{align*}
P^1_{r,s}(P_{X_1X_2X_3X'_1X_2'Y})&\leq\exp{\left(n\left(-I(Y;X_1'X_2'|X_1X_2X_3)+3\epsilon/2\right)\right)}\\
&\le \exp{\left(n\left(I(X_1X_2;X_1'X_2'|X_3)-I(X_1X_2Y;X_1'X_2'|X_3)+3\epsilon/2\right)\right)}\\
&\stackrel{(a)}{\leq} \exp{\left(n\left(\epsilon-\eta+3\epsilon/2\right)\right)}\\
&= \exp{\left(n\left(5\epsilon/2-\eta\right)\right)}\\
&\rightarrow 0 \text{ as }\pink{\eta > 6\epsilon.}
\end{align*}
where $(a)$ uses the fact that $I(X_1X_2;X_1'X_2'|X_3)<\epsilon$ \pink{(see \eqref{eq:rate_cases1bb})} and $I(X_1X_2Y;X_1'X_2'|X_3)\geq \eta$ \pink{(follows from the definition of $\cQ_1$. See \eqref{eq:main_proof_q_1}.)}.

\noindent \pink{{\em \underline{\textbf{Case 2: \eqref{code_eq3_codn2_b} holds.}}}}\\
We consider the case when $\left|\left|R_1-I(X'_1;X_3)\right|^+ +\left|R_2-I(X'_2;X_3)\right|^+-I(X'_1;X'_2|X_3)\right|^+ = R_1-I(X'_1;X'_2X_3)$.  In this case \eqref{code_eq3_codn2} evaluates to $I(X_1;X_2X'_1X'_2X_3)+I(X_2;X_1'X_2'X_3)< R_1-I(X'_1;X'_2X_3) + \epsilon$. This implies the following: 
\begin{align*}
\epsilon&> I(X_1;X_2X'_1X'_2X_3)+I(X_2;X_1'X_2'X_3)+I(X'_1;X'_2X_3)-R_1\\
&\geq  I(X_1;X_1'|X'_2X_2X_3)+I(X_2;X_1'|X_2'X_3)+I(X'_1;X'_2X_3)-R_1\\
&=I(X'_1; X_1X_2X_3X'_2) - R_1.
\end{align*}
Thus,
\begin{align*}
R_1-I(X'_1;X_1X_2X_3X'_2)\geq -\epsilon.
\end{align*}
This implies that 
\begin{align*}
|R_1-I(X_1';X_1X_2X_3)|^{+} \leq R_1-I(X'_1;X_1X_2X_3) + \epsilon
\end{align*}
and  we get the following upper bound on \eqref{eq:upper_bound1}:
\begin{align}
&\exp{\left(n\left(\left||R_1-I(X_1';X_1X_2X_3)|^{+}+|R_2-I(X_2';X_1X_2X_3)|^{+}-I(X_1';X_2'|X_1X_2X_3)\right|^{+}-I(Y;X_1'X_2'|X_1X_2X_3)+3\epsilon/2\right)\right)}\nonumber\\
&\leq \exp{\left(n\left(R_1-I(X_1';X_1X_2X_3X_2') +\epsilon+0-I(Y;X_1'X_2'|X_1X_2X_3)+3\epsilon/2\right)\right)}\label{eq:local_upper_bd}\\
&\leq\exp{\left(n\left(R_1-I(X_1';X_1X_2X_3X_2'Y)+5\epsilon/2\right)\right)}\nonumber\\
&\leq\exp{\left(n\left(R_1-I(X_1';X_2'Y)+5\epsilon/2\right)\right)}\nonumber\\
&{\leq}\exp{\left(n\left(R_1-I(\tilde{X}_1;\tilde{Y}|\tilde{X}_2) + \gamma+5\epsilon/2\right)\right)}\nonumber
\end{align}
where $P_{\tilde{X}_1\tilde{X}_2\tilde{X}_3\tilde{Y}} \defineqq P_{X'_1}\times P_{X_2}\times P_{X'_3}\times W$ and $\gamma$ is chosen to satisfy $I(X_1';X_2'Y) \geq I(\tilde{X}_1;\tilde{Y}|\tilde{X}_2)-\gamma$. Note that $P_{X'_1X_2'X'_3Y}$ is such that $D(P_{X_1'X_2X'_3 Y}||P_{X'_1}\times P_{X_2}\times P_{X'_3}\times W)< \eta$ where $\eta$ can be chosen arbitrarily small. Thus, $P_{X_1'X_2X'_3 Y}$ is arbitrarily close to $P_{\tilde{X}_1\tilde{X}_2\tilde{X}_3\tilde{Y}}$ and $\gamma$ can be chosen arbitrarily small. Thus, $P^1_{r,s}(P_{X_1X_2X_3X'_1X_2'Y})\rightarrow 0$ exponentially, if
\begin{align*}
R_1&< I(\tilde{X}_1;\tilde{Y}|\tilde{X}_2)- \gamma-5\epsilon/2.
\end{align*}
Minimizing this in the limit of $n\rightarrow \infty$ and $\epsilon, \eta \rightarrow 0$ over all $P_{X_1X_2X_3X'_1X_2'Y}\in \cQ_1$ is same as minimizing $I(\tilde{X}_1;\tilde{Y}|\tilde{X}_2)$ over $P_{\tilde{X}_1\tilde{X}_2\tilde{X}_3\tilde{Y}}\in \cP_3$ where $\cP_3$ is defined as 
\begin{align*}
\cP_3\defineqq \{P_{X_1X_2X_3Y}: P_{X_1X_2X_3Y} = P_{X_1}\times P_{X_2}\times Q_{X_3}\times W \text{ for some }Q_{X_3}\}.
\end{align*}
Using definition of $\cP_3$, we obtain the following bound on $R_1$
\begin{align}\label{eq:R_1_bound}
R_1&< \min_{P_{X_1X_2X_3Y}\in \cP_3}I(X_1;Y|X_2).
\end{align}

\noindent \pink{{\em \underline{\textbf{Case 3: \eqref{code_eq3_codn2_c} holds.}}}}\\
Suppose $\left|\left|R_1-I(X'_1;X_3)\right|^+ +\left|R_2-I(X'_2;X_3)\right|^+-I(X'_1;X'_2|X_3)\right|^+ = R_2-I(X'_2;X'_1X_3)$. In this case \eqref{code_eq3_codn2} evaluates to $I(X_1;X_2X'_1X'_2X_3)+I(X_2;X_1'X_2'X_3)< R_2-I(X'_2;X'_1X_3) + \epsilon$. Thus, 
\begin{align*}
\epsilon&> I(X_1;X_2X'_1X'_2X_3)+I(X_2;X_1'X_2'X_3)+I(X'_2;X'_1X_3)-R_2\\
&\geq  I(X_1;X_2'|X'_1X_2X_3)+I(X_2;X_2'|X_1'X_3)+I(X'_2;X'_1X_3)-R_2\\
&=I(X'_2; X_1X_2X_3X'_1) - R_2.
\end{align*}
This implies that
\begin{align*}
R_2-I(X'_2;X_1X_2X_3X'_1)\geq -\epsilon.
\end{align*}
Thus, $|R_2-I(X'_2;X_1X_2X_3X'_1)|^+ \leq R_2-I(X'_2;X_1X_2X_3X'_1) +\epsilon$. Substituting this in ~\eqref{eq:upper_bound1}, we get the following upper bound:
\begin{align*}
\exp{\left(n\left(0 + R_2-I(X_2';X_1X_2X_3X_1') +\epsilon-I(Y;X_1'X_2'|X_1X_2X_3)+3\epsilon/2\right)\right)}
\end{align*}
This is same as the upper bound in \eqref{eq:local_upper_bd} with $X'_1$ and $X'_2$ interchanged, and $R_1$ replaced by $R_2$. Thus, we can do a symmetric analysis as in the previous case to obtain the following bound on $R_2$:
\begin{align}\label{eq:R_2_bound}
R_2&< \min_{P_{X_1X_2X_3Y}\in \cP_3}I(X_2;Y|X_1)
\end{align}

\noindent\pink{{\em \underline{\textbf{Case 4: \eqref{code_eq3_codn2_d} holds.}}}}\\
Suppose $\left|\left|R_1-I(X'_1;X_3)\right|^+ +\left|R_2-I(X'_2;X_3)\right|^+-I(X'_1;X'_2|X_3)\right|^+ = R_1-I(X'_1;X_3) + R_2-I(X'_2;X'_1X_3)$.   In this case \eqref{code_eq3_codn2} evaluates to $I(X_1;X_2X'_1X'_2X_3)+I(X_2;X_1'X_2'X_3)< R_1 + R_2-I(X'_1;X_3) + I(X'_2;X'_1X_3) + \epsilon$. Thus, 
\begin{align*}
\epsilon&> I(X_1;X_2X'_1X'_2X_3)+I(X_2;X_1'X_2'X_3)+I(X'_1;X_3) + I(X'_2;X'_1X_3)-R_1-R_2\\
&\geq  I(X_1;X'_1|X_2X_3) + I(X_1;X'_2|X'_1X_2X_3) +I(X_2;X_1'|X_3) +I(X_2;X_2'|X_1'X_3)  +I(X'_1;X_3)+I(X'_2;X'_1X_3)-R_1-R_2\\
&=I(X'_1;X_1X_2X_3) + I(X'_2; X_1X_2X_3X'_1) - R_1- R_2.
\end{align*}
This implies that
\begin{align*}
R_1 - I(X'_1;X_1X_2X_3) + R_2 - I(X'_2; X_1X_2X_3) - I(X'_1; X_2'|X_1X_2X_3) \geq -\epsilon.
\end{align*}
Note that 
\begin{align*}
&\left||R_1 - I(X'_1;X_1X_2X_3)|^+ + |R_2 - I(X'_2; X_1X_2X_3)|^+ - I(X'_1; X_2'|X_1X_2X_3)\right|^+\\
&\qquad\geq R_1 - I(X'_1;X_1X_2X_3) + R_2 - I(X'_2; X_1X_2X_3) - I(X'_1; X_2'|X_1X_2X_3)\\
&\qquad\geq -\epsilon.
\end{align*}
So, 
\begin{align*}
\left||R_1 - I(X'_1;X_1X_2X_3)|^+ + |R_2 - I(X'_2; X_1X_2X_3)|^+ - I(X'_1; X_2'|X_1X_2X_3)\right|^+\\
\leq R_1 - I(X'_1;X_1X_2X_3) + R_2 - I(X'_2; X_1X_2X_3) - I(X'_1; X_2'|X_1X_2X_3) +\epsilon.
\end{align*}
Thus,
\begin{align}
P^1_{r,s}(P_{X_1X_2X_3X'_1X_2'Y}) &\leq \exp{\left(n\left(R_1-I(X_1';X_1X_2X_3)+R_2-I(X_2';X_1X_2X_3X'_1)+ \epsilon-I(Y;X_1'X_2'|X_1X_2X_3)+3\epsilon/2\right)\right)}\label{eq:r_1_r_2_large}\\
&=\exp{\left(n\left(R_1-I(X_1';X_1X_2X_3)+R_2-I(X_2';X_1X_2X_3X'_1)-I(Y;X_1'X_2'|X_1X_2X_3)+5\epsilon/2\right)\right)}\nonumber\\
&\leq\exp{\left(n\left(R_1 +R_2-I(X_1'X_2';X_1X_2X_3)-I(Y;X_1'X_2'|X_1X_2X_3)+5\epsilon/2\right)\right)}\nonumber\\
&\leq\exp{\left(n\left(R_1 +R_2-I(X_1'X_2';X_1X_2X_3Y)+5\epsilon/2\right)\right)}\nonumber\\
&\leq\exp{\left(n\left(R_1 +R_2-I(X_1'X_2';Y)+5\epsilon/2\right)\right)}\nonumber
\end{align}
Following similar steps as earlier, we obtain the following sum rate bound
\begin{align}\label{eq:R_1_R_2_bound}
R_1 + R_2 &< \min_{P_{X_1X_2X_3Y}\in \cP_3}I(X_1X_2;Y).
\end{align}
\pink{{\em \underline{\textbf{Analysis of $P_{\textsf{b}}$}}}}\\
Now, we will look at the second term in \eqref{eq_error}, which is \pink{(see \eqref{eq:main_proof_error_b})}, 
\begin{align}\label{eq:error12}
\pink{P_{\textsf{b}}:=}\sum_{P_{X_1X_2X_3X'_1Y}\in \cQ_2}\frac{1}{N_1N_2}\sum_{r,s}W^n(\cE_{r,s,2}(P_{X_1X_2X_3X'_1Y})|\vecx_{1r}, \vecx_{2s}, \vecx_3).
\end{align}
Let 
\begin{align*}
P^2_{r,s}(P_{X_1X_2X_3X'_1Y}) := W^n(\cE_{r,s,2}(P_{X_1X_2X_3X'_1Y})|\vecx_{1r}, \vecx_{2s}, \vecx_3).
\end{align*}
From \eqref{code_eq4}, we see that when $P_{X_1X_2X_3X'_1Y}$ satisfies 
\begin{align}\label{eq:p_b_atypical}
I(X_i;X_jX'_iX_k)+I(X_j;X_i'X_k)\geq \left|R_i-I(X'_i;X_k)\right|^+  + \epsilon,
\end{align}
\begin{align}
&\frac{1}{N_1N_2}\sum_{r,s}P^2_{r,s}(P_{X_1X_2X_3X'_1Y})\nonumber\\
&= \frac{1}{N_1N_2}\sum_{\substack{(r,s):\exists u \text{ satisfying }\\(\vecx_{1r},\vecx_{2s}, \vecx_3, \vecx_{1u}) \in T^{n}_{X_1  X_2 X_3X_1'} }}W^n\inp{\inb{\vecy:\vecy\in T^{n}_{Y|X_1X_2X_3X_1'}(\vecx_{1r},\vecx_{2s}, \vecx_3, \vecx_{1u})}\Big|\vecx_{1r}, \vecx_{2s}, \vecx_3}\\
&\leq \frac{1}{N_1N_2}|\{(r, s)\in [1:N_1]\times[1:N_2]: \exists u\in [1:N_1] \, u\neq r,\, (\vecx_{1r}, \vecx_{2s},\vecx_3, \vecx_{1u})\in T^{n}_{X_1X_2 X_3 X_1^{'}} \}|\nonumber\\
& < \exp\left(-\frac{n\epsilon}{2}\right).\label{eq:small_fraction2}
\end{align}
Otherwise, when
\begin{align}
I(X_i;X_jX'_iX_k)+I(X_j;X_i'X_k)<\left|R_i-I(X'_i;X_k)\right|^+  + \epsilon,\label{eq:cond12}
\end{align}
we will show that $P^2_{r,s}(P_{X_1X_2X_3X'_1Y})$ falls doubly exponentially for each $P_{X_1X_2X_3X'_1Y}\in \cQ_2$. We will show this by using the following upper bound.

\begin{align}
&P^2_{r,s}(P_{X_1X_2X_3X'_1Y}) = W^n(\cE_{r,s,2}(P_{X_1X_2X_3X'_1Y})|\vecx_{1r}, \vecx_{2s}, \vecx_3)\nonumber\\
&\qquad\leq \sum_{\substack{u:(\vecx_{1r},\vecx_{2s}, \vecx_{3}, \vecx_{1u})\\\in T^{n}_{X_1 X_2 X_3 X_1'}}}\quad\sum_{\vecy\in T^n_{Y|X_1X_2X_3X'_1}(\vecx_{1r},\vecx_{2s}, \vecx_{3}, \vecx_{1u})}W^n(\vecy|\vecx_{1r},\vecx_{2s}, \vecx_{3})\nonumber\\
&\qquad\leq \sum_{\substack{u:(\vecx_{1r},\vecx_{2s}, \vecx_{3}, \vecx_{1u})\\\in T^{n}_{X_1 X_2 X_3 X_1'}}}\exp{\left(-n(I(Y;X_1'|X_1X_2X_3)-\epsilon)\right)}\nonumber\\
&\qquad\stackrel{(a)}{=}\exp{\left(n\left(\left|R_1-I(X_1';X_1X_2X_3)\right|^{+}-I(Y;X_1'|X_1X_2X_3)+3\epsilon/2\right)\right)}\label{eq:upper_bound12}
\end{align}
where $(a)$ follows from \eqref{lemma_eq7b}.

Suppose $R_1\leq I(X'_1;X_2)$, then \eqref{eq:cond12} evaluates to $I(X_1;X_2X'_1X_3)+I(X_2;X_1'X_3)< \epsilon$. Thus, $I(X_1X_2;X_1'|X_3)<\epsilon$. We analyze \eqref{eq:upper_bound12} for this case.
\begin{align*}
P^2_{r,s}(P_{X_1X_2X_3X'_1Y}) &\leq \exp{\left(n\left(\left|R_1-I(X_1';X_1X_2X_3)\right|^{+}-I(Y;X_1'|X_1X_2X_3)+3\epsilon/2\right)\right)}\\
&=\exp{\left(n\left(0-I(Y;X_1'|X_1X_2X_3)+3\epsilon/2\right)\right)}\\
&=\exp{\left(n\left(I(X_1X_2;X_1'|X_3)-I(X_1X_2;X_1'|X_3)-I(Y;X_1'|X_1X_2X_3)+3\epsilon/2\right)\right)}\\
&=\exp{\left(n\left(I(X_1X_2;X_1'|X_3)-I(X_1X_2Y;X_1'|X_3)+3\epsilon/2\right)\right)}\\
&\stackrel{(a)}{\leq}\exp{\left(n\left(\epsilon-\eta+3\epsilon/2\right)\right)}\\
&\rightarrow 0 \pink{\text{ as }\eta> 6\epsilon}.
\end{align*}where $(a)$ follows by using $I(X_1X_2;X_1'|X_3)<\epsilon$ and $I(X_1X_2Y;X_1'|X_3)\geq\eta$ (see the definition of $\cQ_2$).

Now, we consider the case when $R_1> I(X'_1;X_2)$. In this case, \eqref{eq:cond12} evaluates to 
\begin{align*}
I(X_1;X_2X'_1X_3)+I(X_2;X_1'X_3)<R_1-I(X'_1;X_3) + \epsilon
\end{align*}
This implies that $-\epsilon<R_1- I(X_1';X_1X_2X_3)\leq \left|R_1-I(X_1';X_1X_2X_3)\right|^{+}$. Thus,  
\begin{align*}
\left|R_1-I(X_1';X_1X_2X_3)\right|^{+}-I(Y;X_1'|X_1X_2X_3) \leq R_1-I(X_1';X_1X_2X_3)-I(Y;X_1'|X_1X_2X_3) +\epsilon
\end{align*}
Plugging it into the upper bound on $P^2_{r,s}(P_{X_1X_2X_3X'_1Y})$, we obtain
\begin{align*}
P^2_{r,s}(P_{X_1X_2X_3X'_1Y}) &\leq \exp{\left(n\left(R_1-I(X_1';X_1X_2X_3)-I(Y;X_1'|X_1X_2X_3)+5\epsilon/2\right)\right)}\\
&=\exp{\left(n\left(R_1-I(X_1';X_1X_2X_3Y)+5\epsilon/2\right)\right)}\\
&=\exp{\left(n\left(R_1-I(X_1';X_2Y)+5\epsilon/2\right)\right)}
\end{align*}

Since $P_{X'_1X_2X'_3Y}$ is such that $D(P_{X_1'X_2X'_3 Y}||P_{X'_1}\times P_{X_2}\times P_{X'_3}\times W)< \eta$ where $\eta$ can be chosen arbitrarily small, $P_{X_1'X_2X'_3 Y}$ is arbitrarily close to $P_{\tilde{X}_1\tilde{X}_2\tilde{X}_3\tilde{Y}} \defineqq P_{X'_1}\times P_{X_2}\times P_{X'_3}\times W$. So, for small positive number $\gamma_2$,  $I(X_1';X_2Y) \geq I(\tilde{X}_1;\tilde{Y}|\tilde{X}_2)-\gamma_2 \geq \min_{P_{X'_3}} I(\tilde{X}_1;\tilde{Y}|\tilde{X}_2)-\gamma_2$. Thus, if 
\begin{align*}
R_1&<\min_{P_{\tilde{X_3}}}I(\tilde{X}_1;\tilde{Y}|\tilde{X}_2)-5\epsilon/2-\gamma_2, \\
\text{then, }R_1 &\leq \min_{P_{X'_3}}I(X'_1;Y|X_2)-5\epsilon/2,
\end{align*}
and therefore, $P_{e,_{X_1,X'_1,X_2X_3Y}} \rightarrow 0$ as $n\rightarrow 0$. In the limit of $\epsilon\rightarrow 0$, we get
\begin{align}\label{eq:rate_R1}
R_1 &\leq \min_{\substack{P_{X_3'}: P_{\tilde{X}_1\tilde{X}_2\tilde{X}_3\tilde{Y}} \\= P_{X'_1}\times P_{X_2}\times P_{X'_3}\times W}}I(\tilde{X}_1;\tilde{Y}|\tilde{X}_2)
\end{align}
This is same as the upper bound on $R_1$ given in \eqref{eq:R_1_bound}.\\
\pink{{\em \underline{\textbf{Analysis of $P_{\textsf{c}}$}}}}\\
We are left with the analysis of the third term in \eqref{eq_error}, which is given by \pink{(see \eqref{eq:main_proof_error_c})}
\begin{align}\label{eq:error13}
\pink{P_{\textsf{c}}}:=\sum_{P_{X_1X_2X_3X'_1X'_3Y}\in \cQ_3}\frac{1}{N_1N_2}\sum_{r,s}W^n(\cE_{r,s,3}(P_{X_1X_2X_3X'_1X'_3Y})|\vecx_{1r}, \vecx_{2s}, \vecx_3).
\end{align}
Let 
\begin{align*}
P^3_{r,s}(P_{X_1X_2X_3X'_1X'_3Y}) := W^n(\cE_{r,s,3}(P_{X_1X_2X_3X'_1X'_3Y})|\vecx_{1r}, \vecx_{2s}, \vecx_3).
\end{align*}
When $P_{X_1X_2X_3X'_1X'_3Y}$ satisfies the condition (see \eqref{code_eq5}),
\begin{align}\label{eq:p_c_atypical}
I(X_i;X_jX'_iX'_kX_k)+I(X_j;X_i'X_k'X_k)\geq \left|\left|R_i-I(X'_i;X_k)\right|^+ +\left|R_k-I(X'_k;X_k)\right|^+-I(X'_i;X'_k|X_k)\right|^+ + \epsilon,
\end{align}
\begin{align}
&\frac{1}{N_1N_2}\sum_{r,s}P^3_{r,s}(P_{X_1X_2X_3X'_1X'_3Y})\\
&= \frac{1}{N_1N_2}\sum_{\substack{(r,s):\exists (u, w) \text{ satisfying }\\(\vecx_{1r},\vecx_{2s}, \vecx_3, \vecx_{1u},\vecx_{3w}) \in T^{n}_{X_1  X_2 X_3X_1'X'_3} }}W^n\inp{\inb{\vecy:\vecy\in T^{n}_{Y|X_1X_2X_3X_1'X'_3}(\vecx_{1r},\vecx_{2s}, \vecx_3, \vecx_{1u},\vecx_{3w})}\Big|\vecx_{1r}, \vecx_{2s}, \vecx_3}\\
&\leq \frac{1}{N_1N_2}|\{(r, s)\in [1:N_1]\times[1:N_2]: \exists u, w\in  [1:N_1]\times[1:N_3] \, u\neq r \, (\vecx_{1r}, \vecx_{2s}, \vecx_3,\vecx_{1u}, \vecx_{3w})\in T^{n}_{X_1X_2 X_3 X_1^{'} X'_3} \}|\nonumber\\
& < \exp\left(-\frac{n\epsilon}{2}\right).\label{eq:small_fraction3}
\end{align}
Otherwise, when
\begin{align}
I(X_i;X_jX'_iX'_kX_k)+I(X_j;X_i'X_k'X_k)< \left|\left|R_i-I(X'_i;X_k)\right|^+ +\left|R_k-I(X'_k;X_k)\right|^+-I(X'_i;X'_k|X_k)\right|^+ + \epsilon,\label{eq:cond13}
\end{align}
we will show that $P^3_{r,s}(P_{X_1X_2X_3X'_1X'_3Y})$ falls doubly exponentially for each $P_{X_1X_2X_3X'_1X'_3Y}\in \cQ_3$. We upper bound $P^3_{r,s}(P_{X_1X_2X_3X'_1X'_3Y})$ by the following set of equations.

\begin{align}
&P^3_{r,s}(P_{X_1X_2X_3X'_1X'_3Y}) = W^n(\cE_{r,s,3}(P_{X_1X_2X_3X'_1X'_3Y})|\vecx_{1r}, \vecx_{2s}, \vecx_3)\nonumber\\
&\qquad\leq \sum_{\substack{(u, w):(\vecx_{1r},\vecx_{2s}, \vecx_{3}, \vecx_{1u}, \vecx_{3w})\\\in T^{n}_{X_1 X_2 X_3 X_1' X_3'}}}\quad\sum_{\vecy\in T^n_{Y|X_1X_2X_3X'_1X'_3}(\vecx_{1r},\vecx_{2s}, \vecx_{3}, \vecx_{1u}, \vecx_{3w})}W^n(\vecy|\vecx_{1r},\vecx_{2s}, \vecx_{3})\nonumber\\
&\qquad\leq \sum_{\substack{(u, w):(\vecx_{1r},\vecx_{2s}, \vecx_{3}, \vecx_{1u}, \vecx_{3w})\\\in T^{n}_{X_1 X_2 X_3 X_1' X_3'}}}\exp{\left(-n(I(Y;X_1'X_3'|X_1X_2X_3)-\epsilon)\right)}\nonumber\\
&\stackrel{(a)}{=}\exp{\left(n\left(\left||R_1-I(X_1';X_1X_2X_3)|^{+}+|R_3-I(X_3';X_1X_2X_3)|^{+}-I(X_1';X_3'|X_1X_2X_3)\right|^{+}-I(Y;X_1'X_3'|X_1X_2X_3)+3\epsilon/2\right)\right)}.\label{eq:upper_bound13}
\end{align}
where $(a)$ follows from \eqref{lemma_eq7c}. Now, we need to show that \eqref{eq:upper_bound13} goes to zero under the condition given in \eqref{eq:cond13}. This is same as the previous analysis of \eqref{code_eq3_codn2} under the condition \eqref{eq:upper_bound1} with $R_2$ and $X_2'$ replaced by $R_3$ and $X_3'$. Note that with these replacements, the entire analysis follows through and we obtain the analogues of \eqref{eq:R_1_bound}, \eqref{eq:R_2_bound} and \eqref{eq:R_1_R_2_bound} as given in \eqref{eq:R_1_bound3}, \eqref{eq:R_2_bound3} and \eqref{eq:R_1_R_2_bound3} respectively. For
\begin{align*}
\cP_2\defineqq \{P_{X_1X_2X_3Y}: P_{X_1X_2X_3Y} = P_{X_1}\times Q_{X_2}\times P_{X_3}\times W \text{ for some }Q_{X_2}\}
\end{align*}
\begin{align}\label{eq:R_1_bound3}
R_1&< \min_{P_{X_1X_2X_3Y}\in \cP_2}I(X_1;Y|X_3);
\end{align}
\begin{align}\label{eq:R_2_bound3}
R_3&< \min_{P_{X_1X_2X_3Y}\in \cP_2}I(X_3;Y|X_1);
\end{align}
\begin{align}\label{eq:R_1_R_2_bound3}
R_1 + R_3 &< \min_{P_{X_1X_2X_3Y}\in \cP_2}I(X_1X_3;Y).
\end{align}
Similarly, we will obtain rate bounds while analyzing the cases when user 1 and 2 are adversarial.

\pink{
Thus, for any input distribution $p(x_1)p(x_2)p(x_3)$, we have shown the achievability of the set of rate triples $(R_1,R_2,R_3)$ which, for all permutations $(i,j,k)$ of $(1,2,3)$, satisfy the following conditions:  
\begin{align}
R_i &< \min_{q(x_k)} I(X_i;Y|X_j),\quad\text{and}\label{eq:rateconstraint1_no_time_sharing}\\
R_i+R_j &< \min_{q(x_k)} I(X_i,X_j;Y),\label{eq:rateconstraint2_no_time_sharing}
\end{align} where the mutual information terms are evaluated using the joint distribution $p(x_i)p(x_j)q(x_k)W(y|x_1,x_2,x_3)$.}

\pink{It remains to argue that the rate region $\cR$ given by \eqref{eq:rateconstraint1} and \eqref{eq:rateconstraint2} is achievable. To this end, consider a distribution\footnote{For clarity, in the rest of this proof we introduce subscripts to denote the p.m.f.s involved in \eqref{eq:rateconstraint1} and \eqref{eq:rateconstraint2}.} $p_Up_{X_1|U}p_{X_2|U}p_{X_3|U}$. Without loss of generality, take $\cU=\{1,2,\ldots,|\cU|\}$. It suffices to show the achievability for $p_U(u)$ whose elements are rational numbers. Let $l$ be such that $lp_U(u)$ are integers for all $u\in\cU$. For $u\in \cU$, let $m_u=lp_U(u)$ and $n_u=\sum_{j\leq u} m_j$, and let $n_0=0$.}

\pink{Consider the $l$-fold product $W^{\otimes l}$ of the channel $W$. For this product channel, consider the input distribution $p(\vecx_1)p(\vecx_2)p(\vecx_3)$ defined by
\begin{align*}
p(\vecx_i)=p((x_{i1},\ldots,x_{il}))=\prod_{u\in\cU} \prod_{t=n_{u-1}+1}^{n_u} p_{X_i|U}(x_{it}|u).
\end{align*}
By \eqref{eq:rateconstraint1_no_time_sharing} and \eqref{eq:rateconstraint2_no_time_sharing} applied to the product channel $W^{\otimes l}$, we may conclude that the rate triple $(R_1,R_2,R_3)$ {\em is achievable for} $W$ if, for all permutations $(i,j,k)$ of $(1,2,3)$,
\begin{align}
lR_i &\leq \min_{q(\vecx_k)} I(\vecX_i;\vecY|\vecX_j),\text{ and}\label{eq:rate_new_eqs1}\\
l(R_i+R_j) &\leq \min_{q(\vecx_k)} I(\vecX_i,\vecX_j;\vecY).\label{eq:rate_new_eqs2}
\end{align}}

\pink{The achievability of the theorem follows from the following observation (for concreteness we take $(i,j,k)=(1,2,3)$ below):
\begin{align}
\min_{q(\vecx_3)} I(\vecX_1;\vecY|\vecX_2)
 &= \min_{q(\vecx_3)} \sum_{t=1}^l I(X_{1t};\vecY|\vecX_2,X_1^{t-1}) \notag\\
 &\stackrel{(a)}{=} \min_{q(\vecx_3)} \sum_{t=1}^l I(X_{1t};\vecY,X_1^{t-1}|\vecX_2) \notag\\
 &\geq \min_{q(\vecx_3)} \sum_{t=1}^l I(X_{1t};Y_t|\vecX_2) \notag\\
 &\geq \sum_{t=1}^l \min_{q(\vecx_3)} I(X_{1t};Y_t|\vecX_2) \notag\\
 &\stackrel{(b)}{=} \sum_{t=1}^l \min_{q(x_{3t})} I(X_{1t};Y_t|X_{2t}) \notag\\
 &= \sum_{u\in \cU} \sum_{t=n_{u-1}+1}^{n_u}  \min_{q(x_{3t})} I(X_{1t};Y_t|X_{2t}),\label{eq:time-sharing-midstep}	
\end{align}
where (a) follows from the independence of $X_{11},X_{12},\ldots,X_{1l},\vecX_2$, (b) follows from the memorylessness of the product channel across its components and the independence of $X_{21},X_{22},\ldots,X_{2l}$. Notice that in \eqref{eq:time-sharing-midstep}, the $n_u-n_{u-1}=lp_U(u)$ terms in the inner sum corresponding to each $u\in\cU$ are identical. For $u\in \cU$, let $(X_{1,u}, X_{2,u}, X_{3,u},Y_u)\sim p_{X_1|U}(\cdot|u)p_{X_2|U}(\cdot|u)q_{X_3|U}(\cdot|u)W(\cdot|\cdot,\cdot,\cdot)$. Then, rewriting~\eqref{eq:time-sharing-midstep},
\begin{align*}
\min_{q(\vecx_3)} I(\vecX_1;\vecY|\vecX_2)
  &\geq \sum_{u\in \cU} \left(lp_{U}(u)\right) \min_{q_{X_{3}|U}(.|u)} I(X_{1,u};Y_u|X_{2,u})\\
  &= l \min_{q_{X_3|U}} \sum_{u\in \cU}p_{U}(u)I(X_{1,u};Y_u|X_{2,u})\\
  &= l \min_{q_{X_3|U}} I(X_{1};Y|X_{2}U).
\end{align*} 
Similarly, 
\begin{align*}
\min_{q(\vecx_3)} I(\vecX_1\vecX_2;\vecY)
\geq l  \min_{q_{X_3|U}} I(X_{1}X_{2};Y|U).
\end{align*}}
\pink{Thus, any rate triple satisfying the conditions in \eqref{eq:rateconstraint1}-\eqref{eq:rateconstraint2} also satisfies \eqref{eq:rate_new_eqs1}-\eqref{eq:rate_new_eqs2} and hence is achievable. }
\end{proof}

\subsection{Randomized coding capacity region}\label{app:A}

\newcommand{\Peach}[1]{\textcolor{Peach}{#1}}

\usetikzlibrary{arrows, shapes,positioning,
                chains,
                decorations.markings,
                shadows, shapes.arrows}
  \definecolor{lightgreen}{rgb}{0.4,0.4,0.1}
\definecolor{lightblue}{rgb}{0.7372549019607844,0.8313725490196079,0.9019607843137255}
\definecolor{darkblue}{rgb}{0.08235294117647059,0.396078431372549,0.7529411764705882}
\definecolor{orangered}{rgb}{0.6,0.3,0.1}
\definecolor{gray}{rgb}{0.7529411764705882,0.7529411764705882,0.7529411764705882}

\begin{proof}[Proof (Achievability of Theorem~\ref{thm:random})] For each $k=1,2,3$, let $W^{(k)}$ be the $2$-user AV-MAC formed by channel inputs from node $k$ as the state and the remaining channel inputs as legitimate inputs.  Let $(R_1,R_2,R_3)$ be a rate triple such that, for some \blue{$p(u)p(x_1|u)p(x_2|u)p(x_3|u)$}, the following conditions hold for all permutations $(i,j,k)$ of $(1,2,3)$: \begin{align}
R_i &< \min_{q(x_k\blue{|u})} I(X_i;Y|\blue{U}X_j),\quad\text{and}\label{eq:rateconstraint1c}\\
R_i+R_j &< \min_{q(x_k\blue{|u})} I(X_i,X_j;Y\blue{|U}),\label{eq:rateconstraint2c}
\end{align} 
with the mutual information terms evaluated using the joint distribution
\blue{$p(u)p(x_i|u)p(x_j|u)q(x_k|u)W(y|x_1,x_2,x_3)$}. Note that, by the first part of the direct result of ~\cite[Theorem~1]{Jahn81} (see \cite[Section~III-C]{Jahn81}), the rate pair $(R_i,R_j)$ is achievable for the AV-MAC $W^{(k)}$ \pink{(see the footnote on page~\pageref{footnote:gubner})}. Let $\epsilon>0$. For each $i\in\{1,2,3\}$, let $\tilde{\cM}_i=[1:2^{nR_i}]$ and $\cM_i=[1:2^{nR_i}/v]$ for the largest integer $v\leq 3/\epsilon$. In the following, we show the existence of a randomized $(2^{nR_1}/v,2^{nR_2}/v,2^{nR_3}/v,n)$ code  $(F_1,F_2,F_3,\phi)$ with $\pink{P^{\text{rand}}_{e}}$ no larger than $\epsilon$, for sufficiently large $n$.

\begin{figure}[h]\longonly{\centering} \scalebox{0.85}{\begin{tikzpicture}[line cap=round,line join=round,>=triangle 45,x=1.1cm,y=1.1cm]
\tikzstyle{myarrows}=[line width=0.5mm,draw = black!20!green!40!red,-triangle 45,postaction={draw, line width=2mm, shorten >=4mm, -}]

\fill[line width=0.4pt,color=lightgreen,fill=lightgreen,fill opacity=0.07] (-9.2,2.7) -- (-9.2,-.7) -- (-7.8,-.7) -- (-7.8,2.7) -- cycle;
\draw [rotate around={90:(-11.5,1)},line width=0.4pt,dotted,color=lightblue,fill=lightblue,fill opacity=0.46] (-11.5,1) ellipse (0.73cm and 0.367cm);
\draw [rotate around={90:(-10.25,1.8)},line width=0.4pt,dotted,color=lightblue,fill=lightblue,fill opacity=0.46] (-10.25,1.8) ellipse (0.40cm and 0.300cm);
\draw [rotate around={90:(-10.25,0.75)},line width=0.4pt,dotted,color=lightblue,fill=lightblue,fill opacity=0.46] (-10.25,0.75) ellipse (0.45cm and 0.320cm);
\draw [rotate around={90:(-10.25,-0.2)},line width=0.4pt,dotted,color=lightblue,fill=lightblue,fill opacity=0.46] (-10.25,-0.2) ellipse (0.40cm and 0.300cm);
\draw [rotate around={90:(-10.25,1)},line width=0.4pt,color=gray,fill=gray,fill opacity=0.2] (-10.25,1) ellipse (1.83cm and 0.55cm);
\draw [line width=0.4pt,color=lightgreen] (-9.2,2.7)-- (-9.2,-.7);
\draw [line width=0.4pt,color=lightgreen] (-9.2,-.7)-- (-7.8,-.7);
\draw [line width=0.4pt,color=lightgreen] (-7.8,-.7)-- (-7.8,2.7);
\draw [line width=0.4pt,color=lightgreen] (-7.8,2.7)-- (-9.2,2.7);

\draw [color = blue!60!red,->,line width=0.1pt] (-11.5,0.5) -- (-10.25,1.7);
\draw [color = blue!60!red,->,line width=0.1pt] (-11.5,1.5) -- (-10.25,1);
\draw [color = blue!60!red,->,line width=0.1pt] (-11.5,1.2) -- (-10.25,0);

\draw [color = black!20!green!40!red,->,line width=0.1pt] (-10.25,2) -- (-8.3,1.1);
\draw [color = black!20!green!40!red,->,line width=0.1pt] (-10.25,1) -- (-8.6,2);
\draw [color = black!20!green!40!red,->,line width=0.1pt] (-10.25,1.7) -- (-8.8,0.7);
\draw [color = black!20!green!40!red,->,line width=0.1pt] (-10.25,0.5) -- (-8.3,1.6);
\draw [color = black!20!green!40!red,->,line width=0.1pt] (-10.25,0.0) -- (-8.6,0.3);
\draw [color = black!20!green!40!red,->,line width=0.1pt] (-10.25,-0.41) -- (-8.3,-0.1);

\draw [rotate around={-90:(-8.5,1)},line width=0.4pt,color=gray,fill=gray,fill opacity=0.26] (-8.5,1) ellipse (1.83cm and 0.55cm);

\draw [fill=darkblue] (-11.5,1.2) circle (2.5pt);
\draw [fill=darkblue,opacity=0.8] (-11.5,1.5) circle (2.5pt);
\draw [fill=darkblue,opacity=0.8] (-11.5,0.5) circle (2.5pt);
\draw[color=darkblue] (-11.5,0) node {\small $\cM_i$};
\draw[color=black] (-11.8,1.15) node {{\small  $m_i$}};
\draw [fill=white] (-10.25,0.5) circle (2.5pt);
\draw[fill=darkblue,opacity=0.8] (-10.25,0.0) circle (2.5pt);
\draw [fill=white] (-10.25,-0.4) circle (2.5pt);
\draw [fill=darkblue,opacity=0.8] (-10.25,1.7) circle (2.5pt);%
\draw [fill=white] (-10.25,2) circle (2.5pt);%
\draw [fill=darkblue,opacity=0.8] (-10.25,1) circle (2.5pt);%
\draw (-10.25,-1) node {\small $\tilde{\cM}_i=\left[2^{nR_i}\right]$};
\draw[color = lightgreen] (-8.5,-1) node {\small $\cX_i^n$};
\draw[color = blue!60!red] (-11,1.75) node {\small $L_i$};
\draw[color = blue!60!red] (-9.8,2.2) node {{\small  ${\color{black}\tilde{m}_i=}{\color{blue!60!red}L_i}{\color{black}(m_i)}$}};
\draw[color = black!20!green!40!red] (-9.6,2.7) node {\small $G_i$};
\draw[color = black] (-7.8,0.95) node {{\small  $\vecx_i={\color{black!20!green!40!red}G_i}({\color{blue!60!red}L_i}{\color{black}(m_i)})$}};

\draw [fill=darkblue,opacity=0.8] (-8.6,2) circle (2.5pt);%
\draw [fill=white] (-8.3,1.1) circle (2.5pt);%
\draw [fill=darkblue,opacity=0.8] (-8.8,0.7) circle (2.5pt);%
\draw [fill=white] (-8.3,1.6) circle (2.5pt);
\draw [fill=darkblue,opacity=0.8] (-8.6,0.3) circle (2.5pt);
\draw [fill=white] (-8.3,-0.1) circle (2.5pt);

\draw [line width=0.4pt,color=orangered] (-7.2,2.7)-- (-7.2,-.7);
\draw [line width=0.4pt,color=orangered] (-7.2,-.7)-- (-5.8,-.7);
\draw [line width=0.4pt,color=orangered] (-5.8,-.7)-- (-5.8,2.7);
\draw [line width=0.4pt,color=orangered] (-5.8,2.7)-- (-7.2,2.7);
\fill[line width=0.4pt,color=orangered,fill=orangered,fill opacity=0.07] (-7.2,2.7) -- (-7.2,-.7) -- (-5.8,-.7) -- (-5.8,2.7) -- cycle;
\draw[color = orangered] (-6.5,-1) node {\small $\cY^n$};
\draw [fill=black] (-6.25,1.4) circle (2.5pt);
\draw[color = black] (-6.25,1.65) node {\small $\vecy$};

\draw [rotate around={90:(-4.75,1)},line width=0.4pt,color=gray,fill=gray,fill opacity=0.2] (-4.75,1) ellipse (1.83cm and 0.55cm);
\draw [fill=white] (-4.75,0.5) circle (2.5pt);
\draw (-4.75,-1) node {\small $\tilde{\cM}_i$};
\draw[color = black!20!green!40!red] (-5.45,1.5) node {\small $\Gamma^{(k)}_i$};
\draw[color = blue!60!red] (-4,2) node {\small $\Lambda_i$};

\draw[color = black] (-5.3,3.2) node {\small valid inner messages for $L_i$};

\draw [myarrows](-5.9,1)--(-5.1,1);
\draw [rotate around={90:(-4.75,1.8)},line width=0.4pt,dotted,color=lightblue,fill=lightblue,fill opacity=0.46] (-4.75,1.8) ellipse (0.40cm and 0.300cm);
\draw [rotate around={90:(-4.75,0.75)},line width=0.4pt,dotted,color=lightblue,fill=lightblue,fill opacity=0.46] (-4.75,0.75) ellipse (0.45cm and 0.320cm);
\draw [rotate around={90:(-4.75,-0.2)},line width=0.4pt,dotted,color=lightblue,fill=lightblue,fill opacity=0.46] (-4.75,-0.2) ellipse (0.40cm and 0.300cm);
\draw [color = blue!60!red,->,line width=0.2pt] (-4.75,1.7) -- (-3.5,1);
\draw [color = blue!60!red,->,line width=0.2pt] (-4.75,1) -- (-3.5,2);
\draw [color = blue!60!red,->,line width=0.2pt] (-4.75,2) -- (-3.55,0);

\draw [color = blue!60!red,->,line width=0.2pt] (-4.75,0.5) -- (-3.55,0);
\draw [color = blue!60!red,->,line width=0.2pt] (-4.75,0.0) -- (-3.5,1.7);;
\draw [color = blue!60!red,->,line width=0.2pt] (-4.75,-0.4) -- (-3.55,0);
\draw[fill=darkblue,opacity=0.8] (-4.75,0.0) circle (2.5pt);
\draw [fill=white] (-4.75,-0.4) circle (2.5pt);

\draw [->,color=darkblue] (-5.4,3) to [out=250,in=150] (-4.9,1.7);

\draw [fill=white] (-4.75,0.5) circle (2.5pt);
\draw [fill=darkblue] (-4.75,1.7) circle (2.5pt);%
\draw [fill=white] (-4.75,2) circle (2.5pt);%
\draw [fill=darkblue] (-4.75,1) circle (2.5pt);%

\draw [rotate around={90:(-3.5,1.5)},line width=0.4pt,dotted,color=lightblue,fill=lightblue,fill opacity=0.46] (-3.5,1.5) ellipse (0.73cm and 0.367cm);
\draw [fill=darkblue] (-3.5,1.7) circle (2.5pt);
\draw [fill=darkblue] (-3.5,2) circle (2.5pt);
\draw [fill=darkblue] (-3.5,1) circle (2.5pt);
\draw [color=darkblue] (-3.3,.6) node {\small $\cM_i$};
\draw [color=red!50!blue,fill=red!70!blue!20!white] (-3.5,0) circle (2.5pt);
\draw [color=red!50!blue] (-3.3,-0.1) node {\small $\bot$};

\draw (-9.75,-1.7) node {\small {(a)} Encoder $F_i:L_i\circ G_i$};
\draw (-5.25,-1.7) node {\small {(b)} Pre-decoder $\phi^{(k)}_i:\Gamma_i^{(k)}\circ\Lambda_i$};

\end{tikzpicture}}
\caption{The encoders and pre-decoders for Theorem~\ref{thm:random}.}\label{fig:random_proof}

\end{figure}
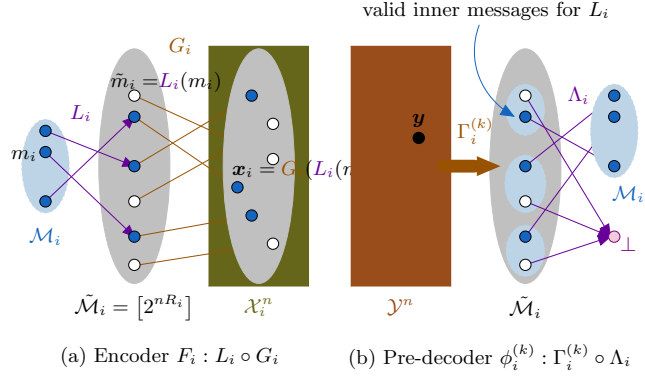

\paragraph*{Code design.} We will first describe some randomized maps which will be used in the code design (see Figure~\ref{fig:random_proof}).  
For each $i\in \{1, 2, 3\}$, let $G_i:\tilde{\cM}_i\to\cX^n_i$ be a randomized map such that it maps $m_i\in \tilde{\cM}_i$ to an $n$-length i.i.d. sequence $G_i(m_i)$ generated according to the distribution $p_i$. 
The sequences $G_i(m)$ are independent across $i\in \{1, 2,3\}$ and $m\in \cM_i$. The realization of $G_i(m_i)$ for all $i\in \{1,2, 3\}$ and $m_i\in \cM_i$ is shared with the decoder. 
For any permutation $(i, j, k)$ of $(1, 2, 3)$, consider the AV-MAC $W^{(k)}$ which corresponds to user-$k$ as the adversary. 
If we consider $\tilde{\cM_i}$ and $\tilde{\cM_j}$ as the message sets and $G_i$ and $G_j$ as the corresponding encoders, then this construction ensures that the randomness of the encoders $G_i$ and $G_j$ is private from each other and also private from the adversarial user-$k$. 
This joint distribution of  $G_i$ and $G_j$ (and the corresponding codewords) is the same as that of the encoders of AV-MAC $W^{(k)}$ in the direct part of \cite[Theorem~1, Section~III-C]{Jahn81}. 
For $G_i$ and $G_j$ as encoders, let $\Gamma^{(k)}$ denote the decoder corresponding to the decoding sets defined in proof of the direct part of \cite[Theorem~1, Section~III-C]{Jahn81} for the AV-MAC $W^{(k)}$. 
Suppose $(\Gamma^{(k)}_i, \Gamma^{(k)}_j):=\Gamma^{(k)}$ where $\Gamma^{(k)}_i:\cY^n\to\tilde{\cM}_i$. 
For all $\epsilon>0$, by \cite[Theorem~1]{Jahn81}, there exists a large enough $n$ such that for all permutations $(i, j, k)$ of $(1, 2, 3)$, the code $(G_i, G_j, \Gamma^{(k)})$ has error probability no larger than $\epsilon/3$.
We consider that $n$.

For each $i\in\{1, 2, 3\}$, the message set $\cM_i$ is randomly embedded into the set $\tilde{\cM_i}$ as follows: We choose an arbitrary partition of $\tilde{\cM_i}$ into $|\cM_i|$ many disjoint equal-sized subsets (each subset size is $v$). 
Let us denote the partition by $\cS_{m_i}, \, m_i\in \cM_i$ where $\cup_{m_i\in \cM_i}\cS_{m_i}=     \tilde{\cM_i}$ and $\cS_{m_i}\cap\cS_{m'_i}=\emptyset$ for all $m_i, m'_i\in \cM_i$ where $m_i\neq m_i'$.
The size of each $\cS_{m_i}, \, m_i\in \cM_i$ is $v$ ($\leq 3/\epsilon$). 
The maps $L_i:\cM_i\to\tilde{\cM}_i$ and $\Lambda_i:\tilde{\cM}_i\to\cM_i$ are the forward and reverse maps for an injection from $\cM_i$ to $\tilde{\cM}_i$ where, independently for each $m_i\in \cM_i$, $L_i(m_i)$ is chosen uniformly at random from $\cS_{m_i}$. 
Both the encoder maps $G_i$ and $L_i$ are independent for $i=1,2,3$ and are made available to the decoder as the shared secret between user-$i$ and the decoder, unknown to other users. 
 
For each $i\in\set{1,2,3}$, the encoder map $F_i:\cM_i\to\cX^n_i$ is defined as $F_i(m_i)= G_i(L_i(m_i))$ for every $m_i\in\cM_i$. For $i\in\set{1,2,3}$ and $k\in\set{1,2,3}\setminus\set{i}$, we define pre-decoder\footnote{In this notation $\phi_i^{(k)}(\vecy)$, we are suppressing the dependence of the pre-decoder (and later the decoder) on the randomness of the encoders.} 
$$\phi_i^{(k)}(\vecy):=\begin{cases}\Lambda_i(\Gamma^{(k)}_i(\vecy)) &\mbox{if }\Gamma^{(k)}_i(\vecy)\in L_i(\cM_i),\\ \bot &\mbox{otherwise}.\end{cases}$$ 
The decoder $\phi:\cY^n\to\cM_1\times\cM_2\times\cM_3$ outputs $\phi(\vecy)=(\hat{m}_1,\hat{m}_2,\hat{m}_3)$, where, for each $i\in\set{1,2,3}$ and $(j,k)$ a permutation of $\set{1,2,3}\setminus\set{i}$, 
$$\hat{m}_i=\begin{cases}\phi_i^{(j)}(\vecy)& \mbox{ if }  \phi_i^{(j)}(\vecy)=\phi_i^{(k)}(\vecy)\neq\bot \\
\phi_i^{(j)}(\vecy)& \mbox{ if } \phi_i^{(j)}(\vecy)\neq\bot\mbox{ and } \phi_i^{(k)}(\vecy)=\bot\\
\phi_i^{(k)}(\vecy)& \mbox{ if } \phi_i^{(k)}(\vecy)\neq\bot\mbox{ and } \phi_i^{(j)}(\vecy)=\bot\\
1&\mbox{otherwise.}
\end{cases}$$
\paragraph*{Error Analysis.}
We first show that as long as the rate triple $(R_1,R_2,R_3)$ satisfy the constraints~\eqref{eq:rateconstraint1c} and~\eqref{eq:rateconstraint2c}, {\em i.e.}, each pair of rates lie in the corresponding AV-MAC randomized coding capacity region, the following hold simultaneously for every honest user $i$ which sends message $m_i\in \cM_i$ and potentially adversarial user $k\neq i$: \emph{(i)} $\phi_i^{(k)}(\vecY)$ is $m_i$ w.h.p. (with probability at least $1-\epsilon/3$) if user-$k$ is indeed adversarial and \emph{(ii)} $\phi_i^{(k)}(\vecY)$ is, w.h.p. (with probability at least $1-\epsilon/3$), either $\bot$ or $m_i$ if user-$k$ is not adversarial. To this end, consider any permutation $(i,j,k)$ of $(1,2,3)$ and assume that the  adversarial user (if any) is user-$k$ which sends $\vecX_k$ as its  potentially adversarial  input to the channel. Suppose, for $(m_i, m_j)\in \cM_i\times\cM_j$, user-$i$ and user-$j$ send $F_i(m_i)$ and $F_j(m_j)$ respectively. Let $\vecY$ denote the channel output.
\begin{description}
\item[$(i)$] First, consider the AV-MAC $W^{(k)}$.  Recall that $\Gamma_i^{(k)}(\vecY)=L_i(m_i)$ with probability at least $1-\epsilon/3$. Thus, with probability at least $1-\epsilon/3$, $\phi_i^{(k)}(\vecY)$ equals $m_i$. 
\item[$(ii)$] Next, consider the AV-MAC $W^{(j)}$.  In this case, $\Gamma_i^{(j)}(\vecY)$ may not equal $L_i(m_i)$  as $\vecX_k$ may not be a valid codeword. 
We would like to compute $\bbP\inp{\phi_i^{(j)}(\vecY)\notin \set{m_i,\bot}}$ where the probability is over $G_i(L_i(m_i))$, $G_j(L_j(m_j))$, $\vecX_k$ and the channel. 
Note that $G_i$ and $L_i$ are independent of (potentially jointly distributed and adversarially chosen) $G_k$, $L_k$ and $\vecX_k$. Thus,  
\begin{align*}
&\bbP\inp{\phi_i^{(j)}(\vecY)\notin \set{m_i,\bot}}\\
& = \bbP\inp{\Gamma_i^{(j)}(\vecY)\in L_i(\cM_i\setminus\{m_i\})}\\
& = \sum_{\tilde{m}_i\in\tilde{\cM}\setminus\cS_{m_i}}\bbP\inp{\Gamma_i^{(j)}(\vecY)=\tilde{m}_i, \tilde{m}_i\in L_i(\cM_i\setminus\{m_i\})}\\
& = \sum_{\tilde{m}_i\in\tilde{\cM}\setminus\cS_{m_i}}\bbP\inp{\Gamma_i^{(j)}(\vecY)=\tilde{m}_i}\bbP\inp{ \tilde{m}_i\in L_i(\cM_i\setminus\{m_i\})\big|\Gamma_i^{(j)}(\vecY)=\tilde{m}_i}\\
& \stackrel{(a)}{=} \sum_{\tilde{m}_i\in\tilde{\cM}\setminus\cS_{m_i}}\bbP\inp{\Gamma_i^{(j)}(\vecY)=\tilde{m}_i}\bbP\inp{ \tilde{m}_i\in L_i(\cM_i\setminus\{m_i\})}\\
& \stackrel{(b)}{=}  \sum_{\tilde{m}_i\in\tilde{\cM}\setminus\cS_{m_i}}\bbP\inp{\Gamma_i^{(j)}(\vecY)=\tilde{m}_i}\cdot\frac{1}{v}\\
& \leq 1/v \leq \epsilon/3.
\end{align*} Here, $(a)$ holds as $\Gamma_i^{(j)}(\vecY)\,\indep \,L_i(\cM_i\setminus\{m_i\})$. This is because $L_i(m_i)\,\indep \,L_i(\cM_i\setminus\{m_i\})$ and  $\Gamma_i^{(j)}\,\indep \,L_i$ as $\Gamma_i^{(j)}$ is a function of AV-MAC encoders $G_i$ and $G_k$ which are independent of $L_i$. 
The equality $(b)$ holds because for $\tilde{m}_i\in\tilde{\cM}\setminus\cS_{m_i}$, 
\begin{align*}
&\bbP\inp{ \tilde{m}_i\in L_i(\cM_i\setminus\{m_i\})}\\
&= \sum_{m_i'\in \cM_i\setminus\{m_i\}}\bbP\inp{ L_i(m_i')=\tilde{m}_i }\\
 &= \sum_{m_i'\in \cM_i\setminus\{m_i\}}1_{\{\tilde{m}_i\in \cS_{m_i'}\}}\cdot\frac{1}{v}\\
 &=1/v.
 \end{align*}

Thus, with probability $1-\epsilon$, for each non-adversarial user $i$, at least one of the decoders $\phi_i^{(j)}$ or $\phi_i^{(k)}$ outputs the true message while the other decoder outputs either the true message or $\bot$.
\end{description}
\end{proof}

\begin{proof}[{Proof (Converse of Theorem~\ref{thm:random})}]
We show the converse for the weak adversary. \pink{Since, $\mathcal{R}_{\random} \subseteq \mathcal{R}_{\random}^{\weak}$, a converse bound on $\mathcal{R}_{\random}^{\weak}$ is also a converse bound on $\mathcal{R}_{\random}$.} Suppose $(F_1,F_2,F_3,\phi)$ is a $(2^{nR_1},2^{nR_2},2^{nR_3},n)$ randomized code such that $\pink{P^{\text{weak}}_e}\leq \epsilon$ for some $\epsilon>0$. Recall that $F_1,F_2,F_3$ are independent. 
Let $M_i\sim\textup{Unif}(\mathcal{M}_i)$, $i=1,2,3$ be independent. Let $\hat{M}_i= \phi_i(\vecY,F_1,F_2,F_3), \, i = 1,2,3$. Then, $\epsilon$ is an upper bound on~\eqref{eq:pe1_weak} which is given by
\begin{align*}
 P_{e,1}^{\text{weak}}&= \max_{\vecx_1} \;\Prob_{F_2,F_3}\Big( (\Msgh_2,\Msgh_3)\neq(\Msg_2,\Msg_3)\Big|\vecX_1=\vecx_1,\vecX_2=F_2(\Msg_2), \vecX_3=F_3(\Msg_3)\Big)\\
 &= \max_{p_{\vecX_1}} \Prob_{F_2,F_3}\Big( (\Msgh_2,\Msgh_3)\neq(\Msg_2,\Msg_3)\Big| \vecX_2=F_2(\Msg_2), \vecX_3=F_3(\Msg_3)\Big).
\end{align*}
For a vector $\vecx_j\in \cX^n_j$, $j=1, 2, 3$,  we use $x_{j, i}$ to denote its $i^{\text{th}}$ index. That is $\vecx_j = \inp{x_{j, 1}, x_{j, 2}, \ldots, x_{j, n}}$. Similarly, a random vector $\vecX_j$ distributed on $\cX^n_j$ can be written as $\vecX_j = \inp{X_{j, 1}, X_{j, 2}, \ldots, X_{j, n}}$. For $i\in [1:n]$, let $q_{X_{1, i}}$ be some distribution on $\cX_1$. We consider the following $p_{\vecX_1}$.
\[ p_{\vecX_1}(\vecx_1) = \prod_{i=1}^n q_{X_{1,i}} (x_{1,i}).\]
By Fano's inequality, under this $p_{\vecX_1}$ and when $\vecX_i=F_i(M_i)$, $i=2,3$, 
\begin{align*}
H(\Msg_2,\Msg_3|\vecY, F_2, F_3) \leq 1+n\epsilon(R_2+R_3).
\end{align*}
Ignoring small terms, we have
\begin{align*}
n(R_2+R_3) &\leq H(M_2,M_3)\\
 &\leq H(M_2,M_3|F_2,F_3)\\
 &\stackrel{\text{(a)}}{\approx} I(M_2,M_3;\vecY|F_2,F_3)\\
 &= \sum_{i=1}^n I(M_2,M_3;Y_i|Y^{i-1},F_2,F_3)\\
 &\leq \sum_{i=1}^n I(M_2,M_3,F_2,F_3,Y^{i-1};Y_i)\\
 &= \sum_{i=1}^n I(M_2,M_3,F_2,F_3,Y^{i-1},X_{2,i},X_{3,i};Y_i)\\
 &\stackrel{\text{(b)}}{=} \sum_{i=1}^n I(X_{2,i},X_{3,i};Y_i),
\end{align*}
where (a) follows from Fano's inequality (ignoring an $O(n\epsilon)$ term), (b) follows from the memorylessness of the channel and the independence of $X_{1,i}$ over $i=1,\ldots,n$ for the particular $p_{\vecX_1}$ under consideration.

Let $U\sim\textup{Unif}\{1,2,\ldots,n\}$ independent of $(M_1,M_2,M_3,F_1,F_2,F_3,\vecY)$. We have (where we ignore an additive $O(\epsilon)$ term) 
\begin{align*}
R_2+R_3\leq I(X_{2,U},X_{3,U};Y_U|U).
\end{align*}
Since, the above bound holds for all $p_{\vecX_1}(\vecx_1)=\prod_{i=1}^n q_{X_{i,i}}(x_{1,i})$, and noticing that conditioned on $X_{1,U},X_{2,U},X_{3,U}$ the channel law $\mach$ gives the conditional probability of $Y_U$, we may write
\begin{align}
R_2+R_3 \leq  \min_{q(x_1|u)} I(X_2,X_3;Y|U) \label{eq:converseR2+R3}
\end{align}
for some $q(x_1|u)$. We note that the distribution of $U,X_1,X_2,X_3,Y$ is $p(u)q(x_1|u)p(x_2|u)p(x_3|u)\mach(y|x_1,x_2,x_3)$ where $p(x_2|u)$ is determined by the distribution of $F_2$ and $p(x_3|u)$ is determined by the distribution of $F_3$.
 
Proceeding similarly, for $p_{\vecX_1}(\vecx_1)=\prod_{i=1}^n q_{X_{1,i}}(x_{1,i})$,
\begin{align*}
nR_2 &\leq H(M_2)\\
 &\leq H(M_2|M_3,F_2,F_3)\\
 &\approx I(M_2;\vecY|M_3,F_2,F_3)\\
 &= \sum_{i=1}^n I(M_2;Y_i|Y^{i-1},M_3,F_2,F_3)\\
 &= \sum_{i=1}^n I(M_2,X_{2,i};Y_i|X_{3,i},Y^{i-1},M_3,F_2,F_3)\\
 &\leq \sum_{i=1}^n I(X_{2,i}, Y^{i-1},M_2,M_3,F_2,F_3; Y_i|X_{3,i})\\
 &= \sum_{i=1}^n I(X_{2,i};Y_i|X_{3,i}).
\end{align*}
Hence, we have
\begin{align}
R_2 \leq \min_{q_{(x_1|u)}} I(X_2;Y|X_3,U), \label{eq:converseR2}
\end{align}
where the joint distribution  of the random variables is $p(u)q(x_1|u)p(x_2|u)p(x_3|u)\mach(y|x_1,x_2,x_3)$. We note that $p(u)p(x_2|u)p(x_3|u)$ are the same as in \eqref{eq:converseR2+R3}. Similarly,
\begin{align}
R_3\leq \min_{q(x_1|u)} I(X_3;Y|X_2,U). \label{eq:converseR3}
\end{align}

Similarly, considering $P_{e,2}^{\text{weak}}$ with $p_{\vecX_2}(x_2^n)=\prod_{i=1}^n q_{X_{2,i}}(x_{2,i})$ (and $\vecX_i=F_i(W_i)$, $i=1,3$), we get
\begin{align}
R_3 &\leq \min_{q(x_2|u)} I(X_3;Y|X_1,U), \label{eq:converseR3B}\\
R_1 &\leq \min_{q(x_2|u)} I(X_1;Y|X_3,U), \label{eq:converseR1B}\\
R_3+R_1 &\leq \min_{q(x_2|u)} I(X_3,X_1;Y|U), \label{eq:converseR3+R1}
\end{align}
where the joint distribution of the random variables is $p(u)p(x_1|u)q(x_2|u)p(x_3|u)\mach(y|x_1,x_2,x_3)$ for some $q_{x_2|u}$. We note that $p(u)$ and $p(x_3|u)$ here are the same as in \eqref{eq:converseR2+R3}-\eqref{eq:converseR3}. Considering $P_{e,3}^{\text{weak}}$ with $p_{\vecX_3}(x_3^n)=\prod_{i=1}^n q_{X_{3,i}}(x_{3,i})$ (and $\vecX_i=F_i(W_i)$, $i=1,2$), we similarly arrive at
\begin{align}
R_1 &\leq \min_{q(x_3|u)} \pink{I(X_1;Y|X_2,U)}, \label{eq:converseR1C}\\
R_2 &\leq \min_{q(x_3|u)} \pink{I(X_2;Y|X_1,U)}, \label{eq:converseR2C}\\
R_1+R_2 &\leq \min_{q(x_3|u)} I(X_1,X_2;Y|U), \label{eq:converseR1+R2}
\end{align}
where the joint distribution of the random variables is $p(u)p(x_1|u)p(x_2|u)q(x_3|u)\mach(y|x_1,x_2,x_3)$. The $p(u)$, $p(x_1|u)$, and $p(x_2|u)$ are the same as in \eqref{eq:converseR2+R3}-\eqref{eq:converseR1+R2}. This completes the proof of converse.
\end{proof}

\section{The $k$-user \bmac}\label{sec:k_user}
In this section, we generalize our model to a $k$-user \bmac for any positive integer $k$. We allow for a set of users to  be controlled by an adversary simultaneously. 

We study the problem under both randomized and deterministic codes. The techniques for the $3$-user \bmac are extended to show the characterization of the randomized capacity region. For the deterministic part, we take the first approach (Section~\ref{intro:first_approach}) as mentioned in the introduction. 
We first show that for any randomized code with vanishing probability of error, there exists another randomized code, also with a vanishing probability of error, which requires only $n^2$-valued randomness at each encoder for a code of blocklength $n$. This argument is along the lines of the extension of the elimination technique~\cite{Ahlswede78} provided in \cite{Jahn81}. Next, we generalize the symmetrizability conditions to show that the deterministic coding capacity region has non-empty interior if and only if the \bmac is not symmetrizable. This allows us to use a small rate positive code to share the random bits with the decoder  whenever the channel is not symmetrizable and then use the randomized scheme to achieve the entire randomized capacity region under deterministic codes (also see Remark~\ref{remark:2}). 

We give the system model in Section~\ref{sec:k_user_model} and discuss the randomized and deterministic coding capacity regions in Sections~\ref{sec:k_user_random} and \ref{sec:k_user_det} respectively. We only give proof sketches in these sections and defer the complete proofs to the appendices.

\subsection{System model}\label{sec:k_user_model}

A  memoryless $k$-user \bmac $(W, \cA)$ consists of a $k$-user memoryless MAC $W$ with input alphabets $\mathcal{X}_1,\mathcal{X}_2,\ldots, \mathcal{X}_k$, and output alphabet ${\mathcal Y}$ along with an adversary who can control a set of users simultaneously. The set of users the adversary controls may be any one of the sets in  $\cA\subseteq 2^{[1:k]}$, where $2^{[1:k]}$ denotes the power set of $[1:k]$. The other users and the decoder are unaware of the identity of the set $\cQ$ of users, $\cQ\in \cA$, that the adversary controls. In the sequel, we refer to the users in this set $\cQ\in \cA$ which the adversary controls as the malicious users and the other users as honest.   
If $\emptyset\in \cA$, then it corresponds to the case when all users are honest. 
For the $3$-user \bmac (Section~\ref{section:MAC2_model_main}) which considers the case when at most one user is malicious, the adversary structure is given by $\cA = \inb{\emptyset, \{1\},\{2\}, \{3\}}$. \pink{Along the lines of Definition~\ref{defn:rand_code} for the three user \bmac, we define randomized codes for $k$-user \bmac $(W, \cA)$ below.}
\begin{defn}[Randomized code]\label{defn:random_code_k}
An $(\nummsg_1,\nummsg_2,\ldots, \nummsg_k,n)$ randomized {\em code} for the \bmac $(W,\cA)$ consists of the following: 
\begin{enumerate}[label=(\roman*)]
\item $k$ message sets, $\mathcal{M}_i = \{1,\ldots,\nummsg_i\}$, $i=1,2,\ldots,k$,
\item $k$ independent randomized encoders, $F_{i}:\mathcal{M}_i\rightarrow \mathcal{X}_i^n$,  where $F_i\sim P_{F_i}$ takes values in {$\mathcal{F}_i\subseteq\{g:\cM_i\rightarrow \cX_i^n\},\, i =1,2,\ldots,k$} and
\item a decoder, $\phi:\mathcal{Y}^n\times\mathcal{F}_1\times\mathcal{F}_2\times\ldots\times\mathcal{F}_k\rightarrow\mathcal{M}_1\times\mathcal{M}_2\times\ldots\times\mathcal{M}_k$ where \\$\phi(\vecy,F_1,F_2,\ldots,F_k) = (\phi_1(\vecy,F_1,F_2,\ldots,F_k),\phi_2(\vecy,F_1,F_2,\ldots,F_k),\ldots,\phi_k(\vecy,F_1,F_2,\ldots,F_k))$ for some deterministic functions $\phi_i:\mathcal{Y}^n\times\mathcal{F}_1\times\mathcal{F}_2\times\ldots\times\mathcal{F}_k\rightarrow\mathcal{M}_1\times\mathcal{M}_2\times\ldots\times\mathcal{M}_k, \, i = 1,2,\ldots,k$.
\end{enumerate}
\end{defn}

\pink{Next, we define the probability of error, achievable rate region and the capacity region.
As mentioned in Section~\ref{section:MAC2_model}, the decoder is a function which maps the channel output as well as the random encoding maps to decoded messages. \blue{Hence, the adversary can mount an attack by selecting the random encoding maps of the users it controls.} Note that while doing this, the adversary does not have access to the random encoding maps of the other (honest) users. Similar to the $3$-user case, the adversary selects the encoding maps and chooses the inputs of all malicious users jointly. Note that while doing this, the adversary is unaware of the realizations of the other users' encoding maps. If the adversary controls the users in $\cQ\in\cA$, then it may  choose the encoding maps $f_{\cQ}$ ({\em i.e.}, $(f_i)_{i\in \cQ}$) in addition to the 
input vectors $\vecx_{\cQ}$. }

\pink{
Let $P^{\text{rand}}_{e, \cQ}$ denote the average probability of error when the  adversary controls the set $\cQ$ of users.
\begin{align}\label{eq:avg_strong}
P^{\text{rand}}_{e, \cQ} &= \max_{\vecx_{\cQ}, f_{\cQ}\in \cF_{\cQ}}\frac{1}{(\prod_{i\in \cQ^c}N_i)}\sum_{m_{\cQ^c}\in \cM_{\cQ^c}}\bbP\inp{\inb{\phi(\vecY, f_{\cQ}, F_{\cQ^c})_{\cQ^c}\neq m_{\cQ^c}}|\vecX_{\cQ^c} = F_{\cQ^c}(m_{\cQ^c}), \vecX_{\cQ} = \vecx_{\cQ}},
\end{align}
where $\phi(\vecy, f_{\cQ}, f_{\cQ^c})_{\cQ^c}$ denotes $\hat{m}_{\cQ^c}$ for $\hat{m}_{[1:k]}\in \cM_{[1:k]}$ such that $\phi(\vecy, f_{\cQ}, f_{\cQ^c}) = \hat{m}_{[1:k]}$. The probability is over independent $F_i\sim P_{F_i}, \, i\in \cQ^c$ and the channel.}

\pink{
The average probability of error $P^{\text{rand}}_{e}$ is given by  
\begin{align*}
P^{\text{rand}}_{e} =\max_{\cQ\in \cA}P^{\text{rand}}_{e, \cQ}.
\end{align*}
Note that though the users controlled by the adversary do not use $f_{\cQ}$ for encoding, the decoder uses it and hence its choice gives the adversary additional power. 
We also emphasize that the decoder is unaware of the identity of the set $\cQ\in\cA$ of users controlled by the adversary ({\em i.e.}, in \eqref{eq:avg_strong}, the decoding map $\phi$ may not depend of $\cQ$.). }

\pink{
We say a rate tuple $(R_1,R_2,\ldots,R_k)$ is {\em achievable}, if there is a sequence of $(\lfloor2^{nR_1}\rfloor,\lfloor2^{nR_2}\rfloor,\ldots,\lfloor2^{nR_k}\rfloor,n)$ codes  $(F_1^{(n)},F_2^{(n)},\ldots,F_k^{(n)},\phi^{(n)})_{n=1}^\infty$ such that $\lim_{n\rightarrow\infty}P^{\text{rand}}_{e}(P_{F_1^{(n)}},P_{F_2^{(n)}},\ldots,P_{F_k^{(n)}},\phi^{(n)})\rightarrow0$. The {\em randomized coding capacity region} $\mathcal{R_{\text{random}}}$ is the closure of the set of all achievable rate triples. }

\pink{We also study the  {\em weak} adversary model for the converse where the adversary does not have any knowledge of the any of the random encoding maps while choosing the inputs of the malicious users.
Probability of error and capacity region for randomized codes with weak adversary can be defined by replacing  $P^{\text{rand}}_{e,\cQ}$ with $P^{\text{weak}}_{e,\cQ}$ for $\cQ\in\cA$ in the above definition,  where
\begin{align}\label{eq:avg_weak}
P^{\weak}_{e, \cQ} = \max_{\vecx_{\cQ}}\frac{1}{(\prod_{i\in \cQ^c}N_i)}\sum_{m_{\cQ^c}\in \cM_{\cQ^c}}\bbP\inp{\inb{\phi(\vecY, F_{\cQ}, F_{\cQ^c})_{\cQ^c} \neq m_{\cQ^c}}|\vecX_{\cQ^c} = F_{\cQ^c}(m_{\cQ^c}), \vecX_{\cQ} = \vecx_{\cQ}},
\end{align}
\noindent The probability is over independent $F_i\sim P_{F_i}, \, i\in [1:k]$ and the channel.}

\pink{
We denote the {\em
randomized coding capacity region} for the weak adversary by  $\mathcal{R}_{\random}^{\weak}$. As was the case in $3$-user \bmac, 
$\mathcal{R}_{\random} \subseteq \mathcal{R}_{\random}^{\weak}$. }
\pink{We define determinsitic codes for $k$-user \bmac $(W, \cA)$ along the lines of Definition~\ref{code:det_3user}}.
\begin{defn}[Deterministic code]\label{code:det_kuser}
An $(\nummsg_1,\nummsg_2,\ldots, \nummsg_k,n)$ deterministic {\em code} for the \bmac $(W, \cA)$ consists of: \vspace{-0.25em}
\begin{enumerate}[label=(\roman*)]
\item $k$ message sets, $\mathcal{M}_i = \{1,\ldots,\nummsg_i\}$, $i\in \{1,2,\ldots, k\}$,
\item $k$ encoders, $f_{i}:\mathcal{M}_i\rightarrow \mathcal{X}_i^n$, $i\in \{1,2,\ldots, k\}$, and
\item a decoder, $\phi:\mathcal{Y}^n\rightarrow\mathcal{M}_1\times\mathcal{M}_2\times\ldots\times\mathcal{M}_k.$ 
\end{enumerate}
\end{defn}
Let $P_{e, \cQ}$ denote the average probability of error when the adversary controls the set $\cQ\in \cA$  of users.
\begin{align}\label{eq:avg_error_k_users}
P_{e, \cQ} = \max_{\vecx_{\cQ}}\frac{1}{(\prod_{i\in \cQ^c}N_i)}\sum_{m_{\cQ^c}\in \cM_{\cQ^c}}\bbP\inp{\inb{\phi(\vecY)_{\cQ^c} \neq m_{\cQ^c}}|\vecX_{\cQ^c} = f_{\cQ^c}(m_{\cQ^c}), \vecX_{\cQ} = \vecx_{\cQ}}.
\end{align}
The average probability of error $P_{e}$ is given by  
\begin{align*}
P_{e} =\max_{\cQ\in \cA}P_{e, \cQ}.
\end{align*}
Similar to the randomized coding case, the decoder is unaware of which set of users from $\cA$ are controlled by the adversary.

We say a rate tuple $(R_1,R_2,\ldots, R_k)$ is {\em achievable} if there is a sequence of $(\lfloor2^{nR_1}\rfloor,\lfloor2^{nR_2}\rfloor,\ldots, \lfloor2^{nR_k}\rfloor,n)$ codes  $(f_1^{(n)},f_2^{(n)},\ldots, f_k^{(n)},\allowbreak\phi^{(n)})_{n=1}^\infty$ such that $\lim_{n\rightarrow\infty}P_{e}(f_1^{(n)},f_2^{(n)},\ldots, f_k^{(n)},\phi^{(n)})\rightarrow0$. The {\em deterministic coding capacity region} $\mathcal{R}_{\deterministic}$ is the set of all achievable rate tuples. 

Recall that for the three user \bmac (Section~\ref{section:MAC2_model_main}, where the adversary structure is $\cA = \inb{\emptyset,\{1\},\{2\}, \{3\}}$, we could show that $P_{e,0}\leq P_{e,1}+P_{e,2}+P_{e,3}$ (see \eqref{eq:honest_upper_bound}). Generalizing this to the $k$-user \bmac $(W, \cA)$, we can show the following lemma whose proof is in Appendix~\ref{appendix:proofs}.
\begin{lemma}\label{lemma:maximal sets}
For any $\cQ_1, \ldots, \cQ_t\in \cA$, $t\in \bbN$ and $\cQ\subseteq[1:k]$ such that $\cQ= \cap_{i=1}^{t}\cQ_i$, $P_{e, \cQ}\leq \sum_{i=1}^{t}P_{e, \cQ_i}$. 
\end{lemma}
This lemma implies that even if a set $\cQ= \cap_{i=1}^{t}\cQ_i$ as in Lemma~\ref{lemma:maximal sets} is removed from $\cA$, the capacity region of a \bmac remains unchanged.

\subsection{Randomized coding capacity region}\label{sec:k_user_random}
Let $\mathcal{R}$ be the closure of the set of all rate tuples $(R_1,R_2,\ldots,R_k)$ such that for some $p(u)p(x_1|u)p(x_2|u)\ldots p(x_k|u)$, the following conditions hold for all $\cQ\in\cA$ and $\cJ\subseteq \cQ^c$,
\begin{align}\label{rate_region_k_user}
\sum_{j\in \cJ}R_j\leq \min_{q(\vecx_{\cQ}|u)}I\inp{X_{\cJ};Y|X_{(\cQ\cup\cJ)^c}, \pink{U}}
\end{align}
where the mutual information above is evaluated using the joint distribution
$p(u) q(\vecx_{\cQ}|u)\prod_{j\in\cQ^c}p(x_j|u)W(y|\vecx_{\cQ},\vecx_{\cQ^c})$. \pink{Here, an upper bound of $2^k$ on $|\mathcal{U}|$ can be shown using the convex cover method~\cite[Appendix~C]{YHKEG}.}
\pink{\begin{remark}\label{remark:sum_Rate}
As discussed after Lemma~\ref{lemma:maximal sets}, the capacity region of a \bmac remains unchanged even if a set $\cQ= \cap_{i=1}^{t}\cQ_i$  is removed from $\cA$. It is easy to verify that the rate region $\cR$ shares this property. For instance, for the three user case, let $\cA = \inb{\emptyset,\{1\},\{2\}, \{3\}}$. Consider the constraint corresponding to $\cQ = \emptyset$ and $\cJ = \inb{1,2,3}$ in \eqref{rate_region_k_user}
\begin{align*}
R_1+R_2 + R_3 \leq I(X_1X_2X_3;Y|U).
\end{align*}This is implied by the following constraints which correspond to $\cQ = \inb{3}$, $\cJ = \inb{1,2}$ and $\cQ = \inb{1}$, $\cJ = \inb{3}$ respectively
\begin{align*}
R_1+R_2&\leq \min_{q(x_3|u)}I(X_1X_2;Y|U)\leq I(X_1X_2;Y|U)\Big|_{p(x_3|u)}, \text{ and}\\
R_3&\leq \min_{q(x_1|u)}I(X_3;Y|X_2U)\leq I(X_3;Y|X_2U)\Big|_{p(x_1|u)}.
\end{align*}Now, the implication follows from 
\begin{align*}
I(X_1X_2X_3;Y|U)&= I(X_1X_2;Y|U) + I(X_3;Y|X_1X_2U)\\
&\stackrel{(a)}{=} I(X_1X_2;Y|U) + I(X_3;YX_1|X_2U)\\
&\geq I(X_1X_2;Y|U) + I(X_3;Y|X_2U),
\end{align*}where $(a)$ follows from the conditional independence of $X_1, X_2, X_3$ given $U$.
Hence, the sum rate constraint (corresponding to $\emptyset$) is redundant in the three user case.
\end{remark}}
\begin{thm}\label{thm:random_k}
For a $k$-user \bmac, 
\pink{\[ \mathcal{R}_{\random} = \mathcal{R}_{\random}^{\weak} = \mathcal{R}.\]}
\end{thm} 
\pink{Similar to the three user case (Section~\ref{sec:random_proof_sketch}), we prove Theorem~\ref{thm:random_k} by showing an achievability in the standard model and a converse for the weak adversary. The converse can be proved by a simple extension of the proof of the converse of Theorem~\ref{thm:random} (three-user randomized coding capacity region) and is skipped}. The achievability uses arguments similar to the proof of achievability of Theorem~\ref{thm:random}. It is shown in Appendix~\ref{section:k_user_rand}. }

\subsection{Deterministic coding capacity region}\label{sec:k_user_det}
Similar to the $3$-user case (Section~\ref{sec:deterministic_capacity_region}), we first give a general symmetrizability condition which characterizes the class of channels under which all users cannot communicate reliably in a \bmac $(W, \cA)$ using deterministic codes. 
For the $3$-user \bmac case, this condition (given below) specializes to the three conditions \eqref{eq:symm1}-\eqref{eq:symm3}.  
\begin{defn}[Symmetrizability and symmetrizable \bmac]\label{defn:symm_k}
For a non-empty set $\cT\subseteq [1:k]$, we say that a \bmac $(W, \cA)$  is {\em\textbf{$\cT$-symmetrizable}}  if there exist sets $\cQ, \cQ'\in  \cA$, not necessarily distinct, satisfying $\cQ\cap \cT = \cQ'\cap\cT = \emptyset$,  and  a pair of conditional distributions $P_{X_{\cQ}|X_{\cT\cup (\cQ\setminus \cQ')}}$ and $P'_{X_{\cQ'}|X_{\cT\cup (\cQ'\setminus \cQ)}}$ such that
\begin{align}
&\sum_{x'_{\cQ}\in \cX_{\cQ}}P_{X_{\cQ}|X_{\cT\cup (\cQ\setminus \cQ')}}(x'_{\cQ}|x_{\cT}, x_{\cQ\setminus \cQ'})W(y|x'_{\cQ}, \tilde{x}_{\cT}, x_{\cQ'\setminus \cQ}, x_{(\cT\cup \cQ \cup \cQ')^c})\nonumber\\
 &= \sum_{\tilde{x}_{\cQ'}\in \cX_{\cQ'}}P'_{X_{\cQ'}|X_{\cT\cup (\cQ'\setminus \cQ)}}(\tilde{x}_{\cQ'}|\tilde{x}_{\cT}, x_{\cQ'\setminus \cQ})W(y|\tilde{x}_{\cQ'}, {x}_{\cT}, x_{\cQ\setminus \cQ'}, x_{(\cT\cup \cQ \cup \cQ')^c})\label{eq:thirty}
\end{align}
for all $x_{\cT}, \tilde{x}_{\cT}\in \cX_{\cT}$, $x_{\cQ\setminus \cQ'}\in \cX_{\cQ\setminus \cQ'}$, $x_{\cQ'\setminus \cQ}\in \cX_{\cQ'\setminus \cQ}$, $x_{(\cT\cup \cQ \cup \cQ')^c}\in \cX_{(\cT\cup \cQ \cup \cQ')^c}$ and $y\in \cY$.
We say that a \bmac $(W, \cA)$ is {\em \textbf{symmetrizable}} if it is $\cT$-symmetrizable for any $\cT\neq\emptyset$.
\end{defn}

\begin{figure}[h]
\centering
\begin{tikzpicture}[scale=0.6]
    \usetikzlibrary {arrows.meta}
    \draw[line width = 1pt] (0,0) rectangle ++(2.5, 6.5) node[pos=.5]{$W$};
    \draw[line width = 1pt,->] (2.5, 3.25) -- ++ (1.7,0) node[right]{$Y$};
    \draw[line width = 1pt,->] (7, 3.25) -- ++ (-1.7,0);
    \draw[line width = 1pt,vermillion] (-5.8,4.5) rectangle ++(4.2, 2) node[pos=.5]{\scriptsize ${ P_{X_{\cQ}|X_{\cT\cup(\cQ\setminus\cQ')}}}$};
    \draw[vermillion,-, line width = 2.5pt, draw opacity=0.7] (0,4.5) -- ++ (0,2);
    \draw[-, line width = 2.5pt, draw opacity=0.6] (0,3) -- ++ (0,1.5);
    \draw[bluishgreen,-, line width = 2.5pt, draw opacity=0.7] (0,2) -- ++ (0,1);
    \draw[myblue,-, line width = 2.5pt, draw opacity=0.7] (0,0) -- ++ (0,2);
    \draw[vermillion, line width=1pt,double distance=1pt,arrows = {-Latex[length=0pt 2.5 0]}] (-1.6,5.5) -- (0,5.5);
    \draw[draw opacity=0.6, line width=1pt,double distance=1pt,arrows = {-Latex[length=0pt 2.5 0]}] (-1.6,3.75) node[left]{\scriptsize $x_{\cQ'\setminus\cQ}$} -- ++ (1.6,0); 
    \draw[bluishgreen, line width=1pt,double distance=1pt,arrows = {-Latex[length=0pt 2.5 0]}] (-1.6,2.5) node[left]{\scriptsize $x_{\inp{\cT\cup\cQ\cup\cQ'}^c}$} -- ++ (1.6,0);
    \draw[myblue, line width=1pt,double distance=1pt,arrows = {-Latex[length=0pt 2.5 0]}] (-1.6,1) node[left]{\scriptsize $\tilde{x}_{\cT}$} -- ++ (1.6,0);
    \draw[vermillion, line width=1pt,double distance=1pt,arrows = {-Latex[length=0pt 2.5 0]}] (-6.8,5) node[left]{\scriptsize $x_{\cQ\setminus\cQ'}$} -- ++ (1,0);
    \draw[vermillion, line width=1pt,double distance=1pt,arrows = {-Latex[length=0pt 2.5 0]}] (-6.8,6) node[left]{\scriptsize $x_{\cT}$} -- ++ (1,0);

    \draw[line width = 1pt] (7,0) rectangle ++(2.5, 6.5) node[pos=.5]{$W$};
    \draw[line width = 1pt,vermillion] (7+2.5+1.6,3) rectangle ++(4.2, 2.5) node[pos=.5]{ \scriptsize ${ P_{X_{\cQ'}|X_{\cT\cup(\cQ'\setminus\cQ)}}}$};
    \draw[vermillion,-, line width = 2.5pt, draw opacity=0.7] (7+2.5,3) -- ++ (0,2.5);
    \draw[-, line width = 2.5pt, draw opacity=0.6] (7+2.5,5.5) -- ++ (0,1);
    \draw[bluishgreen,-, line width = 2.5pt, draw opacity=0.7] (7+2.5,2) -- ++ (0,1);
    \draw[myblue,-, line width = 2.5pt, draw opacity=0.7] (7+2.5,0) -- ++ (0,2);
    \draw[vermillion, line width=1pt,double distance=1pt,arrows = {-Latex[length=0pt 2.5 0]}] (7+2.5+1.6,4.25) -- ++ (-1.6, 0);
    \draw[bluishgreen, line width=1pt,double distance=1pt,arrows = {-Latex[length=0pt 2.5 0]}] (7+2.5+1.6,2.5) node[right]{ \scriptsize $x_{\inp{\cT\cup\cQ\cup\cQ'}^c}$} -- ++ (-1.6, 0); 
    \draw[myblue, line width=1pt,double distance=1pt,arrows = {-Latex[length=0pt 2.5 0]}] (7+2.5+1.6,1) node[right]{ \scriptsize $x_{\cT}$} -- ++ (-1.6, 0);
    \draw[draw opacity=0.7, line width=1pt,double distance=1pt,arrows = {-Latex[length=0pt 2.5 0]}] (7+2.5+1.6,6) node[right]{ \scriptsize $x_{\cQ\setminus\cQ'}$} -- ++ (-1.6, 0);
    \draw[vermillion, line width=1pt,double distance=1pt,arrows = {-Latex[length=0pt 2.5 0]}] (7+2.5+1+4.2+1.6,3.7) node[right]{ \scriptsize $x_{\cQ'\setminus\cQ}$} -- ++ (-1, 0);
    \draw[vermillion, line width=1pt,double distance=1pt,arrows = {-Latex[length=0pt 2.5 0]}] (7+2.5+1+4.2+1.6,4.8) node[right]{ \scriptsize $\tilde{x}_{\cT}$} -- ++ (-1, 0); 
\end{tikzpicture}

\caption{For each $(\tilde{x}_{\cT},{x_{\cT}}, x_{\cQ\setminus \cQ'}, x_{\cQ'\setminus\cQ}, x_{{\cT\cup\cQ\cup\cQ'}^c})$, the conditional output distributions in the two cases above are the same. Thus, the receiver is unable to tell whether the users in set $\cT$ are sending $\tilde{x}_{\cT}$ or $x_{\cT}$.}
\label{fig:symm_k}
\end{figure}
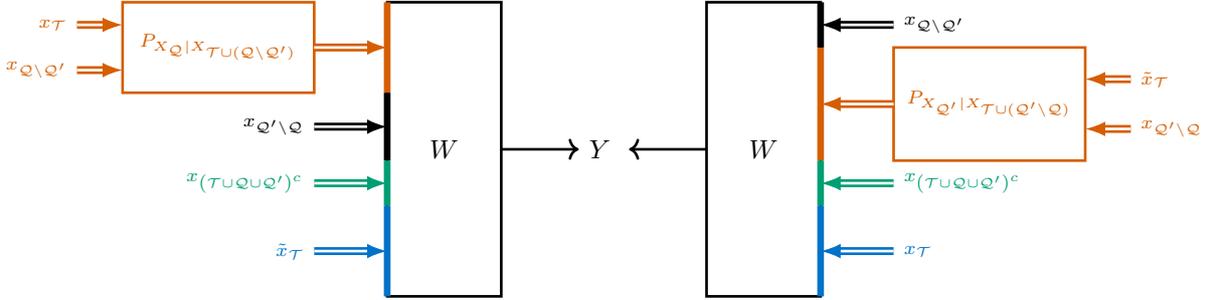

Fig.~\ref{fig:symm_k} illustrates the symmetrizability condition in \eqref{eq:thirty}. The set $\cT$ of users are being symmetrized by the users in sets $\cQ,\cQ'\in \cA$. The users not in ${\cT\cup\cQ\cup\cQ'}$ are not symmetrized. Definition~\ref{defn:symm_k} extends the notion of symmetrizability for the three users case (Definition~\ref{defn:symm}) to the $k$-user \bmac with adversary structure $\cA$. It generalizes the three conditions \eqref{eq:symm1}-\eqref{eq:symm3} to a single condition given by \eqref{eq:thirty}. In particular,  $\cT = \{j, k\}$ and $\cQ=\cQ'=\{i\}$ recovers \eqref{eq:symm1}, $\cT = \{k\}$ and $\cQ=\cQ'=\{i\}$ recovers \eqref{eq:symm2}, and $\cT = \{k\}$, $\cQ=\{i\}$ and $\cQ'=\{j\}$ recovers \eqref{eq:symm3}. 
Now, we are ready to state our main result. 

\begin{thm}\label{achievability_k_users_final}
For a $k$ user \bmac $(W, \cA)$, the interior of the deterministic coding capacity region, $\textsf{int}(\mathcal{R}_{\deterministic})$ is empty if and only if it is symmetrizable. Furthermore, when $(W, \cA)$ is not symmetrizable,  $\mathcal{R}_{\deterministic} = \mathcal{R}_{\random}$.
\end{thm}
\begin{proof}[Proof sketch]
For the {\em converse}, similar to the $3$-user case, we show in Appendix~\ref{sec:empty_interior} that if the channel is symmetrizable then $\textsf{int}(\mathcal{R}_{\deterministic})=\emptyset$.  When the channel is not symmetrizable, the outer bound on the rate region follows from Theorem~\ref{thm:random_k}. 
To show the {\em achievability} direction of the theorem, {\em i.e.}, $\mathcal{R}_{\deterministic} \supseteq \mathcal{R}_{\random}$ if $(W, \cA)$ is not symmetrizable, we take the first approach  discussed in the introduction (Section~\ref{intro:first_approach}). We first show the following lemma (proved in Appendix~\ref{app:k_users_positivity}) which states that all users can communicate at positive rates if a \bmac is not symmetrizable. 
\begin{lemma}\label{achievability_k_users}
If a $k$-user \bmac $(W, \cA)$ is not symmetrizable, then there exists $(R_1, R_2, \ldots, R_k)\in \cR_{\deterministic}$ where $R_i>0$ for all $i\in [1:k]$. 
\end{lemma}
\pink{Next, we show that for every randomized code achieving a small probability of error, there exists another randomized code which also achieves a small probability of error, but requires only $n^2$-valued randomness at each encoder for a code of blocklength $n$.} This {\em randomness reduction} argument is along the lines of the extension of the elimination technique~\cite{Ahlswede78} given in Jahn~\cite[Theorem~1]{Jahn81}. The formal statement and its proof is given in Appendix~\ref{app:randomness_reduction}.

The achievability of Theorem~\ref{achievability_k_users_final} is done in two phases. In the first phase, each user communicates a small number of their uniformly distributed message bits using the positive rate deterministic codes given by Lemma~\ref{achievability_k_users}. These will serve as the shared random bits between the user and the receiver in the second phase. The first phase is short compared to the second phase and only needs to communicate $\log{{n^2}}$ bits for a second phase of blocklength $n$. In the second phase, this small amount of randomness will be used by the new code obtained from the randomness reduction argument to communicate the remaining message bits. Note that the first phase allows the adversary to maliciously choose inputs of the users they control and thus the shared randomness between the malicious users and the decoder.
\pink{This is why in our model, we allow the adversary to select the encoding maps for all users in $\cQ$.}

The above argument is formalized in Lemma~\ref{lem:positivity_implies_capacity} below and is proved in Appendix~\ref{sec:positivity_implies_capacity}. Its proof is along the lines of the proof of \cite[Theorem~12.11]{CK11}. 
\begin{lemma}\label{lem:positivity_implies_capacity}
For a \bmac, if there exists $(R_1, R_2, \ldots, R_k)\in \mathcal{R}_{\deterministic}$ where $R_i>0$ for all $i\in [1:k]$, then $\mathcal{R}_{\deterministic} \supseteq \mathcal{R}_{\random}$.
\end{lemma}

\end{proof}

Similar to the $3$-user case, the proof of Lemma~\ref{achievability_k_users} (formally given in Appendix~\ref{app:k_users_positivity}) 
employs a codebook generated using a random coding argument 
(see Lemma~\ref{codebook_general} in Appendix~\ref{app:k_users_positivity}) and is  along the lines of \cite[Lemma 2]{AhlswedeC99} and \cite[Lemma 3]{CsiszarN88}.  The decoder is a generalization of $3$-user decoder given in Definition~\ref{def:decoder} and is defined below.
\begin{defn}[Decoder]\label{defn:dec_general}
For $\eta>0$,  and encoding maps, $f_i:\mathcal{M}_i\rightarrow \mathcal{X}_i^n$ for $i\in [1:k]$, the decoding set $\cD_{m_1,m_2,\ldots, m_k}  \subseteq \mathcal{Y}^n$ of the message tuple $\inp{m_1,m_2,\ldots, m_k}\in \cM_1\times \cM_2\times\ldots\times \cM_k$ is defined as the intersection of the sets $\dm{i}{m_i}, \, i\in [1:k]$, {\em i.e.}, $\cD_{m_1,m_2,\ldots, m_k}\defineqq \cap_{i=1}^{k}\dm{i}{m_i}$, where the sets $\dm{i}{m_i},\, i\in [1:k]$ are defined as follows:

\noindent A sequence $\vecy \in \dm{i}{m_i},\,i\in [1:k]$, if there exists   
$\cQ\in \cA$, $i\notin\cQ$, $\vecx_{\cQ}\in \cX^n_{\cQ}$, $\tilde{m}_{\cQ^c}\in \cM_{\cQ^c}$ where $\tilde{m}_i=m_i$ and  random variables $X_{\cQ^c}$, $X_{\cQ}$ and $Y$ with  $(f_{\cQ^c}(\tilde{m}_{\cQ^c}), \vecx_{\cQ}, \vecy) \in T^{n}_{X_{\cQ^c} X_{\cQ}Y}$,   satisfying the following:  

\begin{enumerate} 
\item \label{dec:1} $D(P_{X_{\cQ^c}X_{\cQ}Y}||(\prod_{i\in \cQ^c}P_{X_{i}})P_{X_{\cQ}}W)< \eta$.  
\item \label{dec:2}Suppose there exist $\cQ'\in \cA$, not necessarily distinct from $\cQ$, a non-empty set $\cT\subseteq (\cQ\cup\cQ')^c$ with $i\in \cT$,  $\vecx'_{\cQ'}\in \mathcal{X}^n_{\cQ'}$, $m'_{\cQ\setminus \cQ'}\in \cM_{\cQ\setminus \cQ'}$, ${m'}_{\cT}\in \cM_{\cT}$  such that $m'_{t}\neq \tilde{m}_t$ for all $t \in \cT$ such that for the joint distribution $P_{X_{\cQ^c} X_{\cQ} X'_{\cT} X'_{\cQ\setminus \cQ'}X'_{\cQ'}Y}$ defined by $(f_{\cQ^c}(\tilde{m}_{\cQ^c}), \vecx_{\cQ},f_{\cT}(m'_{\cT}),f_{\cQ\setminus \cQ'}(m'_{\cQ\setminus \cQ'}),\vecx'_{\cQ'},  \vecy) \in \allowbreak T^{n}_{X_{\cQ^c} X_{\cQ} X'_{\cT} X'_{\cQ\setminus \cQ'}X'_{\cQ'} Y}$, \\
\begin{align}
     &D(P_{X'_{\cT}X'_{\cQ\setminus \cQ'}X_{\cQ^c\setminus{(\cT\cup\cQ')}} X'_{\cQ'} Y}||(\prod_{t\in \cT}P_{X'_{t}})(\prod_{j\in \cQ\setminus \cQ'}P_{X'_{j}})(\prod_{l\in {\cQ^c\setminus{(\cT\cup\cQ')}}}P_{X_{l}}) P_{X'_{\cQ'}} W)< \eta. \label{dec:2.1}\\
    &\text{Then, }\nonumber\\
    &\hspace{2cm}I(X_{\cQ^c}Y;X'_{\cT} X'_{\cQ\setminus \cQ'}|X_{\cQ}) < \eta. \label{dec:2.2}
\end{align}
\end{enumerate}
\end{defn}
In the definition of $\dm{i}{m_i}$ above, condition \ref{dec:1} checks for typicality with respect to channel inputs $(f_{j}(\tilde{m}_{j}),\, j\in \cQ^c)$ and $\vecx_{\cQ}$. Under condition \eqref{dec:2.1}, where an alternative input to the channel $(f_{t}(m'_{t}),\, t \in \cT),\,\allowbreak (f_{j}(m'_{j}),\,j \in \cQ\setminus \cQ'),\, (f_{l}(m_{l}),\, l\in \cQ^c\setminus(\cT\cup\cQ'))$ and $\vecx'_{\cQ'}$ looks typical, condition \eqref{dec:2.2} implies that the input $(f_{j}(\tilde{m}_{j}),\, j\in \cQ^c)$ and $\vecx_{\cQ}$ is a more plausible explanation for the channel output than the alternative input (see Fig.~\ref{fig:symm_k}). 

As mentioned earlier, Definition~\ref{defn:dec_general} is a generalization of $3$-user decoder in Definition~\ref{def:decoder}. In particular, check \textbf{(a)} can be obtained by setting $\cQ=\cQ'=\{3\}$ and $\cT=\{1,2\}$, \textbf{(b)} by setting $\cQ=\cQ'=\{3\}$ and $\cT=\{1\}$ and \textbf{(c)} by setting $\cQ=\{3\}$, $\cQ'=\{2\}$ and $\cT=\{1\}$.

Similar to Lemma~\ref{lemma:dec}, we can show that for small enough $\eta>0$, $\cD_{m_1,m_2,\ldots,m_k}\cap\cD_{\tilde{m}_1,\tilde{m}_2,\ldots,\tilde{m}_k} = \emptyset$ for every $(m_1,m_2,\ldots,m_k)\neq (\tilde{m}_1,\tilde{m}_2,\ldots, \tilde{m}_k)$. See Lemma~\ref{lemma:dec_disambiguity} in Appendix~\ref{app:k_users_positivity} for the formal statement and proof.

\bibliography{refs} 
\bibliographystyle{ieeetr}
\appendix
\section{Proof of Lemma~\ref{lemma:dec}}\label{app:disamb}
\begin{proof}
Suppose $\vecy\in\mathcal{Y}^n$ is such that $\vecy\in\cD^{(1)}_{m_1}\cap\cD^{(1)}_{\tilde{m}_1}$ for $m_1,\tilde{m}_1\in \mathcal{M}_1$ where $\tilde{m}_1\neq m_1$. Then there exist permutations $(i,j)$ and $(\tilde{i},\tilde{j})$ of $(2,3)$ such that one of the following cases holds.\\
{\bf Case 1}: $(\tilde{i}, \tilde{j}) = (i,j)$\\
There exist $m_j,\tilde{m}_j\in \mathcal{M}_j$, sequences $\vecx_i,\tilde{\vecx}_i\in \mathcal{X}_i^n$, and random variables $X_1,\tilde{X}_1,X_j,\tilde{X}_j,X_i,\tilde{X}_i$ with $(f_1(m_1),f_1(\tilde{m}_1),\allowbreak f_j(m_j),f_j(\tilde{m}_j),\vecx_i,\tilde{\vecx}_i)\in T^{n}_{X_1\tilde{X}_1X_j\tilde{X}_jX_i\tilde{X}_i}$ such that $D(P_{X_1X_jX_i Y}||P_{X_1}\times P_{X_j}\times P_{X_i}\times W),\, D(P_{\tilde{X}_1\tilde{X}_j\tilde{X}_i Y}||P_{\tilde{X}_1}\times P_{\tilde{X}_j}\times P_{\tilde{X}_i}\times W)< \eta$ and
\begin{description}
	\item [Case 1(a)] if $\tilde{m}_j\neq m_j$, then $I(X_1X_jY;\tilde{X}_1\tilde{X}_j|X_i), I(\tilde{X}_1\tilde{X}_jY;X_1X_j|\tilde{X}_i) < \eta $.
	\item [Case 1(b)] if $\tilde{m}_j =  m_j$, then $\tilde{X}_j = X_j$ and $I(X_1X_jY;\tilde{X}_1|X_i), I(\tilde{X}_1X_jY;X_1|\tilde{X}_i) < \eta$. 
\end{description}
{\bf Case 2}: $(\tilde{i}, \tilde{j}) = (j,i)$\\
There exist $m_j\in \mathcal{M}_j$, $\tilde{m}_i\in \mathcal{M}_i$, sequences $\tilde{\vecx}_{j}\in \mathcal{X}_{j}^n$, $\vecx_i\in \mathcal{X}_i^n$ and random variables $X_1$, $\tilde{X}_1$, $X_j$, $\tilde{X}_{j}$, $X_i$, $\tilde{X}_i$ with $(f_1(m_1)$, $f_1(\tilde{m}_1)$, $f_j(m_j)$, $\tilde{\vecx}_{j}$, $\vecx_i$, $f_i(\tilde{m}_i))\in T^{n}_{X_1\tilde{X}_1X_j\tilde{X}_{j}X_i\tilde{X}_i}$ such that $D(P_{X_1X_jX_i Y}||P_{X_1}\times P_{X_j}\times P_{X_i}\times W)$, $ D(P_{\tilde{X}_1\tilde{X}_{j}\tilde{X}_i Y}||P_{\tilde{X}_1}\times P_{\tilde{X}_{j}}\times P_{\tilde{X}_i}\times W)< \eta$ and
 $I(X_1  X_j Y;\tilde{X}_1 \tilde{X}_i|X_i)$, $ I(\tilde{X}_1\tilde{X}_iY;X_1X_j|\tilde{X}_{j})<\eta$.\\

\noindent We first analyze {\bf Case 1(a)}. Let $W_{Y|X_1X_jX_i}$ be denoted by $W$.
\begin{align*}
&D(P_{X_1X_jX_i Y}||P_{X_1}\times P_{X_j}\times P_{X_i}\times W) + D(P_{\tilde{X}_1,\tilde{X}_j}||P_{\tilde{X}_1}\times P_{\tilde{X}_j}) + I(X_1X_jY;\tilde{X}_1\tilde{X}_j|X_i) \stackrel{(a)}{=} \\
&\quad\sum_{x_1,x_j,x_i,y}P_{X_1X_jX_iY}(x_1,x_j,x_i,y)\log{\frac{P_{X_1X_jX_iY}(x_1,x_j,x_i,y)}{P_{X_1}(x_1)P_{X_j}(x_j)P_{X_i}(x_i)W(y|x_1,x_j,x_i)}} + \sum_{\tilde{x}_1,\tilde{x}_j}P_{\tilde{X}_1\tilde{X}_j}(\tilde{x}_1,\tilde{x}_j)\log{\frac{P_{\tilde{X}_1\tilde{X}_j}(\tilde{x}_1,\tilde{x}_j)}{P_{\tilde{X}_1}(\tilde{x}_1)P_{\tilde{X}_j}(\tilde{x}_j)}}\\
&\quad + \sum_{x_1,\tilde{x}_1,x_j,\tilde{x}_j,x_i,y}P_{X_1\tilde{X}_1X_j\tilde{X}_jX_iY}(x_1,\tilde{x}_1,x_j,\tilde{x}_j,x_i,y)\log{\frac{P_{X_1\tilde{X}_1X_j\tilde{X}_jY|X_i}(x_1,\tilde{x}_1,x_j,\tilde{x}_j,y|x_i)}{P_{X_1X_jY|X_i}(x_1,x_j,y|x_i)P_{\tilde{X}_1\tilde{X}_j|X_i}(\tilde{x}_1,\tilde{x}_j|x_i)}}\\
&\quad =\sum_{x_1,\tilde{x}_1,x_j,\tilde{x}_j,x_i,y}P_{X_1\tilde{X}_1X_j\tilde{X}_jX_iY}(x_1,\tilde{x}_1,x_j,\tilde{x}_j,x_i,y)\log{\frac{P_{X_1\tilde{X}_1X_j\tilde{X}_jX_iY}(x_1,\tilde{x}_1,x_j,\tilde{x}_j,x_i,y)}{P_{X_1}(x_1)P_{\tilde{X}_1}(\tilde{x}_1)P_{X_j}(x_j)P_{\tilde{X}_j}(\tilde{x}_j)P_{X_i|\tilde{X}_1\tilde{X}_j}(x_i|\tilde{x}_1,\tilde{x}_j)W(y|x_1,x_j,x_i)}}\\
&\quad = D(P_{X_1\tilde{X}_1X_j\tilde{X}_jX_iY}||P_{X_1}P_{\tilde{X}_1}P_{X_j}P_{\tilde{X}_j}P_{X_i|\tilde{X}_1\tilde{X}_j}W)\\
&\quad \stackrel{\text{(b)}}{\geq} D(P_{X_1\tilde{X}_1X_j\tilde{X}_jY}||P_{X_1}P_{\tilde{X}_1}P_{X_j}P_{\tilde{X}_j}\tilde{V}_1)\text{ where }\tilde{V}_1(y|x_1,\tilde{x}_1,x_j,\tilde{x}_j) = \sum_{x_i}P_{X_i|\tilde{X}_1\tilde{X}_j}(x_i|\tilde{x}_1,\tilde{x}_j)W(y|x_1,x_j,x_i),
\end{align*}
where (b) follows from the log sum inequality. From the given conditions, we know that the term on the LHS of (a) is no greater than $3\eta$. Thus, $D(P_{X_1\tilde{X}_1X_j\tilde{X}_jY}||P_{X_1}P_{\tilde{X}_1}P_{X_j}P_{\tilde{X}_j}\tilde{V}_1) \leq 3\eta$. Using Pinsker's inequality, it follows that 
\begin{equation}\label{eq:a5}
\sum_{x_1,\tilde{x}_1,x_j,\tilde{x}_j,y}\Big|P_{X_1\tilde{X}_1X_j\tilde{X}_jY}(x_1,\tilde{x}_1,x_j,\tilde{x}_j,y)-P_{X_1}(x_1)P_{\tilde{X}_1}(\tilde{x}_1)P_{X_j}(x_j)P_{\tilde{X}_j}(\tilde{x}_j)\tilde{V}_1(y|x_1,\tilde{x}_1,x_j,\tilde{x}_j)\Big| \leq c\sqrt{3\eta},
\end{equation}
where $c$ is some positive constant. Following a similar line of argument, we can show that 
\begin{align*}
3\eta &\geq D(P_{\tilde{X}_1\tilde{X}_j\tilde{X}_i Y}||P_{\tilde{X}_1}\times P_{\tilde{X}_j}\times P_{\tilde{X}_i}\times W) + D(P_{X_1X_j}||P_{X_1}\times P_{X_j}) + I(\tilde{X}_1\tilde{X}_jY;X_1 X_j|\tilde{X}_i)\\
& \geq D(P_{X_1\tilde{X}_1X_j\tilde{X}_jY}||P_{X_1}P_{\tilde{X}_1}P_{X_j}P_{\tilde{X}_j}V_1)\text{ where }V_1(y|x_1,\tilde{x}_1,x_j,\tilde{x}_j) = \sum_{\tilde{x}_i}P_{\tilde{X}_i|X_1X_j}(\tilde{x}_i|x_1,x_j)W(y|\tilde{x}_1,\tilde{x}_j,\tilde{x}_i)
\end{align*}
\noindent Using Pinsker's inequality, it follows that 
\begin{equation}\label{eq:a6}
\sum_{x_1,\tilde{x}_1,x_j,\tilde{x}_j,y}\Big|P_{X_1\tilde{X}_1X_j\tilde{X}_jY}(x_1,\tilde{x}_1,x_j,\tilde{x}_j,y)-P_{X_1}(x_1)P_{\tilde{X}_1}(\tilde{x}_1)P_{X_j}(x_j)P_{\tilde{X}_j}(\tilde{x}_j)V_1(y|x_1,\tilde{x}_1,x_j,\tilde{x}_j)\Big| \leq c\sqrt{3\eta}.
\end{equation}
From \eqref{eq:a5} and \eqref{eq:a6}, 
\begin{equation*}
\sum_{x_1,\tilde{x}_1,x_j,\tilde{x}_j,y}P_{X_1}(x_1)P_{\tilde{X}_1}(\tilde{x}_1)P_{X_j}(x_j)P_{\tilde{X}_j}(\tilde{x}_j)\Big|\tilde{V}_1(y|x_1,\tilde{x}_1,x_j,\tilde{x}_j)-V_1(y|x_1,\tilde{x}_1,x_j,\tilde{x}_j)\Big| \leq 2c\sqrt{3\eta}.
\end{equation*}
This implies that 
\begin{equation}\label{eq:1a}
\max_{x_1,\tilde{x}_1,x_j,\tilde{x}_j,y}\Big|\tilde{V}_1(y|x_1,\tilde{x}_1,x_j,\tilde{x}_j)-V_1(y|x_1,\tilde{x}_1,x_j,\tilde{x}_j)\Big| \leq \frac{2c\sqrt{3\eta}}{\alpha^4}.
\end{equation}
Similar to~\cite[(A.15) on page~748]{AhlswedeC99}, since $\mach$ is not $\cX_1\times\cX_j$-{\em symmetrizable by}~$\cX_i$ ({\em i.e.},~\eqref{eq:symm1} does not hold for $(i,j,k) = (i,j,1)$), we can show that for any pair of channels $P_{\tilde{X}_i|X_1X_j}$ and $P_{X_i|\tilde{X}_1\tilde{X}_j}$, there exists $\zeta_1>0$ such that
\begin{equation*}
\max_{x_1,\tilde{x}_1,x_j,\tilde{x}_j,y}\Big|\tilde{V}_1(y|x_1,\tilde{x}_1,x_j,\tilde{x}_j)-V_1(y|x_1,\tilde{x}_1,x_j,\tilde{x}_j)\Big| \geq \zeta_1.
\end{equation*}
This contradicts \eqref{eq:1a} if $\eta < \frac{\zeta_1^2\alpha^8}{12c^2}$.\\

\noindent We now analyze {\bf Case 1(b)}. 
\begin{align*}
&D(P_{X_1X_jX_i Y}||P_{X_1}\times P_{X_j}\times P_{X_i}\times W)  + I(X_1X_jY;\tilde{X}_1|X_i) \stackrel{(a)}{=} \\
&\quad\sum_{x_1,x_j,x_i,y}P_{X_1X_jX_iY}(x_1,x_j,x_i,y)\log{\frac{P_{X_1X_jX_iY}(x_1,x_j,x_i,y)}{P_{X_1}(x_1)P_{X_j}(x_j)P_{X_i}(x_i)W(y|x_1,x_j,x_i)}} \\
&\quad\quad + \sum_{x_1,\tilde{x}_1,x_j,x_i,y}P_{X_1\tilde{X}_1X_jX_iY}(x_1,\tilde{x}_1,x_j,x_i,y)\log{\frac{P_{X_1\tilde{X}_1X_jY|X_i}(x_1,\tilde{x}_1,x_j,y|x_i)}{P_{X_1X_jY|X_i}(x_1,x_j,y|x_i)P_{\tilde{X}_1|X_i}(\tilde{x}_1|x_i)}}\\
&\quad\quad =\sum_{x_1,\tilde{x}_1,x_j,x_i,y}P_{X_1\tilde{X}_1X_jX_iY}(x_1,\tilde{x}_1,x_j,x_i,y)\log{\frac{P_{X_1\tilde{X}_1X_jX_iY}(x_1,\tilde{x}_1,x_j,x_i,y)}{P_{X_1}(x_1)P_{\tilde{X}_1}(\tilde{x}_1)P_{X_j}(x_j)P_{X_i|\tilde{X}_1}(x_i|\tilde{x}_1)W(y|x_1,x_j,x_i)}}\\
&\quad\quad = D(P_{X_1\tilde{X}_1X_jX_iY}||P_{X_1}P_{\tilde{X}_1}P_{X_j}P_{X_i|\tilde{X}_1}W)\\
&\qquad \stackrel{\text{(b)}}{\geq} D(P_{X_1\tilde{X}_1X_jY}||P_{X_1}P_{\tilde{X}_1}P_{X_j}\tilde{V}_2)\text{ where }\tilde{V}_2(y|x_1,\tilde{x}_1,x_j) = \sum_{x_i}P_{X_i|\tilde{X}_1}(x_i|\tilde{x}_1)W(y|x_1,x_j,x_i),
\end{align*}
where (b) follows from the log sum inequality. From the given conditions, we know that the term on the LHS of (a) is no greater than $2\eta$. Thus, $D(P_{X_1\tilde{X}_1X_jY}||P_{X_1}P_{\tilde{X}_1}P_{X_j}\tilde{V}_2) \leq 2\eta$. Using Pinsker's inequality, it follows that 
\begin{equation}\label{eq:a52}
\sum_{x_1,\tilde{x}_1,x_j,y}\Big|P_{X_1\tilde{X}_1X_jY}(x_1,\tilde{x}_1,x_j,y)-P_{X_1}(x_1)P_{\tilde{X}_1}(\tilde{x}_1)P_{X_j}(x_j)\tilde{V}_2(y|x_1,\tilde{x}_1,x_j)\Big| \leq c\sqrt{2\eta},
\end{equation}
where $c$ is some positive constant. Following a similar line of argument, we can show that 
\begin{align*}
2\eta& \geq D(P_{\tilde{X}_1 X_j \tilde{X}_i Y}||P_{\tilde{X}_1}\times P_{X_j}\times P_{\tilde{X}_i}\times W)  + I(\tilde{X}_1X_jY;X_1|\tilde{X}_i)\\
\qquad & \geq D(P_{X_1\tilde{X}_1X_jY}||P_{X_1}P_{\tilde{X}_1}P_{X_j}V_2)\text{ where }V_2(y|x_1,\tilde{x}_1,x_j) = \sum_{\tilde{x}_i}P_{\tilde{X}_i|X_1}(\tilde{x}_i|x_1)W(y|\tilde{x}_1,x_j,\tilde{x}_i). 
\end{align*}
Using Pinsker's inequality, it follows that 
\begin{equation}\label{eq:a62}
\sum_{x_1,\tilde{x}_1,x_j,y}\Big|P_{X_1\tilde{X}_1X_jY}(x_1,\tilde{x}_1,x_j,y)-P_{X_1}(x_1)P_{\tilde{X}_1}(\tilde{x}_1)P_{X_j}(x_j)V_2(y|x_1,\tilde{x}_1,x_j)\Big| \leq c\sqrt{3\eta}.
\end{equation}
From \eqref{eq:a52} and \eqref{eq:a62}, 
\begin{equation*}
\sum_{x_1,\tilde{x}_1,x_j,y}P_{X_1}(x_1)P_{\tilde{X}_1}(\tilde{x}_1)P_{X_j}(x_j)\Big|\tilde{V}_2(y|x_1,\tilde{x}_1,x_j)-V_2(y|x_1,\tilde{x}_1,x_j)\Big| \leq 2c\sqrt{3\eta}.
\end{equation*}
This implies that 
\begin{equation}\label{eq:1b}
\max_{x_1,\tilde{x}_1,x_j,y}\Big|\tilde{V}_2(y|x_1,\tilde{x}_1,x_j)-V_2(y|x_1,\tilde{x}_1,x_j)\Big| \leq \frac{2c\sqrt{2\eta}}{\alpha^4}.
\end{equation}
Similar to~\cite[(A.5) on page~747]{AhlswedeC99}, since $\mach$ is not $\cX_1|\cX_j$-{\em symmetrizable by}~$\cX_i$ ({\em i.e.},~\eqref{eq:symm2} does not hold for $(i,j,k) = (i,j,1)$), we can show that for any pair for channels $P_{\tilde{X}_i|X_1}$ and $P_{X_i|\tilde{X}_1}$, there exists $\zeta_2>0$ such that
\begin{equation*}
\max_{x_1,\tilde{x}_1,x_j,y}\Big|\tilde{V}_2(y|x_1,\tilde{x}_1,x_j)-V_2(y|x_1,\tilde{x}_1,x_j)\Big| \geq \zeta_2.
\end{equation*}
This contradicts \eqref{eq:1b} if $\eta < \frac{\zeta_2^2\alpha^8}{8c^2}$.\\

 \noindent We now analyse {\bf Case 2}. 
\begin{align*}
&D(P_{X_1X_jX_i Y}||P_{X_1}\times P_{X_j}\times P_{X_i}\times W) + D(P_{\tilde{X}_1,\tilde{X}_i}||P_{\tilde{X}_1}\times P_{\tilde{X}_i}) + I(X_1  X_j Y;\tilde{X}_1 \tilde{X}_i|X_i) \stackrel{(a)}{=} \\
&\quad\sum_{x_1,x_j,x_i,y}P_{X_1X_jX_iY}(x_1,x_j,x_i,y)\log{\frac{P_{X_1X_jX_iY}(x_1,x_j,x_i,y)}{P_{X_1}(x_1)P_{X_j}(x_j)P_{X_i}(x_i)W(y|x_1,x_j,x_i)}} + \sum_{\tilde{x}_1,\tilde{x}_i}P_{\tilde{X}_1\tilde{X}_i}(\tilde{x}_1,\tilde{x}_i)\log{\frac{P_{\tilde{X}_1\tilde{X}_i}(\tilde{x}_1,\tilde{x}_i)}{P_{\tilde{X}_1}(\tilde{x}_1)P_{\tilde{X}_i}(\tilde{x}_i)}}\\
&\quad + \sum_{x_1,\tilde{x}_1,x_j,x_i,\tilde{x}_i,y}P_{X_1\tilde{X}_1X_jX_i\tilde{X}_iY}(x_1,\tilde{x}_1,x_j,x_i,\tilde{x}_i,y)\log{\frac{P_{X_1\tilde{X}_1X_j\tilde{X}_iY|X_i}(x_1,\tilde{x}_1,x_j,\tilde{x}_i,y|x_i)}{P_{\tilde{X}_1\tilde{X}_i|X_i}(\tilde{x}_1,\tilde{x}_i|x_i)P_{X_1X_jY|X_i}(x_1,x_j,y|x_i)}}\\
&\quad =\sum_{x_1,\tilde{x}_1,x_j,x_i,\tilde{x}_i,y}P_{X_1\tilde{X}_1X_jX_i\tilde{X}_iY}(x_1,\tilde{x}_1,x_j,x_i,\tilde{x}_i,y)\log{\frac{P_{X_1\tilde{X}_1X_jX_i\tilde{X}_iY}(x_1,\tilde{x}_1,x_j,x_i,\tilde{x}_i,y)}{P_{X_1}(x_1)P_{\tilde{X}_1}(\tilde{x}_1)P_{X_j}(x_j)P_{\tilde{X}_i}(\tilde{x}_i)P_{X_i|\tilde{X}_1\tilde{X}_i}(x_i|\tilde{x}_1,\tilde{x}_i)W(y|x_1,x_j,x_i)}}\\
&\quad = D(P_{X_1\tilde{X}_1X_jX_i\tilde{X}_iY}||P_{X_1}P_{\tilde{X}_1}P_{X_j}P_{\tilde{X}_i}P_{X_i|\tilde{X}_1\tilde{X}_i}W)\\
&\quad \stackrel{\text{(b)}}{\geq} D(P_{X_1\tilde{X}_1X_j\tilde{X}_iY}||P_{X_1}P_{\tilde{X}_1}P_{X_j}P_{\tilde{X}_i}\tilde{V}_3)\text{ where }\tilde{V}_3(y|x_1,\tilde{x}_1,x_j,\tilde{x}_i) = \sum_{x_i}P_{X_i|\tilde{X}_1\tilde{X}_i}(x_i|\tilde{x}_1,\tilde{x}_i)W(y|x_1,x_j,x_i),
\end{align*}
where (b) follows from the log sum inequality. From the given conditions, we know that the term on the LHS of (a) is no greater than $3\eta$. Thus, $D(P_{X_1\tilde{X}_1X_j\tilde{X}_iY}||P_{X_1}P_{\tilde{X}_1}P_{X_j}P_{\tilde{X}_i}\tilde{V}_3) \leq 3\eta$. Using Pinsker's inequality, it follows that 
\begin{equation}\label{eq:a1}
\sum_{x_1,\tilde{x}_1,x_j,\tilde{x}_i,y}\Big|P_{X_1\tilde{X}_1X_j\tilde{X}_iY}(x_1,\tilde{x}_1,x_j,\tilde{x}_i,y)-P_{X_1}(x_1)P_{\tilde{X}_1}(\tilde{x}_1)P_{X_j}(x_j)P_{\tilde{X}_i}(\tilde{x}_i)\tilde{V}_3(y|x_1,\tilde{x}_1,x_j,\tilde{x}_i)\Big| \leq c\sqrt{3\eta}
\end{equation}
for some constant $c>0$.
\noindent Following a similar line of argument, 
\begin{align*}
&D(P_{\tilde{X}_1\tilde{X}_{j}\tilde{X}_i Y}||P_{\tilde{X}_1}\times P_{\tilde{X}_{j}}\times P_{\tilde{X}_i}\times W) + D(P_{X_1X_j}||P_{X_1}\times P_{X_{j}}) + I(\tilde{X}_1 \tilde{X}_iY ;X_1  X_j|\tilde{X}_{j}) \stackrel{(a)}{=} \\
&\quad\sum_{\tilde{x}_1,\tilde{x}_{j},\tilde{x}_i,y}P_{\tilde{X}_1\tilde{X}_{j}\tilde{X}_iY}(\tilde{x}_1,\tilde{x}_{j},\tilde{x}_i,y)\log{\frac{P_{\tilde{X}_1\tilde{X}_{j}\tilde{X}_iY}(\tilde{x}_1,\tilde{x}_{j},\tilde{x}_i,y)}{P_{\tilde{X}_1}(\tilde{x}_1)P_{\tilde{X}_{j}}(\tilde{x}_{j})P_{\tilde{X}_i}(\tilde{x}_i)W(y|\tilde{x}_1,\tilde{x}_{j},\tilde{x}_i)}} + \sum_{x_1,x_j}P_{X_1X_{j}}(x_1,x_{j})\log{\frac{P_{X_1X_{j}}(x_1,x_{j})}{P_{X_1}(x_1)P_{X_{j}}(x_{j})}}\\
&\quad + \sum_{x_1,\tilde{x}_1,x_j,\tilde{x}_{j},\tilde{x}_i,y}P_{X_1\tilde{X}_1X_j\tilde{X}_{j}\tilde{X}_iY}(x_1,\tilde{x}_1,x_j,\tilde{x}_{j},\tilde{x}_i,y)\log{\frac{P_{X_1\tilde{X}_1X_j\tilde{X}_iY|\tilde{X}_{j}}(x_1,\tilde{x}_1,x_j,\tilde{x}_i,y|\tilde{x}_{j})}{P_{\tilde{X}_1\tilde{X}_iY|\tilde{X}_{j}}(\tilde{x}_1,\tilde{x}_i,y|\tilde{x}_{j})P_{X_1X_j|\tilde{X}_{j}}(x_1,x_j|\tilde{x}_{j})}}\\
&\quad =\sum_{x_1,\tilde{x}_1,x_j,\tilde{x}_{j},\tilde{x}_i,y}P_{X_1\tilde{X}_1X_j\tilde{X}_{j}\tilde{X}_iY}(x_1,\tilde{x}_1,x_j,\tilde{x}_{j},\tilde{x}_i,y)\log{\frac{P_{X_1\tilde{X}_1X_j\tilde{X}_{j}\tilde{X}_iY}(x_1,\tilde{x}_1,x_j,\tilde{x}_{j},\tilde{x}_i,y)}{P_{X_1}(x_1)P_{\tilde{X}_1}(\tilde{x}_1)P_{X_j}(x_j)P_{\tilde{X}_i}(\tilde{x}_i)P_{\tilde{X}_{j}|X_1X_j}(\tilde{x}_{j}|x_1,x_j)W(y|\tilde{x}_1,\tilde{x}_{j},\tilde{x}_i)}}\\
&\quad = D(P_{X_1\tilde{X}_1X_j\tilde{X}_{j}\tilde{X}_iY}||P_{X_1}P_{\tilde{X}_1}P_{X_j}P_{\tilde{X}_{j}P_{\tilde{X}_{j}}|X_1X_{j}}W)\\
&\quad \geq D(P_{X_1\tilde{X}_1X_j\tilde{X}_iY}||P_{X_1}P_{\tilde{X}_1}P_{X_j}P_{\tilde{X}_i}V_3)\text{ where }V_3(y|x_1,\tilde{x}_1,x_j,\tilde{x}_i) = \sum_{\tilde{x}_{j}}P_{\tilde{X}_{j}|X_1X_j}(\tilde{x}_{j}|x_1,x_j)W(y|\tilde{x}_1,\tilde{x}_{j},\tilde{x}_i).
\end{align*}

\noindent From the given conditions, the term on the left of (a) is no larger than $3\eta$. Thus, $D(P_{X_1\tilde{X}_1X_j\tilde{X}_iY}||P_{X_1}P_{\tilde{X}_1}P_{X_j}P_{\tilde{X}_i}V_3) \leq 3\eta$.

\noindent 
Using Pinsker's inequality, it follows that 
\begin{equation}\label{eq:a2}
\sum_{x_1,\tilde{x}_1,x_j,\tilde{x}_i,y}\Big|P_{X_1\tilde{X}_1X_j\tilde{X}_iY}(x_1,\tilde{x}_1,x_j,\tilde{x}_i,y)-P_{X_1}(x_1)P_{\tilde{X}_1}(\tilde{x}_1)P_{X_j}(x_j)P_{\tilde{X}_i}(\tilde{x}_i)V_3(y|x_1,\tilde{x}_1,x_j,\tilde{x}_i)\Big| \leq c\sqrt{3\eta}.
\end{equation}
From \eqref{eq:a1} and \eqref{eq:a2}, 
\begin{equation}
\sum_{x_1,\tilde{x}_1,x_j,\tilde{x}_j,y}P_{X_1}(x_1)P_{\tilde{X}_1}(\tilde{x}_1)P_{X_j}(x_j)P_{\tilde{X}_i}(\tilde{x}_i)\Big|\tilde{V}_3(y|x_1,\tilde{x}_1,x_j,\tilde{x}_i)-V_3(y|x_1,\tilde{x}_1,x_j,\tilde{x}_i)\Big| \leq 2c\sqrt{3\eta}.
\end{equation}

\noindent This implies that 
\begin{equation}\label{eq:1c}
\max_{x_1,\tilde{x}_1,x_j,\tilde{x}_j,y}\Big|\tilde{V}_3(y|x_1,\tilde{x}_1,x_j,\tilde{x}_i)-V_3(y|x_1,\tilde{x}_1,x_j,\tilde{x}_i)\Big| \leq \frac{2c\sqrt{3\eta}}{\alpha^4}.
\end{equation}
Since $\mach$ is not $\cX_1$-{\em symmetrizable by}~$\cX_j/\cX_i$ ({\em i.e.},~\eqref{eq:symm3} does not hold for $(i,j,k) = (i,j,1)$), for any pair of channels $P_{X_i|\tilde{X}_1\tilde{X}_i}$ and $P_{\tilde{X}_{\tilde{j}}|X_1X_j}$, there exists $\zeta_3>0$, such that
\begin{equation*}
\max_{x_1,\tilde{x}_1,x_j,\tilde{x}_j,y}\Big|\tilde{V}_3(y|x_1,\tilde{x}_1,x_j,\tilde{x}_i)-V_3(y|x_1,\tilde{x}_1,x_j,\tilde{x}_i)\Big| \geq \zeta_3.
\end{equation*}
This contradicts \eqref{eq:1c} if $\eta < \frac{\zeta_3^2\alpha^8}{12c^2}$.
Let $\zeta\defineqq \min{\{\zeta_1,\zeta_2,\zeta_3\}}$, any $\eta$ satisfying $0<\eta<\frac{\zeta^2\alpha^8}{12c^2}$ ensures disjoint decoding regions.

\end{proof}


\section{Proof of Lemma~\ref{lemma:codebook} \pink{(Codebook Lemma)}}\label{appendix:codebook}
\pink{To prove Lemma~\ref{lemma:codebook}, we will first define some terminology and prove a concentration result in Lemma~\ref{theorem:concentration}. This will be used to prove Lemma~\ref{lemma:codebook} (Codebook Lemma) as a corollary. }
\pink{\subsection{A concentration result}\label{subsec:concentration}
In this subsection, we restate \cite[Theorem 2.1]{SJ} in a form that can be directly used for proving the properties of the codebook.}

For a positive integer $b$, let $\cS_b$ denote the symmetric group of degree $b$, {\em i.e.}, it contains the permutations of $\{1, 2, \ldots, b\}$. 
For a permutation $\sigma\in \cS_b$, let $\sigma (i), \, i\in \{1, 2, \ldots, b\}$ denote the image of $i$ under $\sigma$. Let $\cA$ be a set. 
For a $b-$length tuple $\inp{\alpha_1, \ldots, \alpha_b}$ consisting of distinct elements of $\cA$, let 
\begin{align*}
\cH_{\inp{\alpha_1, \ldots, \alpha_{b}}} = \inb{a\in\inp{\inb{\alpha_1, \ldots, \alpha_b}\cup\{*\}}^b:\exists \,\sigma\in \cS_b \text{ such that for all $j\in [1:b]$ if } a_j \neq *\text{ then }a_j = \alpha_{\sigma (j)}}, 
\end{align*} 
where $a_j$ represents the $j^\text{th}$ element of the tuple $a$. For $a\in \cH_{\inp{\alpha_1, \ldots, \alpha_{b}}}$, let $|a| = |\inb{i:a_i\neq *}|$. 
For a tuple $\inp{\gamma_1, \ldots, \gamma_{b}}$ consisting of distinct elements of $\cA$, we say that $a\in\cH_{\inp{\alpha_1,  \ldots, \alpha_{b}}}$ and $\inp{\gamma_1, \ldots, \gamma_{b}}$ are $\inp{\alpha_1, \ldots, \alpha_{b}}-$compatible (denoted by $\inp{\gamma_1,  \ldots, \gamma_{b}}{\sim}\insq{a, \inp{\alpha_1, \ldots, \alpha_{b}}}$\pink{)}, if for all $l\in \{1, \ldots, b\}$,
\begin{align*}
\gamma_l &=a_l, \hspace{2.25cm}\text{ if }a_l\neq *, \\
		\gamma_l&\in \pink{\cA}\setminus\inb{\alpha_1, \ldots, \alpha_{b}}, \text{ otherwise.}
\end{align*}

\pink{For example, let $\cA = \{1, 2, \ldots, 9\}$, $b =5$, $(\alpha_1, \ldots, \alpha_b) = \inp{1, 2, 3, 4, 5}$ and $a = (1, 2, *, *, 4)$. Then, $a\in \cH_{(\alpha_1, \ldots, \alpha_b)}$ with $|a| = 3$. Suppose $\inp{\gamma_1,\ldots, \gamma_b}=(1, 2, 6,8, 4)$. Then, $\inp{\gamma_1,\ldots, \gamma_b}\sim\insq{a,(\alpha_1, \ldots, \alpha_b)}$.}
\begin{lemma}\label{theorem:concentration}
For an index set $\cI$, let $\{Y_i: i\in \cI\}$ be a set of independent random variables. Let $\beta$ be a positive integer. Let $\cJ\subseteq \cI^{\beta}$ be a set of $\beta$ length tuples consisting of distinct elements from $\cI$. 
For $\inp{i_1, \ldots, i_{\beta}}\in \cJ$, let $V_{\inp{i_1, \ldots, i_{\beta}}}$ be a binary random variable which is a function of $Y_{i_1}, \ldots, Y_{i_{\beta}}$. 
Suppose $U = \sum_{\inp{i_1, \ldots, i_{\beta}}\in \cJ}V_{\inp{i_1, \ldots, i_{\beta}}}$.
Let 
\begin{align*}
E\geq \max \left\{\max_{\inp{i_1,\ldots,i_{\beta}}\in\cJ}\pink{\left(\max_{\substack{\inp{y_{i_1},\ldots, y_{i_{\beta}}},\\a\in \cH_{\inp{i_1,\ldots,i_{\beta}}}\\1\leq|a|\leq \beta-1}}\bbE\left[ \sum_{\stackrel{{\inp{j_1, \ldots, j_{\beta}}\in \cJ}:}{{\inp{j_1, \ldots, j_{\beta}}\sim \insq{a, \inp{i_1,\ldots,i_{\beta}}}}}}V_{\inp{j_1, \ldots, j_{\beta}}}\Bigg| \inp{Y_{i_1}, \ldots, Y_{i_{\beta}}} = \inp{y_{i_1}, \ldots, y_{i_{\beta}}}\right]\right)}, \bbE[U]\right\},
\end{align*}
For $\gamma>0$, $\nu >1$  
if there exists $\delta_1, \delta_2, \ldots, \delta_{\beta}>1$ such that for all $i\in [1:\beta]$, $\frac{1}{(2\beta)^{\beta}}\inp{\frac{\delta_{i-1}-1}{2\gamma}E - \beta!}>\delta_{i} E$ for $\delta_0 := \nu$, then
\begin{align}\label{eq:concentration_final}
\bbP(U\geq \nu E)\leq {\inp{|\cI|\beta+1}^{{\beta}^2}}e^{-\gamma/3}.
\end{align}
\end{lemma}

\begin{proof}
For $\inp{i_1, \ldots, i_{\beta}}\in \cJ$, let $\tilde{U}_{\inp{i_1, \ldots, i_{\beta}}} =\sum_{\inp{j_1, \ldots, j_{\beta}}\in \cJ:\inb{j_1, \ldots, j_{\beta}}\cap\inb{i_1, \ldots, i_{\beta}}\neq\emptyset}V_{\inp{j_1, \ldots, j_{\beta}}}$. To show \eqref{eq:concentration_final}, we will first show that
\begin{align}\label{eq:concentration}
\bbP(U\geq \nu E)\leq e^{-\gamma/3}+\sum_{\inp{i_1, \ldots, i_{\beta}}\in\cJ}\bbP\inp{\tilde{U}_{\inp{i_1, \ldots, i_{\beta}}}> \frac{(\nu-1)E}{2\gamma}}
\end{align}
using \cite[Theorem 2.1]{SJ}\pink{, which is restated below.}

\begin{duplicatelemma}\cite[Theorem 2.1]{SJ}
Suppose that $Y_{\alpha}, \alpha \in \cA$, is a finite family of non-negative
random variables and that $\sim$ is a symmetric relation on the index set $\cA$
such that each $Y_{\alpha}$ is independent of $\{Y_{\beta} : \beta \nsim \alpha\}$; in other words, the pairs
$(\alpha, \beta)$ with $\alpha\sim\beta$ define the edge set of a (weak) dependency graph for the
variables $Y_{\alpha}$. Let $X:=\sum_{\alpha}Y_{\alpha}$ and $\mu:=\bbE X =\sum_{\alpha} \bbE Y$. Let further, for $\alpha\in \cA$, $\tilde{X}_{\alpha}:=\sum_{\beta\sim\alpha}Y_{\beta}$.
If $t\geq \mu>0$, then for every real $r>0$,
\begin{align*}
\bbP(X>\mu+t)\leq e^{-r/3}+\sum_{\alpha\in \cA}\bbP\inp{\tilde{X}_{\alpha}>\frac{t}{2r}}.
\end{align*}
\end{duplicatelemma}
In order to obtain \eqref{eq:concentration} from \cite[Theorem 2.1]{SJ}, we use $\cJ$ in place of $\cA$ and $V_{\inp{i_1, \ldots, i_{\beta}}}$ in place of $Y_{\alpha}$. 
We note that every $\inp{i_1, \ldots, i_{\beta}},\inp{j_1, \ldots, j_{\beta}}\in \cJ$ such that $\inb{i_1, \ldots, i_{\beta}}\cap\inb{j_1, \ldots, j_{\beta}}\neq\emptyset$, $\inp{i_1, \ldots, i_{\beta}}\sim\inp{j_1, \ldots, j_{\beta}}$ as per the symmetric relation given in \cite[Theorem 2.1]{SJ}. 
Thus, the definitions of $\tilde{U}_{\inp{i_1, \ldots, i_{\beta}}}$ and $\tilde{X}_{\alpha}$ are consistent. 
We upper bound the LHS of \eqref{eq:concentration} as follows:
\begin{align*}
\bbP(X\geq \nu E) &= \bbP(X-\bbE[X]\geq \nu E-\bbE[X])\\
&\leq \bbP(X-\bbE[X]\geq \nu E-E)\\
&=\bbP(X-\bbE[X]\geq (\nu-1)E)
\end{align*}
Theorem~\cite[Theorem 2.1]{SJ} is applied on $\bbP(X\geq \bbE [X]+ (\nu-1)E)$ with $(\nu-1)E$ as $t$ and $\gamma$ as $r$.

Now, we will show \eqref{eq:concentration_final}. We will use \textbf{strong induction on $\beta$}. 
{When $\beta = 1$, \eqref{eq:concentration} implies $\bbP(U\geq \nu E)\leq e^{-\gamma/3}$. This is because for any $1-$length tuple $(i)\in \cJ$, $\tilde{U}_{(i)} = V_{(i)}$, which, being a binary random variable, is at most $1$. However, $\frac{1}{(2\beta)^{\beta}}\inp{\frac{\nu-1}{2\gamma}E - \beta!}>\delta_{1} E\geq0$  implies that for $\beta = 1$, $\frac{1}{2}\inp{\frac{\nu-1}{2\gamma}E - 1}>0$. Therefore, ${\frac{\nu-1}{2\gamma}E}>1$ and the second term on the RHS of \eqref{eq:concentration} is zero.
Thus, \eqref{eq:concentration_final} holds for $\beta = 1$.}

Now, for the \textbf{induction hypothesis}, consider any $\beta'\leq k$ for some positive integer $k$. For an index set $\cI'$, let $\{Y'_i: i\in \cI'\}$ be a set of independent random variables. Let $\cJ'\subseteq \cI'^{\beta'}$ be a set of $\beta'-$length tuples consisting of distinct elements from $\cI'$. 
For $\inp{i_1, \ldots, i_{\beta'}}\in \cJ'$, let $V'_{\inp{i_1, \ldots, i_{\beta'}}}$ be a binary random variable which is a function of $Y'_{i_1}, Y'_{i_2}, \ldots, Y'_{i_{\beta'}}$. 
Suppose $U' = \sum_{\inp{i_1, \ldots, i_{\beta'}}\in \cJ'}V'_{\inp{i_1, \ldots, i_{\beta}}}$. 

Let 
\begin{align*}
E'\geq \max \left\{\max_{\inp{i_1,\ldots,i_{\beta'}}\in\cJ'}\max_{\substack{\inp{y'_{i_1},\ldots, y'_{i_{\beta'}}},\\a\in \cH_{\inp{i_1,\ldots,i_{\beta'}}}\\1\leq|a|\leq \beta'-1}}\bbE\left[ \sum_{\stackrel{{\inp{j_1, \ldots, j_{\beta'}}\in \cJ'}:}{{\inp{j_1, \ldots, j_{\beta'}}  \sim \insq{a, \inp{i_1, \ldots, i_{\beta'}}}}}}V'_{\inp{j_1, \ldots, j_{\beta'}}}\Bigg| \inp{Y'_{i_1}, \ldots, Y'_{i_{\beta'}}} = \inp{y'_{i_1}, \ldots, y'_{i_{\beta'}}}\right], \bbE[U']\right\},
\end{align*}
For $\gamma'>0$, $\nu' >1$, if there exists $\delta'_1, \delta'_2, \ldots, \delta'_{\beta'}>1$ such that for all $i\in [1:\beta']$, $\frac{1}{(2\beta')^{\beta'}}\inp{\frac{\delta'_{i-1}-1}{2\gamma'}E' - \beta'!}>\delta'_{i} E'$ for $\delta'_0 := \nu'$, then

\begin{align*}
\bbP(U'\geq \nu' E')\leq {\inp{|\cI'|\beta' + 1}^{{\beta'}^2}}e^{-\gamma'/3}.
\end{align*}

Now, for $\beta = k+1$ and any $\gamma, \, \nu,$ and $E$ and random variables satisfying the conditions in Lemma~\ref{theorem:concentration},  \eqref{eq:concentration} gives
\begin{align*}
\bbP(U\geq \nu E)\leq e^{-\gamma/3}+\sum_{\inp{i_1, \ldots, i_{k+1}}\in\cJ}\bbP\inp{\tilde{U}_{\inp{i_1, \ldots, i_{k+1}}}> \frac{(\nu-1)E}{2\gamma}}.
\end{align*}

For $\inp{i_1, \ldots, i_{k+1}}\in\cJ$, and any realization  $\inp{y_{i_1}, \ldots, y_{i_{k+1}}}$ of $\inp{Y_{i_1}, \ldots, Y_{i_{k+1}}}$,
\begin{align}
&\bbP\inp{\tilde{U}_{\inp{i_1, \ldots, i_{k+1}}}> \frac{(\nu-1)E}{2\gamma}\,\Bigg|\inp{Y_{i_1}, \ldots, Y_{i_{k+1}}} = \inp{y_{i_1}, \ldots, y_{i_{k+1}}}}\nonumber\\
&\leq\bbP\Bigg( \sum_{\sigma\in \cS_{k+1}}V_{\sigma\inp{i_1, \ldots, i_{k+1}}} + \sum_{\substack{a\in \cH_{\inp{i_1, \ldots, i_{k+1}}}\\1\leq|a|\leq k}}\sum_{\stackrel{{\inp{j_1, \ldots, j_{k+1}}\in\cJ}:}{{\inp{j_1, \ldots, j_{k+1}}\sim \insq{a, \inp{i_1, \ldots, i_{k+1}}}}}}V_{\inp{j_1, \ldots, j_{k+1}}}> \frac{(\nu-1)E}{2\gamma}\,\Bigg|\nonumber\\
&\qquad \qquad\qquad \qquad\qquad \qquad\qquad \qquad\qquad \qquad\qquad \qquad\qquad \qquad\qquad \qquad\inp{Y_{i_1}, \ldots, Y_{i_{k+1}}} = \inp{y_{i_1}, \ldots, y_{i_{k+1}}}\Bigg)\nonumber\\
&=\bbP\Bigg( \sum_{\stackrel{a\in \cH_{\inp{i_1, \ldots, i_{k+1}}}} {1\leq|a|\leq k}}\sum_{\stackrel{{\inp{j_1, \ldots, j_{k+1}}\in\cJ}:}{{\inp{j_1, \ldots, j_{k+1}}\sim \insq{a, \inp{i_1, \ldots, i_{k+1}}}}}}V_{\inp{j_1, \ldots, j_{k+1}}}> \frac{(\nu-1)E}{2\gamma}-\sum_{\sigma\in \cS_{k+1}}V_{\sigma\inp{i_1, \ldots, i_{k+1}}}\,\Bigg|\nonumber\\
&\qquad \qquad\qquad \qquad\qquad \qquad\qquad \qquad\qquad \qquad\qquad \qquad\qquad \qquad\qquad \qquad\inp{Y_{i_1}, \ldots, Y_{i_{k+1}}} = \inp{y_{i_1}, \ldots, y_{i_{k+1}}}\Bigg)\nonumber\\
&\stackrel{(a)}{\leq}\bbP\inp{ \sum_{\stackrel{a\in \cH_{\inp{i_1, \ldots, i_{k+1}}}} {1\leq|a|\leq k}}\sum_{\stackrel{{\inp{j_1, \ldots, j_{k+1}}\in\cJ}:}{{\inp{j_1, \ldots, j_{k+1}}\sim \insq{a, \inp{i_1, \ldots, i_{k+1}}}}}} V_{\inp{j_1, \ldots, j_{k+1}}}> \frac{(\nu-1)E}{2\gamma}-|\cS_{k+1}| \Bigg|\inp{Y_{i_1}, \ldots, Y_{i_{k+1}}} = \inp{y_{i_1}, \ldots, y_{i_{k+1}}}}\nonumber\\
&\stackrel{(b)}{\leq}\bbP\left(\bigcup_{\stackrel{a\in \cH_{\inp{i_1, \ldots, i_{k+1}}}}{1\leq|a|\leq k}}\inp{\sum_{\stackrel{{\inp{j_1, \ldots, j_{k+1}}\in\cJ}:}{{\inp{j_1, \ldots, j_{k+1}}\sim \insq{a, \inp{i_1, \ldots, i_{k+1}}}}}}V_{\inp{j_1, \ldots, j_{k+1}}}> \frac{1}{|\cH_{\inp{i_1, \ldots, i_{k+1}}}|-1-(k+1)!}\inp{\frac{(\nu-1)E}{2\gamma}-(k+1)!} }\Bigg|\right.\nonumber\\
&\qquad \qquad\qquad \qquad\qquad \qquad\qquad \qquad\qquad \qquad\qquad \qquad\qquad \qquad\qquad \qquad\inp{Y_{i_1}, \ldots, Y_{i_{k+1}}} = \inp{y_{i_1}, \ldots, y_{i_{k+1}}}\Bigg)\nonumber\\	
&{\leq}\bbP\left(\bigcup_{\stackrel{a\in \cH_{\inp{i_1, \ldots, i_{k+1}}}}{1\leq|a|\leq k}}\inp{\sum_{\stackrel{{\inp{j_1, \ldots, j_{k+1}}\in\cJ}:}{{\inp{j_1, \ldots, j_{k+1}}\sim \insq{a, \inp{i_1, \ldots, i_{k+1}}}}}}V_{\inp{j_1, \ldots, j_{k+1}}}> \frac{1}{|\cH_{\inp{i_1, \ldots, i_{k+1}}}|}\inp{\frac{(\nu-1)E}{2\gamma}-(k+1)!} }\Bigg|\right.\nonumber\\
&\qquad \qquad\qquad \qquad\qquad \qquad\qquad \qquad\qquad \qquad\qquad \qquad\qquad \qquad\qquad \qquad\inp{Y_{i_1}, \ldots, Y_{i_{k+1}}} = \inp{y_{i_1}, \ldots, y_{i_{k+1}}}\Bigg)\nonumber\\
&\stackrel{(c)}{\leq}\sum_{\stackrel{a\in \cH_{\inp{i_1, \ldots, i_{k+1}}}} {1\leq|a|\leq k}}\bbP\left( \sum_{\stackrel{{\inp{j_1, \ldots, j_{k+1}}\in\cJ}:}{{\inp{j_1, \ldots, j_{k+1}}\sim \insq{a, \inp{i_1, \ldots, i_{k+1}}}}}}V_{\inp{j_1, \ldots, j_{k+1}}}> \frac{1}{(2(k+1))^{k+1}}\inp{\frac{(\nu-1)E}{2\gamma}-(k+1)!} \,\Bigg|\right.\nonumber\\
&\qquad \qquad\qquad \qquad\qquad \qquad\qquad \qquad\qquad \qquad\qquad \qquad\qquad \qquad\qquad \qquad\inp{Y_{i_1}, \ldots, Y_{i_{k+1}}} = \inp{y_{i_1}, \ldots, y_{i_{k+1}}}\Bigg)\nonumber\\
&\stackrel{(d)}{\leq}\sum_{\stackrel{a\in \cH_{\inp{i_1, \ldots, i_{k+1}}}} {1\leq|a|\leq k}}\bbP\left( \sum_{\stackrel{{\inp{j_1, \ldots, j_{k+1}}\in\cJ}:}{{\inp{j_1, \ldots, j_{k+1}}\sim \insq{a, \inp{i_1, \ldots, i_{k+1}}}}}}V_{\inp{j_1, \ldots, j_{k+1}}}> \frac{1}{(2(k+1))^{k+1}}\inp{\frac{(\nu-1)E}{2\gamma}-(k+1)!E} \,\Bigg|\right.\nonumber\\
&\qquad \qquad\qquad \qquad\qquad \qquad\qquad \qquad\qquad \qquad\qquad \qquad\qquad \qquad\qquad \qquad\inp{Y_{i_1}, \ldots, Y_{i_{k+1}}} = \inp{y_{i_1}, \ldots, y_{i_{k+1}}}\Bigg)\nonumber\\
&\stackrel{(e)}{\leq}\sum_{\stackrel{a\in \cH_{\inp{i_1, \ldots, i_{k+1}}}} {1\leq|a|\leq k}}\bbP\left( \sum_{\stackrel{{\inp{j_1, \ldots, j_{k+1}}\in\cJ}:}{{\inp{j_1, \ldots, j_{k+1}}\sim \insq{a, \inp{i_1, \ldots, i_{k+1}}}}}}V_{\inp{j_1, \ldots, j_{k+1}}}> \delta_1 E\,\Bigg|\inp{Y_{i_1}, \ldots, Y_{i_{k+1}}} = \inp{y_{i_1}, \ldots, y_{i_{k+1}}}\right) .\label{eq:previous}
\end{align}
Here, $(a)$ holds because $\sum_{\sigma\in \cS}V_{\sigma\inp{i_1, \ldots, i_{k+1}}}$, being a sum of binary random variables, takes the maximum value $|\cS|$ which is $(k+1)!$. 
The equality $(b)$ holds because $\bbP\inp{\sum_{1\leq i\leq t}A_i>c} \leq \bbP\inp{\cup_{1\leq i\leq t}\inb{A_i>c/t}}$ for any integer $t$, real number $c$  and random variables $A_1, \ldots, A_t$. 
The inequality $(c)$ uses union bound and the fact that $|\cH_{\inp{i_1, \ldots, i_{k+1}}}|\leq (\#\text{ of subsets of }\inb{i_1, \ldots, i_{k+1}}) \times |\cS_{k+1}|$. Thus, $|\cH_{\inp{i_1, \ldots, i_{k+1}}}| \leq 2^{k+1}(k+1)!\leq (2(k+1))^{k+1}$. Inequality $(d)$ holds because $E\geq 1$ and $(e)$ follows from the conditions on $\nu,\gamma$ and $\beta=k+1$ in the statement of Lemma~\ref{theorem:concentration}.


Fix $\inp{i_1, \ldots, i_{k+1}}\in \cJ$, $\inp{Y_{i_1}, \ldots, Y_{i_{k+1}}} = \inp{y_{i_1}, \ldots, y_{i_{k+1}}}$ and $a \in \cH_{\inp{i_1, \ldots, i_{k+1}}}$ such that $|a| = l$ where $l\in [1:k]$. We will use induction hypothesis at this stage. 
To use induction hypothesis, choose $\nu' = \delta_1$, $\gamma' = \gamma$, $E' = E$, $\beta' = k+1-l$ and $\cI' = \cI\setminus\inb{i_1, \ldots, i_{k+1}}$. The set of random variables $\{Y'_i:i\in \cI'\}$ is given by $Y'_{i}=Y_i,\, \forall\,i\in \cI'$. For $i\in [1:k+1]$, let $|a_{1}^{i-1}|= |\inb{j\in [1:i-1]:a_j\neq *}|$.  The set $\cJ'$ consists of $(k+1-l)-$length tuples of distinct elements from $\cI'$ such that for every $(j_1, \ldots, j_{k+1-l})\in \cJ'$, there exists $(m_1, \ldots, m_{k+1})\in \cJ$ such that $m_l = a_l$ if $a_l\neq *$. 
Else, $m_l = j_{l-|a_1^{l-1}|}$. For such $(m_1, \ldots, m_{k+1})\in \cJ$ and $(j_1, \ldots, j_{k+1-l})$, we will say that $(j_1, \ldots, j_{k+1-l}) + a = (m_1, \ldots, m_{k+1})$. Thus, for every $(m_1, \ldots, m_{k+1})\in \cJ$ such that $(m_1, \ldots, m_{k+1})\sim \insq{a, \inp{i_1, \ldots, i_{k+1}}}$, there exists a unique $(j_1, \ldots, j_{k+1-l})\in \cJ'$ with $(j_1, \ldots, j_{k+1-l}) + a = (m_1, \ldots, m_{k+1})$.

For $\inp{j_1, \ldots, j_{k+1-l}}\in \cJ'$, the binary random variable $V'_{\inp{j_1, \ldots, j_{k+1-l}}}$ is the random variable $V_{\inp{m_1,  \ldots, m_{k+1}}}$ where $(m_1, \ldots, m_{k+1})=(j_1, \ldots, j_{k+1-l}) + a$ and the random variables $\inp{Y_{i_1}, \ldots, Y_{i_{k+1}}}$ are fixed to  $\inp{y_{i_1}, \ldots, y_{i_{k+1}}}$. For $U' = \sum_{\inp{j_1, \ldots, j_{k+1-l}}\in \cJ'}V'_{\inp{j_1, \ldots, j_{k+1-l}}}$, we will use induction hypothesis on $\bbP\inp{U'\geq \delta_1E}$.
\begin{align}
&\bbP\inp{U'\geq \delta_1E}\nonumber\\
&=\bbP\inp{ \sum_{\inp{j_1, \ldots, j_{k+1-l}}\in \cJ'}V'_{\inp{j_1, \ldots, j_{k+1-l}}}> \delta_1 E }\nonumber\\
&=\bbP\inp{  \sum_{\stackrel{{\inp{m_1, \ldots, m_{k+1}}\in\cJ}:}{{\inp{m_1, \ldots, m_{k+1}}\sim \insq{a, \inp{i_1, \ldots, i_{k+1}}}}}}V_{\inp{m_1, \ldots, m_{k+1}}}> \delta_1 E \Bigg|\inp{Y_{i_1}, \ldots, Y_{i_{k+1}}} = \inp{y_{i_1}, \ldots, y_{i_{k+1}}}}.\label{eq:star}
\end{align}

We know that 
\begin{align*}
E\geq \max \left\{\max_{\substack{\inp{y_{i_1},\ldots, y_{i_{k+1}}}\\\inp{i_1,\ldots,i_{k+1}}\in\cJ}}\max_{\substack{a\in \cH_{\inp{i_1,\ldots,i_{k+1}}}\\1\leq|a|\leq k}}\bbE\left[ \sum_{\stackrel{{\inp{j_1, \ldots, j_{k+1}}\in \cJ}:}{{\inp{j_1, \ldots, j_{k+1}}\sim \insq{a, \inp{i_1, \ldots, i_{k+1}}}}}}\hspace{-0.5cm}V_{\inp{j_1, \ldots, j_{k+1}}}\Bigg| \inp{Y_{i_1}, \ldots, Y_{i_{k+1}}} = \inp{y_{i_1}, \ldots, y_{i_{k+1}}}\right], \bbE[U]\right\},
\end{align*}
and for $\gamma>0$, $\nu >1$  there exists $\delta_1, \delta_2, \ldots, \delta_{k+1}$ such that for all $i\in [1:k+1]$, $\delta_i>1$ and for $\delta_0 = \nu$, $\frac{1}{(2(k+1))^{k+1}}\inp{\frac{\delta_{i-1}-1}{2\gamma}E - (k+1)!}>\delta_{i} E$. \\
We will use this to show that the choices of $\gamma'$, $\nu',\, E',\, \beta'$ satisfy the conditions in the induction hypothesis.  
\begin{align*}
\bbE[U'] &= \bbE\insq{\sum_{\inp{j_1, \ldots, j_{k+1-l}}\in \cJ'}V'_{\inp{j_1, \ldots, j_{k+1-l}}}}\\
&= \bbE\left[ \sum_{\stackrel{{\inp{m_1, \ldots, m_{k+1}}\in \cJ}:}{{\inp{m_1, \ldots, m_{k+1}}\sim \insq{a, \inp{i_1, \ldots, i_{k+1}}}}}}V_{\inp{m_1, \ldots, m_{k+1}}}\Bigg| \inp{Y_{i_1}, \ldots, Y_{i_{k+1}}} = \inp{y_{i_1}, \ldots, y_{i_{k+1}}}\right]\\
 &\leq E = E'
\end{align*} 
For $\inp{j_1,\ldots,j_{k+1-l}}\in\cJ'$, $a'\in \cH_{\inp{j_1,\ldots,j_{k+1-l}}}$ with $|a'| = l'$ for $1\leq l'\leq k-l$, define a $(k+1)-$length tuple $a''$ as
\begin{align*}
a''_l=\begin{cases}
a_l, &\text{ if }a_l\neq *,\\
a'_{l-|a_1^{l-1}|}, &\text{ if }a_l= *.
\end{cases}
\end{align*} for all $l\in[1:k+1]$.
Let $(m_1, \ldots, m_{k+1})\in \cJ$ be such that $\inp{j_1,\ldots,j_{k+1-l}}+a = (m_1, \ldots, m_{k+1})$. Then, for fixed $\inp{y_{j_1},\ldots, y_{j_{k+1-l}}}$,
\begin{align*}
&\bbE\left[ \sum_{\stackrel{{\inp{g_1, \ldots,g_{k+1-l}}\in \cJ'}:}{{\inp{g_1, \ldots, g_{k+1-l}}\sim \insq{a', \inp{j_1, \ldots, j_{k+1-l}}}}}}V'_{\inp{g_1, \ldots, g_{k+1-l}}}\Bigg| \inp{Y_{j_1}, \ldots, Y_{j_{k+1-l}}} = \inp{y_{j_1}, \ldots, y_{j_{k+1-l}}}\right]\\
&=\bbE\left[ \sum_{\stackrel{{\inp{h_1, \ldots,h_{k+1}}\in \cJ}:}{{\inp{h_1, \ldots, h_{k+1}}\sim \insq{a'', \inp{m_1, \ldots, m_{k+1}}}}}}V_{\inp{h_1, \ldots, h_{k+1}}}\Bigg| \substack{\inp{Y_{j_1}, \ldots, Y_{j_{k+1-l}}} = \inp{y_{j_1}, \ldots, y_{j_{k+1-l}}},\\\inp{Y_{i_1}, \ldots, Y_{i_{k+1}}} = \inp{y_{i_1}, \ldots, y_{i_{k+1}}}\\}\right]\\
&\stackrel{(a)}{=}\bbE\left[ \sum_{\stackrel{{\inp{h_1, \ldots,h_{k+1}}\in \cJ}:}{{\inp{h_1, \ldots, h_{k+1}}\sim \insq{a'', \inp{m_1, \ldots, m_{k+1}}}}}}V_{\inp{h_1, \ldots, h_{k+1}}}\Bigg| \inp{Y_{m_1}, \ldots, Y_{m_{k+1}}} = \inp{y_{m_1}, \ldots, y_{m_{k+1}}}\right]\\
&\leq E=E'.
\end{align*}
In the above, $(a)$ follows from definition of $\cI'$ and $\cJ'$.

Now, we need to show that for $\nu'=\delta'_0=\delta_1,\,\gamma'=\gamma,$ there exists $\delta'_1, \delta'_2, \ldots, \delta'_{\beta'}$ such that for all $i\in [1:\beta']$, $\delta'_i>1$ and 
 $\frac{1}{(2\beta')^{\beta'}}\inp{\frac{\delta'_{i-1}-1}{2\gamma'}E' - \beta'!}>\delta'_{i} E'$.
   First note that as $l\in [1:k]$, $\beta'\in [1:k]$. We know that there exists $\delta_1, \delta_2, \ldots, \delta_{k+1}$ such that for all $i\in [1:k+1]$, $\delta_i>1$ and for $\delta_0 = \nu$, $\frac{1}{(2(k+1))^{k+1}}\inp{\frac{\delta_{i-1}-1}{2\gamma}E - (k+1)!}>\delta_{i} E$. 
Let $\delta_i' = \delta_{i+1}$ for all $i\in[1:\beta']$. Then, for all $i\in [1:\beta']$, $\delta_i' > 1$ and 
\begin{align*}
&\frac{1}{(2\beta')^{\beta'}}\inp{\frac{\delta'_{i-1}-1}{2\gamma'}E' - (\beta')!}\\
&=\frac{1}{(2\beta')^{\beta'}}\inp{\frac{\delta_{i}-1}{2\gamma}E - (\beta')!}\\
&\geq\frac{1}{(2(k+1))^{k+1}}\inp{\frac{\delta_{i}-1}{2\gamma}E - (k+1)!}\\
&\geq \delta_{i+1}E\\
&=\delta'_i E'.
\end{align*}

With this, all the conditions in the induction hypothesis are satisfied and we are ready to apply induction hypothesis.  
Thus,
\begin{align}
&=\bbP\inp{  \sum_{\stackrel{{\inp{m_1, \ldots, m_{k+1}}\in\cJ}:}{{\inp{m_1, \ldots, m_{k+1}}\sim \insq{a, \inp{i_1, \ldots, i_{k+1}}}}}}V_{\inp{m_1, \ldots, m_{k+1}}}> \delta_1 E \Bigg|\inp{Y_{i_1}, \ldots, Y_{i_{k+1}}} = \inp{y_{i_1}, \ldots, y_{i_{k+1}}}}\nonumber\\
&\stackrel{(a)}{=}\bbP\inp{U'\geq \delta_1E}\nonumber\\
&={\inp{|\cI|(k+1-l)+1}^{(k+1-l)^2}} e^{-\gamma/3},\label{eq:plus}
\end{align}
where $(a)$ uses \eqref{eq:star}. Continuing the analysis of \eqref{eq:previous},
\begin{align*}
&\sum_{\stackrel{a\in \cH_{\inp{i_1, \ldots, i_{k+1}}}} {1\leq|a|\leq k}}\bbP\left( \sum_{\stackrel{{\inp{j_1, \ldots, j_{k+1}}\in\cJ}:}{{\inp{j_1, \ldots, j_{k+1}}\sim \insq{a, \inp{i_1, \ldots, i_{k+1}}}}}}V_{\inp{j_1, \ldots, j_{k+1}}}> \delta_1 E\,\Bigg|\inp{Y_{i_1}, \ldots, Y_{i_{k+1}}} = \inp{y_{i_1}, \ldots, y_{i_{k+1}}}\right)\\
&=\sum_{l=1}^k\sum_{\stackrel{a\in \cH_{\inp{i_1, \ldots, i_{k+1}}}} {|a|=l}}\bbP\left( \sum_{\stackrel{{\inp{j_1, \ldots, j_{k+1}}\in\cJ}:}{{\inp{j_1, \ldots, j_{k+1}}\sim \insq{a, \inp{i_1, \ldots, i_{k+1}}}}}}V_{\inp{j_1, \ldots, j_{k+1}}}> \delta_1 E\,\Bigg|\inp{Y_{i_1}, \ldots, Y_{i_{k+1}}} = \inp{y_{i_1}, \ldots, y_{i_{k+1}}}\right)\\
&\stackrel{(b)}{\leq}\sum_{l = 1}^{k}\sum_{\stackrel{a\in \cH_{\inp{i_1, \ldots, i_{k+1}}}} {|a|=l}}{\inp{|\cI|(k+1-l)+1}^{(k+1-l)^2}} e^{-\gamma/3}\\
&{=}\sum_{l = 1}^{k}\binom{k+1}{l}\binom{k+1}{l} l!{\inp{|\cI|(k+1-l)+1}^{(k+1-l)^2}} e^{-\gamma/3}\\
&\stackrel{(c)}{=}\sum_{m = 1}^{k}\binom{k+1}{k+1-m}\binom{k+1}{k+1-m}(k+1-m)!{\inp{|\cI|m+1}^{m^2}} e^{-\gamma/3}\\
&{=}\sum_{m = 1}^{k}\binom{k+1}{k+1-m}\frac{(k+1)!}{m!}{\inp{|\cI|m+1}^{m^2}} e^{-\gamma/3}\\
&{\leq}\sum_{m = 1}^{k}\binom{k+1}{m}((k+1)!)\inp{|\cI|k+1}^{km} e^{-\gamma/3}\\
&{\leq}\sum_{m = 0}^{k+1}\binom{k+1}{m}((k+1)!)\inp{|\cI|k+1}^{km} e^{-\gamma/3}\\
&{=}(k+1)!e^{-\gamma/3}\sum_{m = 0}^{k+1}\binom{k+1}{m}\inp{\inp{|\cI|k+1}^k}^m \\
&{\leq}(k+1)!e^{-\gamma/3}\inp{\inp{|\cI|k+1}^k+1}^{k+1}\\
&{\leq}(k+1)^{(k+1)}e^{-\gamma/3}\inp{\inp{|\cI|k+1}+1}^{k(k+1)}
\end{align*}

Inequality $(b)$ uses  \eqref{eq:plus}.  The equality $(c)$ is obtained by substituting $l$ with $k+1-m$.

Thus, using this analysis, we see that for $\beta= k+1$,
\begin{align*}
\bbP(U\geq \nu E)&\leq e^{-\gamma/3}+\sum_{\inp{i_1, \ldots, i_{k+1}}\in\cJ}\bbP\inp{\tilde{U}_{\inp{i_1, \ldots, i_{k+1}}}> \frac{(\nu-1)E}{2\gamma}}\\
&\leq e^{-\gamma/3}+\sum_{\inp{i_1, \ldots, i_{k+1}}\in\cJ}(k+1)^{(k+1)}e^{-\gamma/3}\inp{|\cI|\inp{k+1}+1}^{k(k+1)}\\
&= e^{-\gamma/3}+|\cJ|(k+1)^{(k+1)}e^{-\gamma/3}\inp{|\cI|\inp{k+1}+1}^{k(k+1)}\\
&\leq e^{-\gamma/3}+|\cI|^{k+1} (k+1)^{(k+1)}e^{-\gamma/3}\inp{|\cI|\inp{k+1}+1}^{k(k+1)}\\
&\leq e^{-\gamma/3}+(|\cI|(k+1))^{k+1}e^{-\gamma/3}\inp{|\cI|\inp{k+1}+1}^{k(k+1)}\\
&\leq \inp{|\cI|\inp{k+1}+1}^{k(k+1)}e^{-\gamma/3}+(|\cI|(k+1))^{k+1}e^{-\gamma/3}\inp{|\cI|\inp{k+1}+1}^{k(k+1)}\\
&\leq ((|\cI|(k+1))^{k+1}+1)e^{-\gamma/3}\inp{|\cI|(k+1)+1}^{k(k+1)}\\
&\leq (|\cI|(k+1)+1)^{k+1}e^{-\gamma/3}\inp{|\cI|(k+1)+1}^{k(k+1)}\\
&\leq e^{-\gamma/3}\inp{|\cI|(k+1)+1}^{(k+1)(k+1)}.
\end{align*}
\end{proof}
\subsection{Preliminary codebook lemma}
\pink{We use the concentration result from Section~\ref{subsec:concentration} (Lemma~\ref{theorem:concentration}) to prove the existence of a codebook with properties as given in Lemma~\ref{codebook}. Roughly speaking, each property counts the number of codewords which are typical with fixed vectors. The codebook lemma (Lemma~\ref{lemma:codebook}) follows from Lemma~\ref{codebook} as a corollary. We first state and prove Lemma~\ref{codebook} and give a proof of Lemma~\ref{lemma:codebook} in the next subsection. }

\pink{We need to define the terminology of {\em Total Correlation} to state the properties of the codebook in Lemma~\ref{codebook}.}
\pink{For random variables $Z_1, Z_2, \ldots, Z_m$, let $C(Z_1;Z_2;\dots; Z_m)$ denote the {\em total correlation} of the random variables $Z_1, Z_2, \ldots, Z_m$ which is given by
\begin{align}\label{defn:tot_corr}
C(Z_1;Z_2;\ldots; Z_m) := \sum_{i=1}^mH(Z_i)-H(Z_1,Z_2,\ldots, Z_m).
\end{align}
Note that $C(Z_1;Z_2; \ldots; Z_m)$ can also be written as 
\begin{align*}
\sum_{i=2}^{m}I(Z_i;Z^{i-1}). 
\end{align*}
Suppose $\bbR_{+}$ denotes the set of positive real numbers. Let $k\in \{1, 2 \ldots\}$. Consider random variables $U_1, U_2, \ldots, U_k,V$ and a set $\cS\subseteq [1:k]$ given by $\cS = \set{\alpha_1, \alpha_2, \ldots, \alpha_{|\cS|}}$. Let $\cS^c = [1:k]\setminus \cS$ be denoted by $\cS^c = \set{\beta_1, \beta_2, \ldots, \beta_{k-|\cS|}}$. We define $g_{U_1, U_2, \ldots, U_k, V}^{\cS}:\bbR_{+}^k\rightarrow \bbR_{+}$ as
\begin{align}\label{eq: notation}
g^{\cS}_{U_1, U_2, \ldots, U_k, V}(R_1, R_2, \ldots, R_k)=\inp{\sum_{i\in \cS} R_i}-C(U_{\alpha_1};U_{\alpha_2}; \ldots; U_{\alpha_{|\cS|}};(U_{\beta_1},U_{\beta_2}, \ldots, U_{\beta_{k-|\cS|}}, V))
\end{align} 
Note that the tuple $(U_{\beta_1},U_{\beta_2}, \ldots, U_{\beta_{k-|\cS|}}, V)$ is treated as a single random variable. Thus, when $|\cS| = 0$,\\ $g^{\cS}_{U_1, U_2, \ldots, U_k, V}(R_1, R_2, \ldots, R_k) = 0$.}
\begin{lemma}\label{codebook}
For any  $\epsilon>0,\,  n\geq n_0(\epsilon), \, N_1, N_2, N_3\geq\exp(n\epsilon)$ and types $P_1\in \cP_{\cX_1}^n$, $ P_2\in \cP_{\cX_2}^n$, $P_3\in \cP_{\cX_3}^n$; there exists codebooks $\left\{\vecx_{11}, \ldots,  \vecx_{1N_1}\in \mathcal{X}_1^n\right\}, \left\{\vecx_{21}, \ldots,  \vecx_{2N_2}\in \mathcal{X}_2^n\right\}$, $ \left\{\vecx_{31}, \ldots,  \vecx_{3N_3}\in \mathcal{X}_3^n\right\}$ whose codewords are of type $P_1, P_2$, $P_3$ respectively such that for every permutation $(i,j,k)$ of $(1,2,3)$; for every $(\vecx_i,\vecx_j,\vecx_k)  \in \mathcal{X}^n_i\times\mathcal{X}^n_j\times\mathcal{X}^n_k$; for every joint type $P_{X_iX'_iX_jX'_jX_kX'_k}\in \cP^n_{\cX_{i}\times \cX_{i}\times\cX_{j}\times\cX_j\times\cX_k\times\cX_k}$; and for $R_i \defineqq (1/ n )\log_2{N_i},R_j \defineqq (1/ n )\log_2{N_j}$, and $R_k \defineqq (1/ n )\log_2{N_k}$; the following holds:
\begin{align}
&\text{\pink{\em (i) Joint typicality of a codeword}}\nonumber\\
&|\{r\in [1:N_i]: (\vecx_{ir},\vecx_{k})\in T^{n}_{X_i X_k} \}| <\exp\left\{n\left(|R_{i}-I(X_i;X_k)|^{+}+\epsilon/2\right)\right\}; \label{lemma_eq1a}\\
&|\{s\in [1:N_j]: (\vecx_{i},\vecx_{js}, \vecx_k)\in T^{n}_{X_i X_j X_k} \}| < \exp\left\{n\left(|R_{j}-I(X_j;X_i,X_k)|^{+}+\epsilon/2\right)\right\};  \label{lemma_eq1b}\\
&|\{u\in [1:N_i]:(\vecx_{iu}, \vecx_i, \vecx_j, \vecx_k)\in T^{n}_{X'_i  X_i X_j X_k}\}|\leq\exp{\left(n\left(\left|R_1-I(X_i';X_iX_jX_k)\right|^{+}+\epsilon/2\right)\right)};\label{lemma_eq7b_}\\
&\text{\pink{\em (ii) Joint typicality of a pair of codewords}}\nonumber\\
&|\{(r, s)\in [1:N_i]\times[1:N_j]:(\vecx_{ir},\vecx_{js}, \vecx_k)\in T^{n}_{X_i X_j X_k}\}|\nonumber\\
&\qquad\qquad\leq \exp\left\{n\left(\left|\left|R_i-I(X_i;X_k)\right|^++\left|R_j-I(X_j;X_k)\right|^+-I(X_i;X_j|X_k)\right|^+ +\epsilon/2\right)\right\}; \label{lemma_eq3a}\\
&|\{(r, u)\in [1:N_i]\times[1:N_i]:(\vecx_{ir},\vecx_{iu}, \vecx_k)\in T^{n}_{X_i X_i' X_k}, \, r\neq u\}|\nonumber\\
&\qquad\qquad\leq \exp\left\{n\left(\left|\left|R_i-I(X_i;X_k)\right|^++\left|R_i-I(X_i';X_k)\right|^+-I(X_i;X_i'|X_k)\right|^+ +\epsilon/2\right)\right\}; \label{lemma_eq3b}\\
&|\{(s, v)\in [1:N_j]\times[1:N_j]:(\vecx_{js},\vecx_{jv}, \vecx_k)\in T^{n}_{X_j X_j' X_k}, \, s\neq v\}|\nonumber\\
&\qquad\qquad\leq \exp\left\{n\left(\left|\left|R_j-I(X_j;X_k)\right|^++\left|R_j-I(X_j';X_k)\right|^+-I(X_j;X_j'|X_k)\right|^+ +\epsilon/2\right)\right\};\label{lemma_eq3e}\\
&|\{(r, w)\in [1:N_i]\times[1:N_k]:(\vecx_{ir},\vecx_{kw}, \vecx_k)\in T^{n}_{X_i, X_k', X_k}\}|\nonumber\\
&\qquad\qquad\leq \exp\left\{n\left(\left|\left|R_i-I(X_i;X_k)\right|^++\left|R_k-I(X_k';X_k)\right|^+-I(X_i;X_k'|X_k)\right|^+ +\epsilon/2\right)\right\}; \label{lemma_eq6c}\\
&|\{(s, w)\in [1:N_j]\times[1:N_k]:(\vecx_{js},\vecx_{kw}, \vecx_k)\in T^{n}_{X_j X_k' X_k}\}|\nonumber\\
&\qquad\qquad\leq \exp\left\{n\left(\left|\left|R_j-I(X_j;X_k)\right|^++\left|R_k-I(X_k';X_k)\right|^+-I(X_j;X_k'|X_k)\right|^+ +\epsilon/2\right)\right\};\label{lemma_eq6e}\\
&|\{(u, v)\in [1:N_i]\times[1:N_j]:(\vecx_{iu},\vecx_{jv}, \vecx_i, \vecx_j, \vecx_k)\in T^{n}_{X'_i X'_j X_i X_j X_k}\}|\nonumber\\
&\qquad\qquad\leq\exp{\left(n\left(\left||R_i-I(X_i';X_iX_jX_k)|^{+}+|R_j-I(X_j';X_iX_jX_k)|^{+}-I(X_i';X_j'|X_iX_jX_k)\right|^{+}+\epsilon/2\right)\right)}\label{lemma_eq7a_}	\\
&|\{(u, w)\in [1:N_i]\times[1:N_k]:(\vecx_{iu},\vecx_{kw}, \vecx_i, \vecx_j, \vecx_k)\in T^{n}_{X'_iX'_jX_iX_j X_k}\}|\nonumber\\
&\qquad\qquad\leq\exp{\left(n\left(\left||R_i-I(X_i';X_iX_jX_k)|^{+}+|R_k-I(X_k';X_iX_jX_k)|^{+}-I(X_i';X_k'|X_iX_jX_k)\right|^{+} + \epsilon/2\right)\right)}\label{lemma_eq7c_}\\
&\text{\pink{\em (iii) Joint typicality of three codewords}}\nonumber\\
&|\{(r, s, u)\in [1:N_i]\times[1:N_j]\times[1:N_i]:(\vecx_{ir},\vecx_{js},\vecx_{iu}, \vecx_k)\in T^{n}_{X_i X_j X_i'X_k}, \, r\neq u\}|\nonumber\\
&\qquad\qquad\leq \max_{\cS\subseteq \{1, 2, 3\}}\exp\left\{n\left({g_{X_i, X_j, X_i', X_k}^{\cS}(R_i, R_j, R_i)}+\epsilon/2\right)\right\}; \text{ and }\label{lemma_eq5}\\
&\text{\pink{\em (iv) Joint typicality of four codewords}}\nonumber\\
&|\{(r, s, u, v)\in [1:N_i]\times[1:N_j]\times[1:N_i]\times[1:N_j]:(\vecx_{ir}, \vecx_{js}, \vecx_{iu}, \vecx_{jv},\vecx_k)\in T^{n}_{X_iX_j X_i^{'} X'_j X_k}, \, r\neq u, s\neq v\}|\nonumber\\
&\qquad\qquad\leq \max_{\cS\subseteq \{1, 2, 3, 4\}}\exp\left\{n\left({g_{X_i, X_j, X_i', X_j', X_k}^{\cS}(R_i, R_j, R_i, R_j)}+\epsilon/2\right)\right\}\label{lemma_eq4}\\
&|\{(r, s, u, w)\in [1:N_i]\times[1:N_j]\times[1:N_i]\times[1:N_k]:(\vecx_{ir}, \vecx_{js}, \vecx_{iu}, \vecx_{jv},\vecx_k)\in T^{n}_{X_iX_j X_i^{'} X'_k X_k}, \, r\neq u\}|\nonumber\\
&\qquad\qquad\leq \max_{\cS\subseteq \{1, 2, 3, 4\}}\exp\left\{n\left({g_{X_i, X_j, X_i', X_k', X_k}^{\cS}(R_i, R_j, R_i, R_k)}+\epsilon/2\right)\right\}\label{lemma_eq6}.
\end{align}
\end{lemma}

\begin{proof}

We will generate the codebooks by a random experiment. For fixed $(\vecx_i,\vecx_j, \vecx_k)  \in \mathcal{X}^n_i \times \mathcal{X}^n_j \times \mathcal{X}^n_k$and joint type $P_{X_iX'_iX_jX'_jX_kX'_k}\in \cP^n_{\cX_{i}\times \cX_{i}\times\cX_{j}\times\cX_j\times\cX_k\times\cX_k}$, we will show that the probability that each of the  statements~\eqref{lemma_eq1a}~-~\eqref{lemma_eq6} does not hold, falls doubly exponentially in $n$. Since $|\cX_i^n|$, $|\cX_j^n|$, $|\cX_k^n|$ and $|\cP^n_{\cX_{i}\times \cX_{i}\times\cX_{j}\times\cX_j\times\cX_k\times\cX_k}|$  grow  only exponentially in $n$, a union bound will imply that the probability that any of the statements~\eqref{lemma_eq1a}~-~\eqref{lemma_eq6} fail for some $\vecx_i,\vecx_j, \vecx_k$ and $P_{X_iX'_iX_jX'_jX_kX'_k}$  also falls doubly exponentially. This will show the existence of a codebook satisfying~\eqref{lemma_eq1a}~-~\eqref{lemma_eq6}. The proof will employ Lemma~\ref{theorem:concentration} which we have restated below for quick reference.

\begin{duplicatelemma}
For an index set $\cI$, let $\{Y_i: i\in \cI\}$ be a set of independent random variables. Let $\beta$ be a positive integer. Let $\cJ\subseteq \cI^{\beta}$ be a set of $\beta$ length tuples consisting of distinct elements from $\cI$. 
For $\inp{i_1, \ldots, i_{\beta}}\in \cJ$, let $V_{\inp{i_1, \ldots, i_{\beta}}}$ be a binary random variable which is a function of $Y_{i_1}, \ldots, Y_{i_{\beta}}$. 
Suppose $U = \sum_{\inp{i_1, \ldots, i_{\beta}}\in \cJ}V_{\inp{i_1, \ldots, i_{\beta}}}$.
Let 
\begin{align}\label{eq:cond_on_E}
E\geq \max \left\{\max_{\inp{i_1,\ldots,i_{\beta}}\in\cJ}\max_{\substack{\inp{y_{i_1},\ldots, y_{i_{\beta}}},\\a\in \cH_{\inp{i_1,\ldots,i_{\beta}}}\\1\leq|a|\leq \beta-1}}\bbE\left[ \sum_{\stackrel{{\inp{j_1, \ldots, j_{\beta}}\in \cJ}:}{{\inp{j_1, \ldots, j_{\beta}}\sim \insq{a, \inp{i_1,\ldots,i_{\beta}}}}}}V_{\inp{j_1, \ldots, j_{\beta}}}\Bigg| \inp{Y_{i_1}, \ldots, Y_{i_{\beta}}} = \inp{y_{i_1}, \ldots, y_{i_{\beta}}}\right], \bbE[U]\right\},
\end{align}
For $\gamma>0$, $\nu >1$  
if there exists $\delta_1, \delta_2, \ldots, \delta_{\beta}>1$ such that for all $i\in [1:\beta]$, $\frac{1}{(2\beta)^{\beta}}\inp{\frac{\delta_{i-1}-1}{2\gamma}E - \beta!}>\delta_{i} E$ for $\delta_0 := \nu$, then
\begin{align}\label{eq:concentration_final_d}
\bbP(U\geq \nu E)\leq {\inp{|\cI|\beta+1}^{{\beta}^2}}e^{-\gamma/3}.
\end{align}
\end{duplicatelemma} 

 Let $T^n_{l}, l\in \{1, 2,3\}$ denote the type class of $P_{l}$. We generate independent random codebooks for each user. The codebook for user $l\in \{1,2,3\}$, denoted by $\inp{\vecX_{l1}, \vecX_{l2}, \ldots, \vecX_{lN_{l}}}$, consists of independent random vectors each distributed uniformly on $T^n_{l}$.  
Fix $(\vecx_i,\vecx_j, \vecx_k)  \in \mathcal{X}^n_i \times \mathcal{X}^n_j \times \mathcal{X}^n_k$and a joint type $P_{X_iX'_iX_jX'_jX_kX'_k}\in \cP^n_{\cX_{i}\times \cX_{i}\times\cX_{j}\times\cX_j\times\cX_k\times\cX_k}$ such that for $l\in \{1,2,3\}, P_{X_l} =P_{X'_l}= P_{l}$ and $(\vecx_i,\vecx_j, \vecx_k)\in T^n_{X_iX_jX_k}$.

In order to obtain ~\eqref{lemma_eq1a}~-~\eqref{lemma_eq6}, we will use $\nu = \exp{\inp{n\epsilon/2}}$ and $\gamma = \exp(n\epsilon/4\beta)$ in Lemma~\ref{theorem:concentration}. For $i\in [1:\beta]$, let $\delta_i = \exp\inp{\frac{(4\beta-3i)n\epsilon}{8\beta}}$. Note that $\delta_i\geq 1$ for all $i\in [1:\beta]$, $\delta_0 = \exp{\inp{n\epsilon/2}}= \nu$ and there exists $n_0$ s.t. for all $n\geq n_0$,
\begin{align*}
\delta_i&=\exp\inp{\frac{(4\beta-3i)n\epsilon}{8\beta}}\\
&<\inp{\exp\inp{\frac{(4\beta-3i+1)n\epsilon}{8\beta}}}\\
&=\frac{\delta_{i-1}}{\gamma}\\
&\approx\frac{1}{\inp{2\beta}^{\beta}}\inp{\frac{\delta_{i-1}-1}{2\gamma}-\frac{\beta!}{E}}\text{ for large n}.
\end{align*}

The choice of $\beta$, $\cI$, $\cJ$ and the random variables will depend on the specific statement among ~\eqref{lemma_eq1a}~-~\eqref{lemma_eq6}. Though, $\beta$ will only range in $\{1, 2, 3, 4\}$. 

\noindent{\pink{\em \underline{(i) Analysis of \eqref{lemma_eq1a}, \eqref{lemma_eq1b} and \eqref{lemma_eq7b_} (Joint typicality of a codeword)}}}\\
To obtain \eqref{lemma_eq1a}, choose $\cI = \{i1, i2, \ldots, iN_i\}$ and the set $\{\vecX_{i1}, \vecX_{i2}, \ldots, \vecX_{iN_i}\}$ corresponding to $\{Y_i:\, i\in \cI\}$. We choose $\beta = 1$ and $\cJ = \inb{\inp{i1}, \inp{i2}, \ldots, \inp{iN_i}}$. For all $r\in [1:N_i]$,
\begin{align*}
V_{\inp{ir}}=\begin{cases}&1,\text{ if }\vecX_{ir}\in T^n_{X_i|X_k}(\vecx_k),\\
& 0,\text{ otherwise.}
\end{cases}
\end{align*}
Note that
\begin{align*}
\bbP\inp{V_{\inp{ir}}=1}  &= \frac{|T^n_{X_i|X_k}(\vecx_k)|}{|T^n_{X_i}|}\\
&\leq \frac{\exp\inb{nH(X_i|X_k)}}{(n+1)^{|\cX_i|}\exp\inb{nH(X_i)}}\\
&= (n+1)^{-|\cX_i|}\exp\inb{-nI(X_i;X_k)}\\
&\stackrel{(a)}{\leq} \exp\inb{-nI(X_i;X_k)}
\end{align*}
where $(a)$ follows because $(n+1)^{-|\cX_i|}\leq 1$.\\
Note that $U = \sum_{r\in [1:N_i]}V_{\inp{ir}} = |\{r\in[1:N_i]: (\vecX_{ir},\vecx_{k})\in T^{n}_{X_i X_k} \}|$. Note that for the case of $\beta = 1$, condition \eqref{eq:cond_on_E} reduces to $E\geq \bbE[U]$. 
Thus,  $\bbE[U] = \sum_{r\in [1:N_i]}\bbE\insq{V_{\inp{ir}}} = \sum_{r\in [1:N_i]}\bbP\inp{V_{\inp{ir}}=1}$$ \leq \exp\left\{n\left(R_i-I(X_i;X_k)\right)\right\}$ $\leq \exp\left\{n\left|R_i-I(X_i;X_k)\right|^+\right\}:=E$. Thus, \eqref{eq:concentration_final_d} gives us
\begin{align}\label{eq:lemma_prob1}
\bbP\left(|\{r\in[1:N_i]: (\vecX_{ir},\vecx_{k})\in T^{n}_{X_i X_k} \}| \geq \exp\left\{n\left(|R_{i}-I(X_i;X_k)|^{+}+\epsilon/2\right)\right\}\right)\leq (N_i+1)e^{-\exp(n\epsilon/4)/3}.
\end{align}

Replacing $x_k$ with $(x_i, x_j, x_k)$, $X_k$ with $(X_i, X_j, X_k)$, and $X_i$ with $X'_i$ in the above argument, one can show that
\begin{align}\label{eq:lemma_prob2}
&\bbP\left(|\{u\in [1:N_i]: (\vecX_{iu},\vecx_i, \vecx_j, \vecx_k)\in T^{n}_{X_i'X_i X_j X_k} \}| \geq \exp\left\{n\left(|R_{i}-I(X_i';X_iX_j X_kX_k)|^{+}+\epsilon/2\right)\right\}\right)\nonumber  \\
&\qquad\qquad\qquad\qquad\qquad\qquad\qquad\qquad\qquad\qquad\qquad\qquad\qquad\qquad\qquad\qquad\qquad\leq (N_i+1)e^{-\exp(n\epsilon/4)/3}.
\end{align}

Similarly, choosing $\cI = \{j1, j2, \ldots, jN_j\}$, $\{Y_i:i\in \cI\}= \{\vecX_{j1}, \vecX_{j2}, \ldots, \vecX_{jN_j}\}$, $\cJ = \inb{\inp{j1}, \inp{j2}, \ldots, \inp{jN_j}}$ and replacing $\vecx_k$ with $(\vecx_i, \vecx_k)$, $X_k$ with $(X_i, X_k)$, and $R_i$ with $R_j$ in the proof of \eqref{eq:lemma_prob1}, we can show that
\begin{align}\label{eq:lemma_prob3}
&\bbP\left(|\{s\in [1:N_j]: (\vecx_{i},\vecX_{js}, \vecx_k)\in T^{n}_{X_i X_j X_k} \}| \geq \exp\left\{n\left(|R_{j}-I(X_j;X_iX_k)|^{+}+\epsilon/2\right)\right\}\right)\nonumber\\
&\qquad\qquad\qquad\qquad\qquad\qquad\qquad\qquad\qquad\qquad\qquad\qquad\qquad\qquad\qquad\qquad\qquad\leq (N_j+1)e^{-\exp(n\epsilon/4)/3}.
\end{align}

\noindent{\em \pink{{\underline{(ii) Analysis of \eqref{lemma_eq3a}~-~\eqref{lemma_eq7c_} (Joint typicality of a pair of codewords)}}}}\\
We will only analyse \eqref{lemma_eq3a} and \eqref{lemma_eq3b}. The analysis of other statements is similar. To show \eqref{lemma_eq3a}, choose $\cI = \{i1, i2, \ldots, iN_i\}\cup\{j1, j2, \ldots, jN_j\}$,\, $\{Y_i:i\in \cI\}= \{\vecX_{i1}, \vecX_{i2}, \ldots, \vecX_{iN_i}\}\cup\{\vecX_{j1}, \vecX_{j2}, \ldots, \vecX_{jN_j}\}$. For $\beta =2$, let $\cJ = \inb{\inp{ir, js}:(r, s)\in [1:N_i]\times[1:N_j]}$. For all $(r, s)\in [1:N_i]\times [1:N_j]$,
\begin{align*}
V_{\inp{ir, js}}=\begin{cases}&1,\text{ if }(\vecX_{ir}, \vecX_{js})\in T^n_{X_iX_j|X_k}(\vecx_k),\\
& 0,\text{ otherwise.}
\end{cases}
\end{align*}
This implies that $U = \sum_{(r, s)\in [1:N_i]\times[1:N_j]}V_{\inp{ir, js}} = \left|\left\{(r, s)\in[1:N_i]\times[1:N_j]: (\vecX_{ir},\vecX_{js},\vecx_{k})\in T^{n}_{X_iX_j X_k} \right\}\right|$.

Note that
\begin{align*}
\bbP\inp{V_{\inp{ir, js}}=1}  &= \frac{|T^n_{X_iX_j|X_k}(\vecx_k)|}{|T^n_{X_i}||T^n_{X_j}|}\\
&\leq \frac{\exp\inb{nH(X_iX_j|X_k)}}{(n+1)^{|\cX_i|+|\cX_j|}\exp\inb{n(H(X_i)+H(X_j)}}\\
&= (n+1)^{-\inp{|\cX_i|+|\cX_j|}}\exp\inb{-n\inp{H(X_iX_j)-H(X_iX_j|X_k) - H(X_iX_j) + H(X_i)+H(X_j)}}\\
&= (n+1)^{-\inp{|\cX_i|+|\cX_j|}}\exp\inb{-n\inp{I(X_iX_j;X_k) + I(X_i;X_j)}}\\
&{\leq} \exp\inb{-n\inp{I(X_iX_j;X_k) + I(X_i;X_j)}}\\
&{=} \exp\inb{-n\inp{I(X_j;X_k) + I(X_i;X_k|X_j) + I(X_i;X_j)}}\\
&{=} \exp\inb{-n\inp{I(X_j;X_k) + I(X_i;X_jX_k)}}
\end{align*}
Thus,  
\begin{align}
\bbE[U] = &\sum_{(r, s)\in [1:N_i]\times[1:N_j]}\bbE\insq{V_{\inp{ir, js}}} = \sum_{(r, s)\in [1:N_i]\times[1:N_j]}\bbP\inp{V_{\inp{ir, js}}=1}\nonumber\\
& \leq \exp\left\{n\left(R_i+R_j-I(X_i;X_jX_k) - I(X_j;X_k)\right)\right\}\nonumber\\
&\leq \exp\left\{n\left|\left|R_i-I(X_i;X_k)\right|^++\left|R_j-I(X_j;X_k)\right|^+-I(X_i;X_j|X_k)\right|^+\right\}:=E.\label{eq:like1}
\end{align}
 We need to show that  
for any $(ir, js)\in \cJ$, and $(\vecX_{ir}, \vecX_{js}) = (\vecx_{ir}, \vecx_{js})$, \\$E\geq \max\inp{\bbE\insq{\sum_{v\neq s}V_{(ir, jv)}|(\vecX_{ir}, \vecX_{js}) = (\vecx_{ir}, \vecx_{js})}, \bbE\insq{\sum_{u\neq r}V_{(iu, js)}|(\vecX_{ir}, \vecX_{js}) = (\vecx_{ir}, \vecx_{js})}}$. Note that
\begin{align}
&\bbE\insq{\sum_{v\neq s}V_{(ir, jv)}|(\vecX_{ir}, \vecX_{js}) = (\vecx_{ir}, \vecx_{js})}\nonumber\\
&= \sum_{v\neq s}\bbE\insq{V_{(ir, jv)}|(\vecX_{ir}, \vecX_{js}) = (\vecx_{ir}, \vecx_{js})}\nonumber\\
&\leq \sum_{v\neq s}\bbP\inp{V_{(ir, jv)} = 1|(\vecX_{ir}, \vecX_{js}) = (\vecx_{ir}, \vecx_{js})}\nonumber\\
&= \sum_{v\neq s}\bbP\inp{\vecX_{jv}\in T^n_{X_j|X_i X_k}(\vecx_{ir}, \vecx_k)}\nonumber\\
&= \sum_{v\neq s}\frac{|T^n_{X_j|X_iX_k}(\vecx_{ir}, \vecx_k)|}{|T^n_{X_j}|}\nonumber\\
&\leq \exp\inb{nR_j}\frac{\exp\inb{nH(X_j|X_iX_k)}}{(n+1)^{|\cX_j|}\exp\inb{nH(X_j)}}\nonumber\\
&\leq \exp\inb{n\inp{|R_j-I(X_j;X_iX_k)|^+}}\label{eq:like2}\\
&\leq E.\nonumber
\end{align} 
Similarly, we can show that $\bbE\insq{\sum_{u\neq r}V_{(iu, js)}|(\vecX_{ir}, \vecX_{js}) = (\vecx_{ir}, \vecx_{js})}\leq E$.

Thus, \eqref{eq:concentration_final_d} implies that
\begin{align}
&\bbP\Bigg(\left|\left\{(r, s)\in[1:N_i]\times[1:N_j]: (\vecX_{ir},\vecX_{js},\vecx_{k})\in T^{n}_{X_iX_j X_k} \right\}\right| \nonumber\\
&\qquad\geq \exp\left\{n\left(\left|\left|R_i-I(X_i;X_k)\right|^++\left|R_j-I(X_j;X_k)\right|^+-I(X_i;X_j|X_k)\right|^++\epsilon/2\right)\right\}\Bigg)\nonumber\\
&\leq ((N_i+N_j)2+1)^4e^{-\exp(n\epsilon/8)/3} \label{eq:lemma_prob3a}.
\end{align}

To show \eqref{lemma_eq3b}, let $\cI = \{i1, i2, \ldots, iN_i\}$ and $\{Y_i: i\in \cI\} = \{\vecX_{i1}, \vecX_{i2}, \ldots, \vecX_{iN_i}\}$. We choose $\beta = 2$ and $\cJ = \inb{\inp{ir, iu}:(r, u)\in [1:N_i]\times[1:N_i] \text{ such that } r\neq u}$. For all $(r, u)$ such that $r\in [1:N_i]$ and $u\in [1:N_i]\setminus\{r\}$,
\begin{align*}
V_{\inp{ir, iu}}=\begin{cases}&1,\text{ if }(\vecX_{ir}, \vecX_{iu})\in T^n_{X_iX_i'|X_k}(\vecx_k),\\
& 0,\text{ otherwise.}
\end{cases}
\end{align*}
By replacing $X_j$ with $X_i'$ in the proof of \eqref{eq:lemma_prob3a} and following similar arguments, we can show that
\begin{align}
&\bbP\Bigg(\left|\left\{(r, u)\in [1:N_i]\times[1:N_i]:(\vecx_{ir},\vecx_{iu}, \vecx_k)\in T^{n}_{X_i X_i' X_k},\, r\neq u\right\}\right| \nonumber\\
&\qquad\geq \exp\left\{n\left(\left|\left|R_i-I(X_i;X_k)\right|^++\left|R_i-I(X_i';X_k)\right|^+-I(X_i;X_i'|X_k)\right|^+ +\epsilon/2\right)\right\}\Bigg)\nonumber\\
&\leq (2N_i+1)^4  e^{-\exp(n\epsilon/8)/3}. \label{eq:lemma_prob3b}
\end{align}

\noindent{\underline{\pink{\em (iii) Analysis of \eqref{lemma_eq5} (Joint typicality of three codewords)}}}\\
Choose  $\cI = \{i1, i2, \ldots, iNi\}\cup\{j1, j2, \ldots, jNj\}$,\, $\{Y_i:i\in \cI\}= \{\vecX_{i1}, \vecX_{i2}, \ldots, \vecX_{iN_i}\}\cup\{\vecX_{j1}, \vecX_{j2}, \ldots, \vecX_{jN_j}\}$. For $\beta = 3$, \\
let $\cJ = \inb{\inp{ir, js, iu}:(r, s, u)\in [1:N_i]\times[1:N_j]\times[1:N_i], \, r\neq u}$. For all $\inp{ir, js, iu}\in \cJ$,
\begin{align*}
V_{\inp{ir, js, iu}}=\begin{cases}&1,\text{ if }(\vecX_{ir}, \vecX_{js}, \vecX_{iu})\in T^n_{X_iX_jX_i'|X_k}(\vecx_k),\\
& 0,\text{ otherwise.}
\end{cases}
\end{align*}
Therefore, $U = \sum_{(ir, js, iu)\in \cJ}V_{\inp{ir, js, iu}} = |\{(r, s, u)\in [1:N_i]\times[1:N_j]\times[1:N_i]:(\vecX_{ir},\vecX_{js},\vecX_{iu}, \vecx_k)\in T^{n}_{X_i X_j X_i'X_k}, \, r\neq u\}|.$

Note that
\begin{align*}
\bbP\inp{V_{\inp{ir, js, iu}}=1}  &= \frac{|T^n_{X_iX_jX_i'|X_k}{(\vecx_k)}|}{|T^n_{X_i}||T^n_{X_j}||T^n_{X_i'}|}\\
&\leq \frac{\exp\inb{nH(X_iX_jX_i'|X_k)}}{(n+1)^{2|\cX_i|+|\cX_j|}\exp\inb{n(H(X_i)+H(X_j)+H(X_i')}}\\
&= (n+1)^{-\inp{2|\cX_i|+|\cX_j|}}\exp\inb{n\inp{H(X_iX_jX_i'X_k)- H(X_i)-H(X_j)-H(X_i')-H(X_k)}}\\
&\stackrel{(a)}{=} (n+1)^{-\inp{|\cX_i|+|\cX_j|}}\exp\inb{-n\inp{C(X_i;X_j;X_i';X_k)}}\\
&{\leq} \exp\inb{-n\inp{C(X_i;X_j;X_i';X_k)}}
\end{align*}
where $(a)$ follows from \eqref{defn:tot_corr}.

Note that 
\begin{align}
\bbE[U] = &\sum_{(ir, js, iu)\in \cJ}\bbE\insq{V_{\inp{ir, js, iu}}} = \sum_{(ir, js, iu)\in \cJ}\bbP\inp{V_{\inp{ir, js, iu}}=1}\\
 &\leq \exp\left\{n\left(2R_i+R_j-C(X_i;X_j;X_i';X_k))\right)\right\}\\
 &\leq \max_{\cS\subseteq \{1, 2, 3\}}\exp\left\{n\left({g_{X_i, X_j, X_i', X_k}^{\cS}(R_i, R_j, R_i)}\right)\right\}\label{eq:like3}.
 \end{align}
This is because for $\cS = \{1,2,3\}$, $\exp\left\{n\left({g_{X_i, X_j, X_i', X_k}^{\cS}(R_i, R_j, R_i)}\right)\right\}$$=\exp\left\{n\left(2R_i+R_j-C(X_i;X_j;X_i';X_k))\right)\right\}$. \\ Let    $E:=\max_{\cS\subseteq \{1, 2, 3\}}\exp\left\{n\left({g_{X_i, X_j, X_i', X_k}^{\cS}(R_i, R_j, R_i)}\right)\right\}$. Using similar arguments as the ones used to obtain \eqref{eq:like1} and \eqref{eq:like2}, we can show that  
for any $(ir, js, iu)\in \cJ$, and $(\vecX_{ir}, \vecX_{js}, \vecX_{iu}) = (\vecx_{ir}, \vecx_{js}, \vecx_{iu})$, \\
\begin{align*}
E&\geq \max\Bigg(\bbE\insq{\sum_{v\neq s}V_{(ir, jv, iu)}|(\vecX_{ir}, \vecX_{js}, \vecX_{iu})= (\vecx_{ir}, \vecx_{js}, \vecx_{iu})},\\
 &\qquad\qquad\bbE\insq{\sum_{r'\notin \{r,u\}}V_{(ir', js, iu)}|(\vecX_{ir}, \vecX_{js}, \vecX_{iu}) = (\vecx_{ir}, \vecx_{js}, \vecx_{iu})},\\
 &\qquad \qquad\bbE\insq{\sum_{u'\notin \{r,u\}}V_{(ir, js, iu')}|(\vecX_{ir}, \vecX_{js}, \vecX_{iu}) = (\vecx_{ir}, \vecx_{js}, \vecx_{iu})},\\
 &\qquad \qquad\bbE\insq{\sum_{r', u'\notin \{r,u\}}V_{(ir', js, iu')}|(\vecX_{ir}, \vecX_{js}, \vecX_{iu}) = (\vecx_{ir}, \vecx_{js}, \vecx_{iu})},\\
 & \qquad \qquad\bbE\insq{\sum_{r'\notin \{r,u\}, v\neq s}V_{(ir', jv, iu)}|(\vecX_{ir}, \vecX_{js}, \vecX_{iu}) = (\vecx_{ir}, \vecx_{js}, \vecx_{iu})},\\
 & \qquad \qquad\bbE\insq{\sum_{u'\notin \{r,u\}, v\neq s}V_{(ir, jv, iv)}|(\vecX_{ir}, \vecX_{js}, \vecX_{iu}) = (\vecx_{ir}, \vecx_{js}, \vecx_{iu})}\Bigg).
 \end{align*}
Thus, we can use \eqref{eq:concentration_final_d} to obtain
\begin{align}
&\bbP\left(|\{(r, s, u)\in [1:N_i]\times[1:N_j]\times[1:N_i]:(\vecX_{ir},\vecX_{js},\vecX_{iu}, \vecx_k)\in T^{n}_{X_i X_j X_i'X_k}, \, r\neq u\}|\right.\nonumber\\
&\qquad\qquad\qquad\qquad> \left.\max_{\cS\subseteq \{1, 2, 3\}}\exp\left\{n\left({g_{X_i, X_j, X_i', X_k}^{\cS}(R_i, R_j, R_i)}+\epsilon/2\right)\right\}\right)\nonumber\\ 
&\leq (3(N_i+N_j)+1)^9 e^{-\exp(n\epsilon/12)/3}.   \label{eq:lemma_prob5}
\end{align}

\noindent{\underline{\pink{\em (iv) Analysis of \eqref{lemma_eq4}  and \eqref{lemma_eq6} (Joint typicality of four codewords)}}}\\
We will start with analysis of \eqref{lemma_eq6}. Choose $\cI = \{i1, i2, \ldots, iN_i\}\cup\{j1, j2, \ldots, jN_j\}\cup\{k1, k2, \ldots, kN_k\}$,\, $\{Y_i:i\in \cI\}= \{\vecX_{i1}, \vecX_{i2}, \ldots, \vecX_{iN_i}\}\cup\{\vecX_{j1}, \vecX_{j2}, \ldots, \vecX_{jN_j}\}\cup\{\vecX_{k1}, \vecX_{k2}, \ldots, \vecX_{kN_k}\}$. For $\beta = 4$, \\
let $\cJ = \inb{\inp{ir, js, iu, kw}:(r, s, u, w)\in [1:N_i]\times[1:N_j]\times[1:N_i]\times[1:N_k], \, r\neq u}$. For all $\inp{ir, js, iu, kw}\in \cJ$,
\begin{align*}
V_{\inp{ir, js, iu, kw}}=\begin{cases}&1,\text{ if }(\vecX_{ir}, \vecX_{js}, \vecX_{iu}, \vecX_{kw})\in T^n_{X_iX_jX_i'X_k'|X_k}(\vecx_k),\\
& 0,\text{ otherwise.}
\end{cases}
\end{align*}
Therefore, $U = \sum_{(ir, js, iu, kw)\in \cJ}V_{\inp{ir, js, iu, kw}} = |\{(r, s, u, w)\in [1:N_i]\times[1:N_j]\times[1:N_i]\times[1:N_k]:(\vecX_{ir}, \vecX_{js}, \vecX_{iu}, \vecX_{jv},\vecx_k)\in T^{n}_{X_iX_j X_i^{'} X'_k X_k}, \, r\neq u\}|.$

Note that
\begin{align*}
\bbP\inp{V_{\inp{ir, js, iu, kw}}=1}  &= \frac{|T^n_{X_iX_jX_i'X_k'|X_k}(\vecx_k)|}{|T^n_{X_i}||T^n_{X_j}||T^n_{X_i'}||T^n_{X_k'}|}\\
&\leq \frac{\exp\inb{nH(X_iX_jX_i'X_k'|X_k)}}{(n+1)^{2|\cX_i|+|\cX_j|+|\cX_k|}\exp\inb{n(H(X_i)+H(X_j)+H(X_i')+H(X_k')}}\\
&= (n+1)^{-\inp{2|\cX_i|+|\cX_j|}}\exp\inb{n\inp{H(X_iX_jX_i'X_k'X_k)- H(X_i)-H(X_j)-H(X_i')-H(X_k')-H(X_k)}}\\
&\stackrel{(a)}{=} (n+1)^{-\inp{|\cX_i|+|\cX_j|}}\exp\inb{-n\inp{C(X_i;X_j;X_i';X_k';X_k)}}\\
&{\leq} \exp\inb{-n\inp{C(X_i;X_j;X_i';X_k';X_k)}}
\end{align*}
where $(a)$ follows from \eqref{defn:tot_corr}.

Note that $\bbE[U] = \sum_{(ir, js, iu, kw)\in \cJ}\bbE\insq{V_{\inp{ir, js, iu, kw}}} = \sum_{(ir, js, iu, kw)\in \cJ}\bbP\inp{V_{\inp{ir, js, iu, kw}}=1}$\\$ \leq \exp\left\{n\left(2R_i+R_j+R_k-C(X_i;X_j;X_i';X_k';X_k))\right)\right\}$$\leq \max_{\cS\subseteq \{1, 2, 3, 4\}}\exp\left\{n\left({g_{X_i, X_j, X_i', X_k',X_k}^{\cS}(R_i, R_j, R_i, R_k)}\right)\right\}$. This is because for $\cS = \{1,2,3, 4\}$, $\exp\left\{n\left({g_{X_i, X_j, X_i',X_k', X_k}^{\cS}(R_i, R_j, R_i,R_k)}\right)\right\}$$=\exp\left\{n\left(2R_i+R_j+R_k-C(X_i;X_j;X_i';X_k';X_k))\right)\right\}$. \\ Using similar arguments as used on \eqref{eq:like1}, \eqref{eq:like2} and \eqref{eq:like3}, one can show that \\$E:=\max_{\cS\subseteq \{1, 2, 3, 4\}}\exp\left\{n\left({g_{X_i, X_j, X_i',X_k', X_k}^{\cS}(R_i, R_j, R_i, R_k)}\right)\right\}$ satisfies condition \eqref{eq:cond_on_E}.
Using \eqref{eq:concentration_final_d}, we obtain that
\begin{align}
&\bbP\left(|\{(r, s, u, w)\in [1:N_i]\times[1:N_j]\times[1:N_i]\times[1:N_k]:(\vecX_{ir},\vecX_{js},\vecX_{iu},\vecX_{kw}, \vecx_k)\in T^{n}_{X_iX_j X_i'X_k' X_k}, \, r\neq u\}|\right.\nonumber\\
&\qquad\qquad\qquad\qquad> \left.\max_{\cS\subseteq \{1, 2, 3, 4\}}\exp\left\{n\left({g_{X_i, X_j, X_i', X_k',X_k}^{\cS}(R_i, R_j, R_i, R_k)}+\epsilon/2\right)\right\}\right)\nonumber\\ 
&\leq (4(N_i+N_j+N_k)+1)^{16} e^{-\exp(n\epsilon/16)/3}.   \label{eq:lemma_prob6}
\end{align}

The analysis of  \eqref{lemma_eq4} is very similar to that of \eqref{lemma_eq6}. For \eqref{lemma_eq4}, we choose  $\cI = \{i1, i2, \ldots, iN_i\}\cup\{j1, j2, \ldots, jN_j\}$,\, $\{Y_i:i\in \cI\}= \{\vecX_{i1}, \vecX_{i2}, \ldots, \vecX_{iN_i}\}\cup\{\vecX_{j1}, \vecX_{j2}, \ldots, \vecX_{jN_j}\}$. For $\beta = 4$, \\
let $\cJ = \inb{\inp{ir, js, iu, jv}:(r, s, u, v)\in [1:N_i]\times[1:N_j]\times[1:N_i]\times[1:N_j], \, r\neq u,\, s\neq v}$. For all $\inp{ir, js, iu, jv}\in \cJ$,
\begin{align*}
V_{\inp{ir, js, iu, jv}}=\begin{cases}&1,\text{ if }(\vecX_{ir}, \vecX_{js}, \vecX_{iu}, \vecX_{jv})\in T^n_{X_iX_jX_i'X_j'|X_k}(\vecx_k),\\
& 0,\text{ otherwise.}
\end{cases}
\end{align*}
We follow the same analysis as done for \eqref{lemma_eq6} to obtain
\begin{align}
&\bbP\left(|\{(r, s, u, v)\in [1:N_i]\times[1:N_j]\times[1:N_i]\times[1:N_j]:(\vecX_{ir},\vecX_{js},\vecX_{iu},\vecX_{jv}, \vecx_k)\in T^{n}_{X_i X_j X_i'X_j' X_k}, \, r\neq u,, s\neq v\}|\right.\nonumber\\
&\qquad\qquad\qquad\qquad> \left.\max_{\cS\subseteq \{1, 2, 3, 4\}}\exp\left\{n\left({g_{X_i, X_j, X_i', X_j',X_k}^{\cS}(R_i, R_j, R_i, R_j)}+\epsilon/2\right)\right\}\right)\nonumber\\ 
&\leq (4(N_i+N_j)+1)^{16} e^{-\exp(n\epsilon/16)/3}.   \label{eq:lemma_prob4}
\end{align}

\end{proof}
\pink{Lemma~\ref{codebook} gives Lemma~\ref{lemma:codebook} as a corollary, which we prove in the next subsection.}
\subsection{Codebook}\label{new_appendix:codebook}
\pink{We restate Lemma~\ref{lemma:codebook} below and show how it follows from Lemma~\ref{codebook}.}
\begin{duplicatelemma}
For any  $\epsilon>0,\,  n\geq n_0(\epsilon), \, N_1, N_2, N_3\geq\exp(n\epsilon)$ and types $P_1\in \cP_{\cX_1}^n$, $ P_2\in \cP_{\cX_2}^n$,  $P_3\in \cP_{\cX_3}^n$, there exists codebooks $\left\{\vecx_{11}, \ldots,  \vecx_{1N_1}\in \mathcal{X}_1^n\right\}, \left\{\vecx_{21}, \ldots,  \vecx_{2N_2}\in \mathcal{X}_2^n\right\}$, $ \left\{\vecx_{31}, \ldots,  \vecx_{3N_3}\in \mathcal{X}_3^n\right\}$ whose codewords are of type $P_1, P_2$, $P_3$ respectively such that for every permutation $(i,j,k)$ of $(1,2,3)$; for every $(\vecx_i,\vecx_j, \vecx_k)  \in \mathcal{X}^n_i \times \mathcal{X}^n_j \times \mathcal{X}^n_k$; for every joint type $P_{X_iX'_iX_jX'_jX_kX'_k}\in \cP^n_{\cX_{i}\times \cX_{i}\times\cX_{j}\times\cX_j\times\cX_k\times\cX_k}$; and for $R_i \defineqq (1/ n )\log_2{N_i},R_j \defineqq (1/ n )\log_2{N_j}$, and $R_k \defineqq (1/ n )\log_2{N_k}$; the following holds:
\begin{align}
&|\{u\in [1:N_i]:(\vecx_{iu}, \vecx_i, \vecx_j, \vecx_k)\in T^{n}_{X'_i  X_i X_j X_k}\}|\leq\exp{\left(n\left(\left|R_1-I(X_1';X_1X_2X_3)\right|^{+}+\epsilon/2\right)\right)};\label{lemma_eq7b_app}\\
&|\{(u, v)\in [1:N_i]\times[1:N_j]:(\vecx_{iu},\vecx_{jv}, \vecx_i, \vecx_j, \vecx_k)\in T^{n}_{X'_i X'_j X_i X_j X_k}\}|\nonumber\\
&\qquad\qquad\leq\exp{\left(n\left(\left||R_i-I(X_i';X_iX_jX_k)|^{+}+|R_j-I(X_j';X_iX_jX_k)|^{+}-I(X_i';X_j'|X_iX_jX_k)\right|^{+}+\epsilon/2\right)\right)};\label{lemma_eq7a_app}\\
&|\{(u, w)\in [1:N_i]\times[1:N_k]:(\vecx_{iu},\vecx_{kw}, \vecx_i, \vecx_j, \vecx_k)\in T^{n}_{X'_iX'_jX_iX_j X_k}\}|\nonumber\\
&\qquad\qquad\leq\exp{\left(n\left(\left||R_i-I(X_i';X_iX_jX_k)|^{+}+|R_k-I(X_k';X_iX_jX_k)|^{+}-I(X_i';X_k'|X_iX_jX_k)\right|^{+} + \epsilon/2\right)\right)};\label{lemma_eq7c_app}\\
&\frac{1}{N_iN_j}|\{(r,s)\in [1:N_i]\times[1:N_j]: (\vecx_{ir},\vecx_{js}, \vecx_k)\in T^{n}_{X_i X_j X_k} \}| < \exp\left(-\frac{n\epsilon}{2}\right), \text{ if } I(X_i;X_k)+I(X_j;X_iX_k)\geq\epsilon; \label{code_eq2_app}\\
&\frac{1}{N_iN_j}|\{(r, s)\in [1:N_i]\times[1:N_j]: \exists \,(u, v)\in  [1:N_i]\times[1:N_j], \, u\neq r, v\neq s, \, (\vecx_{ir}, \vecx_{js}, \vecx_{iu}, \vecx_{jv},\vecx_k)\in T^{n}_{X_iX_j X_i^{'} X'_j X_k} \}|\nonumber\\
&< \exp\left(-\frac{n\epsilon}{2}\right),\nonumber\\
&\qquad\text{if }I(X_i;X_jX'_iX'_jX_k)+I(X_j;X_i'X_j'X_k)\geq \left|\left|R_i-I(X'_i;X_k)\right|^+ +\left|R_j-I(X'_j;X_k)\right|^+-I(X'_i;X'_j|X_k)\right|^+ + \epsilon; \label{code_eq3_app}\\
&\frac{1}{N_iN_j}|\{(r, s)\in [1:N_i]\times[1:N_j]: \exists\, (u, w)\in  [1:N_i]\times[1:N_k], \, u\neq r, \, (\vecx_{ir}, \vecx_{js}, \vecx_{iu}, \vecx_{kw}, \vecx_k)\in T^{n}_{X_iX_j X_i^{'} X'_k X_k} \}|\nonumber\\
&\qquad\qquad<\exp\left(-\frac{n\epsilon}{2}\right),\nonumber\\
&\qquad\text{if }I(X_i;X_jX'_iX'_kX_k)+I(X_j;X_i'X_k'X_k)\geq \left|\left|R_i-I(X'_i;X_k)\right|^+ +\left|R_k-I(X'_k;X_k)\right|^+-I(X'_i;X'_k|X_k)\right|^+ + \epsilon;\label{code_eq5_app}\\
&\frac{1}{N_iN_j}|\{(r, s)\in [1:N_i]\times[1:N_j]: \exists\, u \in [1:N_i], \, u\neq r,\, (\vecx_{ir}, \vecx_{js}, \vecx_{iu},\vecx_k)\in T^{n}_{X_iX_j X_i^{'} X_k} \}|< \exp\left(-\frac{n\epsilon}{2}\right),\nonumber\\
&\qquad\text{if }I(X_i;X_jX'_iX_k)+I(X_j;X_i'X_k)\geq \left|R_i-I(X'_i;X_k)\right|^+  + \epsilon. \label{code_eq4_app}
\end{align}
\end{duplicatelemma}
\begin{proof}
\pink{We will use the following identity, which follows from chain rule, throughout the proof. For random variables $U$, $V$ and $W$, the following holds,
\begin{align}\label{idenity:MI}
I(U;W)+I(V;W)+I(U;V|W)=I(U;VW)+I(V;W)=I(U;W)+I(V;UW).
\end{align}This identity is not often mentioned explicitly while using.}
Statements \eqref{lemma_eq7b_app}-\eqref{lemma_eq7c_app} are statements \eqref{lemma_eq7b_}, \eqref{lemma_eq7a_} and \eqref{lemma_eq7c_}, restated directly from Lemma~\ref{codebook}.
To show \eqref{code_eq2_app}, we divide the LHS and RHS of \eqref{lemma_eq3a} by $N_iN_j$ and substitute the expression for the notation given by \eqref{eq: notation} to obtain
\begin{align*}&\frac{1}{N_iN_j}|\{(r, s)\in [1:N_i]\times[1:N_j]:( \vecx_{ir}, \vecx_{js}, \vecx_k)\in T^{n}_{X_iX_jX_k}\}|\nonumber\\
&\qquad\qquad< \exp\left\{n\inp{\left|\left|R_i-I(X_i;X_k)\right|^+ +\left|R_j-I(X_j;X_k)\right|^+ -I(X_i;X_j|X_k)\right|^+ +\epsilon/2-R_i-R_j}\right\}. 
\end{align*}
We will evaluate the RHS for different values of $R_i$ and $R_j$. We see that when $R_i\leq I(X_i;X_k)$,  the RHS is 
\begin{align*}
&\exp\left\{n\inp{\left|R_j -I(X_j;X_iX_k)\right|^+-R_j +\epsilon/2-R_i}\right\}\\
&\leq \exp\left\{n\inp{\epsilon/2-R_i}\right\}\\
&\stackrel{\text{(a)}}{\leq} \exp\left\{-n\epsilon/2\right\},
\end{align*}
where (a) holds because $R_i>\epsilon$. Similarly, we can show the same upper bound of $\exp\left\{n\inp{-\epsilon/2}\right\}$ when $R_j\leq I(X_j;X_k)$. When  $R_i+R_j\leq I(X_i;X_k) + I(X_j;X_k) + I(X_i;X_j|X_k)$, the RHS is upper bounded by $\exp(n(\epsilon/2-R_i-R_j))\leq \exp(-n\epsilon/2)$. When $R_i> I(X_i;X_k)$, $R_j>I(X_j;X_k)$ and $R_i+R_j> I(X_i;X_k) + I(X_j;X_k) + I(X_i;X_j|X_k)$, the RHS evaluates to \pink{(using \eqref{idenity:MI})}
\begin{align*}
&\exp\left\{n\inp{-I(X_i;X_k)-I(X_j;X_iX_k) +\epsilon/2}\right\}\\
&\leq \exp\left\{-n\epsilon/2\right\},\qquad \text{if }I(X_i;X_k)+I(X_j;X_iX_k)\geq\epsilon.
\end{align*}

Next, we will prove \eqref{code_eq4_app}. For this, we will first show that
\begin{align}
&\frac{1}{N_iN_j}|\{(r, s)\in [1:N_i]\times[1:N_j]: \exists\, u \in [1:N_i], \, u\neq r,\, (\vecx_{ir}, \vecx_{js}, \vecx_{iu},\vecx_k)\in T^{n}_{X_iX_j X_i^{'} X_k} \}|<\exp\left(-\frac{n\epsilon}{2}\right),\label{eq:fall2}\\
&\text{ if any one of the following hold:}\nonumber\\
&\qquad R_i+R_j-\max_{\cS\subseteq \{1, 2, 3\}}{g_{X_i, X_j, X_i', X_k}^{S}(\pink{R_i, R_j, R_i})}\geq \epsilon\label{code_eq4_1},\\
&\qquad R_i+R_j - \left||R_i-I(X_i;X_k)|^++|R_j-I(X_j;X_k)|^+-I(X_i;X_j|X_k)\right|^+ \geq \epsilon. \label{code_eq_2}
\end{align}
We now show that \eqref{code_eq4_1} implies \eqref{eq:fall2}:
\begin{align*}
&\frac{1}{N_iN_j}|\{(r, s)\in [1:N_i]\times[1:N_j]: \exists\, u \in [1:N_i], \, u\neq r,\, (\vecx_{ir}, \vecx_{js}, \vecx_{iu},\vecx_k)\in T^{n}_{X_iX_j X_i^{'} X_k} \}|\\
&\stackrel{(a)}{\leq} \exp\left(-n\left(R_i+R_j-\max_{\cS\subseteq \{1, 2, 3\}}{g_{X_i, X_j, X_i', X_k}^{S}(\pink{R_i, R_j, R_i})}-\epsilon/2\right)\right)\\
&\leq \exp\left(-\frac{n\epsilon}{2}\right), \qquad\text{ if } R_i+R_j-\max_{\cS\subseteq \{1, 2, 3\}}{g_{X_i, X_j, X_i', X_k}^{S}(\pink{R_i, R_j, R_i})}\geq \epsilon,
\end{align*}
where $(a)$ follows from \eqref{lemma_eq5}. 
Next, we show that \eqref{code_eq_2} implies \eqref{eq:fall2}.
\begin{align*}
&\frac{1}{N_iN_j}|\{(r, s)\in [1:N_i]\times[1:N_j]: \exists\, u\in  [1:N_i], \, u\neq r,\, (\vecx_{ir}, \vecx_{js}, \vecx_{iu},\vecx_k)\in T^{n}_{X_iX_j X_i^{'} X_k} \}|\\
&\leq\frac{1}{N_iN_j}|\{(r, s)\in [1:N_i]\times[1:N_j]:  (\vecx_{ir}, \vecx_{js},\vecx_k)\in T^{n}_{X_iX_j X_k} \}|\\
&\stackrel{(a)}{\leq} \exp\left(-n\left(R_i+R_j-\left||R_i-I(X_i;X_k)|^++|R_j-I(X_j;X_k)|^+-I(X_i;X_j|X_k)\right|^+-\epsilon/2\right)\right)\\
&\leq \exp\left(-\frac{n\epsilon}{2}\right), \qquad\text{ if } R_i+R_j - \left||R_i-I(X_i;X_k)|^++|R_j-I(X_j;X_k)|^+-I(X_i;X_j|X_k)\right|^+ \geq \epsilon,
\end{align*}
where $(a)$ follows from \eqref{lemma_eq3a}.
Now, we will show that the condition in \eqref{code_eq4_app} implies at least one of \eqref{code_eq4_1} or \eqref{code_eq_2}. 
We restate the condition in \eqref{code_eq4_app} below for quick reference.
\begin{align}
I(X_i;X_jX'_iX_k)+I(X_j;X_i'X_k)\geq \left|R_i-I(X'_i;X_k)\right|^+  + \epsilon\label{code_eq4_codn}
\end{align}
To show that \eqref{code_eq4_codn} implies at least one of \eqref{code_eq4_1} or \eqref{code_eq_2}, we will do case analysis based on the value of $\hat{\cS}\defineqq \text{argmax}_{\cS}{g_{X_i, X_j, X_i', X_k}^{\cS}(R_i, R_j, R_i)}$, the set of maximizers of the expression  ${g_{X_i, X_j, X_i', X_k}^{{\cS}}(R_i, R_j, R_i)}$ in \eqref{code_eq4_1}. Evaluations of ${g_{X_i, X_j, X_i', X_k}^{\tilde{\cS}}(R_i, R_j, R_i, R_j)}$ under different values of $\tilde{\cS}$ are provided in Table~\ref{table:g2}. 
The table also gives the implications when $\tilde{S}\in \text{argmax}_{\cS}{g_{X_i, X_j, X_i', X_k}^{\cS}(R_i, R_j, R_i)}$ in the fourth column. For example, the $\Circled{6}$ row considers the case of $\{1,3\}\in\text{argmax}_{\cS}{g_{X_i, X_j, X_i', X_k}^{\cS}(R_i, R_j, R_i)}$. Under this case, we have, for instance, ${g_{X_i, X_j, X_i', X_k}^{\{1,3\}}(R_i, R_j, R_i)}\geq {g_{X_i, X_j, X_i', X_k}^{\{1,2,3\}}(R_i, R_j, R_i)}$, {\em i.e.}, $R_i-I(X_i;X_jX_k)+R_i-I(X'_i;X_iX_jX_k)\geq R_j - I(X_j;X_k) + R_i-I(X_i;X_jX_k)+R_i-I(X'_i;X_iX_jX_k)$. Hence, $R_j\leq I(X_j;X_k)$. 
This implication is given in the fifth column of the table against the ``reason'' $\Circled{6}\geq \Circled{8}$ where $\Circled{8}$ is the row corresponding to $\cS = \{1,2,3\}$.
The other implications are also easy to see from the table.
Instead of providing all the implications, the table only provide the ones which we will use in the proof of \eqref{code_eq4_app}. 
\paragraph*{\underline{Case 1: $\tilde{\cS}\in \hat{\cS}$ such that $|\tilde{\cS}| \leq 1$}}
In this case, \eqref{code_eq4_1} holds as $R_i, R_j\geq \epsilon$.
\paragraph*{\underline{Case 2: $\tilde{\cS}\in \hat{\cS}$ such that $|\tilde{\cS}| = 2$}}
If $\{1,2\}\in \hat{\cS}$, then it can be seen from the expression of ${g_{X_i, X_j, X_i', X_k}^{\{1,2\}}(R_i, R_j, R_i)}$ from Table~\ref{table:g2} that \eqref{code_eq4_codn} implies \eqref{code_eq4_1}. If $\{1,3\}\in \hat{\cS}$, then $R_j\leq I(X_j;X_k)$. This implies that the LHS of \eqref{code_eq_2} evaluates to
\begin{align*}
R_i+R_j - \left|R_i-I(X_i;X_jX_k)\right|^+ 
\end{align*} 
which is at least $\epsilon$ because $R_j\geq \epsilon$. Thus, \eqref{code_eq_2} holds. Similarly, when $\{2, 3\}\in \hat{\cS}$, one can use the fact that $R_i\leq I(X_i;X_k)$ and $R_i\geq \epsilon$ to show that \eqref{code_eq_2} holds. 
\paragraph*{\underline{Case 3: $\tilde{\cS}\in \hat{\cS}$ such that $|\tilde{\cS}| = 3$}}
In this case, $R_i\geq I(X'_i;X_k)$. Thus, conditions $\eqref{code_eq4_codn}$ and $\eqref{code_eq4_1}$ are same. Thus, $\eqref{code_eq4_codn}$ implies $\eqref{code_eq4_1}$. 

\begin{table}[h]
\begin{footnotesize}
\begin{center}
\begin{tabular}{ c|c|m{21em}|m{21em}|c} 
Index & $\tilde{\cS}$ &\centering  ${g_{X_i, X_j, X_i', X_k}^{\tilde{\cS}}(R_i, R_j, R_i)}$ & \centering Implications of \\$\tilde{\cS} \in \text{argmax}_{\cS}{g_{X_i, X_j, X_i', X_k}^{S}(R_i, R_j, R_i)}$&reasons \\
\hline
\hline
{\Circled{1}} & $\emptyset$ &$0$ & &\\
\hline	
{\Circled{2}} & $\{1\}$ &$R_i-I(X_i;X_jX'_iX_k)$ & &\\
\hline	
{\Circled{3}} & $\{2\}$ &$R_j-I(X_j;X_iX'_iX_k)$ & &\\
\hline	
{\Circled{4}} & $\{3\}$ &$R_i-I(X'_i;X_iX_jX_k)$ & &\\
\hline
{\Circled{5}} & $\{1,2\}$ &$R_i-I(X_i;X_i'X_k)+R_j-I(X_j;X_iX_i'X_k)$ & &\\
\hline	
\Circled{6} & {$\{1,3\}$} & {$R_i-I(X_i;X_jX_k)+R_i-I(X'_i;X_iX_jX_k)$} & $R_j\leq I(X_j;X_k)$&{$\Circled{6}\geq \Circled{8}$}\\
\hline
{\Circled{7}} & {$\{2,3\}$} & {$R_j-I(X_j;X_iX_k)+R_i-I(X'_i;X_iX_jX_k)$} & $R_i\leq I(X_i;X_k)$&{$\Circled{7}\geq \Circled{8}$}\\
\hline
{\Circled{8}} & {$\{1,2,3\}$} & {$R_i-I(X_i;X_k)+R_j-I(X_j;X_iX_k)+R_i-I(X_i';X_iX_jX_k)$} & $R_i\geq I(X'_i;X_k)$&{$\Circled{8}\geq \Circled{5}$}\\
\end{tabular}
\end{center}
\end{footnotesize}
\caption{Table showing different evaluations of $\max_{\cS\subseteq \{1, 2, 3\}}{g_{X_i, X_j, X_i',  X_k}^{S}(R_i, R_j, R_i)}$ and their implications.}
\label{table:g2}
\end{table}

Now, we will now show \eqref{code_eq3_app}. Its proof is very similar to the proof of \eqref{code_eq4_app}.
To show \eqref{code_eq3_app}, we will first show that
\begin{align}
&\frac{1}{N_iN_j}|\{(r, s)\in [1:N_i]\times[1:N_j]: \exists \,(u, v)\in  [1:N_i]\times[1:N_j], \, u\neq r, v\neq s, \, (\vecx_{ir}, \vecx_{js}, \vecx_{iu}, \vecx_{jv},\vecx_k)\in T^{n}_{X_iX_j X_i^{'} X'_j X_k} \}|\nonumber \\
&< \exp\left(-\frac{n\epsilon}{2}\right)\label{eq:fall1}\\
&\text{ if any one of the following inequalities hold:}\nonumber\\
&\qquad\qquad\qquad\qquad\quad R_i+R_j-\max_{\cS\subseteq \{1, 2, 3, 4\}}{g_{X_i, X_j, X_i', X_j', X_k}^{S}(R_i, R_j, R_i, R_j)}\geq \epsilon, \label{eq:cond}\\
&\qquad\qquad\qquad\qquad\quad R_j - \left|\left|R_j-I(X_j;X_k)\right|^+ +\left|R_i-I(X'_i;X_k)\right|^+-I(X_j;X'_i|X_k)\right|^+\geq \epsilon, \label{eq:cond1}\\
&\qquad\qquad\qquad\qquad\quad R_j - 
\left|\left|R_j-I(X_j;X_k)\right|^+ +\left|R_j-I(X'_j;X_k)\right|^+-I(X_j;X'_j|X_k)\right|^+\geq \epsilon, \label{eq:cond2}\\
&\qquad\qquad\qquad\qquad\quad R_i - \left|\left|R_i-I(X_i;X_k)\right|^+ +\left|R_i-I(X'_i;X_k)\right|^+-I(X_i;X'_i|X_k)\right|^+ \geq \epsilon, \label{eq:cond3}\\
&\qquad\qquad\qquad\qquad\quad R_i - \left|\left|R_i-I(X_i;X_k)\right|^+ +\left|R_j-I(X'_j;X_k)\right|^+-I(X_i;X'_j|X_k)\right|^+\geq \epsilon, \label{eq:cond4}\\
&\qquad\qquad\qquad\qquad\quad R_i+R_j - \left|\left|R_i-I(X_i;X_k)\right|^+ +\left|R_j-I(X_j;X_k)\right|^+-I(X_i;X_j|X_k)\right|^+\geq \epsilon. \label{eq:cond5}
\end{align}
The fact that \eqref{eq:cond} implies \eqref{eq:fall1} can be shown as follows:
\begin{align*}
\frac{1}{N_iN_j}&|\{(r, s)\in [1:N_i]\times[1:N_j]: \exists \,(u, v)\in  [1:N_i]\times[1:N_j], \, u\neq r, v\neq s, \, (\vecx_{ir}, \vecx_{js}, \vecx_{iu}, \vecx_{jv},\vecx_k)\in T^{n}_{X_iX_j X_i^{'} X'_j X_k} \}|\\
{\leq}&\frac{1}{N_iN_j}|\{(r, s, u , v)\in [1:N_i]\times[1:N_j]\times[1:N_i]\times[1:N_j]:  u\neq r, v\neq s, \, (\vecx_{ir}, \vecx_{js}, \vecx_{iu}, \vecx_{jv},\vecx_k)\in T^{n}_{X_iX_j X_i^{'} X'_j X_k} \}| \\
\stackrel{(a)}{\leq}&\exp{\left(-n\left(R_i+R_j - \max_{\cS\in\{1,2,3,4\}}g^{\cS}_{X_i,X_j,X'_i,X'_j,X_k}(R_i,R_j,R_i,R_j)-\epsilon/2\right)\right)}\\
\leq& \exp{(-n\epsilon/2)}, \qquad\text{ if }R_i+R_j - g^{\cS}_{X_i,X_j,X'_i,X'_j,X_k}(R_i,R_j,R_i,R_j)\geq \epsilon,
\end{align*}
where $(a)$ uses \eqref{lemma_eq4}. Next we show that \eqref{eq:cond1} implies \eqref{eq:fall1}.
\begin{align*}
\frac{1}{N_iN_j}&|\{(r, s)\in [1:N_i]\times[1:N_j]: \exists \,(u, v)\in  [1:N_i]\times[1:N_j], \, u\neq r, v\neq s, \, (\vecx_{ir}, \vecx_{js}, \vecx_{iu}, \vecx_{jv},\vecx_k)\in T^{n}_{X_iX_j X_i^{'} X'_j X_k} \}|\\
{\leq} &\frac{1}{N_iN_j}|\{(r, s, u)\in [1:N_i]\times[1:N_j]\times[1:N_i]:  \, (\vecx_{js}, \vecx_{iu}, \vecx_k)\in T^{n}_{X_j X_i^{'} X_k} \}| \\
{=} &\frac{1}{N_iN_j}|\{r\in[1:N_i]\}\times\{(s, u)\in [1:N_j]\times[1:N_i]:  \, (\vecx_{js}, \vecx_{iu}, \vecx_k)\in T^{n}_{X_j X_i^{'} X_k} \}| \\
= &\frac{1}{N_j}|\{(s, u)\in [1:N_j]\times[1:N_i]:  \, (\vecx_{js}, \vecx_{iu}, \vecx_k)\in T^{n}_{X_j X_i^{'} X_k} \}| \\
\stackrel{(a)}{\leq}&\exp{\left(-n\left(R_j-\left|\left|R_j-I(X_j;X_k)\right|^+ +\left|R_i-I(X'_i;X_k)\right|^+-I(X_j;X_i'|X_k)\right|^+ - \epsilon/2\right)\right)}\\
\leq& \exp{(-n\epsilon/2)}, \qquad\text{ if }R_j  -\left|\left|R_j-I(X_j;X_k)\right|^+ +\left|R_i-I(X'_i;X_k)\right|^+-I(X_j;X_i'|X_k)\right|^+ \geq \epsilon.
\end{align*}
Here $(a)$ uses \eqref{lemma_eq3b} with $X_i$ replaced with $X_i'$. The remaining conditions can also be obtained similarly.
We can show that \eqref{eq:cond2} implies \eqref{eq:fall1} by using  \eqref{lemma_eq3e} on the following upper bound.
\begin{align*}
\frac{1}{N_iN_j}&|\{(r, s)\in [1:N_i]\times[1:N_j]: \exists \,(u, v)\in  [1:N_i]\times[1:N_j], \, u\neq r, v\neq s, \, (\vecx_{ir}, \vecx_{js}, \vecx_{iu}, \vecx_{jv},\vecx_k)\in T^{n}_{X_iX_j X_i^{'} X'_j X_k} \}|\\
{\leq} &\frac{1}{N_iN_j}|\{(r, s, v)\in [1:N_i]\times[1:N_j]\times[1:N_j]: s\neq v \, (\vecx_{js}, \vecx_{jv}, \vecx_k)\in T^{n}_{X_j X_j^{'} X_k} \}| \\
= &\frac{1}{N_j}|\{( s, v)\in [1:N_j]\times[1:N_j]: s\neq v \, (\vecx_{js}, \vecx_{jv}, \vecx_k)\in T^{n}_{X_j X_j^{'} X_k} \}|.
\end{align*}
To show that \eqref{eq:cond3} implies \eqref{eq:fall1}, we use the following upper bound and \eqref{lemma_eq3b}.
\begin{align*}
\frac{1}{N_iN_j}&|\{(r, s)\in [1:N_i]\times[1:N_j]: \exists \,(u, v)\in  [1:N_i]\times[1:N_j], \, u\neq r, v\neq s, \, (\vecx_{ir}, \vecx_{js}, \vecx_{iu}, \vecx_{jv},\vecx_k)\in T^{n}_{X_iX_j X_i^{'} X'_j X_k} \}|\\
{\leq} &\frac{1}{N_iN_j}|\{(r, s, u)\in [1:N_i]\times[1:N_j]\times[1:N_i]: r\neq u \, (\vecx_{ir}, \vecx_{iu}, \vecx_k)\in T^{n}_{X_i X_i^{'} X_k} \}|\\
= &\frac{1}{N_i}|\{(r,  u)\in [1:N_i]\times[1:N_i]: r\neq u \, (\vecx_{ir}, \vecx_{iu}, \vecx_k)\in T^{n}_{X_i X_i^{'} X_k} \}|. 
\end{align*}
The condition \eqref{eq:cond4} can be obtained similarly by using \eqref{lemma_eq3b} (with $X_j$ replaced with $X'_j$) on the following upper bound:
\begin{align*}
\frac{1}{N_iN_j}&|\{(r, s)\in [1:N_i]\times[1:N_j]: \exists \,(u, v)\in  [1:N_i]\times[1:N_j], \, u\neq r, v\neq s, \, (\vecx_{ir}, \vecx_{js}, \vecx_{iu}, \vecx_{jv},\vecx_k)\in T^{n}_{X_iX_j X_i^{'} X'_j X_k} \}|\\
{\leq} &\frac{1}{N_iN_j}|\{(r, s, v)\in [1:N_i]\times[1:N_j]\times[1:N_j]: \, (\vecx_{ir}, \vecx_{jv}, \vecx_k)\in T^{n}_{X_i X_j^{'} X_k} \}|\\
= &\frac{1}{N_i}|\{(r, v)\in [1:N_i]\times[1:N_j]: \, (\vecx_{ir}, \vecx_{jv}, \vecx_k)\in T^{n}_{X_i X_j^{'} X_k} \}|. 
\end{align*}
For \eqref{eq:cond5}, we use \eqref{lemma_eq3a} on the following upper bound:
\begin{align*}
\frac{1}{N_iN_j}&|\{(r, s)\in [1:N_i]\times[1:N_j]: \exists \,(u, v)\in  [1:N_i]\times[1:N_j], \, u\neq r, v\neq s, \, (\vecx_{ir}, \vecx_{js}, \vecx_{iu}, \vecx_{jv},\vecx_k)\in T^{n}_{X_iX_j X_i^{'} X'_j X_k} \}|\\
{\leq} &\frac{1}{N_iN_j}|\{(r, s)\in [1:N_i]\times[1:N_j]: \, (\vecx_{ir}, \vecx_{js}, \vecx_k)\in T^{n}_{X_i X_j X_k} \}|.
\end{align*}

Now to show \eqref{code_eq3_app}, we will show 
that the condition in \eqref{code_eq3_app} implies at least one of \eqref{eq:cond}-\eqref{eq:cond4}.  We restate the condition in \eqref{code_eq3_app} below for quick reference. 
\begin{align}
I(X_i;X_jX'_iX'_jX_k)+I(X_j;X_i'X_j'X_k)\geq \left|\left|R_i-I(X'_i;X_k)\right|^+ +\left|R_j-I(X'_j;X_k)\right|^+-I(X'_i;X'_j|X_k)\right|^+ +\epsilon. \label{code_eq3_codn}
\end{align}
To show that \eqref{code_eq3_codn} implies at least one of \eqref{eq:cond}-\eqref{eq:cond4}, we do a case analysis similar to the one in proof of \eqref{code_eq4_app}. The case analysis will be based on the value of $\hat{\cS}\defineqq \text{argmax}_{\cS}{g_{X_i, X_j, X_i', X_j', X_k}^{\cS}(R_i, R_j, R_i, R_j)}$, the set of maximizers of the expression  ${g_{X_i, X_j, X_i', X_j', X_k}^{{\cS}}(R_i, R_j, R_i, R_j)}$ in \eqref{eq:cond}. For ease of reference, evaluations of the expression ${g_{X_i, X_j, X_i', X_j', X_k}^{\tilde{\cS}}(R_i, R_j, R_i, R_j)}$  under different values of $\tilde{\cS}$ are given in Table~\ref{table:g}. Table~\ref{table:g} is similar to Table~\ref{table:g2}. It also gives the implications when $\tilde{\cS}\in \text{argmax}_{\cS\subseteq \{1, 2, 3, 4\}}{g_{X_i, X_j, X_i', X_j', X_k}^{\cS}(R_i, R_j, R_i, R_j)}$ in the fourth column. For example, the $\Circled{9}$-th row considers the case of $\{2,3\}\in \text{argmax}_{\cS}{g_{X_i, X_j, X_i', X_j', X_k}^{\cS}(R_i, R_j, R_i, R_j)}$. Under this case, $g_{X_i, X_j, X_i', X_j', X_k}^{\{2,3\}}(R_i, R_j, R_i, R_j)$ $\geq g_{X_i, X_j, X_i', X_j', X_k}^{\{1,2,3\}}(R_i, R_j, R_i, R_j)$, i.e, $R_j- I(X_j;X_iX_j'X_k)+R_i-I(X'_i;X_iX_jX_j'X_k)\geq R_i-I(X_i;X_j'X_k) + R_j- I(X_j;X_iX_j'X_k)+R_i-I(X'_i;X_iX_jX_j'X_k)$. This implies that $R_i\leq I(X_i;X'_jX_k)$. It is given in the fifth column of the table against the ``reason'' $\Circled{9}\geq \Circled{15}$ where $\Circled{15}$ is the row corresponding to $\cS = \{1,2,3\}$.
 The other implications can also be seen easily from the table.
\paragraph*{\underline{Case 1: $\tilde{\cS}\in \hat{\cS}$ such that $|\tilde{\cS}|\leq 1$}}
For this case, substituting the expression of ${g_{X_i, X_j, X_i', X_j', X_k}^{\tilde{\cS}}(R_i, R_j, R_i, R_j)}$ from Table~\ref{table:g} and noting $R_i, R_j\geq \epsilon$, we see that \eqref{eq:cond} holds. 
\paragraph*{\underline{Case 2: $\tilde{\cS}\in \hat{\cS}$ such that $|\tilde{\cS}| = 2$}}
We start with $\{1, 2\}\in \hat{\cS}$. For this case, \eqref{eq:cond} evaluates to $I(X_i;X'_iX'_jX_k) + I(X_i;X'_iX'_jX_jX_k)\geq \epsilon$ which is directly implied by \eqref{code_eq3_codn}. If $\{1,3\}\in \hat{\cS}$, we see from Table~\ref{table:g} that $R_j\leq I(X_j;X'_jX_k)$, $R_j\leq I(X'_j;X_jX_k)$ and $R_j-I(X_j;X_k)+R_j-I(X'_j;X_k)-I(X_j;X'_j|X_k)\leq 0$. This implies that $\Big|\left|R_j-I(X_j;X_k)\right|^+ +\left|R_j-I(X'_j;X_k)\right|^+\\-I(X_j;X'_j|X_k)\Big|^+ = 0$. Thus, by noting that $R_j\geq \epsilon$, \eqref{eq:cond2} holds. Similarly; if $\{1,4\} \in\hat{\cS}$, \eqref{eq:cond1} holds; if $\{2,3\}\in \hat{\cS}$, \eqref{eq:cond4} holds; if $\{2,4\}\in \hat{\cS}$, \eqref{eq:cond3} holds; and if  $\{3,4\}\in\hat{\cS}$, \eqref{eq:cond5} holds.
\paragraph*{\underline{Case 3: $\tilde{\cS}\in \hat{\cS}$ such that $|\tilde{\cS}| = 3$}}
If $\{2,3,4\}\in \hat{\cS}$, from Table~\ref{table:g}, $R_i\leq I(X_i;X_k)$. This implies that the LHS of \eqref{eq:cond5} is $R_i+R_j-|R_j-I(X_j;X_iX_k)|^+$ which is at least $\epsilon$ because $R_i\geq \epsilon$. Thus, \eqref{eq:cond5} holds. Similarly, for $\{1,3,4\}\in \hat{\cS}$, $R_j\leq I(X_j;X_k)$, which implies that \eqref{eq:cond5} holds. If $\{1, 2, 4\}\in \hat{S}$, from Table~\ref{table:g}, $R_i\leq I(X'_i;X_k)$ and $R_j\geq I(X'_j;X_i'X_k)$. These imply that \eqref{code_eq3_codn} evaluates to
\begin{align*}
I(X_i;X_jX'_iX'_jX_k)+I(X_j;X_i'X_j'X_k)\geq R_j-I(X'_j;X'_iX_k) +\epsilon.
\end{align*} 
Moreover, since $\{1, 2, 4\}\in \hat{S}$,\eqref{eq:cond} evaluates to
\begin{align*}
&R_i+R_j-\left(R_i-I(X_i;X'_jX_k) +R_j-I(X_j;X_iX'_jX_k) + R_j -I(X'_j;X_iX_jX'_iX_k)\right)\geq \epsilon.
\end{align*}  It can be seen upon rearranging that \eqref{code_eq3_codn} and \eqref{eq:cond} are the same.
Thus, \eqref{code_eq3_codn} implies \eqref{eq:cond}. 
Similarly, for $\cS = \{1,2,3\}$, \eqref{code_eq3_codn} evaluates to 
\begin{align*}
I(X_i;X_jX'_iX'_jX_k)+I(X_j;X_i'X_j'X_k)\geq R_i-I(X'_i;X'_jX_k) +\epsilon.
\end{align*} which implies \eqref{eq:cond}.

For $\cS = \{1, 2, 3\}$, $R_j\leq I(X'_j;X_k)$ and $R_i\geq I(X'_i;X_j'X_k)$. This implies that \eqref{code_eq3_codn} evaluates to
\begin{align*}
I(X_i;X_jX'_iX'_jX_k)+I(X_j;X_i'X_j'X_k)\geq R_i-I(X'_i;X'_jX_k) +\epsilon.
\end{align*} 
and for $\cS = \{1, 2, 3\}$,\eqref{eq:cond} evaluates to
\begin{align*}
&R_i+R_j-\left(R_i-I(X_i;X'_jX_k) +R_j-I(X_j;X_iX'_jX_k) + R_i -I(X'_i;X_iX_jX'_jX_k)\right)\geq \epsilon.
\end{align*}  It can be seen upon rearranging that the above two inequalities are the same.
Thus, \eqref{code_eq3_codn} implies \eqref{eq:cond} if $\{1, 2, 4\}\in \hat{S}$. 
Similarly, for $\cS = \{1,2,4\}$, \eqref{code_eq3_codn} evaluates to 
\begin{align*}
I(X_i;X_jX'_iX'_jX_k)+I(X_j;X_i'X_j'X_k)\geq R_j-I(X'_j;X'_iX_k) +\epsilon
\end{align*} which implies \eqref{eq:cond}.
\paragraph*{\underline{Case 4: $\tilde{\cS}\in \hat{\cS}$ such that $|\tilde{\cS}| = 4$}}
For $\{1, 2, 3, 4\}\in \hat{S}$, \eqref{code_eq3_codn} evaluates to  
\begin{align*}
I(X_i;X_jX'_iX'_jX_k)+I(X_j;X_i'X_j'X_k)\geq R_j - I(X'_j;X_k) + R_i-I(X'_i;X'_jX_k) +\epsilon
\end{align*} 
and \eqref{eq:cond} evaluates to
\begin{align*}
R_i+R_j-2R_i-2R_j+ I(X'_j;X_k) +I(X'_i;X'_jX_k)+I(X_j;X_i'X_j'X_k)+I(X_i;X_jX'_iX'_jX_k) \geq \epsilon
\end{align*} which is same as \eqref{code_eq3_codn}.
Thus, \eqref{code_eq3_codn} implies \eqref{eq:cond}.

\begin{table}[h]
\begin{footnotesize}
\begin{center}
\begin{tabular}{ c|c|m{21em}|m{21em}|c} 
Index & $\tilde{\cS}$ &\centering  ${g_{X_i, X_j, X_i', X_j', X_k}^{\tilde{\cS}}(R_i, R_j, R_i, R_j)}$ & \centering Implications of \\$\tilde{\cS} \in \text{argmax}_{\cS}{g_{X_i, X_j, X_i', X_j', X_k}^{S}(R_i, R_j, R_i, R_j)}$&reasons \\
\hline
\hline
{\Circled{1}} & $\emptyset$ &$0$ & &\\
\hline	
{\Circled{2}} & $\{1\}$ &$R_i-I(X_i;X_jX'_iX'_jX_k)$ & &\\
\hline	
{\Circled{3}} & $\{2\}$ &$R_j-I(X_j;X_iX'_iX'_jX_k)$ & &\\
\hline	
{\Circled{4}} & $\{3\}$ &$R_i-I(X'_i;X_iX_jX'_jX_k)$ & &\\
\hline	
{\Circled{5}} & $\{4\}$ &$R_j-I(X'_j;X_iX_jX'_iX_k)$ & &\\
\hline	
{\Circled{6}} & $\{1,2\}$ &$R_i-I(X_i;X_i'X'_jX_k)+R_j-I(X_j;X_iX_i'X'_jX_k)$ & &\\
\hline	
\multirow{3}{*}{\Circled{7}} & \multirow{3}{*}{$\{1,3\}$} & \multirow{3}{*}{$R_i-I(X_i;X_jX'_jX_k)+R_i-I(X'_i;X_iX_jX'_jX_k)$} & $R_j\leq I(X_j;X'_jX_k)$&{$\Circled{7}\geq \Circled{15}$}\\
\cline{4-5}
& & &$R_j\leq I(X'_j;X_jX_k)$&{$\Circled{7}\geq \Circled{13}$}\\
\cline{4-5}& & &$R_j-I(X_j;X_k)+R_j-I(X'_j;X_jX_k)\leq 0$ & {$\Circled{7}\geq \Circled{16}$}\\ 
\hline
\multirow{3}{*}{\Circled{8}} & \multirow{3}{*}{$\{1,4\}$} & \multirow{3}{*}{$R_i-I(X_i;X_jX'_iX_k)+R_j-I(X'_j;X_iX_jX'_iX_k)$} & $R_j\leq I(X_j;X'_iX_k)$&{$\Circled{8}\geq \Circled{14}$}\\
\cline{4-5}
& & &$R_i\leq I(X'_i;X_jX_k)$&{$\Circled{8}\geq \Circled{13}$}\\
\cline{4-5}& & &$R_j-I(X_j;X_k)+R_i-I(X'_i;X_jX_k)\leq 0$ & {$\Circled{8}\geq \Circled{16}$}\\ 
\hline
\multirow{3}{*}{\Circled{9}} & \multirow{3}{*}{$\{2,3\}$} & \multirow{3}{*}{$R_j-I(X_j;X_iX'_jX_k)+R_i-I(X'_i;X_iX_jX'_jX_k)$} & $R_i\leq I(X_i;X'_jX_k)$&{$\Circled{9}\geq \Circled{15}$}\\
\cline{4-5}
& & &$R_j\leq I(X'_j;X_iX_k)$&{$\Circled{9}\geq \Circled{12}$}\\
\cline{4-5}& & &$R_i-I(X_i;X_k)+R_j-I(X'_j;X_iX_k)\leq 0$ & {$\Circled{9}\geq \Circled{16}$}\\ 
\hline
\multirow{3}{*}{\Circled{10}} & \multirow{3}{*}{$\{2,4\}$} & \multirow{3}{*}{$R_j-I(X_j;X_iX'_iX_k)+R_j-I(X'_j;X_iX_jX'_iX_k)$} & $R_i\leq I(X_i;X'_iX_k)$&{$\Circled{10}\geq \Circled{14}$}\\
\cline{4-5}
& & &$R_i\leq I(X'_i;X_iX_k)$&{$\Circled{10}\geq \Circled{12}$}\\
\cline{4-5}& & &$R_i-I(X_i;X_k)+R_i-I(X'_i;X_iX_k)\leq 0$ & {$\Circled{10}\geq \Circled{16}$}\\ 
\hline
\multirow{3}{*}{\Circled{11}} & \multirow{3}{*}{$\{3,4\}$} & \multirow{3}{*}{$R_i-I(X'_i;X_iX_jX_k)+R_j-I(X'_j;X_iX_jX'_iX_k)$} & $R_i\leq I(X_i;X_jX_k)$&{$\Circled{11}\geq \Circled{13}$}\\
\cline{4-5}
& & &$R_j\leq I(X_j;X_iX_k)$&{$\Circled{11}\geq \Circled{12}$}\\
\cline{4-5}& & &$R_i-I(X_i;X_k)+R_i-I(X_j;X_iX_k)\leq 0$ & {$\Circled{11}\geq \Circled{16}$}\\ 
\hline
\Circled{12} & $\{2,3,4\}$& $R_j-I(X_j;X_iX_k)+R_i-I(X_i';X_iX_jX_k)+R_j-I(X'_j;X_iX_jX'_iX_k)$ & $R_i\leq I(X_i;X_k)$&$\Circled{12}\geq\Circled{16}$\\
\hline
\Circled{13} & $\{1,3,4\}$  &$R_i-I(X_i;X_jX_k)+R_i-I(X'_i;X_iX_jX_k)+R_j-I(X'_j;X_iX_jX'_iX_k)$ & $R_j\leq I(X_j;X_k)$&$\Circled{13}\geq \Circled{16}$\\
\hline 
\multirow{2}{*}{\Circled{14}} & \multirow{2}{*}{$\{1,2,4\}$} & {$R_i-I(X_i;X'_iX_k)+R_j-I(X_j;X_iX'_iX_k)$} & $R_i\leq I(X'_i;X_k)$&{$\Circled{14}\geq \Circled{16}$}\\
\cline{4-5}
& & $+R_j-I(X_j';X_iX_jX'_iX_k)$&$R_j\geq I(X'_j;X'_iX_k)$&{$\Circled{14}\geq \Circled{6}$}\\
\hline
\multirow{2}{*}{\Circled{15}} & \multirow{2}{*}{$\{1,2,3\}$} & {$R_i-I(X_i;X'_jX_k)+R_j-I(X_j;X_iX'_jX_k)$} & $R_j\leq I(X'_j;X_k)$&{$\Circled{15}\geq \Circled{16}$}\\
\cline{4-5}
& &$+R_i-I(X_i';X_iX_jX'_jX_k)$ &$R_i\geq I(X'_i;X'_jX_k)$&{$\Circled{15}\geq \Circled{6}$}\\
\hline

\multirow{3}{*}{\Circled{16}} & \multirow{3}{*}{$\{1,2,3,4\}$} & {$R_i-I(X_i;X_k)+R_j-I(X_j;X_iX_k)$} & $R_i\geq I(X'_i;X_k)$&{$\Circled{16}\geq \Circled{14}$}\\
\cline{4-5}
& &$+R_i-I(X'_i;X_iX_jX_k)+R_j-I(X'_j;X'_iX_iX_jX_k)$ &$R_j\geq I(X'_j;X_k)$&{$\Circled{16}\geq \Circled{15}$}\\
\cline{4-5}
& & &$R_i-I(X'_i;X_k) + R_j - I(X'_j;X'_i;X_k) \geq 0 $&{$\Circled{16}\geq \Circled{6}$}
\end{tabular}
\end{center}
\end{footnotesize}
\caption{Table showing different evaluations of $\max_{\cS\subseteq \{1, 2, 3, 4\}}{g_{X_i, X_j, X_i', X_j', X_k}^{S}(R_i, R_j, R_i, R_j)}$ and their implications.}
\label{table:g}
\end{table}  

Statement \eqref{code_eq5_app} can be proved similarly by using \eqref{lemma_eq3a}, \eqref{lemma_eq3b}, \eqref{lemma_eq6c},  \eqref{lemma_eq6e} and \eqref{lemma_eq6}, the equivalent statements of \eqref{lemma_eq3a}-\eqref{lemma_eq3e} and \eqref{lemma_eq4} on replacing $\vecx_{jv}$, $X'_j$ and the corresponding rate $R_j$ with $\vecx_{kw}$,  $X'_k$ and the corresponding rate $R_k$ respectively. In fact, by making these replacements in the proof of \eqref{code_eq3_app}, we obtain the proof of \eqref{code_eq5_app}. Note that the proof of \eqref{code_eq3_app} only depended on \eqref{lemma_eq3a}-\eqref{lemma_eq3e} and \eqref{lemma_eq4}.

\end{proof}
  \section{Proof of Lemma~\ref{lemma:maximal sets}}\label{appendix:proofs}
\pink{This appendix gives the proof of Lemma~\ref{lemma:maximal sets}. We restate it here for completeness.
\begin{duplicatelemma}
For any $\cQ_1, \ldots, \cQ_t\in \cA$, $t\in \bbN$ and $\cQ\subseteq[1:k]$ such that $\cQ= \cap_{i=1}^{t}\cQ_i$, $P_{e, \cQ}\leq \sum_{i=1}^{t}P_{e, \cQ_i}$. 
\end{duplicatelemma}}
\begin{proof}[Proof of Lemma~\ref{lemma:maximal sets}]
\begin{align*}
P_{e, \cQ} &= \max_{\vecx_{\cQ}}\frac{1}{(\prod_{j\in \cQ^c}N_j)}\sum_{m_{\cQ^c}\in \cM_{\cQ^c}}\bbP\inp{\inb{\phi(\vecY)_{\cQ^c} \neq m_{\cQ^c}}|\vecX_{\cQ^c} = f_{\cQ^c}(m_{\cQ^c}), \vecX_{\cQ} = \vecx_{\cQ}}\\
&= \max_{\vecx_{\cQ}}\frac{1}{(\prod_{j\in \cQ^c}N_j)}\sum_{m_{\cQ^c}\in \cM_{\cQ^c}}\bbP\inp{\cup_{i=1}^t\inb{\phi(\vecY)_{\cQ_i^c} \neq m_{\cQ_i^c}}|\vecX_{\cQ^c} = f_{\cQ^c}(m_{\cQ^c}), \vecX_{\cQ} = \vecx_{\cQ}}\\
&\stackrel{(a)}{\leq} \max_{\vecx_{\cQ}}\frac{1}{(\prod_{j\in \cQ^c}N_j)}\sum_{m_{\cQ^c}\in \cM_{\cQ^c}}\sum_{i=1}^t\bbP\inp{\inb{\phi(\vecY)_{\cQ_i^c} \neq m_{\cQ_i^c}}|\vecX_{\cQ^c} = f_{\cQ^c}(m_{\cQ^c}), \vecX_{\cQ} = \vecx_{\cQ}}\\
&\pink{= \sum_{i=1}^t\max_{\vecx_{\cQ}}\frac{1}{(\prod_{j\in \cQ^c}N_j)}\sum_{m_{\cQ^c}\in \cM_{\cQ^c}}\bbP\inp{\inb{\phi(\vecY)_{\cQ_i^c} \neq m_{\cQ_i^c}}|\vecX_{\cQ^c} = f_{\cQ^c}(m_{\cQ^c}), \vecX_{\cQ} = \vecx_{\cQ}}}\\
&\stackrel{(b)}{\leq} \sum_{i=1}^t\max_{\vecx_{\cQ_i}}\frac{1}{(\prod_{j\in \cQ^c}N_j)}\sum_{m_{\cQ^c}\in \cM_{\cQ^c}}\bbP\inp{\inb{\phi(\vecY)_{\cQ_i^c} \neq m_{\cQ_i^c}}|\vecX_{\cQ_i^c} = f_{\cQ_i^c}(m_{\cQ_i^c}), \vecX_{\cQ_i} = \vecx_{\cQ_i}}\\
&=\sum_{i=1}^t\max_{\vecx_{\cQ_i}}\frac{1}{(\prod_{j\in \cQ_i^c}N_j)}\sum_{m_{\cQ_i^c}\in \cM_{\cQ_i^c}}\bbP\inp{\inb{\phi(\vecY)_{\cQ_i^c} \neq m_{\cQ_i^c}}|\vecX_{\cQ_i^c} = f_{\cQ_i^c}(m_{\cQ_i^c}), \vecX_{\cQ_i} = \vecx_{\cQ_i}}\\
&=\sum_{i=1}^t P_{e, \cQ_i},
\end{align*}\pink{where $(a)$ follows from a union bound. To see $(b)$ note that $\cQ\subseteq \cQ_i$ for any $i\in [1:t]$ and thus, the maximization is over a larger set.}
\end{proof}

\section{Proof of achievability of Theorem~\ref{thm:random_k}}\label{section:k_user_rand}
For the achievability, we will require the following theorem which gives the randomized coding capacity region of a $t$-user AV-MAC $W_{Y|X_1, \ldots, X_t, S}$ where $\cX_i, \, i \in [1:k]$ are the input alphabets and $\cY$ and $\cS$ are the output and the state alphabets respectively.
The theorem can be proved along the lines of the two user result given in \cite{Jahn81} \pink{and \cite{370123}( see \cite[Remark~IIA3]{Jahn81}).}
\begin{thm}[AV-MAC randomized capacity region for $t$-users]\label{thm:AV_MAC_t}
The randomized capacity region of the AV-MAC $W_{Y|X_1, \ldots, X_t, S}$ is the set of rate tuples such that
\begin{align}\label{eq:AV_MAC_K_proof}
\sum_{j\in\cJ}R_j\leq \min_{q(s|u)}I(X_{\cJ};Y|X_{\cJ^c}, S, \pink{U}) \text{ for every }\cJ\subseteq [1:t]
\end{align}
for some joint distribution $p(u)q(s|u)\prod_{i=1}^t p(x_i|u)$ with $|\cU|\leq t$.
\end{thm}

\begin{proof}[Proof (Achievability of Theorem~\ref{thm:random_k})] This proof is along the lines of the proof of Theorem~\ref{thm:random}. For each $\cQ\in\cA$, let $W^{\cQ}$ be the $|\cQ^c|$-user AV-MAC  which corresponds to users in the set $\cQ$ as adversary and the users in the set $\cQ^c$ as the legitimate users. For users in $\cQ$, their combined input $\vecx_{\cQ}$ and the product input alphabet $\times_{i\in \cQ}\cX_i$ correspond to the adversarial state input and the state alphabet respectively.  Let $(R_1,R_2,\ldots,R_k)$ be a rate tuple such that for some $p(u)\cdot p(x_1|u)\cdot p(x_2|u)\cdot\ldots\cdot p(x_k|u)$, the following conditions hold for all $\cQ\in\cA$ and $\cJ\subseteq \cQ^c$,
\begin{align}
\sum_{j\in \cJ}R_j\leq \min_{q(\vecx_{\cQ}|u)}I(X_{\cJ};Y|X_{(\cQ\cup\cJ)^c}, \pink{U})\label{eq:proof_rand_k}
\end{align}
where the mutual information above is evaluated using the joint distribution
$p(u) q(\vecx_{\cQ}|u)\prod_{j\in\cJ}p(x_j|u)W(y|\vecx_{\cQ},\vecx_{\cQ^c})$. Here $|\mathcal{U}|\leq k$.
 Let $\epsilon>0$ be arbitrary and let $n$ be large enough. Note that, by Theorem~\ref{thm:AV_MAC_t}, the rate tuple $R_{\cQ^c}$ is an  achievable rate pair for the AVMAC $W^{\cQ}$. 
For each $i\in[1:k]$, let $\tilde{\cM}_i=[1:2^{nR_i}]$ and $\cM_i=[1:2^{nR_i}/v]$ for the largest integer $v\leq (k|\cA|)/\epsilon$. In the following, we show the existence of a randomized $(2^{nR_1}/v,\ldots,2^{nR_k}/v,n)$ code  $(F_1,\ldots,F_k,\phi)$ with \pink{$P^{\text{rand}}_e$} no larger than $\epsilon$, for sufficiently large $n$.

\usetikzlibrary{arrows, shapes,positioning,
                chains,
                decorations.markings,
                shadows, shapes.arrows}
  \definecolor{lightgreen}{rgb}{0.4,0.4,0.1}
\definecolor{lightblue}{rgb}{0.7372549019607844,0.8313725490196079,0.9019607843137255}
\definecolor{darkblue}{rgb}{0.08235294117647059,0.396078431372549,0.7529411764705882}
\definecolor{orangered}{rgb}{0.6,0.3,0.1}
\definecolor{gray}{rgb}{0.7529411764705882,0.7529411764705882,0.7529411764705882}

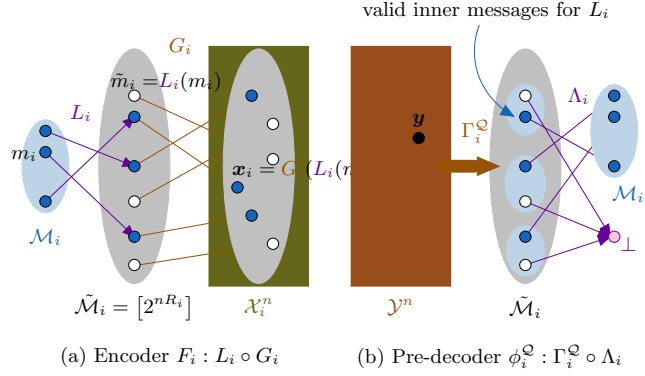
\begin{figure}\longonly{\centering} \scalebox{0.85}{\begin{tikzpicture}[line cap=round,line join=round,>=triangle 45,x=1.1cm,y=1.1cm]
\tikzstyle{myarrows}=[line width=0.5mm,draw = black!20!green!40!red,-triangle 45,postaction={draw, line width=2mm, shorten >=4mm, -}]

\fill[line width=0.4pt,color=lightgreen,fill=lightgreen,fill opacity=0.07] (-9.2,2.7) -- (-9.2,-.7) -- (-7.8,-.7) -- (-7.8,2.7) -- cycle;
\draw [rotate around={90:(-11.5,1)},line width=0.4pt,dotted,color=lightblue,fill=lightblue,fill opacity=0.46] (-11.5,1) ellipse (0.73cm and 0.367cm);
\draw [rotate around={90:(-10.25,1.8)},line width=0.4pt,dotted,color=lightblue,fill=lightblue,fill opacity=0.46] (-10.25,1.8) ellipse (0.40cm and 0.300cm);
\draw [rotate around={90:(-10.25,0.75)},line width=0.4pt,dotted,color=lightblue,fill=lightblue,fill opacity=0.46] (-10.25,0.75) ellipse (0.45cm and 0.320cm);
\draw [rotate around={90:(-10.25,-0.2)},line width=0.4pt,dotted,color=lightblue,fill=lightblue,fill opacity=0.46] (-10.25,-0.2) ellipse (0.40cm and 0.300cm);
\draw [rotate around={90:(-10.25,1)},line width=0.4pt,color=gray,fill=gray,fill opacity=0.2] (-10.25,1) ellipse (1.83cm and 0.55cm);
\draw [line width=0.4pt,color=lightgreen] (-9.2,2.7)-- (-9.2,-.7);
\draw [line width=0.4pt,color=lightgreen] (-9.2,-.7)-- (-7.8,-.7);
\draw [line width=0.4pt,color=lightgreen] (-7.8,-.7)-- (-7.8,2.7);
\draw [line width=0.4pt,color=lightgreen] (-7.8,2.7)-- (-9.2,2.7);

\draw [color = blue!60!red,->,line width=0.1pt] (-11.5,0.5) -- (-10.25,1.7);
\draw [color = blue!60!red,->,line width=0.1pt] (-11.5,1.5) -- (-10.25,1);
\draw [color = blue!60!red,->,line width=0.1pt] (-11.5,1.2) -- (-10.25,0);

\draw [color = black!20!green!40!red,->,line width=0.1pt] (-10.25,2) -- (-8.3,1.1);
\draw [color = black!20!green!40!red,->,line width=0.1pt] (-10.25,1) -- (-8.6,2);
\draw [color = black!20!green!40!red,->,line width=0.1pt] (-10.25,1.7) -- (-8.8,0.7);
\draw [color = black!20!green!40!red,->,line width=0.1pt] (-10.25,0.5) -- (-8.3,1.6);
\draw [color = black!20!green!40!red,->,line width=0.1pt] (-10.25,0.0) -- (-8.6,0.3);
\draw [color = black!20!green!40!red,->,line width=0.1pt] (-10.25,-0.41) -- (-8.3,-0.1);

\draw [rotate around={-90:(-8.5,1)},line width=0.4pt,color=gray,fill=gray,fill opacity=0.26] (-8.5,1) ellipse (1.83cm and 0.55cm);

\draw [fill=darkblue] (-11.5,1.2) circle (2.5pt);
\draw [fill=darkblue,opacity=0.8] (-11.5,1.5) circle (2.5pt);
\draw [fill=darkblue,opacity=0.8] (-11.5,0.5) circle (2.5pt);
\draw[color=darkblue] (-11.5,0) node {\small $\cM_i$};
\draw[color=black] (-11.8,1.15) node {{\small  $m_i$}};
\draw [fill=white] (-10.25,0.5) circle (2.5pt);
\draw[fill=darkblue,opacity=0.8] (-10.25,0.0) circle (2.5pt);
\draw [fill=white] (-10.25,-0.4) circle (2.5pt);
\draw [fill=darkblue,opacity=0.8] (-10.25,1.7) circle (2.5pt);%
\draw [fill=white] (-10.25,2) circle (2.5pt);%
\draw [fill=darkblue,opacity=0.8] (-10.25,1) circle (2.5pt);%
\draw (-10.25,-1) node {\small $\tilde{\cM}_i=\left[2^{nR_i}\right]$};
\draw[color = lightgreen] (-8.5,-1) node {\small $\cX_i^n$};
\draw[color = blue!60!red] (-11,1.75) node {\small $L_i$};
\draw[color = blue!60!red] (-9.8,2.2) node {{\small  ${\color{black}\tilde{m}_i=}{\color{blue!60!red}L_i}{\color{black}(m_i)}$}};
\draw[color = black!20!green!40!red] (-9.6,2.7) node {\small $G_i$};
\draw[color = black] (-7.8,0.95) node {{\small  $\vecx_i={\color{black!20!green!40!red}G_i}({\color{blue!60!red}L_i}{\color{black}(m_i)})$}};

\draw [fill=darkblue,opacity=0.8] (-8.6,2) circle (2.5pt);%
\draw [fill=white] (-8.3,1.1) circle (2.5pt);%
\draw [fill=darkblue,opacity=0.8] (-8.8,0.7) circle (2.5pt);%
\draw [fill=white] (-8.3,1.6) circle (2.5pt);
\draw [fill=darkblue,opacity=0.8] (-8.6,0.3) circle (2.5pt);
\draw [fill=white] (-8.3,-0.1) circle (2.5pt);

\draw [line width=0.4pt,color=orangered] (-7.2,2.7)-- (-7.2,-.7);
\draw [line width=0.4pt,color=orangered] (-7.2,-.7)-- (-5.8,-.7);
\draw [line width=0.4pt,color=orangered] (-5.8,-.7)-- (-5.8,2.7);
\draw [line width=0.4pt,color=orangered] (-5.8,2.7)-- (-7.2,2.7);
\fill[line width=0.4pt,color=orangered,fill=orangered,fill opacity=0.07] (-7.2,2.7) -- (-7.2,-.7) -- (-5.8,-.7) -- (-5.8,2.7) -- cycle;
\draw[color = orangered] (-6.5,-1) node {\small $\cY^n$};
\draw [fill=black] (-6.25,1.4) circle (2.5pt);
\draw[color = black] (-6.25,1.65) node {\small $\vecy$};

\draw [rotate around={90:(-4.75,1)},line width=0.4pt,color=gray,fill=gray,fill opacity=0.2] (-4.75,1) ellipse (1.83cm and 0.55cm);
\draw [fill=white] (-4.75,0.5) circle (2.5pt);
\draw (-4.75,-1) node {\small $\tilde{\cM}_i$};
\draw[color = black!20!green!40!red] (-5.45,1.5) node {\small $\Gamma^{\cQ}_i$};
\draw[color = blue!60!red] (-4,2) node {\small $\Lambda_i$};

\draw[color = black] (-5.3,3.2) node {\small valid inner messages for $L_i$};

\draw [myarrows](-5.9,1)--(-5.1,1);
\draw [rotate around={90:(-4.75,1.8)},line width=0.4pt,dotted,color=lightblue,fill=lightblue,fill opacity=0.46] (-4.75,1.8) ellipse (0.40cm and 0.300cm);
\draw [rotate around={90:(-4.75,0.75)},line width=0.4pt,dotted,color=lightblue,fill=lightblue,fill opacity=0.46] (-4.75,0.75) ellipse (0.45cm and 0.320cm);
\draw [rotate around={90:(-4.75,-0.2)},line width=0.4pt,dotted,color=lightblue,fill=lightblue,fill opacity=0.46] (-4.75,-0.2) ellipse (0.40cm and 0.300cm);
\draw [color = blue!60!red,->,line width=0.2pt] (-4.75,1.7) -- (-3.5,1);
\draw [color = blue!60!red,->,line width=0.2pt] (-4.75,1) -- (-3.5,2);
\draw [color = blue!60!red,->,line width=0.2pt] (-4.75,2) -- (-3.55,0);

\draw [color = blue!60!red,->,line width=0.2pt] (-4.75,0.5) -- (-3.55,0);
\draw [color = blue!60!red,->,line width=0.2pt] (-4.75,0.0) -- (-3.5,1.7);;
\draw [color = blue!60!red,->,line width=0.2pt] (-4.75,-0.4) -- (-3.55,0);
\draw[fill=darkblue,opacity=0.8] (-4.75,0.0) circle (2.5pt);
\draw [fill=white] (-4.75,-0.4) circle (2.5pt);

\draw [fill=white] (-4.75,0.5) circle (2.5pt);
\draw [fill=darkblue] (-4.75,1.7) circle (2.5pt);%
\draw [fill=white] (-4.75,2) circle (2.5pt);%
\draw [fill=darkblue] (-4.75,1) circle (2.5pt);%

\draw [->,color=darkblue] (-5.4,3) to [out=250,in=150] (-4.9,1.7);
\draw [rotate around={90:(-3.5,1.5)},line width=0.4pt,dotted,color=lightblue,fill=lightblue,fill opacity=0.46] (-3.5,1.5) ellipse (0.73cm and 0.367cm);
\draw [fill=darkblue] (-3.5,1.7) circle (2.5pt);
\draw [fill=darkblue] (-3.5,2) circle (2.5pt);
\draw [fill=darkblue] (-3.5,1) circle (2.5pt);
\draw [color=darkblue] (-3.3,.6) node {\small $\cM_i$};
\draw [color=red!50!blue,fill=red!70!blue!20!white] (-3.5,0) circle (2.5pt);
\draw [color=red!50!blue] (-3.3,-0.1) node {\small $\bot$};

\draw (-9.75,-1.7) node {\small {(a)} Encoder $F_i:L_i\circ G_i$};
\draw (-5.25,-1.7) node {\small {(b)} Pre-decoder $\phi^{\cQ}_i:\Gamma_i^{\cQ}\circ\Lambda_i$};

\end{tikzpicture}}
\caption{The encoders and pre-decoders for Theorem~\ref{thm:random_k}.}\label{fig:random_k}

\end{figure}

\paragraph*{Code design} 
Before describing the code, we describe the following maps which will help in describing the encoders and the decoder (see Figure~\ref{fig:random_k}). For each user $i$, let $G_i:\tilde{\cM}_i\to\cX^n_i$ be a randomized map such that it maps $m_i\in \tilde{\cM}_i$ to an $n$-length i.i.d. sequence $G_i(m_i)$ generated according to the distribution $p_i$. 
The sequences $G_i(m)$ are independent across $i\in [1:k]$ and $m\in \cM_i$. The realization of $G_i(m_i)$ for all $i\in [1:k]$ and $m_i\in \cM_i$ is shared with the decoder. 
For any $\cQ\in\cA$, consider the $|\cQ^c|$-user AV-MAC $W^{\cQ}$ as described above.
For each $i\in \cQ^c$, if we consider $\tilde{\cM_i}$  as the message set and $G_i$ as the corresponding encoder, then this construction ensures that the random encoders $G_i, \, i\in \cQ^c$ are independent and their randomness is also private from the adversarial users in the set $\cQ$. 
Thus, the joint distribution of the encoders $G_i, i\in \cQ^c$ (and the corresponding codewords) is the same as that of the encoders of AV-MAC $W^{\cQ}$ in the direct part of \cite[Theorem~1, Section~III-C]{Jahn81} (and its extension to a $t$-user AV-MAC as in Theorem~\ref{thm:AV_MAC_t}). 
For $G_i, i\in \cQ^c$  as encoders, let $\Gamma^{\cQ}$ denote the corresponding decoder for the AV-MAC $W^{\cQ}$ in Theorem~\ref{thm:AV_MAC_t}. 
Suppose $(\Gamma^{\cQ}_j, j\in \cQ^c):=\Gamma^{\cQ}$ where $\Gamma^{\cQ}_j:\cY^n\to\tilde{\cM}_i$. 
For all $\epsilon>0$, by Theorem~\ref{thm:AV_MAC_t}, there exists a large enough $n$ such that for all $\cQ\in \cA$, the code $((G_i, i\in \cQ^c), \Gamma^{\cQ})$ has error probability no larger than $\epsilon/(k|\cA|)$.
We consider that $n$.

For each $i\in [1:k]$, the message set $\cM_i$ is randomly embedded into the set $\tilde{\cM_i}$ as follows: We choose an arbitrary partition of $\tilde{\cM_i}$ into $|\cM_i|$ many disjoint equal-sized subsets (each subset size is $v$). 
Let us denote the partition by $\cS_{m_i}, \, m_i\in \cM_i$ where $\cup_{m_i\in \cM_i}\cS_{m_i}=     \tilde{\cM_i}$ and $\cS_{m_i}\cap\cS_{m'_i}=\emptyset$ for all $m_i\neq m'_i$, $m_i, m'_i\in \cM_i$.
The size of each $\cS_{m_i}, \, m_i\in \cM_i$ is $v$ ($\leq k|\cA|/\epsilon$). 
The maps $L_i:\cM_i\to\tilde{\cM}_i$ and $\Lambda_i:\tilde{\cM}_i\to\cM_i$ are the forward and reverse maps for an injection from $\cM_i$ to $\tilde{\cM}_i$ where, independently for each $m_i\in \cM_i$, $L_i(m_i)$ is chosen uniformly at random from $\cS_{m_i}$. 
Both the encoder maps $G_i$ and $L_i$ are independent for $i=1,2,\ldots, k$ and are made available to the decoder as the shared secret between user-$i$ and the decoder, unknown to other users. 
 
For the code of the \bmac, for each $i\in[1:k]$, the encoder map $F_i:\cM_i\to\cX^n_i$ is defined as $F_i(m_i)= G_i(L_i(m_i))$ for every $m_i\in\cM_i$. For each $\cQ\in\cA$ and $i\in\cQ^c$, we define pre-decoder\footnote{In this notation $\phi_i^{\cQ}(\vecy)$, we are suppressing the dependence of the pre-decoder (and later the decoder) on the randomness of the encoders.}
$$\phi_i^{\cQ}(\vecy)=\begin{cases}\Lambda_i(\Gamma^{\cQ}_i(\vecy)) &\mbox{if }\Gamma^{\cQ}_i(\vecy)\in L_i(\cM_i),\\ \bot &\mbox{otherwise}.\end{cases}$$ 
The decoder $\phi:\cY^n\to\times_{i\in [1:k]}\cM_i$ outputs $\phi(\vecy)=(\hat{m}_1,\ldots,\hat{m}_k)$, where, for each $i\in[1:k]$ and $\cQ\in\cA$,
$$\hat{m}_i=\begin{cases}\phi_i^{\cQ}(\vecy)& \mbox{ if }  |\{\phi_i^{\tilde{\cQ}}(\vecy):\tilde{\cQ}\in\cA\}| = 1\text{ and }\phi_i^{\cQ}(\vecy)\neq\bot \\
\phi_i^{\cQ}(\vecy)& \mbox{ if } \{\phi_i^{\tilde{\cQ}}(\vecy):\tilde{\cQ}\in\cA\} = \{\Phi_i^{\cQ}(\vecy), \bot\} \text{ where }\phi_i^{\cQ}(\vecy)\neq\bot\\
1&\mbox{otherwise.}
\end{cases}$$

\paragraph*{Error Analysis}
We first show that as long as the rate tuple $(R_1,R_2,\ldots,R_k)$ satisfy the rate constraints~\eqref{eq:proof_rand_k}, the following hold simultaneously for every honest user $i$ which sends message $m_i\in \cM_i$, potentially adversarial set of users $\cQ\in\cA$ with $i\notin \cQ$ and for channel output $\vecY$: \emph{(i)} $\phi_i^{\cQ}(\vecY)$ equals $m_i$ with probability at least $1-\epsilon/(k|\cA|)$ if users $\cQ$ are indeed adversarial and \emph{(ii)} $\phi_i^{\cQ}(\vecY)$ either equals $\bot$ or $m_i$, with probability at least $1-\epsilon/(k|\cA|)$, if users $\cQ$ are not adversarial. 
To this end, consider $\cQ\in \cA$ and assume that the adversarial users (if any) are users in set $\cQ$ which send $\vecX_{\cQ}$ as their  potentially adversarial input to the channel. Suppose, for $i\in \cQ^c$ and  $m_i\in \cM_i$, user-$i$ sends $F_i(m_i)$. Let $\vecY$ denote the channel output.
\begin{description}
\item[$(i)$] First, consider the AV-MAC $W^{\cQ}$.  Recall that $\Gamma_i^{\cQ}(\vecY)=L_i(m_i)$ with probability at least $1-\epsilon/(k|\cA|)$. Thus, with probability at least $1-\epsilon/(k|\cA|)$, $\phi_i^{\cQ}(\vecY)$ equals $m_i$. This also holds for any $\tilde{\cQ}\subset \cQ$, as we can think of this as adversarial users $\cQ$ where users in $\cQ\setminus \tilde{\cQ}$ send valid codewords.
\item[$(ii)$] Next, consider the AV-MAC $W^{\tilde{\cQ}}$, for $\tilde{\cQ}\in \cA$ where $\tilde{\cQ}\setminus \cQ\neq \emptyset$.  
We would like to compute $\bbP\inp{\phi_i^{\tilde{\cQ}}(\vecY)\notin \set{m_i,\bot}}$ where for $i\in \cQ^c$, the probability is over $G_i(L_i(m_i))$, $\vecX_{\cQ}$ and the channel.  
\begin{align*}
&\bbP\inp{\phi_i^{\tilde{\cQ}}(\vecY)\notin \set{m_i,\bot}}\\
& = \bbP\inp{\Gamma_i^{\tilde{\cQ}}(\vecY)\in L_i(\cM_i\setminus\{m_i\})}\\
& = \sum_{\tilde{m}_i\in\tilde{\cM}\setminus\cS_{m_i}}\bbP\inp{\Gamma_i^{\tilde{\cQ}}(\vecY)=\tilde{m}_i, \tilde{m}_i\in L_i(\cM_i\setminus\{m_i\})}\\
& = \sum_{\tilde{m}_i\in\tilde{\cM}\setminus\cS_{m_i}}\bbP\inp{\Gamma_i^{\tilde{\cQ}}(\vecY)=\tilde{m}_i}\bbP\inp{ \tilde{m}_i\in L_i(\cM_i\setminus\{m_i\})\big|\Gamma_i^{\tilde{\cQ}}(\vecY)=\tilde{m}_i}\\
& \stackrel{(a)}{=} \sum_{\tilde{m}_i\in\tilde{\cM}\setminus\cS_{m_i}}\bbP\inp{\Gamma_i^{\tilde{\cQ}}(\vecY)=\tilde{m}_i}\bbP\inp{ \tilde{m}_i\in L_i(\cM_i\setminus\{m_i\})}\\
& \stackrel{(b)}{=}  \sum_{\tilde{m}_i\in\tilde{\cM}\setminus\cS_{m_i}}\bbP\inp{\Gamma_i^{\tilde{\cQ}}(\vecY)=\tilde{m}_i}\cdot\frac{1}{v}\\
& \leq 1/v \leq \epsilon/(k|\cA|).
\end{align*} Here, $(a)$ holds as $\Gamma_i^{\tilde{\cQ}}(\vecY)\,\indep \,L_i(\cM_i\setminus\{m_i\})$. This is because $G_i(L_i(m_i))$ which produces $\vecY$ is independent of $L_i(\cM_i\setminus\{m_i\})$ and  $\Gamma_i^{\tilde{\cQ}}\,\indep \,L_i(\cM_i\setminus\{m_i\})$ as $\Gamma_i^{\tilde{\cQ}}$ is a function of AV-MAC encoders $G_i, i\in \tilde{\cQ}^c$ which are independent of $L_i$. 
The equality $(b)$ holds because for $\tilde{m}_i\in\tilde{\cM}\setminus\cS_{m_i}$, 
\begin{align*}
&\bbP\inp{ \tilde{m}_i\in L_i(\cM_i\setminus\{m_i\})}\\
&= \sum_{m_i'\in \cM_i\setminus\{m_i\}}\bbP\inp{ L_i(m_i')=\tilde{m}_i }\\
 &= \sum_{m_i'\in \cM_i\setminus\{m_i\}}1_{\{\tilde{m}_i\in \cS_{m_i'}\}}\cdot\frac{1}{v}\\
 &=1/v.
 \end{align*}

By taking union bound over all users and all $\cQ\in \cA$,  with probability $1-\epsilon$, for each non-adversarial user $i$, at least one of the decoders $\phi_i^{\cQ}$, $\cQ\in \cA$  outputs the true message while the other decoder outputs either the true message or $\bot$.
\end{description}

\end{proof}

\section{Randomness reduction lemma}\label{app:randomness_reduction}
\begin{lemma}[Randomness reduction]\label{lemma:randomness_reduction}
Suppose $\epsilon>0$. For large enough $n$, given any $(N_1, \ldots, N_k,n)$ randomized code $(F_{[1:k]}, \phi(F_{[1:k]}))$ satisfying  \pink{$$P^{\text{rand}}_{e}(P_{F_{[1:k]}}, \phi)<2^{\epsilon/2}-1,$$} there exist $n^2$  deterministic encoding maps  $f_{j,i}, i \in [1:n^2]$ for each user $j\in [1:k]$ such that for every $\cQ\in \cA, j_{\cQ}\in [1:n^2]^{|\cQ|}$, $\vecx_{\cQ}\in \cX^n_{\cQ}$ and the decoder $\phi$,
\begin{align*}
&\frac{1}{(n^2)^{|\cQ^c|}}\sum_{j_{\cQ^c}\in [1:n^2]^{|\cQ^c|}}\frac{1}{(\prod_{i\in \cQ^c}N_i)}\\
&\qquad\qquad\sum_{m_{\cQ^c}\in \cM_{\cQ^c}}\bbP\inb{\inp{\phi(\vecY, f_{\cQ, j_{\cQ}}, f_{\cQ^c, j_{\cQ^c}})= \hat{m}_{[1:k]}\text{ such that }\hat{m}_{\cQ^c}\neq m_{\cQ^c}}|\vecX_{\cQ^c}=f_{\cQ^c, j_{\cQ^c}}(m_{\cQ^c}), \vecX_{\cQ}=\vecx_{\cQ}}<\epsilon.
\end{align*}
Here, $f_{{\cQ}, j_{\cQ}}$ denotes $(f_{i, j_i}: i\in \cQ)$.
\end{lemma}
\begin{remark}
\pink{Lemma~\ref{lemma:randomness_reduction} states that given a randomized code with a small probability of error ($2^{\epsilon/2}-1$), there exists another randomized code of the same rate for all users which uses only $2\log{n}$ random bits at each user such that the new code also has a small probability of error ($\epsilon$).}
\end{remark}
\begin{proof}
The proof follows along the lines of Jahn~\cite[Theorem~1]{Jahn81}, \pink{though there are significant differences because of the byzantine nature of users. In particular, our result needs to incorporate the fact that a malicious user can maliciously choose their encoding map to influence decoding.} For each $i\in [1:K]$, let $\inb{F_{ij}}_{j=1}^{n^2}$ be independent samples of codebook $F_i$ (also independent across $i$). This gives the set of codes $\left\{(F_{ij}, \phi(F_{ij})), i\in [1:K], j\in [1:n^2], \phi:=\phi\right\}$. 
For every $\cQ\in\cA$, define $e_{\cQ}(f_{\cQ}, f_{\cQ^c}, \vecx_{\cQ})$ to be the error probability for fixed encoding maps $f_{\cQ}$ for the adversarial users and $f_{\cQ^c}$ for the non-adversarial users and the channel inputs chosen by the adversarial users as $\vecx_{\cQ}\in \cX^n_{\cQ}$, {\em i.e.},

\begin{align*}
&e_{\cQ}(f_{\cQ}, f_{\cQ^c}, \vecx_{\cQ})\\
&\qquad:=\frac{1}{(\prod_{i\in \cQ^c}N_i)}\sum_{m_{\cQ^c}\in \cM_{\cQ^c}}\hspace{0.2cm}\sum_{\substack{\vecy:\phi(\vecy, f_{\cQ}, f_{\cQ^c})= \hat{m}_{[1:k]},\\\text{ where }\hat{m}_{\cQ^c}\neq m_{\cQ^c}}}W_{Y|\vecX_{\cQ^c}\vecX_{\cQ}} \inp{\vecy|f_{\cQ^c}(m_{\cQ^c}), \vecx_{\cQ}}.
\end{align*}
Note that for $j_{\cQ}\in [1:n^2]^{|\cQ|}$, $j_{\cQ^c}\in [1:n^2]^{|\cQ^c|}$, $e_{\cQ}(F_{\cQ, j_{\cQ}}, F_{\cQ^c, j_{\cQ^c}}, \vecx_{\cQ})$, as a function of $F_{\cQ, j_{\cQ}}$ and  $F_{\cQ^c, j_{\cQ^c}}$ is a random variable.
We wish to show that 
\begin{align*}
&\lim_{n\to\infty}\bbP\inp{\frac{1}{(n^2)^{|\cQ^c|}}\sum_{j_{\cQ^c}\in [1:n^2]^{\pink{|\cQ^c|}}}  e_{\cQ}(F_{\cQ, j_{\cQ}}, F_{\cQ^c, j_{\cQ^c}}, \vecx_{\cQ})\geq \epsilon\text{ for {some} }\cQ\in \cA, j_{\cQ}\in [1:n^2]^{|\cQ|}\text{ and }\vecx_{\cQ}\in \cX^n_{\cQ}}= 0
\end{align*}

Using a union bound over $\cQ\in \cA$, $j_{\cQ}\in [1:n^2]^{|\cQ|}$, and $\vecx_{\cQ}\in \cX^n_{\cQ}$, we have
\begin{align*}
&\bbP\inp{\frac{1}{(n^2)^{|\cQ^c|}}\sum_{j_{\cQ^c}\in [1:n^2]^{\pink{|\cQ^c|}}}  e_{\cQ}(F_{\cQ, j_{\cQ}}, F_{\cQ^c, j_{\cQ^c}}, \vecx_{\cQ})\geq \epsilon\text{ for {some} }\cQ\in \cA, j_{\cQ}\in [1:n^2]^{|\cQ|}\text{ and }\vecx_{\cQ}\in \cX^n_{\cQ}}\\
&\leq \sum_{\cQ\in \cA, j_{\cQ}\in [1:n^2]^{|\cQ|},\vecx_{\cQ}\in \cX^n_{\cQ}}\bbP\inp{\frac{1}{(n^2)^{|\cQ^c|}}\sum_{j_{\cQ^c}\in [1:n^2]^{\pink{|\cQ^c|}}}  e_{\cQ}(F_{\cQ, j_{\cQ}}, F_{\cQ^c, j_{\cQ^c}}, \vecx_{\cQ})\geq \epsilon}
\end{align*}

Note that the summands in $\sum_{j_{\cQ^c}\in [1:n^2]^{\pink{|\cQ^c|}}}  e_{\cQ}(F_{\cQ, j_{\cQ}}, F_{\cQ^c, j_{\cQ^c}}, \vecx_{\cQ})$ are not necessarily independent. Hence, an exponential concentration inequality cannot be directly argued. However, using a similar procedure as Jahn~\cite[Theorem~1]{Jahn81}, we decompose this sum into parts that consist of  summands that are conditionally independent given the adversary's choices.

To this end, let $\Sigma_{n^2}:=\{\tau_i: i\in [0:n^2-1]\}$ be a set of permutations of $\{1, 2, \ldots, n^2\}$ with $$\tau_i(j) = (i+j)\, \textsf{mod} \,n^2\mbox{ for all }j\in [1:n^2].$$ Suppose $|\cQ|=l$ for some $l\in [1:k]$. For ease of notation, let $Q= \{1, 2, \ldots, l\}$. Then, 
\begin{align*}
\frac{1}{(n^2)^{|\cQ^c|}}&\sum_{j_{\cQ^c}\in [1:n^2]^{\pink{|\cQ^c|}}}  e_{\cQ}(F_{\cQ, j_{\cQ}}, F_{\cQ^c, j_{\cQ^c}},\vecx_{\cQ}) = \frac{1}{(n^2)^{k-l}}\sum_{(j_{l+1}, \ldots, j_{k})\in [1:n^2]^{k-l}}  e_{\cQ}(F_{\cQ, j_{\cQ}}, F_{\cQ^c, j_{\cQ^c}}, \vecx_{\cQ})\\
& = \frac{1}{(n^2)^{k-l-1}}\sum_{(\sigma_{l+2}, \sigma_{l+3}, \ldots, \sigma_k)\in \Sigma_{n^2}^{k-l-1}}\inp{\frac{1}{n^2}\sum_{j\in [n^2]}e_{\cQ}(F_{\cQ, j_{\cQ}}, (F_{l+1, j},F_{l+2, \sigma_{l+2}(j)},\ldots, F_{k, \sigma_{k}(j)}),  \vecx_{\cQ})}.
\end{align*} Now, 
\begin{align*}
&\bbP\inp{\frac{1}{(n^2)^{|\cQ^c|}}\sum_{j_{\cQ^c}\in [1:n^2]^{\pink{|\cQ^c|}}}  e_{\cQ}(F_{\cQ, j_{\cQ}}, F_{\cQ^c, j_{\cQ^c}}, \vecx_{\cQ})\geq \epsilon}\\
&= \bbP\inp{\sum_{(\sigma_{l+2}, \sigma_{l+3}, \ldots, \sigma_k)\in \Sigma_{n^2}^{k-l-1}}\inp{\frac{1}{n^2}\sum_{j\in [n^2]}e_{\cQ}(F_{\cQ, j_{\cQ}}, (F_{l+1, j},F_{l+2, \sigma_{l+2}(j)},\ldots, F_{k, \sigma_{k}(j)}), \vecx_{\cQ})}\geq {(n^2)^{k-l-1}}\epsilon}\\
&\leq \bbP\inp{\cup_{(\sigma_{l+2}, \sigma_{l+3}, \ldots, \sigma_k)\in \Sigma_{n^2}^{k-l-1}}\inb{\frac{1}{n^2}\sum_{j\in [n^2]}e_{\cQ}(F_{\cQ, j_{\cQ}}, (F_{l+1, j},F_{l+2, \sigma_{l+2}(j)},\ldots, F_{k, \sigma_{k}(j)}),  \vecx_{\cQ})\geq \epsilon}}\\
&\leq \sum_{(\sigma_{l+2}, \sigma_{l+3}, \ldots, \sigma_k)\in \Sigma_{n^2}^{k-l-1}}\bbP\inb{\frac{1}{n^2}\sum_{j\in [n^2]}e_{\cQ}(F_{\cQ, j_{\cQ}}, (F_{l+1, j},F_{l+2, \sigma_{l+2}(j)},\ldots, F_{k, \sigma_{k}(j)}),  \vecx_{\cQ})\geq \epsilon}.
\end{align*}We note that $e_{\cQ}(F_{\cQ, j_{\cQ}}, (F_{l+1, j},F_{l+2, \sigma_{l+2}(j)},\ldots, F_{k, \sigma_{k}(j)}), \vecx_{\cQ})$ is identically distributed  for all $(\sigma_{l+2}, \sigma_{l+3}, \ldots, \sigma_k)\in \Sigma_{n^2}^{k-l-1}$. Thus, 
\begin{align*}
&\sum_{(\sigma_{l+2}, \sigma_{l+3}, \ldots, \sigma_k)\in \Sigma_{n^2}^{k-l-1}}\bbP\inb{\frac{1}{n^2}\sum_{j\in [n^2]}e_{\cQ}(F_{\cQ, j_{\cQ}}, (F_{l+1, j},F_{l+2, \sigma_{l+2}(j)},\ldots, F_{k, \sigma_{k}(j)}), \vecx_{\cQ})\geq \epsilon}\\
& \leq (n^2)^{k-l-1}\bbP\inb{\frac{1}{n^2}\sum_{j\in [n^2]}e_{\cQ}(F_{\cQ, j_{\cQ}}, (F_{l+1, j},F_{l+2, \tau_{0}(j)},\ldots, F_{k, \tau_{0}(j)}), \vecx_{\cQ})\geq \epsilon}\\
&= (n^2)^{k-l-1}\bbP\inb{\frac{1}{n^2}\sum_{j\in [n^2]}e_{\cQ}(F_{\cQ, j_{\cQ}}, (F_{l+1, j},F_{l+2, j},\ldots, F_{k, j}), \vecx_{\cQ})\geq \epsilon}\\
&= (n^2)^{k-l-1}\bbP\inb{\sum_{j\in [n^2]}e_{\cQ}(F_{\cQ, j_{\cQ}}, (F_{l+1, j},F_{l+2, j},\ldots, F_{k, j}), \vecx_{\cQ})\geq n^2\epsilon}\\
&= (n^2)^{k-l-1}\bbP\inb{\exp\inb{\sum_{j\in [n^2]}e_{\cQ}(F_{\cQ, j_{\cQ}}, (F_{l+1, j},F_{l+2, j},\ldots, F_{k, j}),  \vecx_{\cQ})}\geq \exp\inb{n^2\epsilon}}\\
&\stackrel{(a)}{\leq} (n^2)^{k-l-1}\exp\inb{-n^2\epsilon}\bbE\insq{\exp\inb{\sum_{j\in [n^2]}e_{\cQ}(F_{\cQ, j_{\cQ}}, (F_{l+1, j},F_{l+2, j},\ldots, F_{k, j}),  \vecx_{\cQ})}}\\
&=(n^2)^{k-l-1}\exp\inb{-n^2\epsilon}\bbE\insq{\prod_{j\in [n^2]}\exp\inb{e_{\cQ}(F_{\cQ, j_{\cQ}}, (F_{l+1, j},F_{l+2, j},\ldots, F_{k, j}),  \vecx_{\cQ})}}\\
&=(n^2)^{k-l-1}\exp\inb{-n^2\epsilon}\bbE_{F_{\cQ, j_{\cQ}}}\insq{\bbE\insq{\prod_{j\in [n^2]}\exp\inb{e_{\cQ}(F_{\cQ, j_{\cQ}}, (F_{l+1, j},F_{l+2, j},\ldots, F_{k, j}),  \vecx_{\cQ})}\Bigg|F_{\cQ, j_{\cQ}}}}\\
&\stackrel{(b)}{=}(n^2)^{k-l-1}\exp\inb{-n^2\epsilon}\bbE_{F_{\cQ, j_{\cQ}}}\insq{\prod_{j\in [n^2]}\bbE\insq{\exp\inb{e_{\cQ}(F_{\cQ, j_{\cQ}}, (F_{l+1, j},F_{l+2, j},\ldots, F_{k, j}),  \vecx_{\cQ})}\Bigg|F_{\cQ, j_{\cQ}}}}\\
&\stackrel{(c)}{=}(n^2)^{k-l-1}\exp\inb{-n^2\epsilon}\bbE_{F_{\cQ, j_{\cQ}}}\insq{\inp{\bbE\insq{\exp\inb{e_{\cQ}(F_{\cQ, j_{\cQ}}, (F_{l+1, 1},F_{l+2, 1},\ldots, F_{k, 1}),  \vecx_{\cQ})}\Bigg|F_{\cQ, j_{\cQ}}}^{n^2}}}\\
&\stackrel{(d)}{\leq}(n^2)^{k-l-1}\exp\inb{-n^2\epsilon}\bbE_{F_{\cQ, j_{\cQ}}}\insq{\inp{1+\bbE\insq{{e_{\cQ}(F_{\cQ, j_{\cQ}}, (F_{l+1, 1},F_{l+2, 1},\ldots, F_{k, 1}),  \vecx_{\cQ})}\Bigg|F_{\cQ, j_{\cQ}}}}^{n^2}}\\
&\stackrel{(e)}{\leq}(n^2)^{k-l-1}\exp\inb{-n^2\epsilon}\bbE_{F_{\cQ, j_{\cQ}}}\insq{\inp{1+\pink{P^{\text{rand}}_{e}\inp{P_{F_{[1:k]}}, \phi}}}^{n^2}}\\
&=(n^2)^{k-l-1}\exp\inb{-n^2\epsilon}\inp{1+\pink{P^{\text{rand}}_{e}\inp{P_{F_{[1:k]}}, \phi}}}^{n^2}\\
&=\exp\inb{-n^2\inp{\epsilon-\log\inp{1+\pink{P^{\text{rand}}_{e}(P_{F_{[1:k]}}, \phi))}}-\frac{k-l-1}{n^2}\log(n^2)}}
\end{align*} where $(a)$ follows from Markov's inequality. $(b), (c)$ and $(d)$ hold because for each $j\in [1:n^2]$, conditioned on $F_{\cQ, j_{\cQ}}$,   $e_{\cQ}(F_{\cQ, j_{\cQ}}, (F_{l+1, j},F_{l+2, j},\ldots, F_{k, j}),  \vecx_{\cQ})$ are i.i.d. random variables taking values between $0$ and $1$ (recall that $2^t\leq 1+t$ for $t\in [0:1]$). The inequality $(e)$ follows from the definition of $\pink{P^{\text{rand}}_{e}(P_{F_{[1:k]}}, \phi))}$ by noting that for every realization $f_{\cQ}\in \cF_{\cQ}$ of $F_{\cQ, j_{\cQ}}$, $\bbE\insq{{e_{\cQ}(f_{\cQ}, (F_{l+1, 1},F_{l+2, 1},\ldots, F_{k, 1}),  \vecx_{\cQ})}}\leq \pink{P^{\text{rand}}_{e}(P_{F_{[1:k]}}, \phi))}$. This implies that the random variable $\bbE\insq{{e_{\cQ}(F_{\cQ, j_{\cQ}}, (F_{l+1, 1},F_{l+2, 1},\ldots, F_{k, 1}),  \vecx_{\cQ})}\Bigg|F_{\cQ, j_{\cQ}}}$ is upper bounded by $\pink{P^{\text{rand}}_{e}(P_{F_{[1:k]}}, \phi))}$.
Thus, 
\begin{align*}
&\bbP\inp{\frac{1}{(n^2)^{|\cQ^c|}}\sum_{j_{\cQ^c}\in [1:n^2]^{\pink{|\cQ^c|}}}  e_{\cQ}(F_{\cQ, j_{\cQ}}, F_{\cQ^c, j_{\cQ^c}}, \vecx_{\cQ})\geq \epsilon\text{ for some }\cQ\in \cA, j_{\cQ}\in [1:n^2]^{|\cQ|}\text{ and }\vecx_{\cQ}\in \cX^n_{\cQ}}\\
&\stackrel{(a)}{\leq} 2^k (n^2)^k\prod_{i\in [1:k]}|\cX_i|^{n}\exp\inb{-n^2\inp{\epsilon-\log\inp{1+\pink{P^{\text{rand}}_{e}(P_{F_{[1:k]}}, \phi))}}-\frac{k-l-1}{n^2}\log(n^2)}}\\
&\rightarrow 0\text{ for enough }n.
\end{align*}Here,  $(a)$ follows by recalling that $\pink{P^{\text{rand}}_{e}(P_{F_{[1:k]}}, \phi))}<2^{\epsilon/2}-1$ and thus, $\epsilon> 2\log\inp{1+\pink{P^{\text{rand}}_{e}(P_{F_{[1:k]}}, \phi))}}$.
\end{proof}

\section{Proof of Lemma~\ref{lem:positivity_implies_capacity}}\label{sec:positivity_implies_capacity}
\begin{proof}[Proof of Lemma~\ref{lem:positivity_implies_capacity}]
This can be shown along the lines of the proof of \cite[Theorem~12.11]{CK11}. 
For  $\epsilon>0$ and large enough $n$, let $(F_{[1:k]}, \phi(F_{[1:k]}))$ be an $(N_1, \ldots, N_k,n)$ randomized code  satisfying  $$\pink{P^{\text{rand}}_{e}(P_{F_{[1:k]}}, \phi)}<2^{\epsilon/2}-1.$$   
Applying Lemma~\ref{lemma:randomness_reduction} on this code,  for each user $j\in [1:k]$, we obtain  $n^2$ deterministic encoding maps  $f_{j,i}, i\in[1:n^2]$ such that for every $\cQ\in \cA, l_{\cQ}\in [1:n^2]^{|\cQ|}$, $\vecx_{\cQ}\in \cX^n_{\cQ}$ and the decoder $\phi$,
\begin{align}
&\frac{1}{(n^2)^{|\cQ^c|}}\sum_{l_{\cQ^c}\in [1:n^2]^{|\cQ^c|}}\frac{1}{(\prod_{i\in \cQ^c}N_i)}\nonumber\\
&\quad\sum_{m_{\cQ^c}\in \cM_{\cQ^c}}\bbP\inb{\inp{\phi(\vecY, f_{\cQ, l_{\cQ}}, f_{\cQ^c, l_{\cQ^c}})= \hat{m}_{[1:k]}\text{ such that }\hat{m}_{\cQ^c}\neq m_{\cQ^c}}|\vecX_{\cQ^c}=f_{\cQ^c, l_{\cQ^c}}(m_{\cQ^c}), \vecX_{\cQ}=\vecx_{\cQ}}<\epsilon.
 \label{eq:one}
\end{align}
Further, since $R_i>0$ is achievable for all $i\in [1:k]$, there exists an $(n^2, \ldots, n^2, k_n)$ code $(\hat{f}_{[1:k]}, \hat{\phi})$ where $k_n/n\rightarrow 0$ and 
\begin{align}\pink{\max_{\cQ\in \cA}P^{\text{rand}}_{e, \cQ}(\hat{f}_{[1:k]}, \hat{\phi})}\leq \epsilon\label{eq:two}\end{align}
for large enough $n$. We choose sufficiently large $n$ such that both \eqref{eq:one} and \eqref{eq:two} hold. For a vector sequence $\tilde{\vecs}\in \cS^{n+k_n}$ for any alphabet $\cS$, we write $\tilde{\vecs}= ({\hat{\vecs},\vecs})$, where $\hat{\vecs}$ denotes the first $k_n$-length part of $\tilde{\vecs}$ and ${\vecs}$ denotes the last $n$-length part of the $\tilde{\vecs}$. 
Let $(\tilde{f}_{[1:k]}, \tilde{\phi})$ be a new $(\tilde{N}_1, \ldots, \tilde{N}_k,\tilde{n})$ code where $\tilde{n}:=k_n+n$, message set for user-$i$, $\tilde{\cM}_i=[1:\tilde{N}_i]:=\{1, 2, \ldots, n^2\}\times [1:{N}_i]$. Further, for $l\in [1:n^2], m\in [1:{N}_i]$, let $\tilde{m} := (l,m)$. We define $\tilde{f}_i(\tilde{m}) = \tilde{f}_{i}(l, m):= (\hat{f}_{i}(l),f_{i,l}(m))$. For $\tilde{\vecy}=(\hat{\vecy},\vecy)$, let $\tilde{\phi}(\tilde{\vecy}):=(\hat{l}_{[1:k]}, \phi(\vecy,f_{[1:k],\hat{l}_{[1:k]}}))$ where  $\hat{l}_{[1:k]} = \hat{\phi}(\hat{\vecy})$. 
Then, for $\cQ\in \cA$,    
\begin{align*}
&P^{\text{rand}}_{e, \cQ}(\tilde{f}_{[1:k]}, \tilde{\phi}) \\
&= \max_{\stackrel{({\hat{\vecx}_{\cQ}, \vecx}_{\cQ})}{\in \cX^{k_n}_{\cQ}\times \cX^n_{\cQ}}}\frac{1}{(\prod_{i\in \cQ^c}\tilde{{N}}_i)}\sum_{\tilde{m}_{\cQ^c}\in \tilde{\cM}_{\cQ^c}}\bbP\inp{\inb{\tilde{\phi}(\tilde{\vecY})= {m'}_{[1:k]}\text{ such that }{m'}_{\cQ^c}\neq \tilde{m}_{\cQ^c}}\bigg|\tilde{\vecX}_{\cQ^c} = \tilde{f}_{\cQ^c}(\tilde{m}_{\cQ^c}), \tilde{\vecX}_{\cQ} = ({\hat{\vecx}_{\cQ},\vecx}_{\cQ})}\\
&\leq \max_{\stackrel{\hat{\vecx}_{\cQ},{\vecx}_{\cQ}}{l_{\cQ}}}\frac{1}{(n^2)^{|\cQ^c|}(\prod_{i\in \cQ^c}{{N}}_i)}\sum_{l_{\cQ^c}\in [1:n^2]^{|\cQ^c|}}\sum_{m_{\cQ^c}\in {\cM}_{\cQ^c}}\bbP\Bigg(\inb{\hat{\phi}(\hat{\vecY})= \bar{l}_{[1:k]}\text{ such that }\bar{l}_{\cQ^c}\neq l_{\cQ^c}}\\
&\hspace{0.6cm}\cup\inb{{\phi}({\vecY}, f_{\cQ, {l}_{\cQ}},f_{\cQ^c, {l}_{\cQ^c}})= \bar{m}_{[1:k]}\text{ such that }\bar{m}_{\cQ^c}\neq m_{\cQ^c}}\bigg|(\hat{\vecX}_{\cQ^c},{\vecX}_{\cQ^c}) = (\hat{f}_{\cQ^c}(l_{\cQ^c}),f_{\cQ^c,l_{\cQ^c}}(m_{\cQ^c})), (\hat{\vecX}_{\cQ} = \hat{\vecx}_{\cQ}, {\vecX}_{\cQ} = {\vecx}_{\cQ})\Bigg)\\
&\leq \max_{\hat{\vecx}_{\cQ}}\frac{1}{(n^2)^{|\cQ^c|}}\sum_{l_{\cQ^c}\in [1:n^2]^{|\cQ^c|}}\bbP\Bigg(\inb{\hat{\phi}(\hat{\vecY})= \bar{l}_{[1:k]}\text{ such that }\bar{l}_{\cQ^c}\neq l_{\cQ^c}}\bigg|\hat{\vecX}_{\cQ^c} = \hat{f}_{\cQ^c}(l_{\cQ^c}), \hat{\vecX}_{\cQ} = \hat{\vecx}_{\cQ}\Bigg)\\
&+ \max_{{\vecx}_{\cQ}, l_{\cQ}}\frac{1}{(n^2)^{|\cQ^c|}(\prod_{i\in \cQ^c}{{N}}_i)}\sum_{l_{\cQ^c}\in [1:n^2]^{|\cQ^c|}}\sum_{m_{\cQ^c}\in {\cM}_{\cQ^c}}\bbP\Bigg(\inb{{\phi}({\vecY}, f_{\cQ, {l}_{\cQ}},f_{\cQ^c, {l}_{\cQ^c}})= \bar{m}_{[1:k]}\text{ such that }\bar{m}_{\cQ^c}\neq m_{\cQ^c}}\bigg|\\
&\hspace{8cm}{\vecX}_{\cQ^c} = f_{\cQ^c,l_{\cQ^c}}(m_{\cQ^c}),  {\vecX}_{\cQ} = {\vecx}_{\cQ}\Bigg)\\
&\stackrel{(a)}{\leq}2\epsilon
\end{align*} where $(a)$ follows from \eqref{eq:one} and \eqref{eq:two}.
\end{proof}

\section{Proof of Lemma~\ref{achievability_k_users}}\label{app:k_users_positivity}
\pink{We first give the codebook which is given by Lemma~\ref{codebook_general} below. Its proof is along the lines of \cite[Lemma 2]{AhlswedeC99} and \cite[Lemma 3]{CsiszarN88} and is given later.}
\begin{lemma}\label{codebook_general}
For any  $\epsilon>0$, large enough $n, \, N\geq\exp(n\epsilon)$ and types $P_i\in \cP_{\cX_i}^n:i\in [1:k]$, there exist codebooks $\cC_i, \, i\in [1:k]$ for message sets $\cM_i = [1:N],\, i\in[1:k]$,  whose codewords are of type $P_i, \, i\in [1:k]$ respectively such that for every $\cQ\in \cA$ such that $|\cQ|<k$\footnote{Note that there are no decoding guarantees when all users are malicious, so we only consider the case when at least one user is honest.}, $\vecx_{\cQ}\in \cX^n_{\cQ}$,  $\cT\subseteq \cQ^c$, $\cJ\subseteq\cQ$, and joint type $P_{X_{\cQ^c}X_{\cQ}X'_{\cT}X'_{\cJ}}\in \cP^n_{\cX_{\cQ^c}\times\cX_{\cQ}\times\cX'_{\cT}\times\cX'_{\cJ}}$, the following holds:
\begin{enumerate}[label=(\alph*)]
    \item \label{codebook_general_eq1}  If for any $i\in \cQ^c$, $I(X_i;X_{\cQ^c\setminus\{i\}}X_{\cQ})\geq \epsilon$, then,
    \begin{align*}
\frac{1}{N^{|\cQ^c|}}|\{m_{\cQ^c}\in \cM_{\cQ^c}: (f_{\cQ^c}(m_{\cQ^c}),\vecx_{Q})\in T^{n}_{X_{\cQ^c} X_{\cQ}} \}| <\exp\left\{-n\epsilon/2\right\}.
\end{align*}
\item \label{codebook_general_eq2} If for any $i\in \cQ^c$, $I(X_i;X_{\cQ^c\setminus\{i\}}X'_{\cT} X'_{\cJ}X_{\cQ})\geq (|\cT|+|\cJ|)(1/ n )\log_2{N}+\epsilon$, then,
\begin{align*}
&\frac{1}{N^{|\cQ^c|}}|\{m_{\cQ^c}\in \cM_{\cQ^c}: \exists m'_{\cJ}\in \cM_{\cJ},\, m'_{\cT}\in \cM_{\cT},\, m'_i \neq m_i, \,\forall i \in \cT,\, (f_{\cQ^c}(m_{\cQ^c}),f_{\cT}(m'_{\cT}),f_{\cJ}(m'_{\cJ}),\vecx_{Q})\in T^{n}_{X_{\cQ^c}X'_{\cT} X'_{\cJ} X_{\cQ}} \}|\nonumber \\
&\hspace{4cm} <\exp\left\{-n\epsilon/2\right\}.
\end{align*}
\end{enumerate}
\end{lemma}
\begin{proof}[Proof of Lemma~\ref{achievability_k_users}]
For $\epsilon>0$ \pink{(fixed later)}, large enough $n, \, N\geq\exp(n\epsilon)$ and types $P_i\in \cP_{\cX_i}^n$, $i\in [1:k]$, such that $\min_{i\in [1:k]}\min_{x_i\in \cX_i}P_i>0$, consider the codebooks $\cC_i, \, i\in [1:k]$ for message sets $\cM_i = [1:N],\, i\in[1:k]$  as given by Lemma~\ref{codebook_general}. The rates of the codebooks $R_i = R = \log_2(N)/n$ for some $R\geq\epsilon$. The decoder is given by Definition~\ref{defn:dec_general} for $\eta$ satisfying Lemma~\ref{lemma:dec_disambiguity}.
We will choose $\epsilon$ such that $\eta>(2k+1)(k+1)\epsilon$.

Let $\cQ\in \cA$ be the set of adversarial users who attack with input vector $\vecx_{\cQ}\in \cX^n_{\cQ}$. The probability of error is given by 
\begin{align}\label{prob_error}
\frac{1}{N^{|\cQ^c|}}\sum_{m_{\cQ^c}}\bbP\inp{\inb{\vecy: \phi(\vecy)\neq (m_{\cQ^c}, m_{\cQ}) \text{for some }m_{\cQ}}|\vecX_{\cQ^c} = f_{\cQ^c}(m_{\cQ^c}), \vecX_{\cQ} = \vecx_{\cQ}}
\end{align}
From the decoder definition (Definition~\ref{defn:dec_general}), we know that if $\phi(\vecy) = \tilde{m}_{[1:k]}$ where $m_{\cQ^c}\neq \tilde{m}_{\cQ^c}$, then $\vecy\notin\cap_{i\in \cQ^c}\cD^{(i)}_{m_i}$, that is, $\vecy\in \cup_{i\in \cQ^c}(\cD^{i}_{m_1})^c$. Thus, \eqref{prob_error} can be written as
\begin{align}
&\frac{1}{N^{|\cQ^c|}}\sum_{m_{\cQ^c}}\bbP\inp{\inb{\vecy: \vecy\notin\cap_{i\in \cQ^c}\cD^{(i)}_{m_i}}|\vecX_{\cQ^c} = f_{\cQ^c}(m_{\cQ^c}), \vecX_{\cQ} = \vecx_{\cQ}}\nonumber\\
&\leq\sum_{i\in \cQ^c}\frac{1}{N^{|\cQ^c|}}\sum_{m_{\cQ^c}}\bbP\inp{\inb{\vecy: \vecy\notin\cD^{(i)}_{m_i}}|\vecX_{\cQ^c} = f_{\cQ^c}(m_{\cQ^c}), \vecX_{\cQ} = \vecx_{\cQ}}\label{prob_error_1}
\end{align}
We will show that each term in \eqref{prob_error_1} falls exponentially. 
It holds when for joint distribution $P_{X_{\cQ^c}X_{\cQ}}$ defined by $(f_{\cQ^c}(m_{\cQ^c}),\vecx_{\cQ})\in T^n_{X_{\cQ^c} X_{\cQ}}$,  $I(X_i;X_{\cQ^c\setminus\{i\}}X_{\cQ})\geq \epsilon$ for any  $i\in \cQ^c$. This is because for any $j\in \cQ^c$,
\begin{align*}
&\frac{1}{N^{|\cQ^c|}}\sum_{\stackrel{m_{\cQ^c}:(f_{\cQ^c}(m_{\cQ^c}),\vecx_{Q})\in T^{n}_{X_{\cQ^c} X_{\cQ}},}{I(X_i;X_{\cQ^c\setminus\{i\}}X_{\cQ})\geq \epsilon \text{ for some }i\in \cQ^c}}\bbP\inp{\inb{\vecy: \vecy\notin\cD^{(j)}_{m_j}}|\vecX_{\cQ^c} = f_{\cQ^c}(m_{\cQ^c}), \vecX_{\cQ} = \vecx_{\cQ}}\\
&\leq \sum_{\stackrel{P_{X_{\cQ^c}X_{\cQ}}:}{I(X_i;X_{\cQ^c\setminus\{i\}}X_{\cQ})\geq \epsilon}\text{ for some }i\in \cQ^c}\frac{1}{N^{|\cQ^c|}}|\{m_{\cQ^c}\in \cM_{\cQ^c}: (f_{\cQ^c}(m_{\cQ^c}),\vecx_{Q})\in T^{n}_{X_{\cQ^c} X_{\cQ}} \}|\\
&\stackrel{(a)}{<}\sum_{P_{X_{\cQ^c}X_{\cQ}}}\exp\left\{-n\epsilon/2\right\}\\
&\rightarrow 0
\end{align*} where $(a)$ follows from Lemma~\ref{codebook_general}\ref{codebook_general_eq1}.
Thus, we can assume that $I(X_i;X_{\cQ^c\setminus\{i\}}X_{\cQ})< \epsilon$ for all  $i\in \cQ^c$. This implies that
\begin{align*}
|\cQ^c|\epsilon> &\sum_{i\in \cQ^c}I(X_i;X_{\cQ^c\setminus\{i\}}X_{\cQ})\\
\geq & D\inp{P_{X_{\cQ^c}X_{\cQ}}\Big|\Big|(\prod_{i\in \cQ^c}P_{X_{i}})P_{X_{\cQ}}}.
\end{align*}
Under this case, for any $j\in \cQ^c$,
\begin{align}
&\frac{1}{N^{|\cQ^c|}}\sum_{\stackrel{m_{\cQ^c:(f_{\cQ^c}(m_{\cQ^c}),\vecx_{\cQ})\in T^n_{X_{\cQ^c}X_{\cQ}}},}{I(X_i;X_{\cQ^c\setminus\{i\}}X_{\cQ})<\epsilon \,\forall \,i\in \cQ^c}}\bbP\inp{\inb{\vecy: \vecy\notin\cD^{(j)}_{m_j}}|\vecX_{\cQ^c} = f_{\cQ^c}(m_{\cQ^c}), \vecX_{\cQ} = \vecx_{\cQ}}\label{prob_error_3}\\
&\leq\sum_{\stackrel{P_{X_{\cQ^c}X_{\cQ}}:}{I(X_i;X_{\cQ^c\setminus\{i\}}X_{\cQ})<\epsilon} \,\forall \,i\in \cQ^c}\frac{1}{N^{|\cQ^c|}}\sum_{\stackrel{m_{\cQ^c}:}{ (f_{\cQ^c}(m_{\cQ^c}),\vecx_{\cQ})\in T^n_{X_{\cQ^c}X_{\cQ}}}}\sum_{y\in T^n_{Y|X_{\cQ^c}X_{\cQ}}(f_{\cQ^c}(m_{\cQ^c}),\vecx_{\cQ})}W^n(y|f_{\cQ^c}(m_{\cQ^c}),\vecx_{\cQ})\nonumber\\
&\leq\sum_{\stackrel{P_{X_{\cQ^c}X_{\cQ}}:}{I(X_i;X_{\cQ^c\setminus\{i\}}X_{\cQ})<\epsilon} \,\forall \,i\in \cQ^c}\frac{1}{N^{|\cQ^c|}}\sum_{m_{\cQ^c}: (f_{\cQ^c}(m_{\cQ^c}),\vecx_{\cQ}) \in T^n_{X_{\cQ^c}X_{\cQ}}}\exp\inb{-n D\inp{P_{X_{\cQ^c}X_{\cQ}Y}\Big|\Big|P_{X_{\cQ^c}X_{\cQ}}W}}\nonumber\\
&\leq\sum_{\stackrel{P_{X_{\cQ^c}X_{\cQ}}:}{I(X_i;X_{\cQ^c\setminus\{i\}}X_{\cQ})<\epsilon} \,\forall \,i\in \cQ^c}\exp\inb{-n D\inp{P_{X_{\cQ^c}X_{\cQ}Y}\Big|\Big|P_{X_{\cQ^c}X_{\cQ}}W}}\nonumber\\
&\leq\sum_{\stackrel{P_{X_{\cQ^c}X_{\cQ}}:}{I(X_i;X_{\cQ^c\setminus\{i\}}X_{\cQ})<\epsilon} \,\forall \,i\in \cQ^c}\exp\inb{-n \inp{D\inp{P_{X_{\cQ^c}X_{\cQ}Y}\Big|\Big|(\prod_{i\in \cQ^c}P_{X_{i}})P_{X_{\cQ}}W}-D\inp{P_{X_{\cQ^c}X_{\cQ}}\Big|\Big|(\prod_{i\in \cQ^c}P_{X_{i}})P_{X_{\cQ}}}}}\label{eq:star1}
\end{align}
We will break \eqref{eq:star1} into two terms, first corresponding to joint distributions $P_{X_{\cQ^c}X_{\cQ}Y}$ for which \\$D\inp{P_{X_{\cQ^c}X_{\cQ}Y}\Big|\Big|(\prod_{i\in \cQ^c}P_{X_{i}})P_{X_{\cQ}}W}\geq \eta$ and second corresponding to joint distributions for which \\$D\inp{P_{X_{\cQ^c}X_{\cQ}Y}\Big|\Big|(\prod_{i\in \cQ^c}P_{X_{i}})P_{X_{\cQ}}W} < \eta$. Let us start by considering $P_{X_{\cQ^c}X_{\cQ}Y}$ such that \\$D\inp{P_{X_{\cQ^c}X_{\cQ}Y}\Big|\Big|(\prod_{i\in \cQ^c}P_{X_{i}})P_{X_{\cQ}}W}\geq \eta$.
\begin{align*}
&\sum_{\substack{P_{X_{\cQ^c}X_{\cQ}}:\\I(X_i;X_{\cQ^c\setminus\{i\}}X_{\cQ})<\epsilon,\,\forall \,i\in \cQ^c,\\D\inp{P_{X_{\cQ^c}X_{\cQ}Y}\Big|\Big|(\prod_{i\in \cQ^c}P_{X_{i}})P_{X_{\cQ}}W}\geq \eta} }\exp\inb{-n \inp{D\inp{P_{X_{\cQ^c}X_{\cQ}Y}\Big|\Big|(\prod_{i\in \cQ^c}P_{X_{i}})P_{X_{\cQ}}W}-D\inp{P_{X_{\cQ^c}X_{\cQ}}\Big|\Big|(\prod_{i\in \cQ^c}P_{X_{i}})P_{X_{\cQ}}}}}\\
&\leq\sum_{\stackrel{P_{X_{\cQ^c}X_{\cQ}}:}{I(X_i;X_{\cQ^c\setminus\{i\}}X_{\cQ})<\epsilon} \,\forall \,i\in \cQ^c}\exp\inb{-n \inp{\eta-|\cQ^c|\epsilon}}\\
& \rightarrow 0 \text{ for }\eta>k\epsilon.
\end{align*}
Now, we need to evaluate \eqref{eq:star1}  for joint distributions $P_{X_{\cQ^c}X_{\cQ}Y}$ for which $D\inp{P_{X_{\cQ^c}X_{\cQ}Y}\Big|\Big|(\prod_{i\in \cQ^c}P_{X_{i}})P_{X_{\cQ}}W}< \eta$. In this case, since decoding condition \ref{dec:1} holds, $\vecy\notin\cD^{(j)}_{m_j}$ if decoding condition \ref{dec:2} fails. That is, there exist $\cQ'\in \cA$, not necessarily distinct from $\cQ$, a non-empty set $\cT\subseteq (\cQ\cup\cQ')^c$ with $j\in \cT$,  $\vecx'_{\cQ'}\in \mathcal{X}^n_{\cQ'}$, $m'_{\cQ\setminus \cQ'}\in \cM_{\cQ\setminus \cQ'}$, ${m'}_{\cT}\in \cM_{\cT}$  such that $m'_{t}\neq {m}_t$ for all $t \in \cT$ such that for the joint distribution $P_{X_{\cQ^c}X_{\cQ}X'_{\cT} X'_{\cQ\setminus \cQ'}X'_{\cQ'}Y}$ defined by $(f_{\cQ^c}({m}_{\cQ^c}), \vecx_{\cQ},f_{\cT}(m'_{\cT}),f_{\cQ\setminus \cQ'}(m'_{\cQ\setminus \cQ'}),\vecx'_{\cQ'},  \vecy) \in \allowbreak T^{n}_{X_{\cQ^c}X_{\cQ}X'_{\cT} X'_{\cQ\setminus \cQ'}X'_{\cQ'}Y}$, \\
\begin{align*}
     &D(P_{X'_{\cT}X'_{\cQ\setminus \cQ'}X_{\cQ^c\setminus{(\cT\cup\cQ')}} X'_{\cQ'} Y}||(\prod_{t\in \cT}P_{X'_{t}})(\prod_{j\in \cQ\setminus \cQ'}P_{X'_{j}})(\prod_{l\in {\cQ^c\setminus{(\cT\cup\cQ')}}}P_{X_{l}}) P_{X'_{\cQ'}} W)< \eta \\
    &\text{and }I(X_{\cQ^c}Y;X'_{\cT} X'_{\cQ\setminus \cQ'}|X_{\cQ}) \geq \eta. 
\end{align*}
Let $\cH := {\cQ^c\setminus{(\cT\cup\cQ')}}$ and $\cP^1_{X_{\cQ^c}X_{\cQ}X'_{\cT} X'_{\cQ\setminus \cQ'}X'_{\cQ'}Y}$ be the set of distributions $P_{X_{\cQ^c}X_{\cQ}X'_{\cT} X'_{\cQ\setminus \cQ'}X'_{\cQ'}Y}$ satisfying $D\inp{P_{X_{\cQ^c}X_{\cQ}Y}\Big|\Big|(\prod_{i\in \cQ^c}P_{X_{i}})P_{X_{\cQ}}W}\leq \eta$, $D(P_{X'_{\cT},X'_{\cQ\setminus \cQ'},X_{\cH}, X'_{\cQ'}, Y}||(\prod_{t\in \cT}P_{X'_{t}})(\prod_{j\in \cQ\setminus \cQ'}P_{X'_{j}})(\prod_{l\in \cH}P_{X_{l}}) P_{X'_{\cQ'}} W)< \eta$ and $I(X_{\cQ^c}Y;X'_{\cT} X'_{\cQ\setminus \cQ'}|X_{\cQ}) \geq \eta$. Using these definitions we see that, in this case, \eqref{prob_error_3} is upper bounded by
\begin{align*}
&\sum_{\substack{P_{X_{\cQ^c}X_{\cQ}X'_{\cT} X'_{\cQ\setminus \cQ'}X'_{\cQ'}Y}\\\in\cP^1_{X_{\cQ^c}X_{\cQ}X'_{\cT} X'_{\cQ\setminus \cQ'}X'_{\cQ'}Y}}}\frac{1}{N^{|\cQ^c|}}\Big|\Big\{m_{\cQ^c}\in \cM_{\cQ^c}: \exists m'_{\cQ\setminus\cQ'}\in \cM_{\cQ\setminus\cQ'},\, m'_{\cT}\in \cM_{\cT}\text{ where } m'_i \neq m_i \text{ for all } i \in \cT,\\
&\hspace{5cm}\, \text{ such that }(f_{\cQ^c}(m_{\cQ^c}),f_{\cT}(m'_{\cT}),f_{\cQ\setminus\cQ'}(m'_{\cQ\setminus\cQ'}),\vecx_{Q})\in T^{n}_{X_{\cQ^c}X'_{\cT} X'_{\cQ\setminus\cQ'} X_{\cQ}} \Big\}\Big|\\
\leq &\sum_{\substack{P_{X_{\cQ^c}X_{\cQ}X'_{\cT} X'_{\cQ\setminus \cQ'}X'_{\cQ'}Y}\\\in\cP_{X_{\cQ^c}X_{\cQ}X'_{\cT} X'_{\cQ\setminus \cQ'}X'_{\cQ'}Y}}}\exp\{-n\epsilon/2\}\hspace{1cm}\rightarrow 0
\end{align*}if for any $i\in \cQ^c$, $I(X_i;X_{\cQ^c\setminus\{i\}}X'_{\cT} X'_{\cQ\setminus\cQ'}X_{\cQ})\geq (|\cT|+|(\cQ\setminus\cQ')|)R+\epsilon$. This follows from the codebook property Lemma~\ref{codebook_general}\ref{codebook_general_eq2}. Thus, we only need to consider joint distributions for which $I(X_i;X_{\cQ^c\setminus\{i\}}X'_{\cT} X'_{\cQ\setminus\cQ'}X_{\cQ})< (|\cT|+|(\cQ\setminus\cQ')|)R+\epsilon$ for all $i\in \cQ^c$. This implies that $I(X_{\cQ^c};X'_{\cQ\setminus\cQ'}X'_{\cT}|X_{\cQ})\leq |\cQ^c|((|\cT|+|(\cQ\setminus\cQ')|)R+\epsilon)$. This is because $I(X_{\cQ^c};X'_{\cQ\setminus\cQ'}X'_{\cT}|X_{\cQ})\leq \sum_{i\in \cQ^c}I(X_{i};X'_{\cQ\setminus\cQ'}X'_{\cT}X_{\cQ^c\setminus\{i\}}|X_{\cQ})$.

 Let $\cP^2_{X_{\cQ^c}X_{\cQ}X'_{\cT} X'_{\cQ\setminus \cQ'}X'_{\cQ'}Y} = \{P_{X_{\cQ^c}X_{\cQ}X'_{\cT} X'_{\cQ\setminus \cQ'}X'_{\cQ'}Y}\in\cP^1_{X_{\cQ^c}X_{\cQ}X'_{\cT} X'_{\cQ\setminus \cQ'}X'_{\cQ'}Y}:I(X_{\cQ^c};X'_{\cQ\setminus\cQ'}X'_{\cT}|X_{\cQ})\leq |\cQ^c|((|\cT|+|(\cQ\setminus\cQ')|)R+\epsilon)\}$.
So, for any $j\in \cQ^c$, it is sufficient to analyze the following:
\begin{align*}
&\sum_{\cQ'\in\cA}\sum_{\stackrel{P_{X_{\cQ^c}X_{\cQ}X'_{\cT} X'_{\cQ\setminus \cQ'}X'_{\cQ'}Y}\in:}{\cP^2_{X_{\cQ^c}X_{\cQ}X'_{\cT} X'_{\cQ\setminus \cQ'}X'_{\cQ'}Y}}}\frac{1}{N^{|\cQ^c|}}\sum_{\stackrel{m_{\cQ^c}:}{ (f_{\cQ^c}(m_{\cQ^c}),\vecx_{\cQ})\in T^n_{X_{\cQ^c}X_{\cQ}}}}\sum_{\stackrel{\vecy\in T^n_{Y|X_{\cQ^c}X_{\cQ}}(f_{\cQ^c}(m_{\cQ^c}),\vecx_{\cQ})}{\vecy\notin\cD^{(j)}_{m_j}}}W^n(y|f_{\cQ^c}(m_{\cQ^c}),\vecx_{\cQ})\\
&\leq \sum_{\cQ'\in\cA}\sum_{\stackrel{P_{X_{\cQ^c}X_{\cQ}X'_{\cT} X'_{\cQ\setminus \cQ'}X'_{\cQ'}Y}\in:}{\cP^2_{X_{\cQ^c}X_{\cQ}X'_{\cT} X'_{\cQ\setminus \cQ'}X'_{\cQ'}Y}}}\frac{1}{N^{|\cQ^c|}}\sum_{\stackrel{m_{\cQ^c}:}{ (f_{\cQ^c}(m_{\cQ^c}),\vecx_{\cQ})\in T^n_{X_{\cQ^c}X_{\cQ}}}}\sum_{\substack{m'_{\cQ\setminus \cQ'}\in \cM_{\cQ\setminus \cQ'},\\{m'}_{\cT}\in \cM_{\cT}\\(f_{\cQ\setminus\cQ'}(m'_{\cQ\setminus \cQ'}),f_{\cT}(m'_{\cT}))\in T^n_{X'_{\cQ\setminus\cQ'}X'_{\cT}|X_{\cQ^c}X_{\cQ}}(f_{\cQ^c}(m_{\cQ^c}),\vecx_{\cQ})}}\\
&\sum_{{\vecy\in T^n_{Y|X_{\cQ^c}X_{\cQ}X'_{\cQ\setminus\cQ'}X'_{\cT}}(f_{\cQ^c}(m_{\cQ^c}),\vecx_{\cQ},f_{\cQ\setminus\cQ'}(m'_{\cQ\setminus \cQ'}),f_{\cT}(m'_{\cT}))}}W^n(y|f_{\cQ^c}(m_{\cQ^c}),\vecx_{\cQ})\\
&\leq \sum_{\cQ'\in\cA}\sum_{\stackrel{P_{X_{\cQ^c}X_{\cQ}X'_{\cT} X'_{\cQ\setminus \cQ'}X'_{\cQ'}Y}\in:}{\cP^2_{X_{\cQ^c}X_{\cQ}X'_{\cT} X'_{\cQ\setminus \cQ'}X'_{\cQ'}Y}}}\exp\inb{n(|\cQ\setminus\cQ'|+|\cT|)R-I(Y;X'_{\cQ\setminus\cQ'}X'_{\cT}|X_{\cQ^c}X_{\cQ})+\epsilon)}\\
&= \sum_{\cQ'\in\cA}\sum_{\stackrel{P_{X_{\cQ^c}X_{\cQ}X'_{\cT} X'_{\cQ\setminus \cQ'}X'_{\cQ'}Y}\in:}{\cP^2_{X_{\cQ^c}X_{\cQ}X'_{\cT} X'_{\cQ\setminus \cQ'}X'_{\cQ'}Y}}}\exp\inb{n(|\cQ\setminus\cQ'|+|\cT|)R-I(X_{\cQ^c}Y;X'_{\cQ\setminus\cQ'}X'_{\cT}|X_{\cQ})+ I(X_{\cQ^c};X'_{\cQ\setminus\cQ'}X'_{\cT}|X_{\cQ}) +\epsilon)}\\
\end{align*}
From the definitions of $\cP^2_{X_{\cQ^c}X_{\cQ}X'_{\cT} X'_{\cQ\setminus \cQ'}X'_{\cQ'}Y}$ and $\cP^1_{X_{\cQ^c}X_{\cQ}X'_{\cT} X'_{\cQ\setminus \cQ'}X'_{\cQ'}Y}$, we not that $I(X_{\cQ^c}Y;X'_{\cQ\setminus\cQ'}X'_{\cT}|X_{\cQ})\geq \eta$ and $I(X_{\cQ^c};X'_{\cQ\setminus\cQ'}X'_{\cT}|X_{\cQ})\leq |\cQ^c|((|\cT|+|(\cQ\setminus\cQ')|)R+\epsilon)$. 
This implies that 
\begin{align*}
&\sum_{\cQ'\in\cA}\sum_{\stackrel{P_{X_{\cQ^c}X_{\cQ}X'_{\cT} X'_{\cQ\setminus \cQ'}X'_{\cQ'}Y}\in:}{\cP^2_{X_{\cQ^c}X_{\cQ}X'_{\cT} X'_{\cQ\setminus \cQ'}X'_{\cQ'}Y}}}\exp\inb{n(|\cQ\setminus\cQ'|+|\cT|)R-I(X_{\cQ^c}Y;X'_{\cQ\setminus\cQ'}X'_{\cT}|X_{\cQ})+ I(X_{\cQ^c};X'_{\cQ\setminus\cQ'}X'_{\cT}|X_{\cQ}) +\epsilon)}\\
&\leq \sum_{\cQ'\in\cA}\sum_{\stackrel{P_{X_{\cQ^c}X_{\cQ}X'_{\cT} X'_{\cQ\setminus \cQ'}X'_{\cQ'}Y}\in:}{\cP^2_{X_{\cQ^c}X_{\cQ}X'_{\cT} X'_{\cQ\setminus \cQ'}X'_{\cQ'}Y}}}\exp\inb{n((|\cQ\setminus\cQ'|+|\cT|)R-\eta+ |\cQ^c|((|\cT|+|(\cQ\setminus\cQ')|)R+\epsilon) +\epsilon)}\\
&\leq 2^k\sum_{\stackrel{P_{X_{\cQ^c}X_{\cQ}X'_{\cT} X'_{\cQ\setminus \cQ'}X'_{\cQ'}Y}\in:}{\cP^2_{X_{\cQ^c}X_{\cQ}X'_{\cT} X'_{\cQ\setminus \cQ'}X'_{\cQ'}Y}}}\exp\inb{n(kR-\eta+ k(kR+\epsilon) +\epsilon)}\\
&= 2^k\sum_{\stackrel{P_{X_{\cQ^c}X_{\cQ}X'_{\cT} X'_{\cQ\setminus \cQ'}X'_{\cQ'}Y}\in:}{\cP^2_{X_{\cQ^c}X_{\cQ}X'_{\cT} X'_{\cQ\setminus \cQ'}X'_{\cQ'}Y}}}\exp\inb{n((k+k^2)R-\eta+(k+1)\epsilon)}\\
&\rightarrow 0 \text{ for }R<\frac{\eta-(k+1)\epsilon}{k+k^2}.
\end{align*}
Since $\eta>(2k+1)(k+1)\epsilon$, $\frac{\eta-(k+1)\epsilon}{k+k^2}>2\epsilon$. Thus, we can choose $R$ between $\epsilon$ and $2\epsilon$.
\end{proof}

\begin{proof}[Proof of Lemma~\ref{codebook_general}]
This proof is along the lines of \cite[Lemma 2]{AhlswedeC99} and \cite[Lemma 3]{CsiszarN88}.
We will generate the codebooks by a random experiment. For any $\cQ\in \cA$, $\vecx_{\cQ}\in \cX^n_{\cQ}$ and joint type $P_{X_{\cQ^c}X'_{\cT}X'_{\cJ}X_{\cQ}}\in \cP^n_{\cX_{\cQ^c}\times \cX_{\cT}\times\cX'_{\cJ}\times\cX_{\cQ}}$, we will show that the probability that statement \ref{codebook_general_eq2} does not hold, falls doubly exponentially in $n$. We will only analyze statement \ref{codebook_general_eq2} as choosing $\cT = \cJ = \emptyset$ in \ref{codebook_general_eq2}will also imply that the probability that \ref{codebook_general_eq1} does not hold also falls doubly exponentially. Since $\cA$,  $|\cX_{\cQ^c}^n|$ and $|\cP^n_{\cX_{\cQ^c}\times \cX_{\cT}\times\cX'_{\cJ}\times\cX_{\cQ}}|$  grow  at most exponentially in $n$, a union bound will imply  the existence of codebooks satisfying~\ref{codebook_general_eq1} and \ref{codebook_general_eq2}. The proof will use \cite[Lemma A1]{CsiszarN88} which we restate here for a quick reference. 
\begin{lemma}\cite[Lemma A1]{CsiszarN88}\label{lemma:CN}
Let $\vecZ_1, \ldots, \vecZ_N$ be arbitrary random variables, and let $f_j(\vecZ_1, \ldots, \vecZ_j)$ be arbitrary with $0\leq f_j\leq 1, \, j\in 1, \ldots, N$. Then the condition
\begin{align*}
E\insq{f_j(\vecZ_1, \ldots, \vecZ_j)|\vecZ_1, \ldots, \vecZ_{j-1}}\leq a, \qquad j \in [1:N], 
\end{align*}
implies that for any real number $t$,
\begin{align*}
\bbP\inb{\frac{1}{N}\sum_{j=1}^{N}f_j(\vecZ_1, \ldots, \vecZ_j)>t}\leq \exp{\inb{-N\inp{t-a\log{e}}}}.
\end{align*}
\end{lemma}
 Let $T^n_{i}, \, i \in [1:k]$ denote the type class of $P_{i}$. We generate independent random codebooks for each user. The codebook for user $i\in [1:k]$, denoted by $\cC_i = \inp{\vecX_{i,1}, \vecX_{i,2}, \ldots, \vecX_{i,N}}$, consists of independent random vectors each distributed uniformly on $T^n_{i}$.  
Fix $\cQ\in \cA$,  $\vecx_{\cQ}  \in \mathcal{X}^n_{\cQ}$ and a joint type $P_{X_{\cQ^c}X'_{\cT}X'_{\cJ}X_{\cQ}}$ such that for every $i\in \cQ^c$, $P_{X_{i}} = P_{i}$, for $t\in \cT$, $P_{X'_{t}} = P_{t}$ and for $j\in \cJ$, $P_{X'_{j}} = P_{j}$ and $\vecx_{\cQ}\in T^n_{X_{\cQ}}$.
We will analyze the probability that \ref{codebook_general_eq2} does not hold under the randomness of codebook generation process. Note that the bound in \ref{codebook_general_eq2} is non-trivial only when $\cQ^c\neq \emptyset$. For any $l\in\cQ^c$,
\begin{align}
&\bbP\Big(\frac{1}{N^{|\cQ^c|}}|\{m_{\cQ^c}\in \cM_{\cQ^c}: \exists m'_{\cT}\in \cM_{\cT},\, m'_i \neq m_i \text{ for all } i \in \cT,\, m'_{\cJ}\in \cM_{\cJ},	\,\nonumber\\
&\hspace{3cm}(\vecX_{{\cQ^c},m_{\cQ^c}},\vecX_{{\cT},m'_{\cT}},\vecX_{\cJ,m'_{\cJ}},\vecx_{Q})\in T^{n}_{X_{\cQ^c}X'_{\cT} X'_{\cJ} X_{\cQ}} \}|> \exp\left\{-n\epsilon/2\right\}\Big)\nonumber\\
&=\bbP\Big(\sum_{m_{\cQ^c\setminus\{l\}}\in\cM_{\cQ^c\setminus\{l\}}}\frac{1}{N}|\{m_{l}\in \cM_{l}: \exists m'_{\cT}\in \cM_{\cT}\, m'_i \neq m_i \text{ for all } i \in \cT,\, m'_{\cJ}\in \cM_{\cJ},\, \nonumber\\
&\hspace{3cm}(\vecX_{\cQ^c,m_{\cQ^c}},\vecX_{\cT,m'_{\cT}},\vecX_{\cJ,m'_{\cJ}},\vecx_{Q})\in T^{n}_{X_{\cQ^c}X'_{\cT} X'_{\cJ} X_{\cQ}} \}|> N^{|\cQ^c|-1}\exp\left\{-n\epsilon/2\right\}\Big)\nonumber\\
&\leq\sum_{m_{\cQ^c\setminus\{l\}}\in\cM_{\cQ^c\setminus\{l\}}}\bbP\Big(\frac{1}{N}|\{m_{l}\in \cM_{l}: \exists m'_{\cT}\in \cM_{\cT}\, m'_i \neq m_i \text{ for all } i \in \cT,\, m'_{\cJ}\in \cM_{\cJ},\, \nonumber\\
&\hspace{3cm}(\vecX_{\cQ^c,m_{\cQ^c}},\vecX_{\cT,m'_{\cT}},\vecX_{\cJ,m'_{\cJ}},\vecx_{Q})\in T^{n}_{X_{\cQ^c}X'_{\cT} X'_{\cJ} X_{\cQ}} \}|> \exp\left\{-n\epsilon/2\right\}\Big)\nonumber\\
&\leq\sum_{\stackrel{m_{\cQ^c\setminus\{l\}}}{\in\cM_{\cQ^c\setminus\{l\}}}}\Bigg(\sum_{\stackrel{\vecx_{\cQ^c\setminus\{l\}}}{\in T^n_{X_{\cQ^c\setminus\{l\}}|X_{\cQ}}(\vecx_Q)}}\bbP\big(\vecX_{\cQ^c\setminus\{l\}, m_{\cQ^c\setminus\{l\}}}=\vecx_{\cQ^c\setminus\{l\}}\big)\bbP\Big(\frac{1}{N}|\{m_{l}\in \cM_{l}: \exists m'_{\cT}\in \cM_{\cT},\, m'_i \neq m_i \text{ for all } i \in \cT, \nonumber\\
&\hspace{1.5cm}\,m'_{\cJ}\in \cM_{\cJ}: (\vecX_{l, m_{l}}, \vecx_{\cQ^c\setminus\{l\}},\vecX_{\cT,m'_{\cT}},\vecX_{\cJ,m'_{\cJ}},\vecx_{Q})\in T^{n}_{X_{\cQ^c}X'_{\cT} X'_{\cJ} X_{\cQ}} \}|> \exp\left\{-n\epsilon/2\right\}\Big)\Bigg).\label{eq:last}
\end{align}
To analyze this, we first consider the case when $\cT\neq \emptyset$. Recall that $\cT\subseteq\cQ^c$. Without loss of generality, suppose $1\in \cQ^c\cap\cT$. Then for $l = 1$, we note that
\begin{align}
&\bbP\Big(\frac{1}{N}|\{m_{1}\in \cM_{1}: \exists m'_{\cT}\in \cM_{\cT},\, m'_t \neq m_t \text{ for all } t \in \cT, m'_{\cJ}\in \cM_{\cJ},\,\nonumber\\
&\hspace{1.5cm}\,(\vecX_{1, m_{1}}, \vecx_{\cQ^c\setminus\{1\}},\vecX_{\cT,m'_{\cT}},\vecX_{\cJ,m'_{\cJ}},\vecx_{Q})\in T^{n}_{X_{\cQ^c}X'_{\cT} X'_{\cJ} X_{\cQ}} \}|> \exp\left\{-n\epsilon/2\right\}\Big)\Bigg).\nonumber\\
&=\bbP\Big(\frac{1}{N}|j\in \cM_{1}: \exists i<j, i \in \cM_1, m'_{\cT\setminus\{1\}}\in \cM_{\cT\setminus\{1\}},\, m'_t \neq m_t \text{ for all } t \in \cT\setminus\{1\}, m'_{\cJ}\in \cM_{\cJ},\,\nonumber\\
&\hspace{1.5cm}\,((\vecX_{1, j}, \vecx_{\cQ^c\setminus\{1\}}),(\vecX_{1,i},\vecX_{\cT\setminus\{1\},m'_{\cT\setminus\{1\}}}),\vecX_{\cJ,m'_{\cJ}},\vecx_{Q})\in T^{n}_{X_{\cQ^c}X'_{\cT} X'_{\cJ} X_{\cQ}} \}|\nonumber\\
&\hspace{1.5cm}+\frac{1}{N}|j\in \cM_{1}: \exists i>j, i \in \cM_1, m'_{\cT\setminus\{1\}}\in \cM_{\cT\setminus\{1\}},\, m'_t \neq m_t \text{ for all } t \in \cT\setminus\{1\}, m'_{\cJ}\in \cM_{\cJ},\,\nonumber\\
&\hspace{1.5cm}\,((\vecX_{1, j}, \vecx_{\cQ^c\setminus\{1\}}),(\vecX_{1,i},\vecX_{\cT\setminus\{1\},m'_{\cT\setminus\{1\}}}),\vecX_{\cJ,m'_{\cJ}},\vecx_{Q})\in T^{n}_{X_{\cQ^c}X'_{\cT} X'_{\cJ} X_{\cQ}} \}|> \exp\left\{-n\epsilon/2\right\}\Big)\Bigg).\nonumber\\
&\leq\bbP\Big(\frac{1}{N}|j\in \cM_{1}: \exists i<j, i \in \cM_1, m'_{\cT\setminus\{1\}}\in \cM_{\cT\setminus\{1\}},\, m'_t \neq m_t \text{ for all } t \in \cT\setminus\{1\}, m'_{\cJ}\in \cM_{\cJ},\,\nonumber\\
&\hspace{1.5cm}\,((\vecX_{1, j}, \vecx_{\cQ^c\setminus\{1\}}),(\vecX_{1,i},\vecX_{\cT\setminus\{1\},m'_{\cT\setminus\{1\}}}),\vecX_{\cJ,m'_{\cJ}},\vecx_{Q})\in T^{n}_{X_{\cQ^c}X'_{\cT} X'_{\cJ} X_{\cQ}} \}|> \frac{1}{2}\exp\left\{-n\epsilon/2\right\}\Big)\Bigg)\label{eq:last1}\\
&\hspace{0.5cm}+\bbP\Big(\frac{1}{N}|j\in \cM_{1}: \exists i>j, i \in \cM_1, m'_{\cT\setminus\{1\}}\in \cM_{\cT\setminus\{1\}},\, m'_t \neq m_t \text{ for all } t \in \cT\setminus\{1\}, m'_{\cJ}\in \cM_{\cJ},\,\nonumber\\
&\hspace{1.5cm}\,((\vecX_{1, j}, \vecx_{\cQ^c\setminus\{1\}}),(\vecX_{1,i},\vecX_{\cT\setminus\{1\},m'_{\cT\setminus\{1\}}}),\vecX_{\cJ,m'_{\cJ}},\vecx_{Q})\in T^{n}_{X_{\cQ^c}X'_{\cT} X'_{\cJ} X_{\cQ}} \}|> \frac{1}{2}\exp\left\{-n\epsilon/2\right\}\Big)\Bigg).\label{eq:last2}
\end{align}

We will now analyze \eqref{eq:last1} using Lemma~\ref{lemma:CN}. For $j\in [1:N]$, let $Z_j = (\vecX_{1,j}, \cC_{[2:k]})$ where the codewords for $m_{\cQ^c\setminus\{i\}}$ are fixed to $\vecx_{\cQ^c\setminus\{1\}}$.  Let $f^{\vecx_{\cQ^c\setminus\{1\}}}_j(\vecZ_1, \ldots, \vecZ_j)$ be defined as
\begin{align*}
&f^{\vecx_{\cQ^c\setminus\{1\}}}_j(\vecZ_{1}, \ldots, \vecZ_{j}) = f^{\vecx_{\cQ^c\setminus\{1\}}}_j((\vecX_{1,1},\cC_{[2:k]}), \ldots, (\vecX_{1,j}, \cC_{[2:k]}))\\ 
&\hspace{2cm}= \begin{cases}1,&\text{ if }\exists \,i <j, (m'_{\cT\setminus\{1\}}, m'_{\cJ})\in \cM_{\cT\setminus\{1\}}\times\cM_{\cJ}\text{ such that }\forall t\in \cT\setminus\{1\}, m'_t\neq m_t,
\\&((\vecX_{1,j},\vecx_{\cQ^c\setminus\{1\}}),(\vecX_{1,i},\vecX_{(\cT\setminus\{1\},m'_{\cT\setminus\{1\}})}),\vecX_{(\cJ,m'_{\cJ})}, \vecx_{\cQ})\in T^n_{X_{\cQ^c}X'_{\cT}X'_{\cJ}X_{\cQ}},\\
		0,&\text{ otherwise. }\end{cases}
\end{align*} 
For $t= \frac{1}{2}\exp\inb{-n\epsilon/2}$, \eqref{eq:last1} can be written as
\begin{align*}
\inp{\frac{1}{N}\sum_{j= 1}^{N}f^{\vecx_{\cQ^c\setminus\{1\}}}_j(\vecZ_{1}, \ldots, \vecZ_{j})>t}
\end{align*}

We will compute $a$ in Lemma~\ref{lemma:CN}.
\begin{align*}
&\bbE\insq{f^{\vecx_{\cQ^c\setminus\{1\}}}_j((\vecX_{1,1},\cC_{[2:k]})), \ldots, (\vecX_{1,j}, \cC_{[2:k]}))|(\vecX_{1,1},\cC_{[2:k]}), \ldots, (\vecX_{1,(j-1)}, \cC_{[2:k]})}\\
&\leq \sum_{\substack{i \in \cM_1, i<j\\ (m'_{\cT\setminus\{1\}}, m'_{\cJ})\in \cM_{\cT\setminus\{1\}}\times\cM_{\cJ}\\m'_t\neq m_t, t\in \cT\setminus\{1\}}}\bbP\Big(((\vecX_{1,j},\vecx_{\cQ^c\setminus\{1\}}),(\vecX_{1,i},\vecX_{(\cT\setminus\{1\},m'_{\cT\setminus\{1\}})}),\vecX_{(\cJ,m'_{\cJ})}, \vecx_{\cQ})\in T^n_{X_{\cQ^c}X'_{\cT}X'_{\cJ}X_{\cQ}}\\
&\hspace{12cm}\Big|(\vecX_{1,1},\cC_{[2:k]}), \ldots, (\vecX_{1,(j-1)}, \cC_{[2:k]})\Big)\\
&\leq N^{\inp{|\cT|+|\cJ|}}\frac{\exp{\{nH(X_1|X_{\cQ^c\setminus\{1\}}X'_{\cT}X'_{\cJ}X_{\cQ})\}}}{(n+1)^{-|\cX_1|}\exp{\{nH(X_1)\}}}\\
&=(n+1)^{|\cX_1|}\exp{\inb{n\inp{(|\cT|+|\cJ|)(1/n)\log_2{N}-I(X_1;X_{\cQ^c\setminus\{1\}}X'_{\cT}X'_{\cJ}X_{\cQ})}}}
\end{align*}
Suppose $I(X_1;X_{\cQ^c\setminus\{1\}}X'_{\cT}X'_{\cJ}X_{\cQ})>(|\cT|+|\cJ|)(1/n)\log_2{N} +\epsilon$. Then, $a = (n+1)^{|\cX_1|}\exp{\inb{-n\epsilon}}$.
Thus,
\begin{align*}
&\inp{\frac{1}{N}\sum_{j= 1}^{N}f^{\vecx_{\cQ^c\setminus\{1\}}}_j(\vecZ_{1}, \ldots, \vecZ_{j})>t}\\
&\leq \exp\inb{-N\inp{t-a\log_2e}}\\
& = \exp\inb{-N\inp{\frac{1}{2}\exp\inb{-n\epsilon/2}-(n+1)^{|\cX_1|}\exp\inb{-n\epsilon}}}\\
&\leq \exp\inb{\inp{-\frac{1}{2}\exp\inb{n\epsilon/2}+(n+1)^{|\cX_1|}}}\text{ because }N\leq \exp\inb{n\epsilon}.\\
\end{align*}
Thus, \eqref{eq:last1} falls doubly exponentially. Since \eqref{eq:last2} is symmetric to \eqref{eq:last1}, we can obtain the same upper bound for \eqref{eq:last2} as well. This implies that \eqref{eq:last} falls doubly exponentially when $\cT\neq \emptyset$. Now, we consider the case when $\cT=\emptyset$. In this case, in order to show that \eqref{eq:last} falls doubly exponentially, we need to show that
\begin{align}\label{eq:last3}
\bbP\Big(\frac{1}{N}|\{m_{l}\in \cM_{l}: \exists m'_{\cJ}\in \cM_{\cJ},\,(\vecX_{l, m_{l}}, \vecx_{\cQ^c\setminus\{l\}},\vecX_{\cJ,m'_{\cJ}},\vecx_{Q})\in T^{n}_{X_{\cQ^c} X'_{\cJ} X_{\cQ}} \}|> \exp\left\{-n\epsilon/2\right\}\Big)\Bigg)
\end{align} falls doubly exponentially. This can be shown in a similar manner as the previous case. Again, without loss of generality, suppose $l = 1$. Let  $Z_j = (\vecX_{1,j}, \cC_{[2:k]})$, $j\in [1:N]$, where the codewords $\vecx_{\cQ^c\setminus\{1\}}$ corresponding to messages $m_{\cQ^c\setminus\{1\}}$ are fixed.  Let $g^{\vecx_{\cQ^c\setminus\{1\}}}_j(\vecZ_1, \ldots, \vecZ_j)$ be defined as
\begin{align*}
&g^{\vecx_{\cQ^c\setminus\{1\}}}_j(\vecZ_{1}, \ldots, \vecZ_{j}) = g^{\vecx_{\cQ^c\setminus\{1\}}}_j((\vecX_{1,1},\cC_{[2:k]}) \ldots, (\vecX_{1,j}, \cC_{[2:k]}))\\ 
&\hspace{2cm}= \begin{cases}1,&\text{ if }\exists  m'_{\cJ}\in \cM_{\cJ}\text{ such that }((\vecX_{1,j},\vecx_{\cQ^c\setminus\{1\}}),\vecX_{(\cJ,m'_{\cJ})}, \vecx_{\cQ})\in T^n_{X_{\cQ^c}X'_{\cJ}X_{\cQ}},\\
        0,&\text{ otherwise. }\end{cases}
\end{align*} 
For $t = \exp\left\{-n\epsilon/2\right\}$, \eqref{eq:last3} is
\begin{align*}
\bbP\Big(\frac{1}{N}\sum_{j=1}^Ng^{\vecx_{\cQ^c\setminus\{1\}}}_j(\vecZ_{1}, \ldots, \vecZ_{j})>t\Big).
\end{align*} Computing $a$,
\begin{align*}
&\bbE\insq{g^{\vecx_{\cQ^c\setminus\{1\}}}_j((\vecX_{1,1},\cC_{[2:k]}) \ldots, (\vecX_{1,j}, \cC_{[2:k]}))|(\vecX_{1,1},\cC_{[2:k]}) \ldots, (\vecX_{1,(j-1)}, \cC_{[2:k]})}\\
&\leq \sum_{m'_{\cJ}\in \cM_{\cJ}}\bbP\Big(((\vecX_{1,j},\vecx_{\cQ^c\setminus\{1\}}),\vecX_{(\cJ,m'_{\cJ})}, \vecx_{\cQ})\in T^n_{X_{\cQ^c}X'_{\cT}X'_{\cJ}X_{\cQ}}|(\vecX_{1,1},\cC_{[2:k]}) \ldots, (\vecX_{1,(j-1)}, \cC_{[2:k]})\Big)\\
&\leq N^{|\cJ|}\frac{\exp{\{nH(X_1|X_{\cQ^c\setminus\{1\}}X'_{\cJ}X_{\cQ})\}}}{(n+1)^{-|\cX_1|}\exp{\{nH(X_1)\}}}\\
&=(n+1)^{|\cX_1|}\exp{\inb{n\inp{(|\cJ|)(1/n)\log_2{N}-I(X_1;X_{\cQ^c\setminus\{1\}}X'_{\cJ}X_{\cQ})}}}
\end{align*}
Suppose $I(X_1;X_{\cQ^c\setminus\{1\}}X'_{\cJ}X_{\cQ})>(|\cJ|)(1/n)\log_2{N} +\epsilon$. Then, $a = (n+1)^{|\cX_1|}\exp{\inb{-n\epsilon}}$.
Thus,
\begin{align*}
&\inp{\frac{1}{N}\sum_{j= 1}^{N}g^{\vecx_{\cQ^c\setminus\{1\}}}_j(\vecZ_{1}, \ldots, \vecZ_{j})>t}\\
&\leq \exp\inb{-N\inp{t-a\log_2e}}\\
& = \exp\inb{-N\inp{\exp\inb{-n\epsilon/2}-(n+1)^{|\cX_1|}\exp\inb{-n\epsilon}}}\\
&\leq \exp\inb{-\inp{\exp\inb{n\epsilon/2}+(n+1)^{|\cX_1|}}}
\end{align*}
which falls doubly exponentially.
\end{proof}

\begin{lemma}\label{lemma:dec_disambiguity}
Suppose the Byzantine MAC $(W, \cA)$ is not symmetrizable. Let ${P}_{i}\in \cP_{\cX_i},\, i\in [1:k]$ be distributions such that ${P}_{i}(x_i)> 0,\,x_i\in \cX_i, \, i\in [1:k]$. Let $f_i:\mathcal{M}_i\rightarrow T^{n}_{X_i}, \,i\in [1:k]$ be some encoding maps. There exists a choice of $\eta>0$  such that if $(m_1,m_2,\ldots,m_k)\neq (\tilde{m}_1,\tilde{m}_2,\ldots, \tilde{m}_k)$, $\cD_{m_1,m_2,\ldots,m_k}\cap\cD_{\tilde{m}_1,\tilde{m}_2,\ldots,\tilde{m}_k} = \emptyset$.
\end{lemma} 
\begin{proof}
Suppose for $(\hat{m}_1,\hat{m}_2,\ldots,\hat{m}_k)\neq (\bar{m}_1,\bar{m}_2,\ldots, \bar{m}_k)$, there exists $\vecy\in \cD_{\hat{m}_1,\hat{m}_2,\ldots,\hat{m}_k}\cap \cD_{\bar{m}_1,\bar{m}_2,\ldots,\bar{m}_k}$. This implies that there exists $\gamma\in [1:k]$ such that $\vecy\in \dm{\gamma}{\hat{m}_{\gamma}}\cap\dm{\gamma}{\bar{m}_{\gamma}}$ for $\hat{m}_{\gamma}\neq\bar{m}_{\gamma}$. Then, by the decoder definition, there exist $\cQ,\tilde{\cQ}\in \cA$, not necessarily distinct, with $\gamma\notin\cQ, \tilde{\cQ}$; $\vecx_{\cQ}\in \cX^n_{\cQ}, \, \tilde{\vecx}_{\tilde{\cQ}}\in\cX^n_{\tilde{\cQ}}$; $m_{\cQ^c}\in \cM_{\cQ^c};\, \tilde{m}_{\tilde{\cQ}^c}\in \cM_{\tilde{\cQ}^c}$ with $m_{\gamma} = \hat{m}_{\gamma}$ and $\tilde{m}_{\gamma} = \bar{m}_{\gamma}$ such that 
for $\cT:=\inb{i\in(\cQ\cup\tilde{\cQ})^c:m_i\neq \tilde{m}_i}$ (note that $\gamma\in \cT$), and for the joint distributions $P_{{X}_{{\cT}} {X}_{\tilde{Q}\setminus{\cQ}}{X}'_{(\cT\cup\cQ\cup\tilde{\cQ})^c} {X}_{{\cQ}}Y}$ and $P_{\tilde{X}_{{\cT}} \tilde{X}_{{Q}\setminus\tilde{\cQ}}{X}'_{(\cT\cup\cQ\cup\tilde{\cQ})^c} \tilde{X}_{\tilde{\cQ}}Y}$ defined by \\
$(f_{{\cT}}({m}_{{\cT}}), f_{\tilde{\cQ}\setminus{\cQ}}({m}_{\tilde{\cQ}\setminus{\cQ}}),f_{(\cT\cup\cQ\cup\tilde{\cQ})^c}({m}_{(\cT\cup\cQ\cup\tilde{\cQ})^c}),{\vecx}_{{\cQ}}, \vecy) \in T^{n}_{{X}_{{\cT}} {X}_{\tilde{Q}\setminus{\cQ}}{X}'_{(\cT\cup\cQ\cup\tilde{\cQ})^c} {X}_{{\cQ}}Y}$
 and\\ $(f_{{\cT}}(\tilde{m}_{{\cT}}), f_{{\cQ}\setminus\tilde{\cQ}}(\tilde{m}_{{\cQ}\setminus\tilde{\cQ}}),f_{(\cT\cup\cQ\cup\tilde{\cQ})^c}({m}_{(\cT\cup\cQ\cup\tilde{\cQ})^c}),\tilde{\vecx}_{\tilde{\cQ}}, \vecy) \in T^{n}_{\tilde{X}_{{\cT}} \tilde{X}_{{Q}\setminus\tilde{\cQ}}{X}'_{(\cT\cup\cQ\cup\tilde{\cQ})^c} \tilde{X}_{\tilde{\cQ}}Y}$ respectively, the following holds.

\begin{align}
&D\inp{P_{{X}_{\cT} {X}_{\tilde{Q}\setminus{\cQ}}{X'}_{\inp{\cT\cup\cQ\cup\tilde{\cQ}}^c} {X}_{{\cQ}}Y}\Bigg|\Bigg|\inp{\prod_{i\in {\cT}}{P}_{{i}}}\inp{\prod_{j\in {\tilde{\cQ}\setminus{\cQ}}}P_{{j}}}\inp{\prod_{l\in \inp{\cT\cup\cQ\cup\tilde{\cQ}}^c}P_{{l}}}P_{{X}_{{\cQ}}}W}< \eta\label{disamb:1},\\
&D\inp{P_{\tilde{X}_{\cT} \tilde{X}_{{Q}\setminus\tilde{\cQ}}{X'}_{\inp{\cT\cup\cQ\cup\tilde{\cQ}}^c} \tilde{X}_{\tilde{\cQ}}Y}\Bigg|\Bigg|\inp{\prod_{i\in {\cT}}{P}_{{i}}}\inp{\prod_{j\in {{\cQ}\setminus\tilde{\cQ}}}P_{{j}}}\inp{\prod_{l\in \inp{\cT\cup\cQ\cup\tilde{\cQ}}^c}P_{{l}}}P_{\tilde{X}_{\tilde{\cQ}}}W}< \eta\label{disamb:2}.
\end{align}
 Then, the decoding condition~\ref{dec:2} implies that
\begin{align}
&I({X}_{\cT} {X}_{\tilde{\cQ}\setminus{\cQ}}{X'}_{(\cT\cup\cQ\cup\tilde{\cQ})^c} Y;\tilde{X}_{\cT} \tilde{X}_{{\cQ}\setminus \tilde{\cQ}}|X_{{\cQ}}) < \eta,\label{eq:dec_cond1}\\
&I(\tilde{X}_{\cT} \tilde{X}_{{\cQ}\setminus\tilde{\cQ}}{X'}_{(\cT\cup\cQ\cup\tilde{\cQ})^c} Y;X_{\cT} X_{\tilde{\cQ}\setminus \cQ}|\tilde{X}_{\tilde{\cQ}}) < \eta.\label{eq:dec_cond2}
\end{align}
For ease of notation, let $\cH:=(\cT\cup\cQ\cup\tilde{\cQ})^c$.
From \eqref{disamb:2}, by the chain rule of relative entropy, we get \\$D(P_{\tilde{X}_{\cT} \tilde{X}_{{Q}\setminus\tilde{\cQ}}}||(\prod_{i\in {\cT}}{P}_{{i}})(\prod_{j\in {{Q}\setminus\tilde{\cQ}}}P_{{j}}))< \eta$. Using this, \eqref{disamb:1} and \eqref{eq:dec_cond1}, we get

\begin{align*}
3&\eta>D(P_{{X}_{\cT} {X}_{\tilde{Q}\setminus{\cQ}}{X'}_{\cH} {X}_{{\cQ}}Y}||(\prod_{i\in {\cT}}{P}_{{i}})(\prod_{j\in {\tilde{\cQ}\setminus{\cQ}}}P_{{j}})(\prod_{l\in \cH}P_{{l}})P_{{X}_{{\cQ}}}W) + D(P_{\tilde{X}_{\cT}, \tilde{X}_{{Q}\setminus\tilde{\cQ}}}||(\prod_{i\in {\cT}}{P}_{{i}})(\prod_{j\in {{Q}\setminus\tilde{\cQ}}}P_{{j}}))\\
&\hspace{10cm} + I({X}_{\cT} {X}_{\tilde{\cQ}\setminus{\cQ}}{X'}_{\cH} Y;\tilde{X}_{\cT} \tilde{X}_{{\cQ}\setminus \tilde{\cQ}}|X_{{\cQ}})\\
&=D(P_{{X}_{\cT}{X}_{\tilde{Q}\setminus{\cQ}}{X'}_{\cH}{X}_{{\cQ}}Y}||(\prod_{i\in {\cT}}{P}_{{i}})(\prod_{j\in {\tilde{\cQ}\setminus{\cQ}}}P_{{j}})(\prod_{l\in \cH}P_{{l}})P_{{X}_{{\cQ}}}W) + D(P_{\tilde{X}_{\cT}, \tilde{X}_{{Q}\setminus\tilde{\cQ}}}||(\prod_{i\in {\cT}}{P}_{{i}})(\prod_{j\in {{Q}\setminus\tilde{\cQ}}}P_{{j}}))\\
&\hspace{7cm} + D({P_{{X}_{\cT} {X}_{\tilde{\cQ}\setminus{\cQ}}{X'}_{\cH}X_{\cQ} \tilde{X}_{\cT} \tilde{X}_{\cQ\setminus\tilde{\cQ}}Y}}||{P_{\tilde{X}_{\cT} \tilde{X}_{\cQ\setminus\tilde{\cQ}}}P_{X_{\cQ|\tilde{X}_{\cT} \tilde{X}_{\cQ\setminus\tilde{\cQ}}}}P_{{X}_{\cT} {X}_{\tilde{\cQ}\setminus{\cQ}}{X'}_{\cH}Y|X_{\cQ}}})\\
&\quad = D\inp{{P_{{X}_{\cT} {X}_{\tilde{\cQ}\setminus{\cQ}}{X'}_{\cH}X_{\cQ} \tilde{X}_{\cT} \tilde{X}_{\cQ\setminus\tilde{\cQ}}Y}}||(\prod_{i\in {\cT}}{P}_{{i}})(\prod_{j\in {\tilde{\cQ}\setminus{\cQ}}}P_{{j}})(\prod_{l\in \cH}P_{{l}})(\prod_{i\in {\cT}}{P}_{{i}})(\prod_{j\in {{Q}\setminus\tilde{\cQ}}}P_{{j}})P_{X_{\cQ|\tilde{X}_{\cT} \tilde{X}_{\cQ\setminus\tilde{\cQ}}}}W_{Y|X_{\cT} {X}_{\tilde{\cQ}\setminus{\cQ}}{X'}_{\cH}X_{\cQ}}}\\
&\quad \stackrel{\text{(b)}}{\geq} D\inp{{P_{{X}_{\cT} {X}_{\tilde{\cQ}\setminus{\cQ}}{X'}_{\cH} \tilde{X}_{\cT} \tilde{X}_{\cQ\setminus\tilde{\cQ}}Y}}||(\prod_{i\in {\cT}}{P}_{{i}})(\prod_{j\in {\tilde{\cQ}\setminus{\cQ}}}P_{{j}})(\prod_{l\in \cH}P_{{l}})(\prod_{i\in {\cT}}{P}_{{i}})(\prod_{j\in {{Q}\setminus\tilde{\cQ}}}P_{{j}})V_{Y|X_{\cT} {X}_{\tilde{\cQ}\setminus{\cQ}}{X'}_{\cH}\tilde{X}_{\cT} \tilde{X}_{\cQ\setminus\tilde{\cQ}}}}\\
&\qquad\text{ where }V_{Y|X_{\cT} {X}_{\tilde{\cQ}\setminus{\cQ}}{X'}_{\cH} \tilde{X}_{\cT} \tilde{X}_{\cQ\setminus\tilde{\cQ}}}(y|x_{\cT}, {x}_{\tilde{\cQ}\setminus{\cQ}},{x'}_{\cH},\tilde{x}_{\cT},\tilde{x}_{\cQ\setminus\tilde{\cQ}}) = \sum_{x_{\cQ}}P_{X_{\cQ|\tilde{X}_{\cT} \tilde{X}_{\cQ\setminus\tilde{\cQ}}}}(x_{\cQ}|\tilde{x}_{\cT},\tilde{x}_{\cQ\setminus\tilde{\cQ}})W(y|x_{\cT}, {x}_{\tilde{\cQ}\setminus{\cQ}},{x'}_{\cH},x_{\cQ}),
\end{align*}
and (b) follows from the chain rule of relative entropy. Using Pinsker's inequality, it follows that 
\begin{align}
&\sum_{x_{\cT}, {x}_{\tilde{\cQ}\setminus{\cQ}},{x'}_{\cH},\tilde{x}_{\cT}, \tilde{x}_{\cQ\setminus\tilde{\cQ}},y}\Big|P_{X_{\cT}{X}_{\tilde{\cQ}\setminus{\cQ}}{X'}_{\cH} \tilde{X}_{\cT} \tilde{X}_{\cQ\setminus\tilde{\cQ}}Y}({x_{\cT}, {x}_{\tilde{\cQ}\setminus{\cQ}},{x'}_{\cH}, \tilde{x}_{\cT}, \tilde{x}_{\cQ\setminus\tilde{\cQ}},y})\nonumber\\
&-(\prod_{i\in {\cT}}{P}_{{i}}(x_i))(\prod_{j\in {\tilde{\cQ}\setminus{\cQ}}}P_{{j}}(x_j))(\prod_{l\in \cH}P_{{l}}(x'_l))(\prod_{i\in {\cT}}{P}_{{l}}(\tilde{x_l}))(\prod_{j\in {{Q}\setminus\tilde{\cQ}}}P_{{j}}(\tilde{x}_j))V_{Y|X_{\cT} {X}_{\tilde{\cQ}\setminus{\cQ}}{X'}_{\cH} \tilde{X}_{\cT} \tilde{X}_{\cQ\setminus\tilde{\cQ}}}(y|x_{\cT}, {x}_{\tilde{\cQ}\setminus{\cQ}},{x'}_{\cH},\tilde{x}_{\cT},\tilde{x}_{\cQ\setminus\tilde{\cQ}})\Big|\nonumber\\
& \leq c\sqrt{3\eta},\label{eq:V}
\end{align}
where $c$ is some positive constant. By a symmetric analysis, we can show that 
\begin{align}
&\sum_{x_{\cT}, {x}_{\tilde{\cQ}\setminus{\cQ}},{x'}_{\cH},\tilde{x}_{\cT}, \tilde{x}_{\cQ\setminus\tilde{\cQ}},y}\Big|P_{X_{\cT}{X}_{\tilde{\cQ}\setminus{\cQ}}{X'}_{\cH} \tilde{X}_{\cT} \tilde{X}_{\cQ\setminus\tilde{\cQ}}Y}({x_{\cT}, {x}_{\tilde{\cQ}\setminus{\cQ}},{x'}_{\cH}, \tilde{x}_{\cT}, \tilde{x}_{\cQ\setminus\tilde{\cQ}},y})\nonumber\\
&-(\prod_{i\in {\cT}}{P}_{{i}}(x_i))(\prod_{j\in {\tilde{\cQ}\setminus{\cQ}}}P_{{j}}(x_j))(\prod_{l\in \cH}P_{{l}}(x'_l))(\prod_{i\in {\cT}}{P}_{{l}}(\tilde{x_l}))(\prod_{j\in {{Q}\setminus\tilde{\cQ}}}P_{{j}}(\tilde{x}_j))V'_{Y|X_{\cT} {X}_{\tilde{\cQ}\setminus{\cQ}}{X'}_{\cH} \tilde{X}_{\cT} \tilde{X}_{\cQ\setminus\tilde{\cQ}}}(y|x_{\cT}, {x}_{\tilde{\cQ}\setminus{\cQ}},{x'}_{\cH},\tilde{x}_{\cT},\tilde{x}_{\cQ\setminus\tilde{\cQ}})\Big|\nonumber\\
& \leq c\sqrt{3\eta},\label{eq:V'}
\end{align}
for 
\begin{align*}
V'_{Y|X_{\cT} {X}_{\tilde{\cQ}\setminus{\cQ}}{X'}_{\cH} \tilde{X}_{\cT} \tilde{X}_{\cQ\setminus\tilde{\cQ}}}(y|x_{\cT}, {x}_{\tilde{\cQ}\setminus{\cQ}},{x'}_{\cH},\tilde{x}_{\cT},\tilde{x}_{\cQ\setminus\tilde{\cQ}}) = \sum_{\tilde{x}_{\tilde{\cQ}}}P_{\tilde{X}_{\tilde{\cQ}|{X}_{\cT} {X}_{\tilde{\cQ}\setminus{\cQ}}}}(\tilde{x}_{\tilde{\cQ}}|{x}_{\cT},{x}_{\tilde{\cQ}\setminus{\cQ}})W(y|\tilde{x}_{\cT}, \tilde{x}_{{\cQ}\setminus\tilde{\cQ}},{x'}_{\cH},\tilde{x}_{\tilde{\cQ}}).
\end{align*}
By \eqref{eq:V} and \eqref{eq:V'}, 

\begin{align*}
&\sum_{x_{\cT}, x_{{\tilde{\cQ}\setminus{\cQ}}}, x'_{\cH}, \tilde{x}_{\cT},\tilde{x}_{{Q}\setminus\tilde{\cQ}},y}(\prod_{i\in {\cT}}{P}_{{i}}(x_i))(\prod_{j\in {\tilde{\cQ}\setminus{\cQ}}}P_{{j}}(x_j))(\prod_{l\in \cH}P_{{l}}(x'_l))(\prod_{i\in {\cT}}{P}_{{l}}(\tilde{x_l}))(\prod_{j\in {{Q}\setminus\tilde{\cQ}}}P_{{j}}(\tilde{x}_j))\\
&\Big|V_{Y|X_{\cT} {X}_{\tilde{\cQ}\setminus{\cQ}}{X'}_{\cH} \tilde{X}_{\cT} \tilde{X}_{\cQ\setminus\tilde{\cQ}}}(y|x_{\cT}, {x}_{\tilde{\cQ}\setminus{\cQ}},{x'}_{\cH},\tilde{x}_{\cT},\tilde{x}_{\cQ\setminus\tilde{\cQ}})-V'_{Y|X_{\cT} {X}_{\tilde{\cQ}\setminus{\cQ}}{X'}_{\cH} \tilde{X}_{\cT} \tilde{X}_{\cQ\setminus\tilde{\cQ}}}(y|x_{\cT}, {x}_{\tilde{\cQ}\setminus{\cQ}},{x'}_{\cH},\tilde{x}_{\cT},\tilde{x}_{\cQ\setminus\tilde{\cQ}})\Big| \leq 2c\sqrt{3\eta}.
\end{align*}
This implies that for $\alpha:=\min_{i\in[1:k]}\min_{x_i}{P}_{i}(x_i)$ (note that $\alpha>0$),
\begin{align}
&\max_{x_{\cT}, x_{{\tilde{\cQ}\setminus{\cQ}}}, x'_{\cH}, \tilde{x}_{\cT},\tilde{x}_{{Q}\setminus\tilde{\cQ}}, y}\Big|V_{Y|X_{\cT} {X}_{\tilde{\cQ}\setminus{\cQ}}{X'}_{\cH} \tilde{X}_{\cT} \tilde{X}_{\cQ\setminus\tilde{\cQ}}}(y|x_{\cT}, {x}_{\tilde{\cQ}\setminus{\cQ}},{x'}_{\cH},\tilde{x}_{\cT},\tilde{x}_{\cQ\setminus\tilde{\cQ}})\nonumber\\
&\hspace{4cm}-V'_{Y|X_{\cT} {X}_{\tilde{\cQ}\setminus{\cQ}}{X'}_{\cH} \tilde{X}_{\cT} \tilde{X}_{\cQ\setminus\tilde{\cQ}}}(y|x_{\cT}, {x}_{\tilde{\cQ}\setminus{\cQ}},{x'}_{\cH},\tilde{x}_{\cT},\tilde{x}_{\cQ\setminus\tilde{\cQ}})\Big| \leq \frac{2c\sqrt{3\eta}}{\alpha^j}\label{eq:final}
\end{align}
for some integer $j$.
Since $(W, \cA)$ is not symmetrizable, there exist $\zeta>0$ such that
\begin{align*}
&\max_{x_{\cT}, x_{{\tilde{\cQ}\setminus{\cQ}}}, x'_{\cH}, \tilde{x}_{\cT},\tilde{x}_{{Q}\setminus\tilde{\cQ}},y}\Big|\sum_{x_{\cQ}}P_{X_{\cQ}|\tilde{X}_{\cT} \tilde{X}_{\cQ\setminus\tilde{\cQ}}}(x_{\cQ}|\tilde{x}_{\cT},\tilde{x}_{\cQ\setminus\tilde{\cQ}})W(y|x_{\cT}, {x}_{\tilde{\cQ}\setminus{\cQ}},{x'}_{\cH},x_{\cQ})\nonumber\\
 &\hspace{3cm}-\sum_{\tilde{x}_{\tilde{\cQ}}}P_{\tilde{X}_{\tilde{\cQ}}|{X}_{\cT} {X}_{\tilde{\cQ}\setminus{\cQ}}}(\tilde{x}_{\tilde{\cQ}}|{x}_{\cT},{x}_{\tilde{\cQ}\setminus{\cQ}})W(y|\tilde{x}_{\cT}, \tilde{x}_{{\cQ}\setminus\tilde{\cQ}},{x'}_{\cH},\tilde{x}_{\tilde{\cQ}})\Big|> \zeta.
\end{align*}
That is
\begin{align*}
&\max_{x_{\cT}, x_{{\tilde{\cQ}\setminus{\cQ}}}, x'_{\cH}, \tilde{x}_{\cT},\tilde{x}_{{Q}\setminus\tilde{\cQ}}, y}\Big|V_{Y|X_{\cT} {X}_{\tilde{\cQ}\setminus{\cQ}}{X'}_{\cH} \tilde{X}_{\cT} \tilde{X}_{\cQ\setminus\tilde{\cQ}}}(y|x_{\cT}, {x}_{\tilde{\cQ}\setminus{\cQ}},{x'}_{\cH},\tilde{x}_{\cT},\tilde{x}_{\cQ\setminus\tilde{\cQ}})\\
&\hspace{4cm}-V'_{Y|X_{\cT} {X}_{\tilde{\cQ}\setminus{\cQ}}{X'}_{\cH} \tilde{X}_{\cT} \tilde{X}_{\cQ\setminus\tilde{\cQ}}}(y|x_{\cT}, {x}_{\tilde{\cQ}\setminus{\cQ}},{x'}_{\cH},\tilde{x}_{\cT},\tilde{x}_{\cQ\setminus\tilde{\cQ}})\Big| > \zeta.
\end{align*}
This contradicts \eqref{eq:final} for choice of $\eta$ and $\alpha$ satisfying $\frac{2c\sqrt{3\eta}}{\alpha^j}<\zeta$.
\end{proof}

\section{Proof of  the converse of Theorem~\ref{achievability_k_users_final}}\label{sec:empty_interior}
\begin{proof}
Suppose the given \bmac $(W, \cA)$ is symmetrizable. Then, there exist $\cT\subseteq [1:k]$, $\cQ, \cQ'\in  \cA$, not necessarily distinct, satisfying $\cQ\cap \cT = \cQ'\cap\cT = \emptyset$,  and  a pair of conditional distributions $P_{X_{\cQ}|X_{\cT\cup (\cQ\setminus \cQ')}}$ and $P'_{X_{\cQ'}|X_{\cT\cup (\cQ'\setminus \cQ)}}$ satisfying \eqref{eq:converse_k_user_0} below:
\begin{align}
&\sum_{x'_{\cQ}\in \cX_{\cQ}}P_{X_{\cQ}|X_{\cT\cup (\cQ\setminus \cQ')}}(x'_{\cQ}|x_{\cT}, x_{\cQ\setminus \cQ'})W(y|x'_{\cQ}, \tilde{x}_{\cT}, x_{\cQ'\setminus \cQ}, x_{(\cT\cup \cQ \cup \cQ')^c})\nonumber\\
 &= \sum_{\tilde{x}_{\cQ'}\in \cX_{\cQ'}}P'_{X_{\cQ'}|X_{\cT\cup (\cQ'\setminus \cQ)}}(\tilde{x}_{\cQ'}|\tilde{x}_{\cT}, x_{\cQ'\setminus \cQ})W(y|\tilde{x}_{\cQ'}, {x}_{\cT}, x_{\cQ\setminus \cQ'}, x_{(\cT\cup \cQ \cup \cQ')^c})\label{eq:converse_k_user_0}
\end{align} for every $y,  {x}_{\cT}, x_{\cQ\setminus \cQ'}, x_{(\cT\cup \cQ \cup \cQ')^c}, \tilde{x}_{\cT}$ and  $x_{\cQ'\setminus \cQ}$. 
Let $m_{\cT},\tilde{m}_{\cT}\in\mathcal{M}_3$ be such that $m_{i}\neq \tilde{m}_{i}$ for all $i\in \cT$. We consider two different scenarios in which users in set $\cT$ send $f_{\cT}(m_{\cT})$ and $f_{\cT}(\tilde{m}_{\cT})$ respectively:
\begin{enumerate}
\item[(i)] In the first setting, users in the set $\cQ$ are adversarial. They choose a message tuple $M_{\cQ\setminus\cQ'}\sim\textup{Unif}(\mathcal{M}_{\cQ\setminus\cQ'})$. Let $\vecX_{\cQ\setminus\cQ'}=f_{\cQ\setminus\cQ'}(M_{\cQ\setminus\cQ'})$. To produce their input ${\vecX'}_{\cQ, {m}_{\cT}}$ to the channel, they pass $(f_{\cT}(m_{\cT}),\vecX_{\cQ\setminus\cQ'})$ through $P^n_{X_{\cQ}|X_{\cT\cup (\cQ\setminus \cQ')}}$, the $n$-fold product of the channel $P_{X_{\cQ}|X_{\cT\cup (\cQ\setminus \cQ')}}$. Users in the set $(\cT\cup\cQ)^c$, being non-adversarial, send as their input to the channel $\vecX_{(\cT\cup\cQ)^c}=f_{(\cT\cup\cQ)^c}(M_{(\cT\cup\cQ)^c})$, where $M_{(\cT\cup\cQ)^c}\sim\textup{Unif}(\mathcal{M}_{(\cT\cup\cQ)^c})$. Users in the set $\cT$ send $f_{\cT}(\tilde{m}_{\cT})$. The probability of any vector $\vecy$ under this scenario is
\begin{align}
&\sum_{\vecx'_{\cQ}\in \cX^n_{\cQ}}\inp{\frac{1}{|\cM_{\cQ\setminus\cQ'}|}\sum_{m_{\cQ\setminus \cQ'}\in \cM_{\cQ\setminus \cQ'}}P^n_{X_{\cQ}|X_{\cT\cup (\cQ\setminus \cQ')}}(\vecx'_{\cQ}|f_{\cT}(m_{\cT}), f_{\cQ\setminus \cQ'}(m_{\cQ\setminus \cQ'}))}\nonumber\\
&\hspace{5cm}\frac{1}{|\cM_{(\cT\cup\cQ)^c}|}\sum_{m_{(\cT\cup\cQ)^c}\in \cM_{(\cT\cup\cQ)^c}}W^n\inp{\vecy|\vecx'_{\cQ}, f_{\cT}\inp{\tilde{m}_{\cT}}, f_{(\cT\cup \cQ)^c}\inp{m_{(\cT\cup\cQ)^c}}}\nonumber\\
&=\frac{1}{|\cM_{\cQ\setminus\cQ'}|\times|\cM_{(\cT\cup\cQ)^c}|}\sum_{m_{\cQ\setminus \cQ'}\in \cM_{\cQ\setminus \cQ'}}\sum_{m_{(\cT\cup\cQ)^c}\in \cM_{(\cT\cup\cQ)^c}}\prod_{t=1}^n\inp{\sum_{\vecx'_{\cQ,t}\in \cX_{\cQ}}   P_{X_{\cQ}|X_{\cT\cup (\cQ\setminus \cQ')}}(\vecx'_{\cQ,t}|f_{\cT, t}(m_{\cT}), f_{\cQ\setminus \cQ',t}(m_{\cQ\setminus \cQ'}))}\nonumber\\
&\hspace{8cm}W\inp{\vecy_t|\vecx'_{\cQ,t}, f_{\cT,t}\inp{\tilde{m}_{\cT}}, f_{(\cT\cup \cQ)^c,t}\inp{m_{(\cT\cup\cQ)^c}}}\nonumber\\
&=\frac{1}{|\cM_{(\cT\cup(\cQ\cap\cQ'))^c}|}\sum_{m_{(\cT\cup(\cQ\cap\cQ'))^c}\in \cM_{(\cT\cup(\cQ\cap\cQ'))^c}}\prod_{t=1}^n\inp{\sum_{\vecx'_{\cQ,t}\in \cX_{\cQ}}   P_{X_{\cQ}|X_{\cT\cup (\cQ\setminus \cQ')}}(\vecx'_{\cQ,t}|f_{\cT, t}(m_{\cT}), f_{\cQ\setminus \cQ',t}(m_{\cQ\setminus \cQ'}))}\nonumber\\
&\hspace{8cm}W\inp{\vecy_t|\vecx'_{\cQ,t}, f_{\cT,t}\inp{\tilde{m}_{\cT}}, f_{(\cT\cup \cQ)^c,t}\inp{m_{(\cT\cup\cQ)^c}}}\nonumber\\
&= \bbE_{\vecX'_{\cQ, {m}_{\cT}}}\insq{e_{\cQ, \cT}(\vecy,\tilde{m}_{\cT}, \vecX'_{\cQ, {m}_{\cT}})}.\label{eq:converse_k_user_1}
 \end{align}
 where $e_{\cQ, \cT}(\bar{\vecy},\bar{m}_{\cT}, \vecx_{\cQ})$ denotes $\frac{1}{|\cM_{(\cT\cup\cQ)^c}|}\sum_{m_{(\cT\cup\cQ)^c}\in \cM_{(\cT\cup\cQ)^c}}W^n\inp{\bar{\vecy}|\vecx_{\cQ}, f_{\cT}\inp{\bar{m}_{\cT}}, f_{(\cT\cup \cQ)^c}\inp{m_{(\cT\cup\cQ)^c}}}$ for $\bar{\vecy}$, $\bar{m}_{\cT}\in \cM_{\cT}$ and $\vecx_{\cQ}\in \cX^n_{\cQ}$.
The notation $\vecy_t$ represents the $t^{\text{th}}$ component of the vector $\vecy$ and for any set $\cS$ and message tuple $m_{\cS}\in \cM_{\cS}$, $f_{\cS, t}(m_{\cS})$ and $\vecx_{\cS, t}$ represents the $|\cS|$-length tuple containing the $t^{\text{th}}$ components of the vectors in $f_{\cS}(m_{\cS})$ and $\vecx_{\cS}$ respectively.

\item[(ii)] In the second setting, users in the set $\cQ'$ are adversarial. They choose a message tuple $M_{\cQ'\setminus\cQ}\sim\textup{Unif}(\mathcal{M}_{\cQ'\setminus\cQ})$. Let $\vecX_{\cQ'\setminus\cQ}=f_{\cQ'\setminus\cQ}(M_{\cQ'\setminus\cQ})$. To produce their input $\tilde{\vecX}_{\cQ', \tilde{m}_{\cT}}$ to the channel, they pass $(f_{\cT}(\tilde{m}_{\cT}),\vecX_{\cQ\setminus\cQ'})$ through $P'^{n}_{X_{\cQ'}|X_{\cT\cup (\cQ'\setminus \cQ)}}$, the $n$-fold product of the channel $P'_{X_{\cQ'}|X_{\cT\cup (\cQ'\setminus \cQ)}}$. Users in the set $(\cT\cup\cQ')^c$, being non-adversarial, send  $\vecX_{(\cT\cup\cQ')^c}=f_{(\cT\cup\cQ')^c}(M_{(\cT\cup\cQ')^c})$ as their input to the channel, where $M_{(\cT\cup\cQ')^c}\sim\textup{Unif}(\mathcal{M}_{(\cT\cup\cQ')^c})$. Users in the set $\cT$ send $f_{\cT}({m}_{\cT})$. The probability of any vector $\vecy$ under this scenario is
\begin{align}
&\sum_{\tilde{\vecx}_{\cQ'}\in \cX^n_{\cQ'}}\inp{\frac{1}{|\cM_{\cQ'\setminus\cQ}|}\sum_{m_{\cQ'\setminus \cQ}\in \cM_{\cQ'\setminus \cQ}}P'^n_{X_{\cQ'}|X_{\cT\cup (\cQ'\setminus \cQ)}}(\tilde{\vecx}_{\cQ'}|f_{\cT}(\tilde{m}_{\cT}), f_{\cQ'\setminus \cQ}(m_{\cQ'\setminus \cQ}))}\nonumber\\
&\hspace{6cm}\frac{1}{|\cM_{(\cT\cup\cQ')^c}|}\sum_{m_{(\cT\cup\cQ')^c}\in \cM_{(\cT\cup\cQ')^c}}W^n\inp{\vecy|\tilde{\vecx}_{\cQ'}, f_{\cT}\inp{{m}_{\cT}}, f_{(\cT\cup \cQ')^c}\inp{m_{(\cT\cup\cQ')^c}}}\nonumber\\
&=\frac{1}{|\cM_{\cQ'\setminus\cQ}|\times|\cM_{(\cT\cup\cQ')^c}|}\sum_{m_{\cQ'\setminus \cQ}\in \cM_{\cQ'\setminus \cQ}}\sum_{m_{(\cT\cup\cQ')^c}\in \cM_{(\cT\cup\cQ')^c}}\prod_{t=1}^n\sum_{\tilde{x}_{\cQ',t}\in \cX_{\cQ'}}   P'_{X_{\cQ'}|X_{\cT\cup (\cQ'\setminus \cQ)}}(\tilde{x}_{\cQ',t}|f_{\cT, t}(\tilde{m}_{\cT}), f_{\cQ'\setminus \cQ,t}(m_{\cQ'\setminus \cQ}))\nonumber\\
&\hspace{8cm}W\inp{\vecy_t|\tilde{x}_{\cQ',t}, f_{\cT,t}\inp{{m}_{\cT}}, f_{(\cT\cup \cQ')^c,t}\inp{m_{(\cT\cup\cQ')^c}}}\nonumber\\
&=\frac{1}{|\cM_{(\cT\cup(\cQ\cap\cQ'))^c}|}\sum_{m_{(\cT\cup(\cQ\cap\cQ'))^c}\in \cM_{(\cT\cup(\cQ\cap\cQ'))^c}}\prod_{t=1}^n\sum_{\tilde{x}_{\cQ',t}\in \cX_{\cQ'}}   P'_{X_{\cQ'}|X_{\cT\cup (\cQ'\setminus \cQ)}}(\tilde{x}_{\cQ',t}|f_{\cT, t}(\tilde{m}_{\cT}), f_{\cQ'\setminus \cQ,t}(m_{\cQ'\setminus \cQ}))\nonumber\\
&\hspace{8cm}W\inp{\vecy_t|\tilde{x}_{\cQ',t}, f_{\cT,t}\inp{{m}_{\cT}}, f_{(\cT\cup \cQ')^c,t}\inp{m_{(\cT\cup\cQ')^c}}}\nonumber\\
&= \bbE_{\tilde{\vecX}_{\cQ', \tilde{m}_{\cT}}}\insq{e_{\cQ', \cT}(\vecy, m_{\cT},\tilde{\vecX}_{\cQ', \tilde{m}_{\cT}})}\label{eq:converse_k_user_2}
 \end{align}
where $e_{\cQ', \cT}(\bar{\vecy},\bar{m}_{\cT}, \vecx_{\cQ'})$ denotes $\frac{1}{|\cM_{(\cT\cup\cQ')^c}|}\sum_{m_{(\cT\cup\cQ')^c}\in \cM_{(\cT\cup\cQ')^c}}W^n\inp{\bar{\vecy}|\vecx_{\cQ'}, f_{\cT}\inp{\bar{m}_{\cT}}, f_{(\cT\cup \cQ')^c}\inp{m_{(\cT\cup\cQ')^c}}}$ for $\bar{\vecy}\in \cY^n$, $\bar{m}_{\cT}\in \cM_{\cT}$ and $\vecx_{\cQ'}\in \cX^n_{\cQ'}$.
\end{enumerate}
Note that
\begin{align}
&\bbE_{\tilde{\vecX}_{\cQ', \tilde{m}_{\cT}}}\insq{e_{\cQ', \cT}(\vecy, m_{\cT},\tilde{\vecX}_{\cQ', \tilde{m}_{\cT}})}\\
&=\frac{1}{|\cM_{(\cT\cup(\cQ\cap\cQ'))^c}|}\sum_{m_{(\cT\cup(\cQ\cap\cQ'))^c}\in \cM_{(\cT\cup(\cQ\cap\cQ'))^c}}\prod_{t=1}^n\sum_{\tilde{x}_{\cQ',t}\in \cX_{\cQ'}}   P'_{X_{\cQ'}|X_{\cT\cup (\cQ'\setminus \cQ)}}(\tilde{x}_{\cQ',t}|f_{\cT, t}(\tilde{m}_{\cT}), f_{\cQ'\setminus \cQ,t}(m_{\cQ'\setminus \cQ}))\nonumber\\
&\hspace{8cm}W^n\inp{\vecy_t|\tilde{x}_{\cQ',t}, f_{\cT,t}\inp{{m}_{\cT}}, f_{(\cT\cup \cQ')^c,t}\inp{m_{(\cT\cup\cQ')^c}}}\nonumber\\
&\stackrel{(a)}{=}\frac{1}{|\cM_{(\cT\cup(\cQ\cap\cQ'))^c}|}\sum_{m_{(\cT\cup(\cQ\cap\cQ'))^c}\in \cM_{(\cT\cup(\cQ\cap\cQ'))^c}}\prod_{t=1}^n\sum_{{x'}_{\cQ,t}\in \cX_{\cQ}}   P_{X_{\cQ}|X_{\cT\cup (\cQ\setminus \cQ')}}({x'}_{\cQ,t}|f_{\cT, t}({m}_{\cT}), f_{\cQ\setminus \cQ',t}(m_{\cQ\setminus \cQ'}))\nonumber\\
&\hspace{8cm}W^n\inp{\vecy_t|{x'}_{\cQ,t}, f_{\cT,t}\inp{\tilde{m}_{\cT}}, f_{(\cT\cup \cQ)^c,t}\inp{m_{(\cT\cup\cQ)^c}}}\nonumber\\
&=\bbE_{\vecX'_{\cQ, {m}_{\cT}}}\insq{e_{\cQ, \cT}(\vecy,\tilde{m}_{\cT}, \vecX'_{\cQ, {m}_{\cT}})}\label{eq:why_1}
\end{align} where $(a)$ follows from \eqref{eq:converse_k_user_0}.

Arguing along the lines of~\cite[(3.29) in page~187]{CsiszarN88},\shortonly{ we can show that the average probability of error is at least 1/8. See the extended version~\cite{ExtendedDraft} for more details.}

\longonly{
\begin{align*}
2P_{e}&\geq \frac{1}{|\cM_{\cT}|}\sum_{\tilde{m}_{\cT}}\sum_{\vecy:\phi(\vecy)_{\cT}\neq \tilde{m}_{\cT}}\bbE_{\vecX'_{\cQ}}\insq{e_{\cQ, \cT}(\vecy,\tilde{m}_{\cT}, \vecX'_{\cQ})} + \frac{1}{|\cM_{\cT}|}\sum_{{m}_{\cT}}\sum_{\vecy:\phi(\vecy)_{\cT}\neq {m}_{\cT}}\bbE_{\tilde{\vecX}_{\cQ'}}\insq{e_{\cQ', \cT}(\vecy, m_{\cT},\tilde{\vecX}_{\cQ'})}
\end{align*}
for any attack vectors $\vecX'_{\cQ}$ and $\tilde{\vecX}_{\cQ'}$. In particular, for the attack vectors $\frac{1}{|\cM_{\cT}|}\sum_{\tilde{m}_{\cT}}\tilde{\vecX}_{\cQ', \tilde{m}_{\cT}}$ and $\frac{1}{|\cM_{\cT}|}\sum_{{m}_{\cT}}\vecX'_{\cQ, {m}_{\cT}}$,
\begin{align*}
2P_{e}&\geq \frac{1}{|\cM_{\cT}|^2}\sum_{\tilde{m}_{\cT}}\sum_{{m}_{\cT}}\inp{\sum_{\vecy:\phi(\vecy)_{\cT}\neq \tilde{m}_{\cT}}\bbE_{\vecX'_{\cQ}}\insq{e_{\cQ, \cT}(\vecy,\tilde{m}_{\cT}, \vecX'_{\cQ, m_{\cT}})} + \sum_{\vecy:\phi(\vecy)_{\cT}\neq {m}_{\cT}}\bbE_{\tilde{\vecX}_{\cQ'}}\insq{e_{\cQ', \cT}(\vecy, m_{\cT},\tilde{\vecX}_{\cQ',\tilde{m}_{\cT}})}}\\
&\stackrel{(a)}{=}\frac{1}{|\cM_{\cT}|^2}\sum_{\tilde{m}_{\cT}}\sum_{{m}_{\cT}}\inp{\sum_{\vecy:\phi(\vecy)_{\cT}\neq \tilde{m}_{\cT}}\bbE_{\vecX'_{\cQ}}\insq{e_{\cQ, \cT}(\vecy,\tilde{m}_{\cT}, \vecX'_{\cQ, m_{\cT}})} + \sum_{\vecy:\phi(\vecy)_{\cT}\neq {m}_{\cT}}\bbE_{\vecX'_{\cQ}}\insq{e_{\cQ, \cT}(\vecy,\tilde{m}_{\cT}, \vecX'_{\cQ, m_{\cT}})}}
\end{align*} where $(a)$ follows from \eqref{eq:why_1}.
For $m_{\cT}\neq \tilde{m}_{\cT}$, the term in brackets on the right is upper bounded by 1, otherwise it is upper bounded by zero. Thus, 
\begin{align*}
P_{e}&\geq \frac{|\cM_{\cT}|(|\cM_{\cT}|-1)/2}{2|\cM_{\cT}|^2}\geq \frac{1}{8}.
\end{align*}}
This completes the proof of the converse.
\end{proof}

	\end{document}